%% file: main.tex
\documentclass[a4paper,UKenglish,cleveref, autoref, thm-restate]{lipics-v2021}
\RequirePackage[normalem]{ulem} 
\pdfoutput=1
\hideLIPIcs

\title{Branching Bisimilarity for Processes with Time-outs}

\author{Gaspard Reghem}{ENS Paris-Saclay, Université Paris-Saclay, France}{gaspard.reghem@ens-paris-saclay.fr}{}{}

\author{Rob J. van Glabbeek}{School of Informatics, University of Edinburgh, UK \and School of Computer Science and Engineering, University of New South Wales, Sydney, Australia \and \url{http://theory.stanford.edu/~rvg/}}{rvg@cs.stanford.edu}{https://orcid.org/0000-0003-4712-7423}{}

\authorrunning{G. Reghem and R.\,J. van Glabbeek}

\Copyright{Gaspard Reghem and Robert J. van Glabbeek}



\relatedversiondetails[linktext={doi:~10.4230/LIPIcs.CONCUR.2024.36}, cite=RG24a]{Extended abstract in CONCUR'24}{https://doi.org/10.4230/LIPIcs.CONCUR.2024.36}

\funding{Supported by Royal Society Wolfson Fellowship RSWF{\textbackslash}R1{\textbackslash}221008}

\nolinenumbers
\usepackage{macros}
\usepackage{stmaryrd}
\usepackage{multicol}
\usepackage{bussproofs}
\usepackage{tikz}

\begin{document}

\maketitle

\begin{abstract}
This paper provides an adaptation of branching bisimilarity to reactive systems with time-outs. Multiple equivalent definitions are procured, along with a modal characterisation and a proof of its congruence property for a standard process algebra with recursion. The last section presents a complete axiomatisation for guarded processes without infinite sequences of unobservable actions.
\end{abstract}

\section{Introduction} \label{sec:intro}

\emph{Strong bisimilarity} \cite{Mi90ccs} is the default semantic equivalence on labelled transition systems (LTSs), modelling systems that move from state to state by performing discrete, uninterpreted actions. In \cite{strongreactivebisimilarity}, it has been generalised, under the name \emph{strong reactive bisimilarity},
to LTSs that feature, besides the hidden action $\tau$ \cite{Mi90ccs}, an unobservable \emph{time-out} action $\rt$ \cite{vG21}, modelling the end of a time-consuming activity from which we abstract. This addition significantly increases the expressiveness of the model \cite{vG23a,strongreactivebisimilarity}.

Applied to the verification of realistic distributed systems, strong bisimilarity is too fine an equivalence, especially because it does not cater to abstraction from internal activity. \emph{Branching bisimilarity} \cite{branching} is a variant that does abstract from internal activity, and lies at the basis of many verification toolsets \cite{BGKLNVWWW19,GLMS11}. The present paper generalises branching bisimilarity to LTSs with time-outs, thereby combining the virtues of \cite{strongreactivebisimilarity} and \cite{branching}. It supports the resulting notion of \emph{branching reactive bisimilarity} through a modal characterisation, congruence results for a standard process algebra with recursion, and a complete axiomatisation.

The addition of the time-out action $\rt$ aims at modelling the passage of time while staying in the realm of \emph{untimed} process algebra. Here, ``untimed'' means that our framework does not facilitate measuring time, even though it models whether a system can pause in some state or not. We assume that the execution of any action is instantaneous; thus, time elapses in states only. The amount of time spent in a state is dictated by the interaction of the system with an external entity called its \emph{environment}.

We call a system \emph{reactive} if it interacts with an environment able to allow or disallow visible actions. The environment represents a user or other systems, running in parallel, which has no control over $\tau$ or $\rt$ actions. If $X$ is the set of visible actions currently allowed by the environment and the system can perform any transition labelled by an element of $X \cup \{\tau\}$ then it will perform one of those transitions immediately. When a visible action is performed, it triggers the environment to choose a new set of allowed actions. If the environment is allowing $X$ and the system cannot perform any transition labelled by $\tau$ or any allowed action, then the system is said to be \emph{idling}. When the system idles, time-outs become executable, but the environment can also get impatient and choose a new $X$ before any time-out occurs. 
\advance\textheight 3pt

\advance\textheight -3pt
We have supposed that the environment cannot synchronise with the execution of a time-out, thus implying that, right after executing a time-out, the environment is still allowing the same set of allowed actions as before this execution. For example, the process $a.P + \rt.(a.Q + \tau.R)$ will never reach $Q$ because, for the time-out to happen, the environment has to block $a$ and so $a.Q + \tau.R$ can only be reached when the environment blocks $a$. In this case, the $\tau$-transition is always executed before the environment can allow $a$ again. 

Similarly, strong and branching reactive bisimilarity satisfy the process algebraic law $\tau.P + \rt.Q = \tau.P$, essentially giving $\tau$ priority over $\rt$. Whereas this could have been formalised through an operational semantics in which the process $\tau.P + \rt.Q$ lacks an outgoing $\rt$-transition, here, and in \cite{strongreactivebisimilarity}, we derive an LTS for a standard process algebra with time-outs in a way that treats $\rt$ just like any other action. Instead, the priority of $\tau$ over $\rt$ is implemented in the reactive bisimilarity: its says that even though the transition
$\tau.P + \rt.Q \step\rt Q$ is present in our LTS, it will never be taken. This approach is not only simpler, it also generalises better to choices like $b.P + \rt.Q$, where the priority of $b$ over $\rt$ is conditional on the environment in which the system is placed, namely on whether or not this environment allows the $b$-action to occur.

From the system's perspective, the environment can be in two kinds of states: either allowing a specific set of actions, or being triggered to change. Our model does not stipulate how much time the environment takes to choose a new set of allowed actions once triggered, or even if it will ever make such a choice. Thus, the system could perform some transitions while the environment is triggered, especially those labelled $\tau$. In our view, the most natural way to see the environment is as another system executed in parallel, while enforcing synchronisation on all visible actions. This implies that the environment allows a set $X$ of actions when it idles in a state whose set of initial actions is $X$, and the environment is triggered when it is not idling, especially when it can perform a $\tau$-transition. In this paradigm, while the environment is triggered, any action can be allowed for a brief amount of time. However, there is no reason to believe that it will necessarily settle down on a specific set. For instance, this can happen if the environment reaches a \emph{divergence}: an infinite sequence of $\tau$-transitions.

In \cite{vG93}, seven (or nine) forms of branching bisimilarity are classified; they differ only in the treatment of divergence. In the present paper we are chiefly interested in divergence-free processes, on grounds that in the intuition of \cite{strongreactivebisimilarity} any sequence of $\tau$-transitions could be executed in time zero; yet we do wish to allow infinite sequences of $\rt$-transitions.
For divergence-free process all these forms of branching bisimilarity coincide.
Nevertheless, we do not formally exclude divergences, and in their presence our branching reactive bisimilarity generalises the \emph{stability respecting branching bisimilarity} of \cite{vG93}, which differs from the default version from \cite{branching} through the presence of Clause 2.e of Definition~\ref{def:intuitive}. There does not exist a plausible reactive generalisation of the default version.

Section \ref{sec:brb} supplies the formal definition of branching reactive bisimilarity as well as its rooted version, which will be shown to be its congruence closure. It also provides equivalent definitions that reduce our bisimilarity to a non-reactive one and illustrate that branching reactive bisimilarity coincides with stability respecting branching bisimilarity in the absence of time-outs.

Section \ref{sec:modal} gives a modal characterisation of branching reactive bisimilarity and its rooted version on an extension of the Hennessy-Milner logic. Section \ref{sec:congruence} introduces the process algebra $\ccsp$ along with an alternative characterisation of branching reactive bisimilarity that will be used to prove that rooted branching reactive bisimilarity is a full congruence for $\ccsp$.

Section \ref{sec:axiom} displays a complete axiomatisation of our bisimilarity on different fragments of $\ccsp$. Most completeness proofs rely on standard techniques like equation merging, but the very last one uses a relatively new method called ``canonical representatives''. 

\section{Branching Reactive Bisimilarity}\label{sec:brb}

A \emph{labelled transition system} (LTS) is a triple $(\closed,Act,\rightarrow)$ with $\closed$ a set (of \emph{states} or \emph{processes}), $Act$ a set (of \emph{actions}) and ${\rightarrow}\in\closed\times Act\times\closed$. In this paper we consider LTSs with $Act:= A\uplus\{\tau,\rt\}$, where $A$ is a set of \emph{visible actions}, $\tau$ is the \emph{hidden or invisible action}, and $\rt$ the \emph{time-out action}. Let $A_\tau := A \cup \{\tau\}$. $P \step{\alpha} P'$ stands for $(P,\alpha,P') \in {\rightarrow}$ and these triplets are called \emph{transitions}. Moreover, $P \step{\opt{\alpha}} P'$ denotes that either $\alpha = \tau$ and $P = P'$, or $P \step{\alpha} P'$. Furthermore, \emph{paths} are sequences of connected transitions and $\pathtau$ is the reflexive-transitive closure of $\steptau$. The set of \emph{initial} actions of a process $P \in\closed$ is $\init{P}:=\{\alpha\in A_\tau \mid P{\step \alpha}\}$. Here $P{\step \alpha}$ means that there is a $Q$ with $P \step\alpha Q$.

\begin{definition}\rm \label{def:intuitive}
    A \emph{\tb reactive bisimulation} is a symmetric\footnote{meaning that $(P,Q)\in \R \Leftrightarrow (Q,P)\in \R$ and $(P,X,Q)\in \R \Leftrightarrow (Q,X,P)\in \R$} relation $\R \subseteq (\closed\times\closed) \cup (\closed\times\Pow(A)\times\closed)$ such that, for all $P,Q \in \closed$ and $X \subseteq A$,
    \begin{enumerate}
        \item if $\R(P,Q)$ then
        \begin{enumerate}
            \item if $P \step{\alpha} P'$ with $\alpha \in A_\tau$ then there is a path $Q \pathtau Q_1 \step{\opt{\alpha}} Q_2$ with $\R(P,Q_1)$ and $\R(P',Q_2)$,
            \item for all $Y \subseteq A$, $\R(P,Y,Q)$;
        \end{enumerate}
        \item if $\R(P,X,Q)$ then
        \begin{enumerate} 
            \item if $P \steptau P'$ then there is a path $Q \pathtau Q_1 \step{\opt{\tau}} Q_2$ with $\R(P,X,Q_1)$ and $\R(P',X,Q_2)$,
            \item if $P \step{a} P'$ with $a \in X$ then there is a path $Q \pathtau Q_1 \step{a} Q_2$ with $\R(P,X,Q_1)$ and $\R(P',Q_2)$,
            \item if $\deadend{P}{X}$ then there is a path $Q \pathtau Q_0$ with $\R(P,Q_0)$,
            \item if $\deadend{P}{X}$ and $P \step{\rt} P'$ then there is a path $Q \mathrel{=:} Q_0 \pathtau Q_1 \step{\rt} Q_2 \pathtau Q_3 \step{\rt} \dots \pathtau Q_{2r{-}1} \step{\opt{\rt}} Q_{2r}$ with $r>0$, such that $\forall i \in [0,r{-}1], \R(P,X,Q_{2i}) \wedge \deadend{Q_{2i+1}}{X}$ and $\R(P',X,Q_{2r})$,
            \item if $P \nsteptau$ then there is a path $Q \pathtau Q_0 \nsteptau$.
        \end{enumerate}
    \end{enumerate}
For $P,Q \mathbin\in \closed$, if there exists a \tb reactive bisimulation $\R$ with $\R(P,Q)$ (resp.\ $\R(P,X,Q)$) then $P$ and $Q$ are said to be \emph{\tb reactive bisimilar} (resp.\ \emph{\tb $X$-bisimilar}), which is denoted $P \bisimtbr Q$ (resp.\ $P \bisimtbr[X] Q$).
\end{definition}

\noindent
To build the above definition, the definition of a strong reactive bisimulation \cite{strongreactivebisimilarity} was modified in a branching manner \cite{branching}. Intuitively, a triplet $\R(P,X,Q)$ affirms that $P$ and $Q$ behave similarly when the environment allows (only) the set of actions in $X$ to occur, whereas a couple $\R(P,Q)$ says that $P$ and $Q$ behave in the same way when the environment has been triggered to change. As said before, the environment can be seen as a system executed in parallel while enforcing the synchronisation of all visible actions. 

Clause 1 captures the scenario of a triggered environment: if $P$ can perform a visible or invisible action then $Q$ has to be able to match it; and the environment can settle on a set $Y$ of allowed actions at any moment. Time-outs are not considered because these can occur only when the system idles, and idling can happen only when the environment has stabilised on a set of allowed actions. One might notice that, in \cite{strongreactivebisimilarity}, the first clause was only required for invisible actions. However, there the case $\alpha\neq\tau$ is actually implied by the other clauses. If in our definition Clause 1.a were restricted to invisible actions then $\bisimtbr$ would not be a congruence for the parallel operator, as shown in Appendix \ref{app:examples}.

Clause 2 depicts the scenario of an environment allowing $X$. $\tau$-transitions have to be matched since the environment cannot disallow them, and their execution does not trigger the environment to change. Visible actions have to be matched only if they are allowed, and their execution triggers the environment. Triggering the environment or not explains why Clause 2a matches $Q_2$ in a triplet and Clause 2b in a couple. If $P$ idles (i.e.\ $\deadend{P}{X}$) then the environment can be triggered, thus, $Q$ has to be able to instantaneously reach a state $Q_0$ related to $P$ in a triggered environment.\footnote{By Lemma~\ref{lem:obvious}.4 we can even choose $Q_0$ such that $Q_0 \nsteptau$, so that $\init{Q_0}=\init{P}$.} If $P$ idles and has an outgoing time-out transition then $Q$ has to be able to match it in a branching manner. This involves $Q$ performing any sequence of $\tau$ and $\rt$-transitions, such that all states encountered prior to the last optional $\rt$ are related to $P$.\footnote{Clause 2.d requires this only for states of the form $Q_{2i}$ with $i\in [0,r{-}1]$, but by Lemma~\ref{lem:obvious}.1 it holds for all of them. Clause 2.c further implies that in Clause 2.d we have $\R(P,Q_{2i+1})$ for all $i \in [0,r{-}1]$.} Lastly, a stability respecting clause \cite{vG93} was added for practical reasons. In Appendix \ref{app:examples}, an example shows that without it $\bisimtbr$ would not even be an equivalence. For the important class of \emph{divergence-free} systems, without infinite sequences $Q_0 \steptau Q_{1} \steptau \dots$, Clause 2.e is easily seen to be redundant.

\begin{lemma}\label{lem:obvious}
  Let $\R$ be a \tb reactive bisimulation.
  \begin{enumerate}
  \item If $\R(P,X,Q)$, $P \nsteptau$ and $Q \pathtau Q'$ then also $\R(P,X,Q')$.
  \item If $\R(P,Q)$ or $\R(P,X,Q)$, $P \nsteptau$ and $Q \nsteptau$ then $\init{Q}=\init{P}$.
  \item If $\R(P,X,Q)$, $\deadend{P}{X}$ and $Q \nsteptau$ then $\R(P,Q)$.
  \item If $\R(P,X,Q)$ and $\deadend{P}{X}$ then there is a path $Q \pathtau Q_0$ with $\R(P,Q_0)$, $Q_0\nsteptau$ and $\init{Q_0}=\init{P}$.
  \end{enumerate}
\end{lemma}

\begin{proof}
\begin{enumerate}
    \item This is an immediate consequence of the symmetric counterpart of Clause 2.a (where $Q$ takes a $\tau$-step).
    When that clause yields $P \pathtau P_1 \step{\opt{\tau}} P_2$ we have $P_2=P$.
    \item This is a direct consequence of Clause 1.a or 2.b and its symmetric counterpart.
    \item By Clause 2.e there is path $Q \pathtau Q_0$ with $Q_0 \nsteptau$. By Claim 1 of this lemma, $\R(P,X,Q_0)$. Thus, by Clause 2.c there is a path $Q_0 \mathbin{\pathtau} Q_1$ with $\R(\hspace{-1pt}P,\hspace{-1pt}Q_1\hspace{-1pt})$, but $Q_1\mathop{=}Q_0\mathop{=}Q$ since $Q \!\nsteptau$.
    \item By Clause 2.e there is path $Q \pathtau Q_0$ with $Q_0 \nsteptau$. By Claim 1 of this lemma, $\R(P,X,Q_0)$. That $\init{Q_0}=\init{P}$ and $\R(P,Q_0)$ follows by Claims 2 and 3 of this lemma.
\popQED
\end{enumerate}
\end{proof}

\noindent
Definition~\ref{def:intuitive} enables us to elide some time-outs. Using the process algebra notation to be formally introduced in Section~\ref{sec:process algebra}, the processes $a.\rt.b.0$ and $a.\rt.\rt.b.0$ (as well as $a.\rt.\tau.\rt.b.0$) are \tb reactive bisimilar. Both require an unquantified positive but finite amount of rest between the actions $a$ and $b$. To support this example, Clause 2.d of Definition~\ref{def:intuitive} must allow a single time-out transition of one process to be matched by either zero or multiple time-outs of the other. An alternative definition, treating time-outs more like visible transitions, is obtained by replacing Clause 2.d by
\begin{enumerate}
    \setcounter{enumi}{1}
    \item 
    \begin{enumerate}
        \setcounter{enumii}{3}
        \item if $\deadend{P}{X}$ and $P \step{\rt} P'$ then there is a path $Q \pathtau Q_1 \step{\rt} Q_2$ with $\R(P',X,Q_2)$.
    \end{enumerate}
\end{enumerate}
Requiring that the matching time-out is executable (i.e.\ $\deadend{Q_1}{X}$) is not necessary here, as it is implied by the other clauses. Indeed, Lemma~\ref{lem:obvious}.3, which is not affected by changing Clause 2.d, implies the existence of a path $Q \pathtau Q_1 \nsteptau$ such that $\R(P,Q_1)$ and $\deadend{Q_1}{X}$. Since $Q_1 \nsteptau$, $\deadend{P}{X}$ and $P\step{\rt} P'$, Clause 2d yields $Q_1 \step{\rt} Q_2$ with $\R(P',X,Q_2)$. This version of the definition has been studied \cite{Reghem24} and has properties similar to $\bisimtbr$\,, which are recapped in Appendix \ref{app:concrete time-out}.

In \cite{branching}, branching bisimilarity is expressed in multiple equivalent ways. For practical purposes, our definition uses the semi-branching format, which is equivalent to the branching format thanks to the following lemma.

\begin{lemma}[Stuttering Lemma] \label{lem:stuttering}
    Let $P, P^\dag, P^\ddag, Q \in \closed$, if $P \bisimtbr Q$, $P^\ddag \bisimtbr Q$ (resp.\ $P \bisimtbr[X] Q$, $P^\ddag \bisimtbr[X] Q$) and $P \steptau P^\dag \steptau P^\ddag$ then $P^\dag \bisimtbr Q$ (resp.\ $P^\dag \bisimtbr[X] Q$).
\end{lemma}

\begin{proof}
    Let $\R$ be a \tb reactive bisimulation. Let's define $\R' := \R \cup \{(P^\dag,Q),(Q,P^\dag) \mid \exists P,P^\ddag \in \closed, P \pathtau P^\dag \pathtau P^\ddag \wedge \R(P,Q) \wedge \R(P^\ddag,Q)\} \cup \{(P^\dag,X,Q),(Q,X,P^\dag) \mid \exists P,P^\ddag \in \closed,\linebreak[3] P \pathtau P^\dag \pathtau P^\ddag \wedge \R(P,X,Q) \wedge \R(P^\ddag,X,Q)\}$. $\R'$ is symmetric by definition and $\R'$ is a \tb reactive bisimulation, as proven in Appendix \ref{app:intro}.
\end{proof}

\begin{proposition} \label{prop:equivalence}
    $\bisimtbr$ and $(\bisimtbr[X])_{X \subseteq A}$ are equivalence relations.
\end{proposition}

\begin{proof}
    Reflexivity and symmetry are trivial following the definition. For transitivity, consider two \tb reactive bisimulations $\R_1$ and $\R_2$. Let's define $\R := (\R_1 \circ \R_2) \cup (\R_2 \circ \R_1)$. Here $\R_1 \circ \R_2 := \{(P,Q) \mid \exists R.~\R(P,R) \wedge \R(R,Q)\} \cup \{(P,X,Q) \mid \exists R.~\R(P,X,R) \wedge \R(R,X,Q)\}$.\linebreak[3] $\R$ is symmetric by definition and $\R$ is a \tb reactive bisimulation, as proven in Appendix \ref{app:intro}.
\end{proof}

\subsection{Rooted Version}

A well-known limitation of branching bisimilarity $\bisimb$ is that it fails to be a congruence for the choice operator $+$. For example, $a \bisimb \tau.a$ but $a + b \,\not\!\bisimb \tau.a + b$. Since the objective is to define a congruence, instead of $\bisimtbr$ we use the \emph{congruence closure} of $\bisimtbr$\,, which is the coarsest congruence included in $\bisimtbr$\,. 

\begin{definition}\rm\label{def:rooted intuitive}
    A \emph{rooted \tb reactive bisimulation} is a symmetric relation $\R \subseteq (\closed\times\closed)\cup(\closed\times\Pow(A)\times\closed)$ such that, for all $P,Q \in \closed$ and $X \subseteq A$,
    \begin{enumerate}
        \item if $\R(P,Q)$
        \begin{enumerate}
            \item if $P \step\alpha P'$ with $\alpha \in A_\tau$ then there is a transition $Q \step\alpha Q'$ with $P' \bisimtbr Q'$,
            \item for all $Y \subseteq A$, $\R(P,Y,Q)$;
        \end{enumerate}
        \item if $\R(P,X,Q)$
        \begin{enumerate}
            \item if $P \steptau P'$ then there is a transition $Q \steptau Q'$ with $P' \bisimtbr[X] Q'$,
            \item if $P \step{a} P'$ with $a \in X$ then there is a transition $Q \step{a} Q'$ with $P' \bisimtbr Q'$,
            \item if $\deadend{P}{X}$ then $\R(P,Q)$,
            \item if $\deadend{P}{X}$ and $P \step{\rt} P'$ then there is a transition $Q \step{\rt} Q'$ with $P' \bisimtbr[X] Q'$.\vspace{3pt}
        \end{enumerate}
    \end{enumerate}
For $P,Q \in \closed$, if there exists a rooted \tb reactive bisimulation $\R$ with $\R(P,Q)$ (resp.\ $\R(P,X,Q)$) then $P$ and $Q$ are said to be \emph{rooted \tb reactive bisimilar} (resp.\ rooted \tb $X$-bisimilar), which is denoted $P \bisimrtbr Q$ (resp.\ $P \bisimrtbr[X] Q$).
\end{definition}
A rooted version of a bisimulation consists in enforcing a stricter matching on the first transition of a system. In the branching case, the first transition is matched in the strong manner. The stability respecting clause can be removed, as it is now implied by the other clauses. Rooting the bisimilarity is the standard technique to obtain its congruence closure; later $\bisimrtbr$ will be proven to be a congruence. As any \tb reactive bisimulation relating $P+b$ and $Q+b$, for a fresh action $b$, induces a rooted \tb reactive bisimulation relating $P$ and $Q$, it then follows that $\bisimrtbr$ is the coarsest included in $\bisimtbr$. Since $\bisimtbr$ is an equivalence, the proof of Proposition \ref{prop:equivalence} can be adapted to $\bisimrtbr$ in a straightforward way.

\begin{proposition} \label{prop:rooted equivalence}
    $\bisimrtbr$ and $(\bisimrtbr[X])_{X \subseteq A}$ are equivalence relations.
\end{proposition}

\subsection{Alternative Forms of Definition~\ref{def:intuitive}}

Definition~\ref{def:intuitive} can be rephrased in various ways. First of all, using Requirements 1.b and 2.c, one can move Requirement 2.d from Clause 2 (dealing with triples $(P,X,Q)$) to Clause~1 (dealing with pairs $(P,Q)$), now adding a universal quantifier over $X$ to the requirement. Next, Requirement 2.e can be copied under Clause~1. This makes Clause 1.b unnecessary, thereby obtaining a definition in which the triples $(P,X,Q)$ are encountered only after taking a $\rt$-transition. In this form it is obvious that branching reactive bisimilarity reduces to the classical stability respecting branching bisimilarity for systems without $\rt$-transitions. We have chosen the form of Definition~\ref{def:intuitive} over the above alternatives, because we believe it comes with more natural intuitions for its plausibility.

In Appendix~\ref{app:gbrb} a further modification of Definitions~\ref{def:intuitive} and~\ref{def:rooted intuitive} is proposed, called \emph{generalised [rooted] \tb reactive bisimulation}. We show that each [rooted] \tb reactive bisimulation is a generalised [rooted] \tb reactive bisimulation, and two systems are [rooted] \tb reactive bisimilar iff they are related by a generalised [rooted] \tb reactive bisimulation.
This characterisation of $\bisimtbr$ and $\bisimrtbr$ will be used in the proofs of Theorem~\ref{thm:modal characterisation} and Proposition~\ref{prop:time-out bisim}.

In \cite{Pohlmann}, Pohlmann introduces an encoding which maps strong reactive bisimilarity to strong bisimilarity where time-outs are considered as any visible action. This encoding in essence places a given process in a most general environment, one that features environment time-out actions $\rt_\varepsilon$, as well as actions $\varepsilon_X$ for settling in a state that allows exactly the actions in $X$. This proves that reactive equivalences can be expressed as non-reactive ones at the cost of increasing the processes' size. Thus, any tool set able to work on strong bisimulation could theoretically deal with its reactive counterpart. 

In Appendix \ref{app:Pohlmann}, this encoding is slightly modified to yield a similar result for branching reactive bisimulation and its rooted version, for the latter result also employing actions $\rt_X$. It appears that these modifications do not impact its effect on strong reactive bisimilarity. Since our bisimilarity has some time-out eliding properties, it is not mapped to stability respecting branching bisimilarity, but to a new bisimilarity, defined below.

\begin{definition}\rm \label{def:non-reactive}
    A \emph{\rt-branching bisimulation} is a symmetric relation $\R \subseteq \closed\times\closed$ such that, for all $P,Q \in \closed$, if $\R(P,Q)$ then
    \begin{enumerate}
        \item if $P \step{\alpha} P'$ with $\alpha \in A_\tau \cup\{\rt_\epsilon,\epsilon_X \mid X \subseteq A\}$ then there is a path $Q \pathtau Q_1 \step{\opt{\alpha}} Q_2$ with $\R(P,Q_1)$ and $\R(P',Q_2)$,
        \item if $P \step{\rt} P'$ then there is a path $Q = Q_0 \pathtau Q_1 \step{\rt} Q_2 \pathtau Q_3 \step{\rt} ... \pathtau Q_{2r{-}1} \step{\opt{\rt}} Q_{2r}$ with $r>0$, such that $\forall i \in [0,2r{-}1],\; \R(P,Q_i)$ and $\R(P',Q_{2r})$,
        \item if $P \nsteptau$ then there is a path $Q \pathtau Q_0 \nsteptau$.
\vspace{3pt}

\end{enumerate}
For $P,Q \in \closed$, if there exists a \rt-branching bisimulation $\R$ with $\R(P,Q)$ then $P$ and $Q$ are said to be \emph{\rt-branching bisimilar}, which is denoted $P \bisimtb Q$.
\end{definition}

\noindent
The encoding also sends $\bisimrtbr$ to the rooted version of $\bisimtb$\,. 

\begin{definition}\rm \label{def:rooted non-reactive}
    A \emph{rooted \rt-branching bisimulation} is a symmetric relation $\R \subseteq \closed\times\closed$ such that, for all $P,Q \in \closed$, if $\R(P,Q)$ then
    \begin{enumerate}
        \item if $P \mathord{\step{\alpha}} P'$ with $\alpha \mathord\in Act \mathop\cup\{\rt_\epsilon,\hspace{-.5pt}\rt_X,\hspace{-0.5pt}\epsilon_X \mathbin{\mid} X \mathord\subseteq A\}$ then there is a transition $Q \mathop{\step{\alpha}} Q'$ with $P' \mathbin{\bisimtb} Q'\!$.\vspace{3pt}
    \end{enumerate}
For $P,Q \in \closed$, if there exists a rooted \rt-branching bisimulation $\R$ with $\R(P,Q)$ then $P$ and $Q$ are said to be \emph{rooted \rt-branching bisimilar}, which is denoted $P \bisimrtb Q$.
\end{definition}

\noindent
Providing a complete axiomatisation of rooted $\rt$-branching bisimilarity will be useful in the proof of completeness of the axiomatisation of rooted branching reactive bisimilarity (Lemma~\ref{lem:simplication}). 

\section{Modal Characterisation} \label{sec:modal}

The Hennessy-Milner logic \cite{HM85} expresses properties of the behaviour of processes in an LTS\@. In \cite{strongreactivebisimilarity}, the modality $\langle X\rangle\varphi$ was added to obtain a modal characterisation of strong reactive bisimilarity ($\rbis{}{r}$). In order to capture branching reactive bisimilarity we add another modality $X\varphi$. To avoid confusion, $\langle X\rangle\varphi$ is renamed $\langle \rt_X\rangle\varphi$.

\begin{definition}\rm \label{def:Hennessy-Milner}
    The class $\logic$ of \emph{reactive Hennessy-Milner formulas} is defined as follows, where $I$ is an index set, $\alpha \in Act$, $a \in A$ and $X \subseteq A$,
    \begin{center}
        $\varphi ::= \top \mid \bigwedge\limits_{i \in I} \varphi_i \mid \neg\varphi \mid 
      \langle\alpha\rangle\varphi \mid X\varphi$
    \end{center}
    \vspace{-14pt}
\end{definition}

\begin{table}[ht]
\begin{center}
    \begin{tabular}{l c l l c l}
        $P \models \top$ &&& $P \models_Y \top$ \\
        $P \models \bigwedge_{i\in I}\varphi_i$ & iff & $\forall i \in I, P \models \varphi_i$ & $P \models_Y \bigwedge_{i\in I}\varphi_i$ & iff & $\forall i \in I, P \models_Y \varphi_i$ \\
        $P \models \neg\varphi$ & iff & $P \not\models \varphi$ & $P \models_Y \neg\varphi$ & iff & $P \not\models_Y \varphi$ \\
        $P \models \langle\alpha\rangle\varphi$ & iff & $\exists P \step{\alpha} P', P' \models \varphi$ & $P \models_Y \langle\tau\rangle\varphi$ & iff & $\exists P \step{\tau} P', P' \models_Y \varphi$ \\
        &&& $P \models_Y \langle\rt\rangle\varphi$ & iff & $\exists P \step{\rt} P', P' \models_Y \varphi$ \\
        $P \models_Y \langle a\rangle\varphi$ & iff & \multicolumn{4}{l}{$(a 
        \in Y \vee \deadend{P}{Y}) \wedge \exists P \step{a} P', P' \models \varphi$} \\
        $P \models X\varphi$ & iff & \multicolumn{4}{l}{$\deadend{P}{X} \wedge P \models_X \varphi$} \\
        $P \models_Y X\varphi$ & iff & \multicolumn{4}{l}{$\deadend{P}{X \cup Y} \wedge P \models_X \varphi$} \\
    \end{tabular}
\end{center}
    \caption{Semantics of $\models$ and $(\models_Y)_{Y \subseteq A}$}
    \label{tab:formulas semantics}
\vspace{-2.5ex}
\end{table}

\noindent
The satisfaction rules of $\logic$ are given in Table \ref{tab:formulas semantics}. $P \models \varphi$ means that $P$ satisfies $\varphi$ when the environment is triggered, and $P \models_Y \varphi$ indicates that $P$ satisfies $\varphi$ when the environment allows $Y$. The modality $X\varphi$ expresses that a process can idle in its current state during a period in which the environment allows the actions in $X\!$, after which it behaves according~to~$\varphi$.

The modality $\langle \rt_X\rangle\varphi$ from \cite{strongreactivebisimilarity} can now be defined as $\langle \rt_X\rangle\varphi := X\langle \rt\rangle\varphi$. Write $\logic_s$ for the fragment of $\logic$ from \cite{strongreactivebisimilarity}, which includes $\langle \rt_X\rangle\varphi$ but does not feature $X\varphi$ or $\langle \rt\rangle\varphi$. Then the modal characterisation theorem of \cite{strongreactivebisimilarity} says \hfill
\( P \rbis{}{r} Q ~~\Leftrightarrow~~ \forall \varphi\in\logic_s.~(P \models \varphi \Leftrightarrow Q \models \varphi)\;. \)

Here we restrict attention to the fragment of $\logic$ that includes $\langle \rt_X\rangle\varphi$ and $X\varphi$, but not $\langle \rt\rangle\varphi$. On this fragment $\models_Y$ is defined such that whenever $\init{P}\cap (Y\cup\{\tau\}) = \emptyset$ then $P \models_Y \varphi$ iff $P \models \varphi$. This is because the environment may choose to change during a period of idling.

To obtain a modal characterisation of [rooted] branching relative bisimilarity, we need a few other derived modalities. First of all, $\langle\epsilon\rangle\varphi := \bigvee_{i \in \nat}\langle\tau\rangle^i\varphi$. To lessen the notations, for all $\alpha \in A_\tau$, $\langle\hat{\alpha}\rangle\varphi$ denotes $\varphi \vee \langle\tau\rangle\varphi$ if $\alpha = \tau$, $\langle\alpha\rangle\varphi$ otherwise, and the modality $\langle\hat{\rt}_X\rangle\varphi$ denotes $\langle \rt_X\rangle\varphi \vee X\varphi$ or $X\langle\hat{\rt}\rangle\varphi$. Moreover, $\varphi \wedge \langle\hat{\alpha}\rangle\varphi'$ is shortened to $\varphi\langle\hat{\alpha}\rangle\varphi'$.
Furthermore, we define $\varphi\langle\epsilon_X\rangle\varphi' := \bigvee_{i \in \nat}\varphi\langle\epsilon_X\rangle^{(i)}\varphi'$, where $\varphi\langle\epsilon_X\rangle^{(0)}\varphi' := \langle\epsilon\rangle(\varphi\wedge\langle\hat{\rt}_X\rangle\varphi')$ and, for all $i>0$,\linebreak[3]
$\varphi\langle\epsilon_X\rangle^{(i)}\varphi' :=  \langle\epsilon\rangle(\varphi\wedge\langle \rt_X\rangle(\varphi\wedge (\varphi\langle\epsilon_X\rangle^{(i-1)}\varphi')))$.
The satisfaction rules of these new modalities can be derived from the basic ones: see Table \ref{tab:operator semantics}.

\begin{table}[t]
\begin{center}
    \begin{tabular}{l c l l c l}
        $P \models \langle\hat{\alpha}\rangle\varphi$ & iff & $\exists P \step{\opt{\alpha}} P', P' \models \varphi$ & $P \models_Y \langle\hat{\tau}\rangle\varphi$ & iff & $\exists P \step{\opt{\tau}} P', P' \models_Y \varphi$ \\
        $P \models \langle\epsilon\rangle\varphi$ & iff & $\exists P \pathtau P', P' \models \varphi$ & $P \models_Y \langle\epsilon\rangle\varphi$ & iff & $\exists P \pathtau P', P' \models_Y \varphi$ \\
        $P \models \varphi\langle\epsilon_X\rangle\varphi'$ & iff & \multicolumn{4}{l}{$\exists P \pathtau P_1 \step{\rt} P_2 \pathtau P_3 \step{\rt} ... \pathtau P_{2r{-}1} \step{\opt{\rt}} P_{2r}$ with $r>0$, such that}\\
         & & \multicolumn{4}{l}{$\forall i \in [1,2r{-}1]\; P_{i} \models_X \varphi \wedge P_{2r} \models_X \varphi'$ and}\\
         & & \multicolumn{4}{l}{$\forall i \in [0,r{-}1]\; \deadend{P_{2i+1}}{X}$} \\
        $P \models_Y \varphi\langle\epsilon_X\rangle\varphi'$ & iff & \multicolumn{4}{l}{$\exists P \pathtau P_1 \step{\rt} P_2 \pathtau P_3 \step{\rt} ... \pathtau P_{2r{-}1} \step{\opt{\rt}} P_{2r}$ with $r>0$, such that}\\
        & & \multicolumn{4}{l}{$\forall i \in [1,2r{-}1]\; P_{i} \models_X \varphi \wedge P_{2r} \models_X \varphi'$ and}\\
        & & \multicolumn{4}{l}{$\deadend{P_{1}}{Y} \wedge \forall i \in [0,r{-}1]\; \deadend{P_{2i+1}}{X}$} \\
        $P \models \langle \rt_X\rangle\varphi$ & iff & \multicolumn{4}{l}{$\deadend{P}{X} \wedge \exists P \step{\rt} P', P' \models_X \varphi$} \\
    \end{tabular}
\end{center}
\caption{Semantics of $\models$ and $(\models_Y)_{Y \subseteq A}$ for the derived modalities}
\label{tab:operator semantics}
\vspace{-2ex}
\end{table}

\begin{definition}\rm \label{def:subclass}
    The sub-classes $\logic_b$ and $\logic_b^r$ are defined as follows, where $I$ is an index set, $\alpha \in A_\tau$, $X \subseteq A$, $\varphi, \varphi' \in \logic_b$ and $\psi \in \logic_b^r$,
    \begin{align*}
        \varphi &::= \top \mid \bigwedge_{i \in I} \varphi_i \mid \neg\varphi \mid \langle\epsilon\rangle(\varphi\langle\hat{\alpha}\rangle\varphi') \mid \varphi\langle\epsilon_X\rangle\varphi' \mid \langle\epsilon\rangle\neg\langle\tau\rangle\top \tag{$\logic_b$} \\
        \psi &::= \top \mid \bigwedge_{i \in I} \psi_i \mid \neg\psi \mid \langle\alpha\rangle\varphi \mid \langle \rt_X\rangle\varphi \tag{$\logic_b^r$}
    \end{align*}
\end{definition}

\noindent
The last option for $\logic_b$, inspired by \cite{modalstab}, is used to encompass the stability respecting Clause 2.e of Definition~\ref{def:intuitive}.

\begin{theorem} \label{thm:modal characterisation}
    Let $P,Q \in \closed$. For all $X \subseteq A$,
    \begin{itemize}
        \item $P \bisimtbr Q$ iff $\forall \varphi \in \logic_b, P \models \varphi \Leftrightarrow Q \models \varphi$,
        \item $P \bisimtbr[X] Q$ iff $\forall \varphi \in \logic_b, P \models_X \varphi \Leftrightarrow Q \models_X \varphi$,
        \item $P \bisimrtbr Q$ iff $\forall \psi \in \logic_b^r, P \models \psi \Leftrightarrow Q \models \psi$,
        \item $P \bisimrtbr[X] Q$ iff $\forall \psi \in \logic_b^r, P \models_X \psi \Leftrightarrow Q \models_X \psi$.
    \end{itemize}
\end{theorem}

\begin{proof}
    $(\Rightarrow)$ The four propositions are proven simultaneously by structural induction on $\logic_b$ and $\logic_b^r$ in Appendix \ref{app:modal}.

    $(\Leftarrow)$ Let $\equiv \; := \{(P,Q) \mid \forall \varphi \in \logic_{b}, P \models \varphi \Leftrightarrow Q \models \varphi\} \cup \{(P,X,Q) \mid \forall \varphi \in \logic_{b}, P \models_X \varphi \Leftrightarrow Q \models_X \varphi\}$, and $\equiv^r \; := \{(P,Q) \mid \forall \psi \in \logic_{b}^r, P \models \psi \Leftrightarrow Q \models \psi\} \cup \{(P,X,Q) \mid \forall \psi \in \logic_{b}^r, P \models_X \psi \Leftrightarrow Q \models_X \psi\}$. It suffices to check that $\equiv$ [resp.\ $\equiv^r$] is a generalised [rooted] \tb reactive bisimulation. This is done in Appendix \ref{app:modal}.
\end{proof}

\section{Process Algebra and Congruence} \label{sec:congruence}
\label{sec:process algebra}

The process algebra $\ccsp$ is composed of classical operators from the well-known process algebras CCS \cite{Mi90ccs}, CSP \cite{BHR84,OH86} and ACP \cite{BW90,Fok00}, as well as the time-out action $\rt$ and two \emph{environment operators} from \cite{strongreactivebisimilarity}, that were added in order to enable a complete axiomatisation.

\begin{definition}\rm \label{def:ccsp syntax}
    Let $V$ be a countable set of variables, the \emph{expressions} of $\ccsp$ are recursively defined as follows:
    \begin{align*}
        E ::= 0 \mid x \mid \alpha.E \mid E + F \mid E \parallel_S F \mid \tau_I(E) \mid \rename(E) \mid \theta_L^U(E) \mid \psi_X(E) \mid \langle y | \equa \rangle
    \end{align*}
    where $x \in V$, $\alpha \in Act$, $S, I, L, U, X \subseteq A$, $L \subseteq U$, $\rename \subseteq A\times A$, $\equa$ is a \emph{recursive specification}: a set of equations $\{x\mathbin=\equa_x \mid x \mathbin\in V_\equa\}$ with $V_\equa \subseteq V$ and each $\equa_x$ a $\ccsp$ expression, and $y\mathbin\in V_\equa$.\linebreak[3] We require that all sets ${\{b\mid (a,b)\in \rename\}}$ for $a\in A$ are finite.
\end{definition}

\noindent
$0$ stands for a system which cannot perform any action. The expression $\alpha.E$ represents a system that first performs $\alpha$ and then $E$. The expression $E + F$ represents a choice to behave like $E$ or $F$. The parallel composition $E \parallel_S F$ synchronises the execution of $E$ and $F$, but only when performing actions in $S$. $\tau_I(E)$ represents the system $E$ where all actions $a \mathbin\in I$ are transformed into~$\tau$. The operator $\rename$ renames a given action $a\mathbin\in A$ into a choice between all actions $b$ with $(a,b)\mathbin\in \rename$. $\langle y | \equa\rangle$ is the $y$-component of a solution of $\equa$. 

$\ccsp$ also has two environment operators that help to develop a complete axiomatisation (like the left merge for ACP). $\theta_L^U(E)$ is the expression $E$ plunged into an environment $X$ such that $L \subseteq X \subseteq U$. $\theta_X^X(E)$ is denoted $\theta_X(E)$. $\psi_X(E)$ plunges $E$ into the environment $X$ if a time-out occurs, but, has no effect if any other action is performed. The operational semantics of $\ccsp$ is given in Figure~\ref{fig:ccsp semantics}. All operators except the environment ones follow the semantics of CCS, CSP or ACP\@. As $\theta_L^U(E)$ simulates the expression $E$ plunged in an environment $L \subseteq X \subseteq U$, it has no effect on $\tau$-transitions, which do not trigger the environment. Moreover, $\theta_L^U$ restricts the ability to perform visible actions to those allowed by the environment (i.e.\ included in $U$) and performing these actions triggers the environment. However, if the expression idles (i.e.\ $\deadend{E}{L}$) then it might trigger the environment and $\theta_L^U(E)$ acts like $E$. $\psi_X(E)$ supposes that time-outs are performed while the environment allows $X$, thus, it has no effect on actions that are not $\rt$. However, if $E$ can perform a time-out while the environment allows $X$ (i.e. $\deadend{E}{X}$) then $\psi_X(E)$ can perform the time-out while plunging the expression in the environment $X$.

\begin{figure}[ht]
\centering
    \begin{prooftree}
        \AxiomC{\textcolor{white}{$x \step{\alpha} y$}}
        \UnaryInfC{$\alpha.x \step{\alpha} x$}
        \DisplayProof
        \hskip 2em
        \AxiomC{$x \step{\alpha} x'$}
        \UnaryInfC{$x + y \step{\alpha} x'$}
        \DisplayProof
        \hskip 2em
        \AxiomC{$y \step{\alpha} y'$}
        \UnaryInfC{$x + y \step{\alpha} y'$}
    \end{prooftree}
    
    \begin{prooftree}
        \AxiomC{$x \step{a} x' \wedge \rename(a,b)$}
        \UnaryInfC{$\rename(x) \step{b} \rename(x')$}
        \DisplayProof
        \hskip 2em
        \AxiomC{$x \steptau x'$}
        \UnaryInfC{$\rename(x) \steptau \rename(x')$}
        \DisplayProof
        \hskip 2em
        \AxiomC{$x \step{\rt} x'$}
        \UnaryInfC{$\rename(x) \step{\rt} \rename(x')$}
    \end{prooftree}

    \begin{prooftree}
        \AxiomC{$x \step{\alpha} x' \wedge \alpha \not\in S$}
        \UnaryInfC{$x\parallel_S y \step{\alpha} x' \parallel_S y$}
        \DisplayProof
        \hskip 2em
        \AxiomC{$y \step{\alpha} y' \wedge \alpha \not\in S$}
        \UnaryInfC{$x \parallel_S y \step{\alpha} x \parallel_S y'$}
        \DisplayProof
        \hskip 2em
        \AxiomC{$x \step{a} x' \wedge y \step{a} y' \wedge a \in S$}
        \UnaryInfC{$x \parallel_S y \step{a} x' \parallel_S y'$}
    \end{prooftree}

    \begin{prooftree}
        \AxiomC{$x \step{\alpha} x' \wedge \alpha \not\in I$}
        \UnaryInfC{$\tau_I(x) \step{\alpha} \tau_I(x')$}
        \DisplayProof
        \hskip 2em
        \AxiomC{$x \step{a} x' \wedge a \in I$}
        \UnaryInfC{$\tau_I(x) \steptau \tau_I(x')$}
        \DisplayProof
        \hskip 2em
        \AxiomC{$\langle \equa_x | \equa \rangle \step{\alpha} x'$}
        \UnaryInfC{$\langle x | \equa\rangle \step{\alpha} x'$}
    \end{prooftree}

    \begin{prooftree}
        \AxiomC{$x \steptau x'$}
        \UnaryInfC{$\theta_L^U(x) \steptau \theta_L^U(x')$}
        \DisplayProof
        \hskip 2em
        \AxiomC{$x \step{a} x' \wedge a \in U$}
        \UnaryInfC{$\theta_L^U(x) \step{a} x'$}
        \DisplayProof
        \hskip 2em
        \AxiomC{$x \step{\alpha} x' \wedge \alpha \neq t$}
        \UnaryInfC{$\psi_X(x) \step{\alpha} x'$}
    \end{prooftree}

    \begin{prooftree}
        \AxiomC{$x \step{\alpha} x' \wedge \deadend{x}{L}$}
        \UnaryInfC{$\theta_L^U(x) \step{\alpha} x'$}
        \DisplayProof
        \hskip 2em
        \AxiomC{$x \step{\rt} x' \wedge \deadend{x}{X}$}
        \UnaryInfC{$\psi_X(x) \step{\rt} \theta_X(x')$}
    \end{prooftree}
    \caption{Operational semantics of $\ccsp$}
    \label{fig:ccsp semantics}
\end{figure}

All $\equa_x$ are considered to be sub-expressions of $\langle y | \equa\rangle$. An occurrence of a variable $x$ is \emph{bound} in $E \in \ccsp$ iff it occurs in a sub-expression $\langle y |\equa\rangle$ of $E$ such that $x \in V_\equa$; otherwise it is \emph{free}. An expression $E$ is \emph{invalid} if it has a sub-expression $\theta_L^U(F)$ or $\psi_X(F)$ such that a variable occurrence is free in $F$, but bound in $E$. An example justifying this condition can be found in \cite{strongreactivebisimilarity}. The set of valid expressions of $\ccsp$ is denoted $\expr$. If an expression is valid and all of its variable occurrences are bound then it is \emph{closed} and we call it a \emph{process}; the set of processes is denoted $\closed$.

A \emph{substitution} is a partial function $\rho: V \rightharpoonup E$. The application $E[\rho]$ of a substitution $\rho$ to an expression $E \in \expr$ is the result of the simultaneous replacement, for all $x \in \dom{\rho}$, of each free occurrence of $x$ by the expression $\rho(x)$, while renaming bound variables to avoid name clashes. We write $\langle E|\equa\rangle$ for the expression $E$ where any $y \in V_\equa$ is substituted by $\langle y | \equa\rangle$. 

\subsection{Time-out Bisimulation}

Thanks to the environment operator $\theta_L^U$, it is possible to express our bisimilarity in a much more succinct way. Indeed, $\theta_X$ was defined so that $P \bisimtbr[X] Q$ if and only if $\theta_X(P) \bisimtbr \theta_X(Q)$.

\begin{definition}\rm \label{def:time-out bisim}
    A \emph{\tb time-out bisimulation} is a symmetric relation ${\tbisim} \subseteq \closed\times\closed$ such that, for all $P,Q \in \closed$, if $P \tbisim Q$ then
    \begin{enumerate}
        \item if $P \step{\alpha} P'$ with $\alpha \in A_\tau$ then there is a path $Q \pathtau Q_1 \step{\opt{\alpha}} Q_2$ with $P \tbisim Q_1$ and $P' \tbisim Q_2$
        \item if $\deadend{P}{X}$ and $P \step{\rt} P'$ then there is a path $Q \pathtau Q_1
          \step{\rt} Q_2 \pathtau Q_3 \step{\rt} ... \pathtau Q_{2r{-}1} \step{\opt{\rt}} Q_{2r}$
          with $r>0$, such that $Q_1\nsteptau$, $\forall i \in [1,r{-}1],\; \theta_X(P) \tbisim \theta_X(Q_{2i})\linebreak[2] \wedge\linebreak[2] \deadend{Q_{2i+1}}{X}$ and $\theta_X(P') \tbisim \theta_X(Q_{2r})$
        \item if $P \nsteptau$ then there is a path $Q \pathtau Q_0 \nsteptau$.
    \end{enumerate}
\end{definition}

\noindent
Note that in Condition~2 above one also has $P \tbisim Q_1$ and consequently $\deadend{Q_{1}}{X}$. A rooted version of branching time-out bisimulation can be defined in the same vein.

\begin{definition}\rm \label{def:rooted time-out bisim}
    A \emph{rooted \tb time-out bisimulation} is a symmetric relation ${\tbisim} \subseteq \closed\times\closed$ such that, for all $P,Q \in \closed$ such that $P \tbisim Q$,
    \begin{enumerate}
        \item if $P \step{\alpha} P'$ with $\alpha \in A_\tau$ then there is a step $Q \step{\alpha} Q'$ such that $P' \bisimtbr Q'$
        \item if $\init{P}\cap(X\cup\{\tau\}) \mathbin= \emptyset$ and $P \mathbin{\step{\rt}} P'$ then there is a step $Q \mathbin{\step{\rt}} Q'$ such that $\theta_X(P') \mathbin{\bisimtbr} \theta_X(Q').$
    \end{enumerate}
\end{definition}

\begin{proposition} \label{prop:time-out bisim}
    Let $P,Q \in \closed$, 
    \begin{enumerate}
        \item $P \bisimtbr Q$ (resp.\ $P \bisimtbr[X] Q$) iff there exists a \tb time-out bisimulation $\tbisim$ with $P \tbisim Q $ (resp.\ $(\theta_X(P) \tbisim \theta_X(Q)$),
        \item $P \bisimtbr[X] Q$ if and only if $\theta_X(P) \bisimtbr \theta_X(Q)$,\label{corr}
        \item $P \bisimrtbr Q$ (resp.\ $P \bisimrtbr[X] Q$) iff there exists a rooted \tb time-out bisimulation $\tbisim$ with $P \tbisim Q $ (resp.\ $(\theta_X(P) \tbisim \theta_X(Q)$).
    \end{enumerate}
\end{proposition}

\begin{proof}
  Note that Proposition~\ref{prop:time-out bisim}.\ref{corr} is a trivial corollary of \ref{prop:time-out bisim}.1.
  
  Let $\R$ be a [generalised rooted] \tb reactive bisimulation, let's define ${\tbisim} := \{(P,Q) \mid \R(P,Q)\} \cup \{(\theta_X(P),\theta_X(Q)) \mid \R(P,X,Q)\}$. $\tbisim$ is a [rooted] \tb time-out bisimulation, as proven in Appendix \ref{app:time-out}. Let $\tbisim$ be a [rooted] \tb time-out bisimulation, let's define $\R = \{(P,Q) \mid P \tbisim Q\} \cup \{(P,X,Q) \mid \theta_X(P) \tbisim \theta_X(Q)\}$. $\R$ is a [rooted] generalised \tb reactive bisimulation, as proven in Appendix \ref{app:time-out}.
\end{proof}

\noindent
Time-out bisimulations are very practical as there are no triplets to deal with anymore.

\subsection{Congruence}

Until now, bisimilarity was only defined between closed expressions, but any relation ${\sim} \subseteq \closed\times\closed$ can be extended to $\expr\times\expr$ in the following way: $E \sim F$ iff $\forall \rho: V \rightarrow \closed,\; E[\rho] \sim F[\rho]$. It can be extended further to substitutions $\rho, \nu \in V \rightharpoonup \expr$ by $\rho \sim \nu$ iff $\dom{\rho} = \dom{\nu}$ and $\forall x \in \dom{\rho},\; \rho(x) \sim \nu(x)$.

\begin{definition}\rm \label{def:congruence}
    An equivalence ${\sim} \subseteq \expr\times\expr$ is a congruence for an $n$-ary operator $f$ if $P_i \sim Q_i$ for all $i=0,\dots,n{-}1$ implies $f(P_0,...,P_{n-1}) \sim f(Q_0,...,Q_{n-1})$. It is a \emph{lean congruence} if, for all $E \in \expr$ and all $\rho, \nu \in V \rightharpoonup \expr$ such that $\rho \sim \nu$, $E[\rho] \sim E[\nu]$. It is a \emph{full congruence} if 
    \begin{enumerate}
        \item it is a congruence for all operators in the language, and
        \item for all recursive specifications $\equa, \equa'$ with $V_\equa = V_{\equa'}$ and $x \in V_\equa$ such that $\langle x|\equa\rangle,\langle x|\equa'\rangle \in \closed$, if $\forall y \in V_\equa,\; \equa_y \sim \equa'_y$ then $\langle x|\equa\rangle \sim \langle x|\equa'\rangle$.
    \end{enumerate}
\end{definition}

\noindent
To show that $\sim$ is a lean congruence it suffices to restrict attention to closed substitutions $\rho, \nu \in V \rightarrow \closed$, because the general property will then follow by composition of substitutions. A full congruence is a lean congruence, and a lean congruence is a congruence for all operators in the language, but both implications are strict, as shown in \cite{vG17b}.

To show that $\bisimrtbr$ and $\bisimrtb$ are full congruences, it is first necessary to prove that $\bisimtbr$ and $\bisimtb$ are congruences for some of the operators of $\ccsp$.

\begin{proposition} \label{prop:stability}
    $\bisimtbr$ and $\bisimtb$ are congruences for action prefixing, parallel composition, abstraction, renaming and the environment operator $\theta_L^U$, for all $L \subseteq U \subseteq A$.
\end{proposition}

\begin{proof}
    Let $\tbisim$ be the smallest relation such that, for all $P,Q \in \closed$,
    \begin{itemize}
        \item if $P \bisimtbr Q$ then $P \tbisim Q$;
        \item if $P \tbisim Q$ then, for all $\alpha \in Act$, $I\subseteq A$, $\rename \in A\times A$ and $L \subseteq U \subseteq A$, $\alpha.P \tbisim \alpha.Q$, $\tau_I(P) \tbisim \tau_I(Q)$, $\rename(P) \tbisim \rename(Q)$ and $\theta_L^U(P) \tbisim \theta_L^U(Q)$; 
        \item if $P_1 \tbisim Q_1$, $P_2 \tbisim Q_2$ and $S \subseteq A$ then $P_1 \parallel_S P_2 \tbisim Q_1 \parallel_S Q_2$.
    \end{itemize}
    It suffices to show that $\tbisim$ is a \tb time-out bisimulation up to $\bisim$\,, which implies ${\tbisim}\subseteq{\bisimtbr}\,$. A bisimulation ``up to'' is a notion introduced by Milner in \cite{Mi90ccs}; it is commonly used when proving congruence properties. The proof uses some lemmas which were obtained in \cite{strongreactivebisimilarity}. Details can be found in Appendix \ref{app:stability}. A similar proof yields the result for $\bisimtb$\,. \end{proof}

\begin{theorem} \label{thm:congruence}
    $\bisimrtbr$ and $\bisimrtb$ are full congruences.
\end{theorem}

\begin{proof}
    Let ${\tbisim} \subseteq \closed\times\closed$ be the smallest relation such that 
    \begin{itemize}
        \item if $P \bisimrtbr Q$ then $P \tbisim Q$;
        \item if $P_1 \tbisim Q_1$ and $P_2 \tbisim Q_2$ then $P_1 + P_2 \tbisim Q_1 + Q_2$ and $\forall S \subseteq A,\; P_1 \parallel_S P_2 \tbisim Q_1 \parallel_S Q_2$;
        \item if $P \tbisim Q$ then $\forall \alpha \mathbin\in Act,\; \alpha.P \tbisim \alpha.Q$, $\forall I \subseteq A,\; \tau_I(P) \tbisim \tau_I(Q)$, $\forall \rename \subseteq A\times A,\; \rename(P) \tbisim \rename(Q)$, $\forall L \mathbin\subseteq U \mathbin\subseteq A,\;\theta_L^U(P) \tbisim \theta_L^U(Q)$ and $\forall X \subseteq A,\; \psi_X(P) \tbisim \psi_X(Q)$;
        \item if $\equa$ is a recursive specification with $z \in V_\equa$ and $\rho, \nu \in V\setminus V_\equa \rightarrow \closed$ are substitutions such that $\forall x \in V\setminus V_\equa,\; \rho(x) \tbisim \nu(x)$, then $\langle z|\equa\rangle[\rho] \tbisim \langle z|\equa\rangle[\nu]$;
        \item if $\equa$ and $\equa'$ are recursive specifications and $x \in V_\equa=V_{\equa'}$ with $\langle x|\equa\rangle, \langle x|\equa'\rangle \in \closed$ such that $\forall y \in V_\equa,\; \equa_y \bisimrtbr \equa'_y$, then $\langle x|\equa\rangle \tbisim \langle x|\equa'\rangle$.
    \end{itemize}
    Since ${\bisimrtbr} \subseteq {\tbisim}$, it suffices to prove that $\tbisim$ is a rooted \tb time-out bisimulation up to $\bisimtbr$\,, as done in Appendix \ref{app:congruence}. This implies ${\tbisim} = {\bisimrtbr}$\, and the definition will then give us that $\bisimrtbr$ is a lean congruence. Moreover, the last condition of $\tbisim$ adds that it is a full congruence. A similar proof yields the result for $\bisimrtb$\,.
\end{proof}

\section{Axiomatisation} \label{sec:axiom}

We will provide complete axiomatisations for $\bisimrtbr$ and $\bisimrtb$ on various fragments of $\ccsp$.

\subsection{Recursive Principles} \label{subsec:recursive}

The expression $\langle x |\equa\rangle$ is intuitively defined as the $x$-component of the solution of $\equa$. However, $\equa$ could perfectly well have multiple solutions that are not bisimilar to each other. For instance, take $\equa = \{x \mathbin= x\}$; any expression is an $x$-component of a solution of $\equa$.
For our complete\linebreak[3] axiomatisation, we need to restrict attention to recursive specifications which have a unique solution with respect to our notion of bisimilarity. This property can be decomposed into two principles \cite{BW90,Fok00}: the \emph{recursive definition principle} (RDP) states that a system of recursive equations has at least one solution and the \emph{recursive specification principle} (RSP) that it has at most one solution. The latter holds under a condition traditionally called \emph{guardedness}.

\begin{definition}\rm \label{def:solution}
    Let $\equa$ be a recursive specification and ${\sim} \subseteq \closed\times\closed$, a \emph{solution up to $\sim$} of $\equa$ is a substitution $\rho \in \expr^{V_\equa}$ such that $\rho \sim \equa[\rho]$. Here $\rho$ and ${\equa} \in\expr^{V_\equa}$ are seen as $V_\equa$-tuples.
\end{definition}

\noindent
In \cite{BW90,Fok00} RDP was proven for the classical notion of strong bisimilarity $\bisim$. Since $\bisimrtbr$ and $\bisimrtb$ are included in $\bisim$, it holds for both of these relations as well.

\begin{proposition}[RDP] \label{prop:rdp}
    Let $\equa$ be a recursive specification. The substitution $\rho: x \mapsto \langle x|\equa\rangle$ for all $x \mathbin\in V_\equa$ is a solution of $\equa$ up to $\bisim$. It is called the \emph{default solution} of $\equa$.
\end{proposition}

\noindent
An occurrence of a variable $x$ in an expression $E$ is \emph{well-guarded} if $x$ occurs in a subexpression $a.F$ of $E$, with $a\mathbin\in A$. Here we do not allow $\tau$ and $\rt$ as guards. An expression $E$ is \emph{well-guarded} if no operator $\tau_I$ occurs in $E$ and all free occurrences of variables in $E$ are well-guarded.\linebreak[3]
A recursive specification $\equa$ is \emph{manifestly well-guarded} if no operator $\tau_I$ occurs in $\equa$ and for all $x,y\in V_\equa$ all occurrences of $x$ in the expression $\equa_y$ are well-guarded; it is \emph{well-guarded} if it can be made manifestly well-guarded by repeated substitution of $\equa_y$ for $y$ within terms $\equa_x$.\linebreak[3]
A $\ccsp$ process $P\in\closed$ is \emph{guarded} if each recursive specification occurring in $E$ is well-guarded. It is \emph{strongly guarded} if moreover there is no infinite path of $\tau$ and $\rt$-transitions $P_0 \step{\alpha_1} P_1 \step{\alpha_2} P_2 \step{\alpha_1} \dots$ with $\alpha_i\mathbin\in \{\tau,\!\rt\}$ for all $i\mathbin>0$, starting in a state $P_0$ reachable from $P\!$.

\begin{proposition}[RSP] \label{prop:rsp}
    Let $\equa$ be a well-guarded recursive specification and $\rho, \nu \mathbin\in \expr^{V_\equa}\!\!$. If $\rho$ and $\nu$ are solutions of $\equa$ up to $\bisimrtbr$ (or $\bisimrtb$) then $\rho \bisimrtbr \nu$ (resp.\ $\rho \bisimrtb \nu$).
\end{proposition}

\begin{proof}
    Modifying $\equa$ by substituting $\equa_y$ for $y$ within terms $\equa_x$ with $x,y\in V_\equa$ does not affect the set of its solutions. Hence we can restrict attention to manifestly well-guarded $\equa$.

    Thanks to the composition of substitutions, it suffices to prove the proposition when $\rho, \nu \in \closed^{V_\equa}$ and only variables of $V_\equa$ can occur in $\equa_x$ for $x \in V_\equa$. It suffices to show that the symmetric closure of \({\tbisim} := \{(H[\equa[\rho]],H[\equa[\nu]]) \mid H \in \expr\) is without $\tau_I$ operators and with free variables from $V_\equa \mbox{ only}\}$ is a rooted \tb time-out bisimulation up to $\bisimtbr$\,. Here $\equa[\rho]\in\closed^{V_\equa}$ is seen as a substitution. Details can be found in Appendix \ref{app:RSP}. An almost identical strategy can be applied to get RSP for $\bisimrtb$\,.
\end{proof}

\noindent
The following lemma, whose proof can be found in Appendix \ref{app:contraintuitive}, states that, when considering strongly guarded processes, eliding a time-out is independent of the set of allowed actions. 

\begin{lemma} \label{lem:independent}
    Let $P$ be a strongly {guarded} $\ccsp$ process and $X \subseteq A$. If $\deadend{P}{X}$, $P \step{\rt} P'$ and $\theta_X(P) \bisimtbr \theta_X(P')$ then $\forall Y \subseteq A,\; \deadend{P}{Y} \Rightarrow \theta_Y(P) \bisimtbr \theta_Y(P')$.
\end{lemma}

\noindent
Actually, this lemma holds because our restriction of strong guardedness is too strong. Indeed, the equation $x = t.(a + \tau.x)$ has a single solution, but it is not well-guarded. The process $P= \langle x \mid \{x = t.(a + \tau.x)\}\rangle$ is not guarded, yet satisfies \raisebox{0pt}[0pt][0pt]{$P \step\rt P' := a + \tau.P$} and  \mbox{$\theta_\emptyset(P) \bisimtbr \theta_\emptyset(P')$}, while $\theta_{\{a\}}(P) \,\not\!\bisimtbr \theta_{\{a\}}(P')$. Even if we write $P$ as $\tau_{\{b\}}(\langle x \mid \{x = t.(a + b.x)\}\rangle$ it fails to be strongly guarded. This restriction was kept because being more precise is very challenging. For instance, the equation $x = t.(a + \tau.x) + t.a$ has multiple solutions: the default one, $\langle x \mid \{x = t.(a + \tau.x) + t.a + t.(a + t.b)\}\rangle$ and others. Notice that adding a branch $t.a$ to an equation with one solution can lead to it having multiple ones. Intuitively, there are situations where time-out contraction enables to hide the existence of other time-outs. Characterising these situations requires the use of semantic conditions that are difficult to verify, thus, making them undesirable. Moreover, applying the Pohlmann encoding to the processes in order to, then, use the axiomatisation of $\bisimrtb$ leads to the similar complications. This limitation deserves to be studied properly because it will appear for all bisimilarities authorising time-out contraction.

\subsection{Axioms and Soundness} \label{subsec:soundness}

The set of axioms provided is composed of the axiomatisation of $\bisim_r$ \cite{strongreactivebisimilarity}, together with three branching axioms. The \emph{branching axiom} is well-known since it is used in the axiomatisation of rooted branching bisimilarity \cite{branching}. The \emph{$\rt$-branching axiom} and the \emph{$\tau/\rt$-branching axiom} are newly introduced; they are the adaptation of the branching axiom to time-out contraction.

\renewcommand{\arraystretch}{1.2}
\begin{table}[ht]
    \centering
    \begin{tabular}{|l l l|}
        \hline
        $x+(y+z) = (x+y)+z$ & $\tau_I(x+y) = \tau_I(x) + \tau_I(y)$ & $\rename(x+y) = \rename(x) + \rename(y)$ \\
        $x+y = y+x$ & $\tau_I(\alpha.x) = \alpha.\tau_I(x)$ if $\alpha \not\in I$ & $\rename(\tau.x) = \tau.\rename(x)$\\
        $x+x = 0$ & $\tau_I(\alpha.x) = \tau.\tau_I(x)$ if $\alpha \in I$ & $\rename(\rt.x) = \rt.\rename(x)$ \\
        $x+0 = x$ & & $\rename(a.x) = \sum_{\{b \mid \rename(a,b)\}}b.\rename(x)$ \\
        \multicolumn{3}{|l|}{\textbf{Expansion Theorem:} if $P = \sum\limits_{i \in I}\alpha_i.P_i$ and $Q = \sum\limits_{j \in J} \beta_j.Q_j$ then} \\
        \multicolumn{3}{|c|}{$P \parallel_S Q = \sum\limits_{i \in I,\alpha_i \not\in S}(\alpha_i.P_i \parallel_S Q) + \sum\limits_{j \in J, \beta_j \not\in S}(P \parallel_S \beta_j.Q_j) + \sum\limits_{i \in I, j \in J, \alpha_i = \beta_j \in S}\alpha_i.(P_i \parallel_S Q_j)$} \\
        \multicolumn{3}{|c|}{$\alpha.(\tau.(x+y)+x) = \alpha.(x+y)$ \quad\textbf{(Branching Axiom)}} \\
        \multicolumn{3}{|c|}{$\alpha.(\rt.(x+\sum_{i \in I}\rt.y_i)+x) = \alpha.(x + \sum_{i \in I}\rt.y_i)$ \quad\textbf{($\rt$-Branching Axiom)}} \\
        \multicolumn{3}{|c|}{$\alpha.(\tau.(x+y)+\rt.(x+y)+x) = \alpha.(x+y)$ \quad\textbf{($\tau/\rt$-Branching Axiom)}} \\
        $\langle x |\equa\rangle = \langle \equa_x |\equa\rangle$ \quad\textbf{(RDP)} & \multicolumn{2}{c|}{$\equa \Rightarrow x = \langle x | \equa\rangle$ \quad with $\equa$ well-guarded\quad\textbf{(RSP)}} \\
        \hline
        \hline
        \multicolumn{2}{|l}{$\theta_L^U(\sum_{i \in I}\alpha_i.x_i) = \sum_{i\in I} \alpha_i.x_i$} & $(\forall i \in I, \alpha_i \not\in L\cup\{\tau\})$ \\
        \multicolumn{2}{|l}{$\theta_L^U(x + \alpha.y + \beta.z) = \theta_L^U(x + \alpha.y)$} & $(\alpha \in L\cup\{\tau\} \wedge \beta \not\in U\cup\{\tau\})$ \\
        \multicolumn{2}{|l}{$\theta_L^U(x + \alpha.y + \beta.z) = \theta^U_L(x + \alpha.y) + \theta_L^U(\beta.z)$} & $(\alpha \in L\cup\{\tau\} \wedge \beta\in U\cup\{\tau\})$ \\
        \multicolumn{2}{|l}{$\theta_L^U(\alpha.x) = \alpha.x$} & $(\alpha \neq \tau)$ \\
        $\theta_L^U(\tau.x) = \tau.\theta_L^U(x)$ & & \\
        \multicolumn{2}{|l}{$\psi_X(x + \alpha.y) = \psi_X(x) + \alpha.y$} & $(\alpha \not\in X\cup\{\tau,t\})$ \\
        \multicolumn{2}{|l}{$\psi_X(x + \alpha.y + \rt.z) = \psi_X(x + \alpha.y)$} & $(\alpha \in X\cup\{\tau\})$ \\
        \multicolumn{2}{|l}{$\psi_X(x + \alpha.y + \beta.z) = \psi_X(x + \alpha.y) + \beta.z$} & $(\alpha, \beta \in X\cup\{\tau\})$ \\
        \multicolumn{2}{|l}{$\psi_X(\alpha.x) = \alpha.x$} & $(\alpha \neq \rt)$ \\
        $\psi_X(\sum_{i \in I}\rt.y_i) = \sum_{i \in I}\rt.\theta_X(y_i)$ & & \\
        \hline
        \hline
        \multicolumn{3}{|c|}{$(\forall X \subseteq A,\; \psi_X(x) = \psi_X(y)) \Rightarrow x= y$ \quad\textbf{(Reactive Approximation Axiom)}} \\
        \hline
    \end{tabular}
\vspace{2ex}
    \caption{Axiomatisation of $\bisimrtbr$ and $\bisimrtb$}
    \label{tab:axioms}
\end{table}

Let $Ax^\infty$ be the set of all axioms in the first two rectangles in Table \ref{tab:axioms} and $Ax := Ax^\infty \setminus \{\mbox{RDP},\mbox{RSP}\}$. Let $Ax^\infty_r$ be the set of all axioms in Table \ref{tab:axioms} except the $\tau/\rt$-branching one and $Ax_r := Ax^\infty_r \setminus \{\mbox{RDP},\mbox{RSP}\}$. The $\tau/\rt$-branching axiom is removed from $Ax^\infty_r$ because the law \hypertarget{Lt}{$\mbox{\bf L}\tau$: $\tau.x + \rt.y = \tau.x$} can be derived from the reactive approximation axiom \cite{strongreactivebisimilarity}, and applying $\mbox{\bf L}\tau$ to the branching axiom yields the $\tau/\rt$-branching axiom, thus making it redundant.

\begin{proposition} \label{prop:soundness}
    Let $P,Q$ be two $\ccsp$ processes.
    \begin{itemize}
        \item If $Ax^\infty \vdash P = Q$ then $P \bisimrtb Q$.
        \item If $Ax^\infty_r \vdash P = Q$ then $P \bisimrtbr Q$.
    \end{itemize}
\end{proposition}

\begin{proof}
    Since $\bisimrtbr$ and $\bisimrtb$ are congruences, it suffices to prove that each axiom is sound, meaning that replacing, in each axiom, $=$ by the desired bisimilarity and each variable by a process produces a true statement. Most of these axioms were proven to be sound for the classical notion $\bisim$ of strong bisimilarity \cite{Mi90ccs} in \cite{strongreactivebisimilarity}. Thus, since both $\bisimrtbr$ and $\bisimrtb$ are included in $\bisim$, most of them are sound for $\bisimrtbr$ and $\bisimrtb$\,. 
    
    Only the branching axioms, RSP and the reactive approximation axiom remain to be proven sound. The soundness of the branching axioms is trivial and the soundness of RSP is exactly Proposition \ref{prop:rsp}. For the reactive approximation axiom, it suffices to show that ${\tbisim} := \; \bisimrtbr \cup \{(P,Q),(Q,P) \mid \forall X \subseteq A, \psi_X(P) \bisimrtbr \psi_X(Q)\}$ is a rooted \tb time-out bisimulation, as done in Appendix~\ref{app:RA}.
\end{proof}

\subsection{Completeness}

A well-known feature of most process algebras is that the standard collection of axioms allows one to bring any {guarded} process expression in the following normal form \cite{BW90,Fok00}.

\begin{definition}\rm \label{def:head-normal form}
    Let $P$ be a {guarded} $\ccsp$ process. The \emph{head-normal form} of $P$ is $\hat{P} := \sum_{\{(\alpha,Q) \mid P\step{\alpha}Q\}}\alpha.Q$.
\end{definition}

\noindent
In \cite{strongreactivebisimilarity}, it is proven that the axiomatisation of $\bisim_r$ enables one to equate any {guarded} process with its head-normal form (using a definition of guardedness that is more liberal than the one employed here, with $\tau$ and $\rt$ allowed as guards). Since the axiomatisation of $\bisim_r$ is included in $Ax^\infty$ and $Ax^\infty_r$, this yields the property for them as well.

\begin{lemma} \label{lem:head-normal form}
    Let $P$ be a guarded $\ccsp$ process. Then $Ax^\infty \vdash P = \hat{P}$ and $Ax^\infty_r \vdash P = \hat{P}$. Moreover, $Ax$ or $Ax_r$ are sufficient if $P$ is recursion-free.
\qed
\end{lemma}

\noindent
This lemma is used extensively in the proof of the following completeness results.

\begin{proposition} \label{prop:collapse}
    Let $P,Q$ be two recursion-free $\ccsp$ processes. If $P \bisimtbr Q$ (resp.\ $P \bisimtb Q$) then, for all $\alpha \in Act$, $Ax_r \vdash \alpha.\hat{P} = \alpha.\hat{Q}$ (resp.\ $Ax \vdash \alpha.\hat{P} = \alpha.\hat{Q}$).
\end{proposition}

\begin{proof}
    The \emph{depth} $d(p)$ of a process $P$ is the length of the longest path starting from $P$. Note that it is properly defined for recursion-free processes only. The proof proceeds by induction on $\max(d(P),d(Q))$. The technique is fairly standard and the details can be found in Appendix \ref{app:completeness finite}.
\end{proof}

\begin{theorem} \label{thm:completeness finite}
    Let $P,Q$ be two recursion-free $\ccsp$ processes. If $P \bisimrtbr Q$ (resp.\ $P \bisimrtb Q$) then $Ax_r \vdash P = Q$ (resp.\ $Ax \vdash P = Q$).
\end{theorem}

\begin{proof}
    It suffices to express both processes in their head-normal form and then to equate each pair of matching branches using Proposition \ref{prop:collapse}. Details are in Appendix \ref{app:completeness finite}.
\end{proof}

\noindent
The following theorem lifts this result for $\bisimrtb$ from finite (recursion-free) processes to arbitrary (infinite) ones, subject to the restriction of strong guardedness.

\begin{theorem} \label{thm:completeness}
    Let $P,Q$ be strongly guarded $\ccsp$ processes.\\ If $P \bisimrtb Q$ then $Ax^\infty \vdash {P} = {Q}$.
\end{theorem}

\begin{proof}
A well-known technique called \emph{equation merging} can be applied. Details can be found in Appendix \ref{app:completeness}.
\end{proof}

\subsection{Canonical Representative} \label{subsec:canonical}

Unfortunately, equation merging does not work on reactive bisimulations \cite{strongreactivebisimilarity}. Thus, another technique is used \cite{GF20,LY20}, called \emph{canonical representatives}. The idea is to build the simplest process for each equivalence class of $\bisimrtbr$ and use them as intermediary to equate processes.

Let us denote with $\closed^g$ the strongly guarded fragment of $\closed$. For all $P \in \closed^g$, $[P] := \{Q \in \closed^g \mid P \bisimtbr Q\}$ is the $\bisimtbr$-equivalence class of $P$. $[\closed^g]$ denotes the set of all $\bisimtbr$-equivalence classes. Using the axiom of choice, a choice function $\chi: [\closed^g] \rightarrow \closed^g$ can be defined such that $\forall R \in [\closed^g], \chi(R) \in R$. A transition relation can be defined between $\bisimtbr$-equivalence classes:
\begin{align*}
    \forall \alpha \in A_\tau, (R \step{\alpha} R' \Leftrightarrow\; & \chi(R) \pathtau P_1 \step{\alpha} P_2 \wedge P_1\in R \wedge P_2\in R' \wedge (\alpha \in A \vee R \ne R')) \\
    R \step{\rt} R' \Leftrightarrow\; & \exists X \mathbin\subseteq A,r\mathbin>0, \chi(R) \pathtau P_1 \step{\rt} P_2 \pathtau P_3 \step{\rt} ... \pathtau P_{2r{-}1} \step{\rt} P_{2r} \\
    & \wedge \forall i \in [0,r{-}1], \theta_X(P_{2i}) \in [\theta_X(\chi(R))] \wedge \deadend{P_{2i+1}}{X} \\
    &\wedge P_1 \in R \wedge P_{2r} \in R' \wedge [\theta_X(\chi(R))] \ne [\theta_X(\chi(R'))]
\end{align*}

\noindent
All bisimulations can be extended to $\bisimtbr$-equivalence classes. It suffices to consider the set of states $\closed^g \uplus [\closed^g] \uplus \{\theta_X([P]) \mid X \subseteq A \wedge P\in \closed^g\}$.

\begin{proposition} \label{prop:class}
    Let $P \in \closed^g$, $P \bisimtbr[] [P]$.
\end{proposition}

\begin{proof}
    It suffices to prove that $\tbisim := \{(P,[P]),([P],P) \mid P \in \closed^g\}$ is a branching time-out bisimulation up to $\bisimtbr$\,. Details can be found in Appendix \ref{app:canonical rep}.
\end{proof}

\begin{definition}\rm \label{def:canonical representative}
    Let $P, Q \in \closed^g$, the \emph{canonical representative} of $P$ and $Q$ is a recursive specification $\equa$ such that $V_{\equa} := \{x_P,x_Q\} \cup \{x_R \mid R \in \bigcup_{P' \in \reach{P} \cup \reach{Q}}\reach{[P']}\}$, and $\forall R \in \bigcup_{P' \in \reach{P}\cup \reach{Q}}\reach{[P']}$,
    \begin{align*}
        \equa_{x_P} :=  \hspace{-10pt}\sum_{\{(\alpha,P') \mid P \step{\alpha} P'\}} \hspace{-10pt}\alpha.x_{[P']} \mbox{ ; } \equa_{x_Q} := \hspace{-10pt}\sum_{\{(\alpha,Q') \mid Q \step{\alpha} Q'\}} \hspace{-10pt}\alpha.x_{[Q']} \mbox{ and } \equa_{x_R} :=  \hspace{-10pt}\sum_{\{(\alpha,R') \mid R \step{\alpha} R'\}} \hspace{-10pt}\alpha.x_{R'}
    \end{align*}
\end{definition}

\noindent
The canonical representative is well-defined since $P$, $Q$, as well as processes $[P']\in[\closed^g]$ are finitely branching~\cite{strongreactivebisimilarity}. Additionally, $\bigcup_{P' \in \reach{P}\cup\reach{Q}}\reach{[P']}$ is countable. Moreover, $\equa$ is strongly guarded. Furthermore, \hypertarget{recall}{by construction $\langle x_R |\equa\rangle \bisim R$} for all $R \in \bigcup_{P' \in \reach{P}\cup\reach{Q}}\reach{[P']}$. 

\begin{proposition} \label{prop:canonical}
    Let $P,Q \in \closed^g$ and $\equa$ be the canonical representative of $P$ and $Q$. $Ax^\infty_r \vdash P = \langle x_P |\equa\rangle$.
\end{proposition}

\begin{proof}
  It suffices to show that $P$ and $\langle x_P|\equa\rangle$ are $y_P$-components of solutions of
  $\{y_{P^\dag} = \sum_{\{(\alpha,P^\ddag) \mid P^\dag\step{\alpha} P^\ddag\}}\alpha.y_{P^\ddag} \mid P^\dag \in \reach{P}\}$.
  Details can be found in Appendix \ref{app:canonical}.
\end{proof}

\begin{theorem} \label{thm:canonical}
    Let $P,Q \in \closed^g$. If $P \bisimrtbr Q$ then $Ax^\infty_r \vdash P = Q$.
\end{theorem}

\begin{proof}
    It suffices to equate $\langle x_P|\equa\rangle$ and $\langle x_Q|\equa\rangle$ using RDP and the reactive approximation axiom. Details can be found in Appendix \ref{app:canonical}.
\end{proof}

\section*{Conclusion}

This paper defined a form of branching bisimilarity for processes with time-out transitions, and provided a modal characterisation, congruence results, and a complete axiomatisation for strongly {guarded} processes. The restriction to strongly {guarded} processes is rather severe; it rules out processes that may engage in an infinite sequence of time-out transitions, interspersed with $\tau$s. Relaxing this restriction is a suitable topic for further work. Another task is to combine this work with the ideas behind \emph{justness} \cite{GH19}, a weaker form of fairness that allows the formulation and derivation of useful liveness properties. In a setting with time-outs, justness would demand that once a parallel component reaches a state in which a time-out transition is enabled, it cannot stay in that state forever after.

As an example of the use of branching reactive bisimulation, one could verify the correctness of a non-trivial system, such as Peterson's mutual exclusion protocol, as modelled in \cite{vG23a}. There it was argued that a similar model without time-out transitions is not possible. The model from \cite{vG23a} features eight visible actions of entering or leaving the critical or non-critical section of process A or B\@. Abstracting from all actions pertaining to process B yields a protocol that only deals with process A, and a correctness claim could be validated by showing it branching reactive bisimilar with a simple specification of the intended behaviour of A that would apply when B were not around. Although doing such a verification is entirely feasible, for now, it can not be achieved by algebraic means, using our complete axiomatisation. The reason is that abstraction from process B yields infinite sequences of unobservable actions, which are currently not covered by our work.

\bibliography{biblio}

\input{appendix_CONCUR}

\input{appendix}

\end{document}

%% file: appendix_CONCUR.tex
\newpage
\appendix

\section{Examples} \label{app:examples}

\paragraph*{Scope of the First Clause of Definition \ref{def:intuitive}}

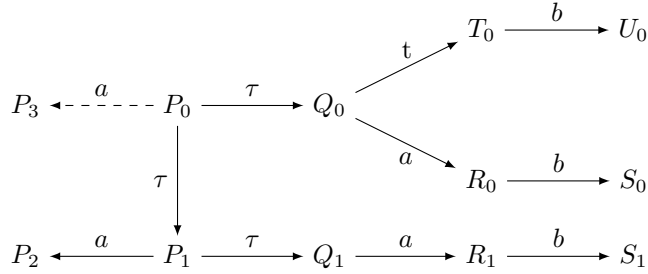
\begin{figure}[ht]
    \centering
    \begin{tikzpicture}
        \node (a) at (0,0) {$P_0$};
        \node (b) at (0,-2) {$P_1$};
        \node (c) at (-2,-2) {$P_2$};
        \node (c1) at (-2,0) {$P_3$};
        \node (d) at (2,0) {$Q_0$};
        \node (e) at (4,1) {$T_0$};
        \node (e1) at (6,1) {$U_0$};
        \node (f) at (4,-1) {$R_0$};
        \node (g) at (6,-1) {$S_0$};
        \node (d1) at (2,-2) {$Q_1$};
        \node (f1) at (4,-2) {$R_1$};
        \node (g1) at (6,-2) {$S_1$};

        \draw[->, >=latex] (a) -- node[midway,left]{$\tau$} (b);
        \draw[->, >=latex, dashed] (a) -- node[midway,above]{$a$} (c1);
        \draw[->, >=latex] (b) -- node[midway,above]{$a$} (c);
        \draw[->, >=latex] (a) -- node[midway,above]{$\tau$} (d);
        \draw[->, >=latex] (d) -- node[midway,above]{$\rt$} (e);
        \draw[->, >=latex] (e) -- node[midway,above]{$b$} (e1);
        \draw[->, >=latex] (d) -- node[midway,below]{$a$} (f);
        \draw[->, >=latex] (f) -- node[midway,above]{$b$} (g);
        \draw[->, >=latex] (b) -- node[midway,above]{$\tau$} (d1);
        \draw[->, >=latex] (d1) -- node[midway,above] {$a$} (f1);
        \draw[->, >=latex] (f1) -- node[midway,above]{$b$} (g1);
    \end{tikzpicture}
    \caption{Counter-Example to a Naive Clause 1.a.}
    \label{fig:counter-example 1a}
\end{figure}

\noindent
In Figure \ref{fig:counter-example 1a}, the process $a.0 + \tau.(\rt.b.0 + a.b.0) + \tau.(\tau.a.b.0 + a.0)$ is represented as an LTS\@. Let $A := \{a,b\}$. Removing the dashed $a$-transition generates the process $\tau.(\rt.b.0 + a.b.0) + \tau.(\tau.a.b.0 + a.0)$.

First, we are going to show that these two processes are not branching reactive bisimilar. Let's try to build a branching reactive bisimulation between them. The only way to match the dashed $a$-transition of $a.0 + \tau.(\rt.0 + a.b.0) + \tau.(\tau.a.b.0 + a.0)$ is by the $a$-transition between $P_1$ and $P_2$, because all other $a$-transitions are followed by a $b$-transition. This requires to elide the $\tau$-transition between $P_0$ and $P_1$, who must be branching reactive bisimilar. Since $P_0 \bisimtbr P_1$, when considering the $\tau$-transition between $P_0$ and $Q_0$, $Q_0$ has to be branching reactive bisimilar to $P_1$ or $Q_1$. Now, the $a$-transition between $Q_0$ and $R_0$ has to be matched by the $a$-transition between $Q_1$ and $R_1$ because of the following $b$-transition. This implies $Q_0 \bisimtbr Q_1$, thus, $Q_0 \bisimtbr[\emptyset] Q_1$. One has $\deadend{Q_0}{\emptyset}$ and $Q_0 \step{\rt} T_0$, i.e., when the environment temporary allows no visible actions, $Q_0$ can time-out into a state in which $b$ is possible. This behaviour cannot be matched by $Q_1$---a contradiction.

Now, consider the alternative to Definition \ref{def:intuitive} in which the first clause has been changed to 
\begin{enumerate}
    \item \begin{enumerate}
        \item if $P \steptau P'$ then there is a path $Q \pathtau Q_1 \step{\opt{\tau}} Q_2$ with $\R(P,Q_1)$ and $\R(P',Q_2)$.
    \end{enumerate}
\end{enumerate}
In other words, the scope of the first clause is restricted to $\tau$-transitions. This modification enables building a bisimulation between the two processes. Indeed, the dashed $a$-transition is only considered when the environment allows $a$. Thus, it is sufficient to get $P_0 \bisimtbr[A] P_1$ and $P_0 \bisimtbr[\{a\}] P_1$ and not $P_0 \bisimtbr P_1$ anymore. Therefore, it is sufficient to match $Q_0$ and $Q_1$ in environments allowing $a$. As a result, the outgoing time-out transition of $Q_0$ is never considered when matching $Q_0$ with $Q_1$, solving our previous issue. Once this observation is made, building the bisimulation is trivial.

Finally, place both processes in the context $\_\!\_ \parallel_{\{a\}} (\tau.0 + a.0)$. It behaves like a one-way switch enabling to block all $a$-transitions forever as soon as the $\tau$-transition is performed. Let's try to build a branching reactive bisimulation between the two processes. Following the same reasoning as before, it is necessary to get $P_0 \parallel_{\{a\}} (\tau.0 + a.0) \bisimtbr[A] P_1 \parallel_{\{a\}} (\tau.0 + a.0)$ because of the dashed $a$-transition, and then $Q_0 \parallel_{\{a\}} (\tau.0 + a.0) \bisimtbr[A] Q_1 \parallel_{\{a\}} (\tau.0 + a.0)$ because of the $a$-transition between $Q_0$ and $R_0$. Note that $Q_0 \parallel_{\{a\}} (\tau.0 + a.0) \steptau Q_0 \parallel_{\{a\}} 0\linebreak[2] \step{\rt} T_0 \parallel_{\{a\}} 0 \step{b} U_0 \parallel_{\{a\}} 0$ and $\deadend{Q_0 \parallel_{\{a\}} 0}{A}$. As before, $Q_0 \parallel_{\{a\}} (\tau.0 + a.0)$ can time-out into a state in which $b$ is executable, whereas this behaviour is impossible in $Q_1 \parallel_{\{a\}} (\tau.0 + a.0)$. As a result, restricting the scope of the first clause of Definition \ref{def:intuitive} to $\tau$-transitions prevents $\bisimtbr$ from being a congruence for parallel composition.

\paragraph*{Necessity of the Stability Respecting Clause}

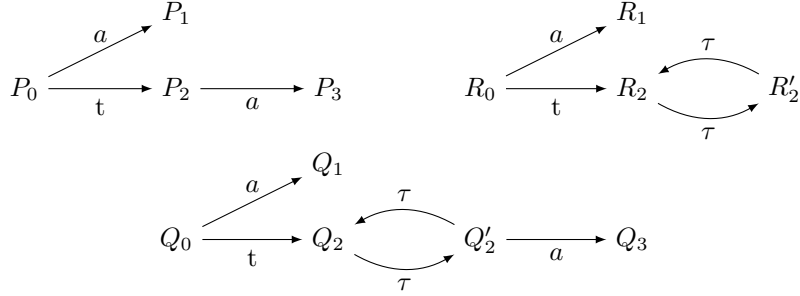
\begin{figure}[ht]
    \centering
    \begin{tikzpicture}
        \node (a) at (0,0) {$P_0$};
        \node (b) at (2,0) {$P_2$};
        \node (c) at (2,1) {$P_1$};
        \node (d) at (4,0) {$P_3$};
        \draw[->, >=latex] (a) -- node[midway,below]{$\rt$} (b);
        \draw[->, >=latex] (a) -- node[midway,above]{$a$} (c);
        \draw[->, >=latex] (b) -- node[midway,below]{$a$} (d);
        \node (a) at (2,-2) {$Q_0$};
        \node (b) at (4,-2) {$Q_2$};
        \node (b') at (6,-2) {$Q'_2$};
        \node (c) at (4,-1) {$Q_1$};
        \node (d) at (8,-2) {$Q_3$};
        \draw[->, >=latex] (a) -- node[midway,below]{$\rt$} (b);
        \draw[->, >=latex] (a) -- node[midway,above]{$a$} (c);
        \draw[->, >=latex] (b) to[bend right] node[midway,below]{$\tau$} (b');
        \draw[->, >=latex] (b') to[bend right] node[midway,above]{$\tau$} (b);
        \draw[->, >=latex] (b') -- node[midway,below]{$a$} (d);
        \node (a) at (6,0) {$R_0$};
        \node (b) at (8,0) {$R_2$};
        \node (b') at (10,0) {$R'_2$};
        \node (c) at (8,1) {$R_1$};
        \draw[->, >=latex] (a) -- node[midway,below]{$\rt$} (b);
        \draw[->, >=latex] (a) -- node[midway,above]{$a$} (c);
        \draw[->, >=latex] (b) to[bend right] node[midway,below]{$\tau$} (b');
        \draw[->, >=latex] (b') to[bend right] node[midway,above]{$\tau$} (b);
    \end{tikzpicture}
    \caption{Counter-Example to the Absence of a Stability Respecting Clause}
    \label{fig:example stability}
\end{figure}

\noindent
In Figure \ref{fig:example stability}, three processes are represented as LTSs. Take $A := \{a\}$. According to Definition~\ref{def:intuitive}, $\neg(P_0 \bisimtbr Q_0)$ and $Q_0 \bisimtbr R_0$. 

Let's try to build a branching reactive bisimulation between the top-left and bottom processes. Matching the time-out between $Q_0$ and $Q_2$ implies that $Q_2 \bisimtbr[\emptyset] P_0$ or $Q_2 \bisimtbr[\emptyset] P_2$. However, $P_0 \nsteptau$ and $P_2 \nsteptau$, thus, there should be a path $Q_2 \pathtau Q_2' \nsteptau$, but this is not the case.

The symmetric closure of
\begin{align*}
    \R := \{(Q_0,R_0), (Q_1,R_1), (Q_2,\emptyset,R_2), (Q_2',\emptyset,R_2')\} \cup \{(Q_0,X,R_0),(Q_1,X,R_1) \mid X \subseteq A\}
\end{align*}
is a branching reactive bisimulation. The $a$-transition between $Q_2'$ and $Q_3$ does not have to be matched since $Q_2'$ is considered only when the environment disallows $a$.

Now, suppose that the stability respecting condition is removed from Definition \ref{def:intuitive}. As a result, a branching reactive bisimulation can be built between the top-left and bottom processes. The symmetric closure of
\begin{align*}
    \R' := &  \{(P_0,Q_0),(P_1,Q_1),(P_2,Q_2),(P_2,Q_2'),(P_3,Q_3)\} \\
    & \cup \{(P_0,X,Q_0),(P_1,X,Q_1),(P_2,X,Q_2),(P_2,X,Q_2'),(P_3,X,Q_3) \mid X \subseteq A\}
\end{align*}
would be a branching reactive bisimulation. Moreover, $\R$ would still be a branching reactive bisimulation, since Definition \ref{def:intuitive} has merely been weakened. Therefore, according to the modified Definition \ref{def:intuitive}, $P_0 \bisimtbr Q_0$ and $Q_0 \bisimtbr R_0$. However, when trying to construct a branching reactive bisimulation between $P_0$ and $R_0$, because of the time-out transition, $R_2$ has to be matched to $P_0$ or $P_2$ and no $a$-transition is reachable from $R_2$; therefore, $\neg(P_0 \bisimtbr R_0)$. As a result, removing the stability respecting clause from Definition~\ref{def:intuitive} prevents $\bisimtbr$ from being an equivalence relation.

\section{Concrete Time-out Version} \label{app:concrete time-out}

Before studying $\bisimtbr$\,, we looked at another version which is not eliding any time-out transitions. More formally, it is defined by replacing Clause 2.d of Definition \ref{def:intuitive} by 
\begin{enumerate}
    \setcounter{enumi}{1}
    \item \begin{enumerate}
        \setcounter{enumii}{3}
        \item if $\deadend{P}{X}$ and $P \step{\rt} P'$ then there exists a path $Q \pathtau Q_1 \step{\rt} Q_2$ with $\R(P',X,Q_2)$.
    \end{enumerate}
\end{enumerate}
It is not necessary to require to match with an executable time-out (i.e. $\deadend{Q_1}{X}$) since this is implied by the other clauses. It is also implied that $\R(P,X,Q_1)$ in the above clause. This bisimilarity has properties similar to $\bisimtbr$\,, to be recapped below. No proof will be provided here since they rely on the same strategies and are actually simpler because of the absence of time-out omission. However, a technical report \cite{Reghem24} is available. In the remainder of this appendix, $\bisimtbrc$ stands for the concrete time-out version.

The stuttering lemma (Lemma \ref{lem:stuttering}) still holds and $\bisimtbrc$ and $(\bisimtbrc[X])_{X \subseteq A}$ are still equivalence relations (Proposition \ref{prop:equivalence}). The rooted version $\bisimrtbrc$ of $\bisimtbrc$ is exactly Definition \ref{def:rooted intuitive} and $\bisimrtbrc$ and $(\bisimrtbrc[X])_{X \subseteq A}$ are still equivalence relations (Proposition \ref{prop:rooted equivalence}). The Pohlmann encoding (Table~\ref{tab:Pohlmann operator}) is simplified as the rooted variants are no longer needed. If $\bisimsb$ stands for the classical stability respecting branching bisimulation \cite{branching,vG93}, and $\bisimrb$ for its rooted version, $P \bisimtbrc Q \Leftrightarrow \vartheta(P) \bisimsb \vartheta(Q)$; $P \bisimtbrc[X] Q \Leftrightarrow \vartheta_X(P) \bisimsb \vartheta_X(Q)$; $P \bisimrtbrc Q \Leftrightarrow \vartheta(P) \bisimrb \vartheta(Q)$ and $P \bisimrtbrc[X] Q \Leftrightarrow \vartheta_X(P) \bisimrb \vartheta_X(Q)$.

The generalised definition of $\bisimtbrc$ can be obtained by replacing Clause 1.b. and 2.c. in  Definition \ref{def:generalised} by
\begin{enumerate}
    \item \begin{enumerate}
        \setcounter{enumii}{1}
        \item If $\deadend{P}{X}$ and $P \step{\rt} P'$ then there exists a path $Q \pathtau Q_1 \step{\rt} Q_2$ with $\R(P',X,Q_2)$
    \end{enumerate}
    \item \begin{enumerate}
        \setcounter{enumii}{2}
        \item If $\deadend{P}{X\cup Y}$ and $P \step{\rt} P'$ then there exists a path $Q \pathtau Q_1 \step{\rt} Q_2$ with $\R(P',Y,Q_2)$
    \end{enumerate}
\end{enumerate}
The rooted generalised version is exactly Definition \ref{def:generalised rooted} and they induce the same bisimilarities as the previous definitions (Proposition \ref{prop:generalised}). In
the modal characterisation, $X \varphi$ is not useful anymore, nor
$\varphi\langle\epsilon_X\rangle\varphi'$. Replacing the fifth induction rule of $\logic_b$ by
$\langle\epsilon\rangle\langle \rt_X\rangle\varphi$ yields the counterpart of Theorem~\ref{thm:modal characterisation}.

The corresponding time-out bisimulation can be obtained by replacing Clause 2.\ of Definition \ref{def:time-out bisim} by
\begin{enumerate}
    \setcounter{enumi}{1}
    \item if $\deadend{P}{X}$ and $P \step{\rt} P'$ then there exists a path $P \pathtau P_1 \step{\rt} P_2$ with $\theta_X(P') \tbisim \theta_X(Q_2)$.
\end{enumerate}
The rooted time-out bisimulation is exactly Definition \ref{def:rooted time-out bisim} and they agree with the previous definitions (Proposition \ref{prop:time-out bisim}). $\bisimtbrc$ is a congruence for prefixing, parallel composition, abstraction, renaming and the operator $\theta_L^U$ (Proposition \ref{prop:stability}). $\bisimrtbrc$ is a full congruence (Theorem~\ref{thm:congruence}).

As ${\bisimrtbrc} \subseteq {\bisim}\,$, RDP holds for $\bisimrtbrc$\,. The definition of well-guarded recursion can be weakened by allowing $\rt$ as a guard and RSP holds for $\bisimrtbrc$ on processes that are guarded in this sense. Lemma \ref{lem:independent} is not useful anymore since time-out omissions are not considered. The set of all axioms of Table \ref{tab:axioms} except the $\rt$-branching and $\tau/\rt$-branching ones is a complete axiomatisation of $\bisimtbrc$ (Theorem \ref{thm:canonical}). Moreover, to obtain the complete axiomatisation of $\bisimtbrc$ on recursion-free processes, it suffices to remove RDP and RSP.

\section{Generalised \tb reactive bisimulation} \label{app:gbrb}

The second clause of Definition~\ref{def:intuitive} is quite tedious to check; thus, an equivalent definition of the bisimilarity would be useful. Actually, it is possible to define the exact same notion in a more general way at the cost of some clear motivations.

\begin{definition}\rm \label{def:generalised}
    A \emph{generalised \tb reactive bisimulation} is a symmetric relation $\R \subseteq (\closed\times\closed)\cup(\closed\times\mathcal{P}(A)\times\closed)$ such that, for all $P,Q \in \closed$ and $X \subseteq A$,
    \begin{enumerate}
        \item if $\R(P,Q)$
        \begin{enumerate}
            \item if $P \step{\alpha} P'$ with $\alpha \in A_\tau$ then there is a path $Q \pathtau Q_1 \step{\opt{\alpha}} Q_2$ with $\R(P,Q_1)$ and $\R(P',Q_2)$,
            \item if $\deadend{P}{X}$ and $P \step{\rt} P'$ then there is a path $Q = Q_0 \pathtau Q_1 \step{\rt} Q_2 \pathtau Q_3 \step{\rt} ... \pathtau Q_{2r{-}1} \step{\opt{\rt}} Q_{2r}$ with $r > 0$, such that $Q_1  \nsteptau$, $\forall i \in [1,r{-}1],\linebreak[3] \R(P,X,Q_{2i}) \wedge \deadend{Q_{2i+1}}{X}$ and $\R(P',X,Q_{2r})$,
            \item if $P \nsteptau$ then there exists a path $Q \pathtau Q_0 \nsteptau$;
        \end{enumerate}
        \item if $\R(P,X,Q)$
        \begin{enumerate}
            \item if $P \steptau P'$ then there is a path $Q \pathtau Q_1 \step{\opt{\tau}} Q_2$ with $\R(P,X,Q_1)$ and $\R(P',X,Q_2)$,
            \item if $P \step{a} P'$ with $a \in X \vee \deadend{P}{X}$ then there is a path $Q \pathtau Q_1 \step{a} Q_2$ with $\R(P,X,Q_1)$ and $\R(P',Q_2)$,
            \item if $\deadend{P}{(X\cup Y)}$ and $P \step{\rt} P'$ then there is a path $Q = Q_0 \pathtau Q_1 \step{\rt} Q_2 \pathtau Q_3 \step{\rt} ... \pathtau Q_{2r{-}1} \step{\opt{\rt}} Q_{2r}$ with $r > 0$, such that $Q_1  \nsteptau$, $\forall i \in [1,r{-}1],\linebreak[3] \R(P,Y,Q_{2i}) \wedge \deadend{Q_{2i+1}}{Y}$ and $\R(P',Y,Q_{2r})$,
            \item if $P \nsteptau$ then there is a path $Q \pathtau Q_0 \nsteptau$.
        \end{enumerate}
    \end{enumerate}
\end{definition}

\noindent
The strong point of the generalised definitions is the restriction on the use of triplets, making use of them only after performing a time-out. A generalised version of rooted \tb reactive bisimulation can be defined in a similar fashion.

\begin{definition}\rm \label{def:generalised rooted}
    A \emph{generalised rooted \tb reactive bisimulation} is a symmetric relation $\R \subseteq (\closed\times\closed)\cup(\closed\times\mathcal{P}(A)\times\closed)$ such that, for all $P,Q \in \closed$ and $X \subseteq A$,
    \begin{enumerate}
        \item if $\R(P,Q)$
        \begin{enumerate}
            \item if $P \step{\alpha} P'$ with $\alpha \in A_\tau$ then there is a transition $Q \step{\alpha} Q'$ such that $P' \bisimtbr Q'$,
            \item if $\deadend{P}{X}$ and $P \step{\rt} P'$ then there is a transition $Q \step{\rt} Q'$ with $P' \bisimtbr[X] Q'$,
        \end{enumerate}
        \item if $\R(P,X,Q)$
        \begin{enumerate}
            \item if $P \steptau P'$ then there is a transition $Q \steptau Q'$ such that $P' \bisimtbr[X] Q'$,
            \item if $P \step{a} P'$ with $a \in X \vee \deadend{P}{X}$ then there is a transition $Q \step{a} Q'$ such that $P' \bisimtbr Q'$,
            \item if $\deadend{P}{(X\cup Y)}$ and $P \step{\rt} P'$ then there is a transition $Q \step{\rt} Q'$ such that $P' \bisimtbr[Y] Q'$.
        \end{enumerate}
    \end{enumerate}
\end{definition}

\noindent
Note that if a system has no time-out, then a generalised [rooted] \tb reactive bisimulation is a stability respecting [rooted] branching bisimulation, thus proving that [rooted] \tb reactive bisimilarity is indeed an extension of stability respecting [rooted] branching bisimilarity to reactive systems with time-outs.

\begin{proposition} \label{prop:generalised}
    Let $P, Q \in \closed$ and $X \subseteq A$,
    \begin{itemize}
        \item $P \bisimtbr Q$ (resp.\ $P \bisimtbr[X] Q$) iff there exists a generalised \tb reactive bisimulation $\R$ with $\R(P,Q)$ (resp.\ $\R(P,X,Q)$),
        \item $P \bisimrtbr Q$ (resp.\ $P \bisimrtbr[X] Q$) iff there exists a rooted generalised \tb reactive bisimulation $\R$ with $\R(P,Q)$ (resp.\ $\R(P,X,Q)$).
    \end{itemize}
\end{proposition}

\begin{proof}
    Let $\R$ be a \tb reactive bisimulation. Let's check that it is a generalised \tb reactive bisimulation. Let $P,Q \in \closed$ and $X \subseteq A$.
    \begin{enumerate}
        \item If $\R(P,Q)$
        \begin{enumerate}
            \item this condition is shared by both definitions
            \item if $\deadend{P}{X}$ and $P \step{\rt} P'$ then, since $\R(P,Q)$, $\R(P,X,Q)$. Since $\deadend{P}{X}$ and $P \step{\rt} P'$, there exists a path $Q = Q_0 \pathtau Q_1 \step{\rt} Q_2 \pathtau Q_3 \step{\rt} ... \pathtau Q_{2r-1} \step{\opt{\rt}} Q_{2r}$ with $r>0$, such that $\forall i \in [0,r{-}1], \R(P,X,Q_{2i}) \wedge \deadend{Q_{2i+1}}{X}$ and $\R(P',X,Q_{2r})$. In particular, $Q_1 \nsteptau$.
            \item if $P \nsteptau$ then, since $\R(P,Q)$, $\R(P,\emptyset,Q)$, so there exists a path $Q \pathtau Q_0 \nsteptau$.
        \end{enumerate}
        \item If $\R(P,X,Q)$
        \begin{enumerate}
            \item this condition is shared by both definitions
            \item if $P \step{a} P'$ with $a \in X \vee \deadend{P}{X}$ then if $a \in X$ then there exists a path $Q \pathtau Q_1 \step{a} Q_2$ such that $\R(P,X,Q_1)$ and $\R(P',Q_2)$. Otherwise, $\deadend{P}{X}$ and so $P \nsteptau$, thus there exists a path $Q \pathtau Q_1 \nsteptau$. Since $\R(P,X,Q)$, $P \nsteptau$ and $Q \pathtau Q_1$, $\R(P,X,Q_1)$. As $\deadend{P}{X}$ and $Q_1 \nsteptau$, $\R(P,Q_1)$. Because $P \step{a} P_1$ and $Q_1 \nsteptau$, there exists a transition $Q_1 \step{a} Q_2$ such that $\R(P',Q_2)$. As a result, there exists a path $Q \pathtau Q_1 \step{a} Q_2$ such that $\R(P,X,Q_1)$ and $\R(P',Q_2)$.
            \item if $\deadend{P}{(X\cup Y)}$ and $P \step{\rt} P'$ then, since $P \nsteptau$, there exists a path $Q \pathtau Q_1 \nsteptau$. Furthermore, using Clause 2.a of Definition~\ref{def:intuitive}, $\R(P,X,Q_1)$. Moreover, $\deadend{P}{X}$ and $Q_1 \nsteptau$, thus, $\R(P,Q_1)$ and $\init{Q_1}=\init{P}$. Since $\R(P,Y,Q_1)$, $\deadend{P}{Y}$ and $P \step{\rt} P'$, there exists a path $Q_1 = Q_0' \pathtau Q_1' \step{\rt} Q_2' \pathtau Q_3' \step{\rt} ... \pathtau Q'_{2r-1} \step{\opt{\rt}} Q'_{2r}$ with $r>0$, such that $\forall i \in [0,r{-}1], \R(P,Y,Q'_{2i}) \wedge \deadend{Q'_{2i+1}}{Y}$ and $\R(P',Y,Q'_{2r})$. Since $Q_1 \nsteptau$, $Q_1 = Q_1'$. As a result, there exists a path $Q = Q_0 \pathtau Q_1 \step{\rt} Q_2 \pathtau Q_3 \step{\rt} ... \pathtau Q_{2r-1} \step{\opt{\rt}} Q_{2r}$ such that $Q_1\nsteptau$, $\forall i \in [1,r{-}1],\; \R(P,Y,Q_{2i}) \wedge \deadend{Q_{2i+1}}{Y}$ and $\R(P',Y,Q_{2r})$.
            \item this condition is shared by both definitions.
        \end{enumerate}
    \end{enumerate}
    Let $\R$ be a generalised \tb reactive bisimulation and define 
    \begin{align*}
        \R' := \R & \cup \{(P,X,Q) \mid \R(P,Q) \wedge X \subseteq A\} \cup \{(P,Y,Q),(P,Q) \mid \exists X \subseteq A,\; \R(P,X,Q) \\
        & \wedge (\init{P}\cup\init{Q})\cap(X\cup\{\tau\}) = \emptyset \wedge Y \subseteq A\}
    \end{align*}
    $\R'$ is symmetric by definition. Let's check that $\R'$ is a \tb reactive bisimulation. Let $P,Q \in \closed$ and $X \subseteq A$.
    \begin{enumerate}
        \item If $\R'(P,Q)$ then $\R(P,Q)$ or there exists a set $Y \subseteq A$ such that $\R(P,Y,Q)$ and $(\init{P}\cup\init{Q})\cap(Y\cup\{\tau\}) = \emptyset$.
        \begin{enumerate}
            \item If $P \step{\alpha} P'$ then
            \begin{itemize}
                \item if $\R(P,Q)$ then there exists a path $Q \pathtau Q_1 \step{\opt{\alpha}} Q_2$ such that $\R(P,Q_1)$ and $\R(P,Q_2)$ and, since $\R \subseteq \R'$, $\R'(P,Q_1)$ and $\R'(P',Q_2)$
                \item if there exists $Y \subseteq A$ such that $\R(P,Y,Q)$ and $(\init{P}\cup\init{Q})\cap(Y\cup\{\tau\}) = \emptyset$ then, since $\R(P,Y,Q)$ and $\deadend{P}{Y}$, $\alpha \ne \tau$, so there exists a path $Q \pathtau Q_1 \step{\alpha} Q_2$ such that $\R(P,Y,Q_1)$ and $\R(P,Q_2)$. Since $\deadend{Q}{Y}$ and $\R\subseteq \R'$, $Q = Q_1$ so there exists a path $Q \step{\alpha} Q_2$ such that $\R'(P,Q)$ and $\R'(P',Q_2)$.
            \end{itemize}
            \item For all $Z \subseteq A$,
            \begin{itemize}
                \item if $\R(P,Q)$ then, by definition of $\R'$, $\R'(P,Z,Q)$
                \item if there exists $Y \subseteq A$ such that $\R(P,Y,Q)$ and $(\init{P}\cup\init{Q})\cap(Y\cup\{\tau\}) = \emptyset$ then, by definition of $\R'$, $\R'(P,Z,Q)$.
            \end{itemize}
        \end{enumerate}
        \item If $\R'(P,X,Q)$ then $\R(P,X,Q)$, or $\R(P,Q)$, or there exists $Y \subseteq A$ such that $\R(P,Y,Q)$ and $(\init{P}\cup\init{Q})\cap(Y\cup\{\tau\}) = \emptyset$.
        \begin{enumerate}
            \item If $P \steptau P'$ then
            \begin{itemize}
                \item if $\R(P,X,Q)$ then there exists a path $Q \pathtau Q_1 \step{\opt{\tau}} Q_2$ such that $\R(P,X,Q_1)$ and $\R(P,X,Q_2)$ and, since $\R \subseteq \R'$, $\R'(P,X,Q_1)$ and $\R'(P',X,Q_2)$
                \item if $\R(P,Q)$ then there exists a path $Q \pathtau Q_1 \step{\opt{\tau}} Q_2$ such that $\R(P,Q_1)$ and $\R(P,Q_2)$ and, by definition of $\R'$, $\R'(P,X,Q_1)$ and $\R'(P',X,Q_2)$
                \item if there exists $Y \subseteq A$ such that $\R(P,Y,Q)$ and $(\init{P}\cup\init{Q})\cap(Y\cup\{\tau\}) = \emptyset$ then $P \nsteptau$, so this case is impossible.
            \end{itemize}
            \item If $P \step{a} P'$ with $a \in X$ then
            \begin{itemize}
                \item if $\R(P,X,Q)$ then there exists a path $Q \pathtau Q_1 \step{a} Q_2$ such that $\R(P,X,Q_1)$ and $\R(P,Q_2)$ and, since $\R \subseteq \R'$, $\R'(P,X,Q_1)$ and $\R'(P',Q_2)$
                \item if $\R(P,Q)$ then there exists a path $Q \pathtau Q_1 \step{a} Q_2$ such that $\R(P,Q_1)$ and $\R(P,Q_2)$ and, by definition of $\R'$, $\R'(P,X,Q_1)$ and $\R'(P',Q_2)$
                \item if there exists $Y \subseteq A$ such that $\R(P,Y,Q)$ and $(\init{P}\cup\init{Q})\cap(Y\cup\{\tau\}) = \emptyset$ then, since $\deadend{P}{Y}$, there exists a path $Q \pathtau Q_1 \step{a} Q_2$ such that $\R(P,Y,Q_1)$ and $\R(P',Q_2)$. Since $\deadend{Q}{Y}$, $Q = Q_1$ so there exists a path $Q \step{a} Q_2$ such that $\R'(P,X,Q)$ and $\R'(P',Q_2)$.
            \end{itemize}
            \item If $\deadend{P}{X}$ then
            \begin{itemize}
                \item if $\R(P,X,Q)$ then, since $P \nsteptau$, there exists a path $Q \pathtau Q_0 \nsteptau$. By Clause 2.a of Definition~\ref{def:generalised}, $\R(P,X,Q_0)$. Since $\R(P,X,Q_0)$, $\deadend{P}{X}$ and $Q_0 \nsteptau$, $\deadend{Q}{X}$, therefore, by definition, $\R'(P,Q_0)$.
                \item if $\R(P,Q)$ then, since $\R \subseteq \R'$, $\R'(P,Q)$.
                \item if there exists $Y \subseteq A$ such that $\R(P,Y,Q)$ and $(\init{P}\cup\init{Q})\cap(Y\cup\{\tau\}) = \emptyset$ then, by definition of $\R'$, $\R'(P,Q)$.
            \end{itemize}
            \item If $\deadend{P}{X}$ and $P \step{\rt} P'$ then
            \begin{itemize}
                \item if $\R(P,X,Q)$ then, there exists a path $Q = Q_0 \pathtau Q_1 \step{\rt} Q_2 \pathtau Q_3 \step{\rt} ... \pathtau Q_{2r-1} \step{\opt{\rt}} Q_{2r}$ with $r>0$, such that $Q_1\nsteptau$, $\forall i \in [1,r{-}1],\; \R(P,X,Q_{2i}) \wedge \deadend{Q_{2i+1}}{X}$ and $\R(P',X,Q_{2r})$. Hence also $\R'(P,X,Q_{2i})$ and $\R'(P',X,Q_{2r})$. By Clause 2.a of Definition~\ref{def:generalised}, $\R(P,X,Q_1)$, so $\deadend{Q_1}{X}$.
                \item if $\R(P,Q)$ then there exists a path $Q = Q_0 \pathtau Q_1 \step{\rt} Q_2 \pathtau Q_3 \step{\rt} ... \pathtau Q_{2r-1} \step{\opt{\rt}} Q_{2r}$ with $r>0$, such that $Q_1\nsteptau$, $\forall i \in [1,r{-}1], \R(P,X,Q_{2i}) \wedge \deadend{Q_{2i+1}}{X}$ and $\R(P',X,Q_{2r})$. Hence also $\R'(P,X,Q_{2i})$ and $\R'(P',X,Q_{2r})$. By Clause 1.a of Definition~\ref{def:generalised}, $\R(P,Q_1)$, so $\deadend{Q_1}{X}$.
                \item if there exists $Y \subseteq A$ such that $\R(P,Y,Q)$ and $(\init{P}\cup\init{Q})\cap(Y\cup\{\tau\}) = \emptyset$ then $\deadend{P}{(Y\cup X)}$ so there exists a path $Q = Q_0 \pathtau Q_1 \step{\rt} Q_2 \pathtau Q_3 \step{\rt} ... \pathtau Q_{2r-1} \step{\opt{\rt}} Q_{2r}$ with $r>0$, such that $Q_1\nsteptau$, $\forall i \in [1,r{-}1],\linebreak[4] \R(P,X,Q_{2i}) \wedge \deadend{Q_{2i+1}}{X}$ and $\R(P',X,Q_{2r})$. Hence also $\R'(P,X,Q_{2i})$ and $\R'(P',X,Q_{2r})$. Since $Q \nsteptau$, $Q \mathbin= Q_1$. Since $P \nsteptau$, $Q_1 \nsteptau$, $\R(P,Y,Q_1)$ and $\deadend{P}{Y}$, Clause 2.b of Definition~\ref{def:generalised} yields $\deadend{Q_1}{X}$. As a result, there is a path $Q = Q_0 \pathtau Q_1 \step{\rt} Q_2 \pathtau Q_3 \step{\rt} ... \pathtau Q_{2r-1} \step{\opt{\rt}} Q_{2r}$ with $r>0$, such that $\forall i \in [0,r{-}1],\; \R(P,X,Q_{2i}) \wedge \deadend{Q_{2i+1}}{X}$ and $\R(P',X,Q_{2r})$.
            \end{itemize}
            \item If $P \nsteptau$ then 
            \begin{itemize}
                \item if $\R(P,X,Q)$ then there exists a path $Q \pathtau Q_0 \nsteptau$.
                \item if $\R(P,Q)$ then there exists a path $Q \pathtau Q_0 \nsteptau$.
                \item if there exists $Y \subseteq A$ such that $\R(P,Y,Q)$ and $(\init{P}\cup\init{Q})\cap(Y\cup\{\tau\}) = \emptyset$ then $Q \nsteptau$.
            \end{itemize}
        \end{enumerate}
    \end{enumerate}
    Let $\R$ be a rooted \tb reactive bisimulation. Let's check that it is a generalised rooted \tb reactive bisimulation. Let $P,Q \in \closed$ and $X \subseteq A$.
    \begin{enumerate}
        \item If $\R(P,Q)$
        \begin{enumerate}
            \item this condition is shared by both definitions
            \item if $\deadend{P}{X}$ and $P \step{\rt} P'$ then, since $\R(P,Q)$, $\R(P,X,Q)$. Since $\deadend{P}{X}$ and $P \step{\rt} P'$, there exists a transition $Q \step{\rt} Q'$ such that $P' \bisimtbr[X] Q'$.
        \end{enumerate}
        \item If $\R(P,X,Q)$
        \begin{enumerate}
            \item this condition is shared by both definitions
            \item if $a\in X$, this condition is shared by both definitions; otherwise, apply Clauses 2.c and 1.a of Definition~\ref{def:rooted intuitive}
            \item if $\deadend{P}{(X\cup Y)}$ and $P \step{\rt} P'$ then, since $\deadend{P}{X}$, $\R(P,Q)$ and so $\R(P,Y,Q)$. Since $\deadend{P}{Y}$ and $P \step{\rt} P'$, there exists a transition $Q \step{\rt} Q'$ such that $P' \bisimtbr[Y] Q'$.
        \end{enumerate}
    \end{enumerate}
    Let $\R$ be a generalised rooted \tb reactive bisimulation and define 
    \begin{align*}
        \R' := \R & \cup \{(P,X,Q) \mid \R(P,Q) \wedge X \subseteq A\} \cup \{(P,Y,Q),(P,Q) \mid \exists X \subseteq A, \R(P,X,Q) \\
        & \wedge (\init{P}\cup\init{Q})\cap(X\cup\{\tau\}) = \emptyset \wedge Y \subseteq A\}
    \end{align*}
    $\R'$ is symmetric by definition. Let's check that $\R'$ is a rooted \tb reactive bisimulation. Let $P,Q \in \closed$ and $X \subseteq A$.
    \begin{enumerate}
        \item If $\R'(P,Q)$ then $\R(P,Q)$ or there exists $Y \subseteq A$ such that $\R(P,Y,Q)$ and $(\init{P}\cup\init{Q})\cap(Y\cup\{\tau\}) = \emptyset$.
        \begin{enumerate}
            \item If $P \step{\alpha} P'$ then
            \begin{itemize}
                \item if $\R(P,Q)$ then there exists a transition $Q \step{\alpha} Q'$ such that $P' \bisimtbr Q'$.
                \item if there exists $Y \subseteq A$ such that $\R(P,Y,Q)$ and $(\init{P}\cup\init{Q})\cap(Y\cup\{\tau\}) = \emptyset$ then, since $\R(P,Y,Q)$ and $\deadend{P}{Y}$, $\alpha \ne \tau$ so there exists a transition $Q \step{\alpha} Q'$ such that $P' \bisimtbr Q'$.
            \end{itemize}
            \item For all $Z \subseteq A$,
            \begin{itemize}
                \item if $\R(P,Q)$ then, by definition of $\R'$, $\R'(P,Z,Q)$
                \item if there exists $Y \subseteq A$ such that $\R(P,Y,Q)$ and $(\init{P}\cup\init{Q})\cap(Y\cup\{\tau\}) = \emptyset$ then, by definition of $\R'$, $\R(P,Z,Q)$.
            \end{itemize}
        \end{enumerate}
        \item If $\R'(P,X,Q)$ then $\R(P,X,Q)$, or $\R(P,Q)$, or there exists $Y \subseteq A$ such that $\R(P,Y,Q)$ and $(\init{P}\cup\init{Q})\cap(Y\cup\{\tau\}) = \emptyset$.
        \begin{enumerate}
            \item If $P \steptau P'$ then
            \begin{itemize}
                \item if $\R(P,X,Q)$ then there exists a transition $Q \step{\tau} Q'$ such that $P' \bisimtbr[X] Q'$,
                \item if $\R(P,Q)$ then there exists a step $Q \step{\tau} Q'$ such that $P' \bisimtbr Q'$ and so $P' \bisimtbr[X] Q'$
                \item if there exists $Y \subseteq A$ such that $\R(P,Y,Q)$ and $(\init{P}\cup\init{Q})\cap(Y\cup\{\tau\}) = \emptyset$ then $P \nsteptau$, so this case is impossible.
            \end{itemize}
            \item If $P \step{a} P'$ with $a \in X$ then
            \begin{itemize}
                \item if $\R(P,X,Q)$ then there exists a transition $Q \step{a} Q'$ such that $P' \bisimtbr Q'$
                \item if $\R(P,Q)$ then there exists a transition $Q \step{a} Q'$ such that $P' \bisimtbr Q'$
                \item if there exists $Y \subseteq A$ such that $\R(P,Y,Q)$ and $(\init{P}\cup\init{Q})\cap(Y\cup\{\tau\}) = \emptyset$ then, since $\deadend{P}{Y}$, there exists a transition $Q \step{a} Q'$ such that $P' \bisimtbr Q'$.
            \end{itemize}
            \item If $\deadend{P}{X}$ then 
            \begin{itemize}
                \item if $\R(P,X,Q)$ then, since $\deadend{P}{X}$, $\deadend{Q}{X}$, therefore, by definition, $\R'(P,Q)$,
                \item if $\R(P,Q)$ then, by definition of $\R'$, $\R'(P,Q)$,
                \item if there exists $Y \subseteq A$ such that $\R(P,Y,Q)$ and $(\init{P}\cup\init{Q})\cap(Y\cup\{\tau\}) = \emptyset$ then, by definition of $\R'$, $\R'(P,Q)$,
            \end{itemize}
            \item If $\deadend{P}{X}$ and $P \step{\rt} P'$ then
            \begin{itemize}
                \item if $\R(P,X,Q)$ then there exists a transition $Q \step{\rt} Q'$ such that $P' \bisimtbr[X] Q'$.
                \item if $\R(P,Q)$ then there exists a transition $Q  \step{\rt} Q'$ such that $P' \bisimtbr[X] Q'$.
                \item if there exists $Y \subseteq A$ such that $\R(P,Y,Q)$ and $(\init{P}\cup\init{Q})\cap(Y\cup\{\tau\}) = \emptyset$ then $\deadend{P}{(Y\cup X)}$ so there exists a step $Q \step{\rt} Q'$ such that $P' \bisimtbr[X] Q'$.
            \popQED
            \end{itemize}
        \end{enumerate}
    \end{enumerate}
\end{proof}

\section{Pohlmann Encoding} \label{app:Pohlmann}

Reactive bisimulations are sometimes complicated to check because of the large number of potential sets of allowed actions. In \cite{Pohlmann}, Pohlmann introduces an encoding which reduces strong reactive bisimilarity to strong bisimilarity. To this end he introduces unary operators $\vartheta$ and $\vartheta_X$ for $X\subseteq A$
that model placing their argument process in an environment that is triggered to change, or allows
exactly the actions in $X$, respectively. The actions $\rt_\varepsilon\notin A$ and $\varepsilon_X\notin A$ for $X \subseteq A$ are generated by the new operators, but may not be used by processes substituted for their arguments $P$.
They model a time-out action taken by the environment, and the stabilisation of an environment into
one that allows exactly the set of actions $X$, respectively. After a slight modification of the encoding, a similar result can be obtained for \tb reactive bisimilarity. We also introduce variants $\vartheta^r$ and $\vartheta^r_X$ of these operators that are targeting rooted \tb reactive bisimilarity.

\begin{table}[ht]
\vspace{-1.5ex}
    \[
    \begin{array}{l c l l l}
        \vartheta(P) \step{\alpha} \vartheta(P') &\wedge& \vartheta^r(P) \step{\alpha} \vartheta(P') & \Leftrightarrow & P \step{\alpha} P' \wedge \alpha \in A_\tau \\
        &&
        \vartheta^r(P) \step{\rt_X} \vartheta_X(P') & \Leftrightarrow & \deadend{P}{X} \wedge P \step{\rt} P'\\
        \vartheta(P) \step{\varepsilon_X} \vartheta_X(P) &\wedge& \vartheta^r(P) \step{\varepsilon_X} \vartheta_X(P) & & \\
        \vartheta_X(P) \steptau \vartheta_X(P') &\wedge& \vartheta^r_X(P) \steptau \vartheta_X(P') & \Leftrightarrow & P \steptau P' \\
        \vartheta_X(P) \step{a} \vartheta(P') &\wedge& \vartheta^r_X(P) \step{a} \vartheta(P') & \Leftrightarrow & P \step{a} P' \wedge \alpha \in X \\
        \vartheta_X(P) \step{\rt_\varepsilon} \vartheta(P) &\wedge& \vartheta^r_X(P) \step{\rt_\varepsilon} \vartheta^r(P) & \Leftrightarrow & \deadend{P}{X} \\
        \vartheta_X(P) \step{\rt} \vartheta_X(P') &\wedge& \vartheta^r_X(P) \step{\rt} \vartheta_X(P') & \Leftrightarrow & \deadend{P}{X} \wedge P \step{\rt} P'
    \end{array}
    \]
    \caption{Operational semantics of $\vartheta$, $\vartheta^r$, $(\vartheta_X)_{X \subseteq A}$ and $(\vartheta^r_X)_{X \subseteq A}$}
    \label{tab:Pohlmann operator}
\vspace{-1.5ex}
\end{table}

\noindent
In \cite{Pohlmann}, the first rule only applies to $\tau$-transitions; this echoes the previous remark about applying the first clause of Definition \ref{def:intuitive} only to invisible actions. As the intermediary actions $\rt_\epsilon$ and $(\epsilon_X)_{X\subseteq A}$ interfere with rootedness, the actions $(\rt_X)_{X\subseteq A}$ are added when rootedness has to be preserved. One can think of these as doing the actions $\varepsilon_X$ and $\rt$ in one (instead of two) steps. Note that the encoding rules mirror the clauses of Definition~\ref{def:intuitive}. The encoding transforms $\bisimtbr$ into $\bisimtb$ (see Definition~\ref{def:non-reactive}), and $\bisimrtbr$ in $\bisimrtb$ (Definition~\ref{def:rooted non-reactive}).

\begin{proposition} \label{prop:reduction}
  Let $P, Q \in \closed$.\\[1ex]
    \begin{minipage}{2.5in}
    \begin{itemize}
        \item $P \bisimtbr Q \Leftrightarrow \vartheta(P) \bisimtb \vartheta(Q)$
        \item $P \bisimrtbr Q \Leftrightarrow \vartheta^r(P) \bisimrtb \vartheta^r(Q)$
    \end{itemize}
    \end{minipage}\hfill
    \begin{minipage}{2.5in}
    \begin{itemize}
        \item $P \bisimtbr[X] Q \Leftrightarrow \vartheta_X(P) \bisimtb \vartheta_X(Q)$
        \item $P \bisimrtbr[X] Q \Leftrightarrow \vartheta^r_X(P) \bisimrtb \vartheta^r_X(Q)$
    \end{itemize}
    \end{minipage}
\end{proposition}

\begin{proof}
    It suffices to prove that: if $\R$ is a \tb reactive bisimulation then $\R' := \{(\vartheta(P),\vartheta(Q)) \mid \R(P,Q)\} \cup \{(\vartheta_X(P),\vartheta_X(Q)) \mid \R(P,X,Q)\}$ is a \rt-branching bisimulation; and if $\R$ is a \rt-branching bisimulation then $\R' := \{(P,Q), (P,X,Q) \mid \R(\vartheta(P),\vartheta(Q)) \wedge X\subseteq A\} \cup \{(P,X,Q) \mid \R(\vartheta_X(P),\vartheta_X(Q))\}$ is a \tb reactive bisimulation. The rooted case is very similar.
    
    Let $\R$ be a \tb reactive bisimulation and define 
    \begin{align*}
        \R' := \{(\vartheta(P),\vartheta(Q)) \mid \R(P,Q)\} \cup \{(\vartheta_X(P),\vartheta_X(Q)) \mid \R(P,X,Q)\}
    \end{align*}
    We are going to check that $\R'$ is a $\rt$-branching bisimulation. Let $P,Q \in \closed$ such that $\R'(P,Q)$.
    \begin{itemize}
        \item If $P = \vartheta(P^\dag)$ and $Q = \vartheta(Q^\dag)$ then, by definition of $\R'$, $\R(P^\dag,Q^\dag)$.
        \begin{enumerate}
            \item If $P \step{\alpha} P'$ with $\alpha \in A_\tau \cup\{\rt_\epsilon,\epsilon_X \mid X \subseteq A\}$ then
            \begin{itemize}
                \item if $\alpha \in A_\tau$ then, by the semantics of $\vartheta$, $P' = \vartheta(P^\ddag)$ and $P^\dag \step{\alpha} P^\ddag$. Since $\R(P^\dag,Q^\dag)$, there exists a path $Q^\dag \pathtau Q^\star \step{\opt{\alpha}} Q^\ddag$ such that $\R(P^\dag,Q^\star)$ and $\R(P^\ddag,Q^\ddag)$. By the semantics, there exists a path $Q \pathtau \vartheta(Q^\star) \step{\opt{\alpha}} \vartheta(Q^\ddag)$ such that, by definition of $\R'$, $\R'(P, \vartheta(Q^\star))$ and $\R'(P', \vartheta(Q^\ddag))$.
                \item if $\alpha = \rt_\epsilon$ then this case is not possible according to the semantics of $\vartheta$.
                \item if $\alpha = \epsilon_X$ with $X \subseteq A$ then, by the semantics of $\vartheta$, $P' = \vartheta_X(P^\dag)$. Since $\R(P^\dag,Q^\dag)$, $\R(P^\dag,X,Q^\dag)$. By the semantics, $Q \step{\epsilon_X} \vartheta_X(Q^\dag)$ such that, by the definition of $\R'$, $\R'(P',\vartheta_X(Q^\dag))$.
            \end{itemize}
            \item If $P \step{\rt} P'$ then, by the semantics, this is impossible.
            \item If $P \nsteptau$ then, by the semantics of $\vartheta$, $P^\dag \nsteptau$. Since $\R(P^\dag,Q^\dag)$, $\R(P^\dag,\emptyset,Q^\dag)$, so there exists a path $Q^\dag \pathtau Q^\star \nsteptau$. By the semantics, there exists a path $Q \pathtau \vartheta(Q^\star) \nsteptau$.
        \end{enumerate}
        \item If there exists $X \subseteq A$ such that $P = \vartheta_X(P^\dag)$ and $Q = \vartheta_X(Q^\dag)$ then, by definition of $\R'$, $\R(P^\dag,X,Q^\dag)$.
        \begin{enumerate}
            \item If $P \step{\alpha} P'$ with $\alpha \in A_\tau \cup \{\rt_\epsilon,\epsilon_X \mid X\subseteq A\}$ then
            \begin{itemize}
                \item if $P \steptau P'$ then, by the semantics, $P' = \vartheta_X(P^\ddag)$ and $P^\dag \steptau P^\ddag$. Since $\R(P^\dag,X,Q^\dag)$, there exists a path $Q^\dag \pathtau Q^\star \step{\opt{\tau}} Q^\ddag$ such that $\R(P^\dag,X,Q^\star)$ and $\R(P^\ddag,X,Q^\ddag)$. By the semantics, there exists a path $Q \pathtau \vartheta_X(Q^\star) \step{\opt{\tau}} \vartheta_X(Q^\ddag)$ such that, by the definition of $\R'$, $\R'(P,\vartheta_X(Q^\star))$ and $\R'(P',\vartheta_X(Q^\ddag))$.
                \item if $P \step{a} P'$ with $a \in A$ then, by the semantics, $a \in X$, $P' = \vartheta(P^\ddag)$ and $P^\dag \step{a} P^\ddag$. Since $\R(P^\dag,X,Q^\dag)$, there exists a path $Q^\dag \pathtau Q^\star \step{a} Q^\ddag$ such that $\R(P^\dag,X,Q^\star)$ and $\R(P^\ddag,Q^\ddag)$. By the semantics, there exists a path $Q \pathtau \vartheta_X(Q^\star) \step{a} \vartheta(Q^\ddag)$ such that, by the definition of $\R'$, $\R'(P,\vartheta_X(Q^\star))$ and $\R'(P',\vartheta(Q^\ddag))$.
                \item if $P \step{\rt_\epsilon} P'$ then, by the semantics, $P' = \vartheta(P^\dag)$ and $\deadend{P^\dag}{X}$. Since $\R(P^\dag,X,Q^\dag)$ and $P^\dag \nsteptau$, there exists a path $Q^\dag \pathtau Q^\star \nsteptau$. Moreover, $\R(P^\dag,X,Q^\star)$. Since $\deadend{P^\dag}{X}$, $\R(P^\dag,X,Q^\star)$ and $Q^\star \nsteptau$, $\deadend{Q^\star}{X}$ and $\R(P^\dag,Q^\star)$. By the semantics, there exists a path $Q \pathtau \vartheta_X(Q^\star) \step{\rt_\epsilon} \vartheta(Q^\star)$ such that, by the definition of $\R'$, $\R'(P,\vartheta_X(Q^\star))$ and $\R'(P',\vartheta(Q^\star))$.
                \item if $\alpha = \epsilon_X$ with $X \subseteq A$ then this case is impossible according to the semantics of $\vartheta_X$.
            \end{itemize}
            \item if $P \step{\rt} P'$ then, by the semantics, $P' = \vartheta_X(P^\ddag)$, $\deadend{P^\dag}{X}$ and $P^\dag \step{\rt} P^\ddag$. Since $\R(P^\dag,X,Q^\dag)$, there exists a path $Q^\dag = Q^\dag_0 \pathtau Q^\dag_1 \step{\rt} Q^\dag_2 \pathtau Q^\dag_3 \step{\rt} ... \pathtau Q^\dag_{2r-1} \step{\opt{\rt}} Q^\ddag_{2r}$ with $r>0$, such that $\forall i \in [0,r{-}1], \R(P^\dag,X,Q^\dag_{2i}) \wedge \deadend{Q^\dag_{2i+1}}{X}$ and $\R(P^\ddag,X,Q^\ddag_{2r})$. For all $i \in [0,r{-}1]$, since $P^\dag \nsteptau$, $Q^\dag_{2i} \pathtau Q^\dag_{2i+1}$ and $\R(P^\dag,X,Q^\dag_{2i})$, $\R(P^\dag,X,Q_{2i+1}^\dag)$. By the semantics, there exists a path $Q \pathtau \vartheta_X(Q^\dag_1) \step{\rt} \vartheta_X(Q^\dag_2) \pathtau \vartheta_X(Q^\dag_3) \step{\rt} ... \pathtau \vartheta_X(Q_{2r-1}) \step{\opt{\rt}} \vartheta_X(Q_{2r})$ such that, by definition of $\R'$, $\forall i \in [0,2r{-}1], \R'(P,\vartheta_X(Q^\dag_i))$ and $\R'(P',\vartheta_X(Q^\ddag_{2r}))$.
            \item if $P \nsteptau$ then, by the semantics of $\vartheta_X$, $P^\dag \nsteptau$. Since $\R(P^\dag,X,Q^\dag)$, there exists a path $Q^\dag \pathtau Q_0 \nsteptau$. By the semantics, there exists a path $Q \pathtau \vartheta_X(Q_0) \nsteptau$.
        \end{enumerate}
    \end{itemize}
    Let $\R$ be a $\rt$-branching bisimulation and define
    \begin{align*}
        \R' := & \{(P,Q), (P,X,Q) \mid \R(\vartheta(P),\vartheta(Q)) \wedge X \subseteq A\} \cup \{(P,X,Q) \mid \R(\vartheta_X(P),\vartheta_X(Q))\}
    \end{align*}
    We are going to show that $\R'$ is a \tb reactive bisimulation. Let $P, Q \in \closed$ and $X \subseteq A$.
    \begin{enumerate}
        \item If $\R'(P,Q)$ then $\R(\vartheta(P),\vartheta(Q))$.
        \begin{enumerate}
            \item If $P \step{\alpha} P'$ with $\alpha \in A_\tau$ then, by the semantics, $\vartheta(P) \step{\alpha} \vartheta(P')$. Since $\R(\vartheta(P),\vartheta(Q))$, there exists a path $\vartheta(Q) \pathtau Q^\star \step{\opt{\alpha}} Q^\ddag$ such that $\R(\vartheta(P),Q^\star)$ and $\R(\vartheta(P'),Q^\ddag)$. By the semantics, $Q^\star = \vartheta(Q_1)$, $Q^\ddag = \vartheta(Q_2)$ and $Q \pathtau Q_1 \step{\opt{\alpha}} Q_2$ such that, by definition of $\R'$, $\R'(P,Q_1)$ and $\R'(P',Q_2)$.
            \item For all $Y \subseteq A$, by definition of $\R'$, $\R'(P,Y,Q)$.
        \end{enumerate}
        \item If $\R'(P,X,Q)$ then $\R(\vartheta(P),\vartheta(Q))$ or $\R(\vartheta_X(P),\vartheta_X(Q))$. If $\R(\vartheta(P),\vartheta(Q))$ then $\vartheta(P) \!\step{\epsilon_X} \vartheta_X(P)$, thus there exists a path $\vartheta(Q) \pathtau Q^\star \step{\epsilon_X} Q^\ddag$ such that $\R(\vartheta(P),Q^\star)$ and $\R(\vartheta_X(P),Q^\ddag)$. By the semantics, $Q^\star = \vartheta(Q_0)$, $Q^\ddag = \vartheta_X(Q_0)$ and $Q \pathtau Q_0$. Therefore, there exists a path $Q \pathtau Q_0$ such that $\R(\vartheta_X(P),\vartheta_X(Q_0))$.
        \begin{enumerate}
            \item If $P \step{\tau} P'$ then, by the semantics, $\vartheta_X(P) \step{\tau} \vartheta_X(P')$. Since $\R(\vartheta_X(P),\vartheta_X(Q_0))$, there exists a path $\vartheta_X(Q_0) \pathtau Q^\star \step{\opt{\tau}} Q^\ddag$ such that $\R(\vartheta_X(P),Q^\star)$ and $\R(\vartheta_X(P'),Q^\ddag)$. By the semantics, $Q^\star = \vartheta_X(Q_1)$, $Q^\ddag = \vartheta_X(Q_2)$ and $Q \pathtau Q_1 \step{\opt{\tau}} Q_2$ such that, by definition of $\R'$, $\R'(P,X,Q_1)$ and $\R'(P',X,Q_2)$.
            \item If $P \step{a} P'$ with $a \in X$ then, by the semantics, $\vartheta_X(P) \step{a} \vartheta(P')$. As $\R(\vartheta_X(P),\vartheta_X(Q_0))$, there exists a path $\vartheta_X(Q_0) \pathtau Q^\star \step{a} Q^\ddag$ such that $\R(\vartheta_X(P),Q^\star)$ and $\R(\vartheta(P'),Q^\ddag)$. By the semantics, $Q^\star = \vartheta_X(Q_1)$, $Q^\ddag = \vartheta(Q_2)$ and $Q \pathtau Q_1 \step{a} Q_2$ such that, by definition of $\R'$, $\R'(P,X,Q_1)$ and $\R'(P',Q_2)$.
            \item If $\deadend{P}{X}$ then, by the semantics, $\vartheta_X(P) \step{\rt_\epsilon} \vartheta(P)$. As $\R(\vartheta_X(P),\vartheta_X(Q_0))$, there exists a path $\vartheta_X(Q_0) \pathtau Q^\star \step{\rt_\epsilon} Q^\ddag$ such that $\R(\vartheta_X(P),Q^\star)$ and $\R(\vartheta(P),Q^\ddag)$. By the semantics, $Q^\star = \vartheta_X(Q_0')$, $Q^\ddag = \vartheta(Q_0')$ and $Q \pathtau Q_0'$ such that, by definition of $\R'$, $\R'(P,Q_0')$.
            \item If $\deadend{P}{X}$ and $P \step{\rt} P'$ then, by the semantics, $\vartheta_X(P) \step{\rt} \vartheta_X(P')$. Since $\R(\vartheta_X(P),\vartheta_X(Q_0))$, there exists a path $\vartheta_X(Q_0) \pathtau Q^\dag_1 \step{\rt} Q^\dag_2 \pathtau Q^\dag_3 \step{\rt} ... \pathtau Q^\dag_{2r-1} \step{\opt{\rt}} Q^\ddag_{2r}$ with $r>0$, such that $\forall i \in [0,2r{-}1], \R(\vartheta_X(P),Q^\dag_i)$ and $\R(\vartheta_X(P'),Q^\ddag_{2r})$. By the semantics, there exists a path $Q_0 \pathtau Q_1 \step{\rt} Q_2 \pathtau Q_3 \step{\rt} ... \pathtau Q_{2r-1} \step{\opt{\rt}} Q_{2r}$ such that $\forall i \in [0,r{-}1], Q^\dag_{2i} = \vartheta_X(Q_{2i}) \wedge Q^\dag_{2i+1} = \vartheta_X(Q_{2i+1})$ and $Q^\ddag_{2r} = \vartheta_X(Q_{2r})$. Thus, by definition of $\R'$, $\forall i \in [0,r{-}1], \R'(P,X,Q_{2i})$ and $\R'(P',X,Q_{2r})$. With the possible exception of $i=r{-}1$, for all $i \in [0,r{-}1]$ we have $\deadend{Q_{2i+1}}{X}$. In the case that $Q_{2r}=Q_{2r-1}$ we can choose $Q_{2r-1}$ such that $Q_{2r-1}\nsteptau$, and hence $Q^\dag_{2r-1}\nsteptau$. Since $\R(\vartheta_X(P),Q^\dag_{2r-1})$ and $\vartheta_X(P)\step{\rt_\epsilon}$, also $Q^\dag_{2r-1}\step{\rt_\epsilon}$. Thus $\deadend{Q_{2r-1}}{X}$.
            \item If $P \nsteptau$ then, by the semantics, $\vartheta_X(P) \nsteptau$. Since $\R(\vartheta_X(P),\vartheta_X(Q_0))$, there exists a path $\vartheta_X(Q_0) \pathtau Q^\star \nsteptau$. By the semantics, $Q^\star = \vartheta_X(Q_1)$ and $Q \pathtau Q_1 \nsteptau$.
        \end{enumerate}
    \end{enumerate}
    Let $\R$ be a rooted \tb reactive bisimulation and define 
    \begin{align*}
        \R' := \{(\vartheta^r(P),\vartheta^r(Q)) \mid \R(P,Q)\} \cup \{(\vartheta^r_X(P),\vartheta^r_X(Q)) \mid \R(P,X,Q)\}
    \end{align*}
    We are going to check that $\R'$ is a rooted $\rt$-branching bisimulation. Let $P,Q \in \closed$ such that $\R'(P,Q)$.
    \begin{itemize}
        \item If $P = \vartheta^r(P^\dag)$ and $Q = \vartheta^r(Q^\dag)$ then, by definition of $\R'$, $\R(P^\dag,Q^\dag)$.
        \begin{enumerate}
            \item Let $P \step{\alpha} P'$ with $\alpha \in Act \cup\{\rt_\epsilon,\rt_X,\epsilon_X \mid X \mathop\subseteq A\}$.
            \begin{itemize}
                \item If $\alpha \in A_\tau$ then, by the semantics of $\vartheta^r$, $P' = \vartheta(P^\ddag)$ and $P^\dag \step{\alpha} P^\ddag$. Since $\R(P^\dag,Q^\dag)$, there exists a transition $Q^\dag \step{\alpha} Q^\ddag$ such that $P^\ddag \bisimtbr Q^\ddag$. By the semantics, there exists a transition $Q \step{\alpha} \vartheta(Q^\ddag)$ such that, by the first part of this proof, $\vartheta(P') \bisimtb \vartheta(Q^\ddag)$.
                \item The case $\alpha = \rt_\epsilon$ is not possible according to the semantics of $\vartheta^r$.
                \item If $\alpha = \epsilon_X$ with $X \subseteq A$ then, by the semantics of $\vartheta^r$, $P' = \vartheta_X(P^\dag)$. Since $\R(P^\dag,Q^\dag)$, $P^\dag \bisimrtbr Q^\dag$ so $P^\dag \bisimtbr[X] Q^\dag$. By the semantics, $Q \step{\epsilon_X} \vartheta_X(Q^\dag)$ such that, by the first part of this proof, $P' \bisimtb \vartheta_X(Q^\dag)$.
                \item The case $\alpha=\rt$, by the semantics, is not possible.
                \item If $P \step{\rt_X} P'$ then, by the semantics, $\deadend{P^\dag}{X}$, $P^\dag \step{t} P^\ddag$ and $P' = \vartheta_X(P^\ddag)$. Since $\R(P^\dag,Q^\dag)$, $\deadend{Q^\ddag}{X}$ and there is a transition $Q^\dag \step{t} Q^\ddag$ such that $P^\ddag \bisimtbr[X] Q^\ddag$. By the semantics, $Q \step{\rt_X} \vartheta_X(Q^\ddag)$ and, by the first part of this proof, $P' \bisimtb \vartheta_X(Q^\ddag)$.
            \end{itemize}
        \end{enumerate}
        \item If there exists $X \subseteq A$ such that $P = \vartheta^r_X(P^\dag)$ and $Q = \vartheta^r_X(Q^\dag)$ then, by definition of $\R'$, $\R(P^\dag,X,Q^\dag)$.
        \begin{enumerate}
            \item Let $P \step{\alpha} P'$ with $\alpha \in Act \cup\{\rt_\epsilon,\rt_X,\epsilon_X \mid X \mathop\subseteq A\}$.
            \begin{itemize}
                \item If $P \steptau P'$ then, by the semantics, $P' = \vartheta_X(P^\ddag)$ and $P^\dag \steptau P^\ddag$. Since $\R(P^\dag,X,Q^\dag)$, there exists a transition $Q^\dag \steptau Q^\ddag$ such that $P^\ddag \bisimtbr[X] Q^\ddag$. By the semantics, there exists a transition $Q \steptau \vartheta_X(Q^\ddag)$ such that, by the first part, $P' \bisimtb \vartheta_X(Q^\ddag)$.
                \item If $P \step{a} P'$ with $a \in A$ then, by the semantics, $a \in X$, $P' = \vartheta(P^\ddag)$ and $P^\dag \step{a} P^\ddag$. Since $\R(P^\dag,X,Q^\dag)$, there exists a transition $Q^\dag \step{a} Q^\ddag$ such that $P^\ddag \bisimtbr Q^\ddag$. By the semantics, there exists a transition $Q \step{a} \vartheta(Q^\ddag)$ such that, by the first part, $P' \bisimtb \vartheta(Q^\ddag)$.
                \item If $P \step{\rt_\epsilon} P'$ then, by the semantics, $P' = \vartheta^r(P^\dag)$ and $\deadend{P^\dag}{X}$. Since $\R(P^\dag,X,Q^\dag)$ and $\deadend{P^\dag}{X}$, $\deadend{Q^\dag}{X}$ and $\R(P^\dag,Q^\dag)$. By the semantics, there exists a path $Q \step{\rt_\epsilon} \vartheta^r(Q^\dag)$ such that, by the definition of $\R'$, $\R'(P',\vartheta^r(Q^\dagger))$. Considering the previous case, this implies that $P' \bisimrtb \vartheta^r(Q^\dag)$ and so $P' \bisimtb \vartheta^r(Q^\dag)$.
                \item The case $\alpha = \epsilon_X$ or $\rt_X$ with $X \subseteq A$ is impossible according to the semantics of $\vartheta_X$.
                \item If $P \step{\rt} P'$ then, by the semantics, $P' = \vartheta_X(P^\ddag)$, $\deadend{P^\dag}{X}$ and $P^\dag \step{\rt} P^\ddag$. Since $\R(P^\dag,X,Q^\dag)$, there exists a path $Q^\dag \step{\rt} Q^\ddag$ such that $P^\ddag \bisimtbr[X] Q^\ddag$. Moreover, $\deadend{Q^\dag}{X}$. By the semantics, there exists a path $Q \step{\rt} \vartheta_X(Q^\ddag)$ such that, by the first part, $P' \bisimtb \vartheta_X(Q^\ddag)$.
            \end{itemize}
        \end{enumerate}
    \end{itemize}
    Let $\R$ be a rooted $t$-branching bisimulation and define
    \begin{align*}
        \R' := & \{(P,Q), (P,X,Q) \mid \R(\vartheta^r(P),\vartheta^r(Q)) \wedge X \subseteq A\} \cup \{(P,X,Q) \mid \R(\vartheta^r_X(P),\vartheta^r_X(Q))\} \\
        & \cup \left\{(P,Q), (P,X,Q) \left|\; \begin{array}{@{}l@{}} \R(\vartheta^r_Y(P),\vartheta^r_Y(Q)) \wedge X \subseteq A \wedge \mbox{}\\ (\init{\vartheta^r_Y(P)}\cup\init{\vartheta^r_Y(Q)})\cap(Y\cup\{\tau\})=\emptyset \end{array}\right.\right\}
    \end{align*}
    We are going to show that $\R'$ is a rooted \tb reactive bisimulation. Let $P, Q \in \closed$ and $X \subseteq A$.
    \begin{enumerate}
        \item If $\R'(P,Q)$ then $\R(\vartheta^r(P),\vartheta^r(Q))$ or $\R(\vartheta^r_Y(P),\vartheta^r_Y(Q))$ and $(\init{\vartheta^r_Y(P)}\cup\init{\vartheta^r_Y(Q)})\cap(X\cup\{\tau\})=\emptyset$.
        \begin{enumerate}
            \item If $P \step{\alpha} P'$ with $\alpha \in A_\tau$ then, by the semantics, $\vartheta^r(P) \step{\alpha} \vartheta(P')$. 
            \begin{itemize}
                \item If $\R(\vartheta^r(P),\vartheta^r(Q))$, there exists a transition $\vartheta^r(Q) \step{\alpha} Q^\ddag$ such that $\vartheta(P') \bisimtb Q^\ddag$. By the semantics, $Q^\ddag = \vartheta(Q')$ and $Q \step{\alpha} Q'$ such that, by the second part, $P' \bisimtbr Q'$.
                \item If $\R(\vartheta^r_Y(P),\vartheta^r_Y(Q))$ and $(\init{\vartheta^r_Y(P)}\cup\init{\vartheta^r_Y(Q)})\cap(X\cup\{\tau\})=\emptyset$ then $\vartheta_Y(P) \step{\rt_\epsilon} \vartheta^r(P)$, thus there exists a transition $\vartheta^r_Y(Q) \step{t_\epsilon} \vartheta^r(Q)$ with $\vartheta^r(P) \bisimtb \vartheta^r(Q)$. Since $\vartheta^r(P) \step{\alpha} \vartheta(P')$, there exists a path $\vartheta^r(Q) \pathtau Q_1 \step{(\alpha)} Q_2$ such that $\vartheta^r(P) \bisimtb Q_1$ and $\vartheta(P') \bisimtb Q_2$. Since $\vartheta^r(Q)\nsteptau$, $\vartheta^r(Q) \step{\alpha} Q_2$ so, by the semantics, $Q \step{\alpha} Q'$ and $Q_2 = \vartheta(Q')$. As a result, there exists a transition $Q \step{\alpha} Q'$ such that, by the second part, $P' \bisimtbr Q'$.
            \end{itemize}
            \item For all $Y \subseteq A$, by definition of $\R'$, $\R'(P,Y,Q)$.
        \end{enumerate}
        \item If $\R'(P,X,Q)$ then $\R(\vartheta^r(P),\vartheta^r(Q))$, $\R(\vartheta^r_X(P),\vartheta^r_X(Q))$ or\\ $\R(\vartheta^r_Y(P),\vartheta^r_Y(Q))$ and $(\init{\vartheta^r_Y(P)}\cup\init{\vartheta^r_Y(Q)})\cap(X\cup\{\tau\})=\emptyset$.
        \begin{enumerate}
            \item If $P \steptau P'$ then, by the semantics, $\vartheta^r(P) \steptau \vartheta(P')$ and $\vartheta^r_X(P) \steptau \vartheta_X(P')$.
            \begin{itemize}
                \item If $\R(\vartheta^r(P),\vartheta^r(Q))$, there exists a path $\vartheta^r(Q) \steptau Q^\ddag$ such that $\vartheta(P') \bisimtb Q^\ddag$. By the semantics, $Q^\ddag = \vartheta(Q')$ and $Q \steptau Q'$ so that, by the second part, $P' \bisimtbr Q'$ and thus $P' \bisimtbr[X] Q'$.
                \item If $\R(\vartheta^r_X(P),\vartheta^r_X(Q))$, there exists a path $\vartheta^r_X(Q) \steptau Q^\ddag$ such that $\vartheta_X(P') \bisimtb Q^\ddag$. By the semantics, $Q^\ddag = \vartheta_X(Q')$ and $Q \steptau Q'$ so that, by the second part, $P' \bisimtbr[X] Q'$.
                \item The case $\R(\vartheta^r_Y(P),\vartheta^r_Y(Q))$ and $(\init{\vartheta^r_Y(P)}\cup\init{\vartheta^r_Y(Q)})\cap(X\cup\{\tau\})=\emptyset$ is impossible.
            \end{itemize}
            \item If $P \step{a} P'$ with $a \in X$ then, by the semantics, $\vartheta^r(P) \step{a} \vartheta(P')$, $\vartheta^r_X(P) \step{a} \vartheta(P')$. 
            \begin{itemize}
                \item If $\R(\vartheta^r(P),\vartheta^r(Q))$, there exists a step $\vartheta^r(Q) \step{a} Q^\ddag$ such that $\vartheta(P') \bisimtb Q^\ddag$. By the semantics, $Q^\ddag = \vartheta(Q')$ and $Q \step{a} Q'$ such that, by the second part, $P' \bisimtbr Q'$.
                \item If $\R(\vartheta^r_X(P),\vartheta^r_X(Q))$, there exists a path $\vartheta^r_X(Q) \step{a} Q^\ddag$ such that $\vartheta(P') \bisimtb Q^\ddag$. By the semantics, $Q^\ddag = \vartheta(Q')$ and $Q \step{a} Q'$ such that, by the second part, $P' \bisimtbr Q'$.
                \item If $\R(\vartheta^r_Y(P),\vartheta^r_Y(Q))$ and $(\init{\vartheta^r_Y(P)}\cup\init{\vartheta^r_Y(Q)})\cap(X\cup\{\tau\})=\emptyset$ then $\vartheta^r_Y(P) \step{\rt_\epsilon} \vartheta^r(P)$, thus there exists a transition $\vartheta^r_Y(Q) \step{t_\epsilon} \vartheta^r(Q)$ with $\vartheta^r(P) \bisimtb \vartheta^r(Q)$. Since $\vartheta^r(P) \step{a} \vartheta(P')$, therefore, there exists a path $\vartheta^r(Q) \pathtau Q_1 \step{a} Q_2$ such that $\vartheta^r(P) \bisimtb Q_1$ and $\vartheta(P') \bisimtb Q_2$. Since $\vartheta^r(Q)\nsteptau$, $\vartheta^r(Q) \step{a} Q_2$ so, by the semantics, $Q \step{a} Q'$ and $Q_2 = \vartheta(Q')$. As a result, there exists a transition $Q \step{a} Q'$ such that, by the second part, $P' \bisimtbr Q'$.
            \end{itemize}
            \item If $\deadend{P}{X}$ then
            \begin{itemize}
                \item if $\R(\vartheta^r(P),\vartheta^r(Q))$ then, by definition of $\R'$, $\R'(P,Q)$.
                \item if $\R(\vartheta^r_X(P),\vartheta^r_X(Q))$ then $(\init{\vartheta^r_X(P)}\cup\init{\vartheta^r_X(Q)})\cap(X\cup\{\tau\})=\emptyset$ thus, by definition of $\R'$, $\R(P,Q)$.
                \item If $\R(\vartheta^r_Y(P),\vartheta^r_Y(Q))$ and $(\init{\vartheta^r_Y(P)}\cup\init{\vartheta^r_Y(Q)})\cap(X\cup\{\tau\})=\emptyset$ then, by definition of $\R'$, $\R(P,Q)$.
            \end{itemize}
            \item If $\deadend{P}{X}$ and $P \step{\rt} P'$ then, by the semantics, $\vartheta^r(P) \step{\rt_X} \vartheta_X(P')$, $\vartheta_X(P) \step{\rt} \vartheta_X(P')$.
            \begin{itemize}
                \item If $\R(\vartheta^r(P),\vartheta^r(Q))$ then there exists a transition $\vartheta^r(Q) \step{\rt_X} Q^\ddag$ with $\vartheta_X(P') \bisimtb Q^\ddag$. By the semantics, $Q^\ddag = \vartheta_X(Q')$ and $Q \step{\rt_X} Q'$. Moreover, by the second part, $P' \bisimtbr[X] Q'$.
                \item if $\R(\vartheta^r_X(P),\vartheta^r_X(Q))$ then there exists a transition $\vartheta^r_X(Q) \step{\rt} Q^\ddag$ with $\vartheta_X(P') \bisimtb Q^\ddag$. By the semantics, $Q^\ddag = \vartheta_X(Q')$ and $Q \step{\rt} Q'$. Moreover, by the second part, $P' \bisimtbr[X] Q'$.
                \item if $\R(\vartheta^r_Y(P),\vartheta^r_Y(Q))$ and $(\init{\vartheta^r_Y(P)}\cup\init{\vartheta^r_Y(Q)})\cap(X\cup\{\tau\})=\emptyset$ then $\vartheta^r_Y(P) \step{\rt_\epsilon} \vartheta^r(P)$, so there exists a transition $\vartheta^r_Y(Q) \step{\rt_\epsilon} \vartheta^r(Q)$ with $\vartheta^r(P) \bisimtb \vartheta^r(Q)$. Since $\vartheta^r(P) \step{\rt_X} \vartheta_X(P')$, there exists a path $\vartheta^r(Q) \pathtau Q_1 \step{\rt_X} Q_2$ such that $\vartheta^r(P) \bisimtb Q_1$ and $\vartheta_X(P') \bisimtb Q_2$. As $\vartheta^r(Q) \nsteptau$, $\vartheta^r(Q) \step{\rt_X} Q_2$ so, by the semantics, $Q \step{\rt} Q'$ and $Q_2 = \vartheta_X(Q')$. Moreover, by the second part, $P' \bisimtbr[X] Q'$.
            \popQED
            \end{itemize}
        \end{enumerate}
    \end{enumerate}
\end{proof}

\noindent
It would have been possible to define the \rt-branching bisimilarity differently while preserving the same result. The encoded processes are part of a sub-class with specific properties. For instance, an encoded process cannot have an outgoing $\tau$-transition and an outgoing time-out by definition of $\vartheta$ and $(\vartheta_X)_{X \subseteq A}$, i.e., for any encoded process $P$, ${P \step{\rt}} \Rightarrow {P\nsteptau}$. Thus, adding the condition $\forall i \in [0,r-1], Q_{2i+1} \nsteptau$ in  clause 2 of Definition \ref{def:non-reactive} does not interfere with our result even though it obviously defines a different bisimilarity. We settled on Definition 7 because it is the one that yields the simplest proofs.

%% file: appendix.tex
\newpage
\section{Proofs of Stuttering Property and Transitivity} \label{app:intro}

\begin{proof}[Proof of Lemma \ref{lem:stuttering}]
    Let $\R$ be a \tb reactive bisimulation. Let's define
    \begin{align*}
        \R' := & \{(P^\dag,Q),(Q,P^\dag) \mid \exists P,P^\ddag \in \closed, P \pathtau P^\dag \pathtau P^\ddag \wedge \R(P,Q) \wedge \R(P^\ddag,Q)\} \cup {}\\
        & \{(P^\dag\!\!,X,Q),(Q,X,P^\dag) \mid \exists P,P^\ddag \mathbin\in \closed, P \mathbin{\pathtau} P^\dag \mathbin{\pathtau} P^\ddag \wedge \R(P,X,Q) \wedge \R(P^\ddag\!\!,X,Q)\}
    \end{align*}
    $\R'$ is symmetric by definition and we are going to prove that $\R'$ is a \tb reactive bisimulation. Note that $\R \subseteq \R'$ (by taking $P^\ddag=P^\dag$). Let $P,Q \in \closed$ and $X \subseteq A$.
    \begin{enumerate}
        \item Let $\R'(P,Q)$.
        \begin{enumerate}
            \item Suppose $P \step{\alpha} P'$ with $\alpha \in A_\tau$.
            \begin{itemize}
                \item Let there exist $P^\dag, P^\ddag \in \closed$ such that $P^\dag \pathtau P \pathtau P^\ddag$, $\R(P^\dag,Q)$ and $\R(P^\ddag,Q)$. Since $P^\dag \pathtau P$ and $\R(P^\dag,Q)$, there exists a path $Q \pathtau Q_0$ such that $\R(P,Q_0)$. Since $P \step{\alpha} P'$, there exists a path $Q_0 \pathtau Q_1 \step{\opt{\alpha}} Q_2$ such that $\R(P,Q_1)$ and $\R(P',Q_2)$. Thus, there exists a path $Q \pathtau Q_1 \step{\opt{\alpha}} Q_2$ such that, since $\R \subseteq \R'$, $\R'(P,Q_1)$ and $\R'(P',Q_2)$.
                \item Let there exist $Q^\dag, Q^\ddag \in \closed$ such that $Q^\dag \pathtau Q \pathtau Q^\ddag$, $\R(P,Q^\dag)$ and $\R(P,Q^\ddag)$. Since $\R(P,Q^\ddag)$, there exists a path $Q^\ddag \pathtau Q_1 \step{\opt{\alpha}} Q_2$ such that $\R(P,Q_1)$ and $\R(P',Q_2)$. Since $Q \pathtau Q^\ddag$, there exists a path $Q \pathtau Q_1 \step{\opt{\alpha}} Q_2$ such that, since $\R \subseteq \R'$, $\R'(P,Q_1)$ and $\R'(P',Q_2)$.
            \end{itemize}
            \item For all $Y \subseteq A$, $\R'(P,Y,Q)$ by definition of $\R'$.
        \end{enumerate}
        \item Let $\R'(P,X,Q)$.
        \begin{enumerate}
            \item Suppose $P \steptau P'$.
            \begin{itemize}
                \item Let there exist $P^\dag, P^\ddag \mathbin\in \closed$ such that $P^\dag \mathbin{\pathtau} P \mathbin{\pathtau} P^\ddag$, $\R(P^\dag\!,X,Q)$ and $\R(P^\ddag\!,X,Q)$. Since $P^\dag \pathtau P$ and $\R(P^\dag,X,Q)$, there exists a path $Q \pathtau Q_0$ such that $\R(P,X,Q_0)$. Since $P \steptau P'$, there exists a path $Q_0 \pathtau Q_1 \step{\opt{\tau}} Q_2$ such that $\R(P,X,Q_1)$ and $\R(P',X,Q_2)$. Thus, there exists a path $Q \pathtau Q_1 \step{\opt{\tau}} Q_2$ such that, since $\R \subseteq \R'$, $\R'(P,X,Q_1)$ and $\R'(P',X,Q_2)$.
                \item Let there exist $Q^\dag, Q^\ddag \mathbin\in \closed$ such that $Q^\dag \mathbin{\pathtau} Q \mathbin{\pathtau} Q^\ddag$, $\R(P,X,Q^\dag)$ and $\R(P,X,Q^\ddag)$. Since $\R(P,X,Q^\ddag)$, there exists a path $Q^\ddag \pathtau Q_1 \step{\opt{\tau}} Q_2$ such that $\R(P,X,Q_1)$ and $\R(P',X,Q_2)$. Since $Q \pathtau Q^\ddag$, there exists a path $Q \pathtau Q_1 \step{\opt{\tau}} Q_2$ such that, since $\R \subseteq \R'$, $\R'(P,X,Q_1)$ and $\R'(P',X,Q_2)$.
            \end{itemize}
            \item Suppose $P \step{a} P'$ with $a \in X$.
            \begin{itemize}
                \item Let there exist $P^\dag, P^\ddag \mathbin\in \closed$ such that $P^\dag \mathbin{\pathtau} P \mathbin{\pathtau} P^\ddag$, $\R(P^\dag\!,X,Q)$ and $\R(P^\ddag\!,X,Q)$. Since $P^\dag \pathtau P$ and $\R(P^\dag,X,Q)$, there exists a path $Q \pathtau Q_0$ such that $\R(P,X,Q_0)$. Since $P \step{a} P'$, there exists a path $Q_0 \pathtau Q_1 \step{a} Q_2$ such that $\R(P,X,Q_1)$ and $\R(P',Q_2)$. Thus, there exists a path $Q \pathtau Q_1 \step{a} Q_2$ such that, since $\R \subseteq \R'$, $\R'(P,X,Q_1)$ and $\R'(P',Q_2)$.
                \item Let there exist $Q^\dag, Q^\ddag \mathbin\in \closed$ such that $Q^\dag \mathbin{\pathtau} Q \mathbin{\pathtau} Q^\ddag$, $\R(P,X,Q^\dag)$ and $\R(P,X,Q^\ddag)$. Since $\R(P,X,Q^\ddag)$, there exists a path $Q^\ddag \pathtau Q_1 \step{a} Q_2$ such that $\R(P,X,Q_1)$ and $\R(P',Q_2)$. Since $Q \pathtau Q^\ddag$, there exists a path $Q \pathtau Q_1 \step{a} Q_2$ such that, since $\R \subseteq \R'$, $\R'(P,X,Q_1)$ and $\R'(P',Q_2)$.
            \end{itemize}
            \item Suppose $\deadend{P}{X}$.
            \begin{itemize}
                \item Let there exist $P^\dag, P^\ddag \mathbin\in \closed$ such that $P^\dag \mathbin{\pathtau} P \mathbin{\pathtau} P^\ddag$, $\R(P^\dag\!,X,Q)$ and $\R(P^\ddag\!,X,Q)$. Since $P^\dag \pathtau P$ and $\R(P^\dag,X,Q)$, there exists a path $Q \pathtau Q_0$ such that $\R(P,X,Q_0)$. Since $\deadend{P}{X}$, there exists a path $Q_0 \pathtau Q_0'$ such that $\R(P,Q_0')$. Thus, there exists a path $Q \pathtau Q_0'$ such that, since $\R \subseteq \R'$, $\R'(P,Q'_0)$.
                \item Let there exist $Q^\dag, Q^\ddag \mathbin\in \closed$ such that $Q^\dag \mathbin{\pathtau} Q \mathbin{\pathtau} Q^\ddag$, $\R(P,X,Q^\dag)$ and $\R(P,X,Q^\ddag)$. Since $\R(P,X,Q^\ddag)$, there exists a path $Q^\ddag \pathtau Q_0$ such that $\R(P,Q_0)$. Since $Q \pathtau Q^\ddag$, there exists a path $Q \pathtau Q_0$ such that, since $\R \subseteq \R'$, $\R'(P,Q_0)$.
            \end{itemize}
            \item Suppose $\deadend{P}{X}$ and $P \step{\rt} P'$.
            \begin{itemize}
                \item Let there exist $P^\dag, P^\ddag \mathbin\in \closed$ such that $P^\dag \mathbin{\pathtau} P \mathbin{\pathtau} P^\ddag$, $\R(P^\dag\!,X,Q)$ and $\R(P^\ddag\!,X,Q)$. Since $P^\dag \pathtau P$ and $\R(P^\dag,X,Q)$, there exists a path $Q \pathtau Q_0$ such that $\R(P,X,Q_0)$. Since $\deadend{P}{X}$ and $P \step{\rt} P'$, there exists a path $Q_0 \pathtau Q_1 \step{\rt} Q_2 \pathtau Q_3 \step{\rt} ... \pathtau Q_{2r-1} \step{\opt{\rt}} Q_{2r}$ with $r>0$, such that $\forall i \in [0,r{-}1],\linebreak \R(P,X,Q_{2i}) \wedge \deadend{Q_{2i+1}}{X}$ and $\R(P',X,Q_{2r})$. Thus, there exists a path $Q =: Q_0 \pathtau Q_1 \step{\rt} Q_2 \pathtau Q_3 \step{\rt} ... \pathtau Q_{2r-1} \step{\opt{\rt}} Q_{2r}$ with $r>0$, such that $\forall i \in [0,r{-}1], \R(P,X,Q_{2i}) \wedge \deadend{Q_{2i+1}}{X}$ and $\R(P',X,Q_{2r})$.
                \item Let there exist $Q^\dag, Q^\ddag \mathbin\in \closed$ such that $Q^\dag \mathbin{\pathtau} Q \mathbin{\pathtau} Q^\ddag$, $\R(P,X,Q^\dag)$ and $\R(P,X,Q^\ddag)$. Since $\R(P,X,Q^\ddag)$, there exists a path $Q^\ddag =: Q_0 \pathtau Q_1 \step{\rt} Q_2 \pathtau Q_3 \step{\rt} ... \pathtau Q_{2r-1} \step{\opt{\rt}} Q_{2r}$ with $r>0$, such that $\forall i \in [0,r{-}1], \R(P,X,Q_{2i}) \wedge \deadend{Q_{2i+1}}{X}$ and $\R(P',X,Q_{2r})$. Since $Q \pathtau Q^\ddag$, there exists a path $Q =: Q_0 \pathtau Q_1 \step{\rt} Q_2 \pathtau Q_3 \step{\rt} ... \pathtau Q_{2r-1} \step{\opt{\rt}} Q_{2r}$ with $r>0$, such that $\forall i \in [0,r{-}1], \R(P,X,Q_{2i}) \wedge \deadend{Q_{2i+1}}{X}$ and $\R(P',X,Q_{2r})$.
            \end{itemize}
            \item Suppose $P \nsteptau$.
            \begin{itemize}
                \item Let there exist $P^\dag, P^\ddag \mathbin\in \closed$ such that $P^\dag \mathbin{\pathtau} P \mathbin{\pathtau} P^\ddag$, $\R(P^\dag\!,X,Q)$ and $\R(P^\ddag\!,X,Q)$. Since $P^\dag \pathtau P$ and $\R(P^\dag,X,Q)$, there exists a path $Q \pathtau Q_0$. Since $P \nsteptau$, there exists a path $Q_0 \pathtau Q_1 \nsteptau$. Thus, there exists a path $Q \pathtau Q_1 \nsteptau$.
                \item Let there exist $Q^\dag, Q^\ddag \mathbin\in \closed$ such that $Q^\dag \mathbin{\pathtau} Q \mathbin{\pathtau} Q^\ddag$, $\R(P,X,Q^\dag)$ and $\R(P,X,Q^\ddag)$. Since $\R(P,X,Q^\ddag)$, there exists a path $Q^\ddag \pathtau Q_1 \nsteptau$. Since $Q \pathtau Q^\ddag$, there exists a path $Q \pathtau Q_1 \nsteptau$.
           \popQED
            \end{itemize}
        \end{enumerate}
    \end{enumerate}
\end{proof}

\begin{proof}[Proof of Proposition \ref{prop:equivalence}]
    Let $\R_1$ and $\R_2$ be two \tb reactive bisimulations and define
    \begin{align*}
        \R := (\R_1 \circ \R_2)\cup(\R_2\circ \R_1)
    \end{align*}
    $\R$ is clearly symmetric by definition. Let's check that $\R$ is a \tb reactive bisimulation. Let $P,Q \in \closed$ and $X \subseteq A$.
    \begin{enumerate}
        \item If $\R(P,Q)$ then there exists $R \in \closed$ such that $\R_1(P,R)$ and $\R_2(R,Q)$, or $\R_2(P,R)$ and $\R_1(R,Q)$. The two possibilities are similar; thus, suppose without loss of generality that $\R_1(P,R)$ and $\R_2(R,Q)$.
        \begin{enumerate}
            \item If $P \step{\alpha} P'$ with $\alpha \in A_\tau$ then, since $\R_1(P,R)$, there exists a path $R \pathtau R_1 \step{\opt{\alpha}} R_2$ such that $\R_1(P,R_1)$ and $\R_1(P',R_2)$. Since $\R_2(R,Q)$ and $R \pathtau R_1$, there exists a path $Q \pathtau Q_0$ such that $\R_2(R_1,Q_0)$. Since $R_1 \step{\opt{\alpha}} R_2$, there exists a path $Q_0 \pathtau Q_1 \step{\opt{\alpha}} Q_2$ such that $\R_2(R_1,Q_1)$ and $\R_2(R_2,Q_2)$. By definition of $\R$, there exists a path $Q \pathtau Q_1 \step{\opt{\alpha}} Q_2$ such that $\R(P,Q_1)$ and $\R(P',Q_2)$.
            \item For all $Y \subseteq A$, since $\R_1(P,R)$ and $\R_2(R,Q)$, $\R_1(P,Y,R)$ and $\R_2(R,Y,Q)$, thus, $\R(P,Y,Q)$.
        \end{enumerate}
        \item If $\R(P,X,Q)$ then there exists $R \in \closed$ such that $\R_1(P,X,R)$ and $\R_2(R,X,Q)$, or $\R_2(P,X,R)$ and $\R_1(R,X,Q)$. The two possibilities are similar; thus, suppose without loss of generality that $\R_1(P,X,R)$ and $\R_2(R,X,Q)$.
        \begin{enumerate}
            \item If $P \steptau P'$ then, since $\R_1(P,X,R)$, there exists a path $R \pathtau R_1 \step{\opt{\tau}} R_2$ such that $\R_1(P,X,R_1)$ and $\R_1(P',X,R_2)$. Since $\R_2(R,X,Q)$ and $R \pathtau R_1$, there exists a path $Q \pathtau Q_0$ such that $\R_2(R_1,X,Q_0)$. Since $R_1 \step{\opt{\tau}} R_2$, there exists a path $Q_0 \pathtau Q_1 \step{\opt{\tau}} Q_2$ such that $\R_2(R_1,X,Q_1)$ and $\R_2(R_2,X,Q_2)$. By definition of $\R$, there exists a path $Q \pathtau Q_1 \step{\opt{\tau}} Q_2$ such that $\R(P,X,Q_1)$ and $\R(P',X,Q_2)$.
            \item If $P \step{a} P'$ with $a \in X$ then, since $\R_1(P,X,R)$, there exists a path $R \pathtau R_1 \step{a} R_2$ such that $\R_1(P,X,R_1)$ and $\R_1(P',R_2)$. Since $\R_2(R,X,Q)$ and $R \pathtau R_1$, there exists a path $Q \pathtau Q_0$ such that $\R_2(R_1,X,Q_0)$. Since $R_1 \step{a} R_2$, there exists a path $Q_0 \pathtau Q_1 \step{a} Q_2$ such that $\R_2(R_1,X,Q_1)$ and $\R_2(R_2,Q_2)$. By definition of $\R$, there exists a path $Q \pathtau Q_1 \step{a} Q_2$ such that $\R(P,X,Q_1)$ and $\R(P',Q_2)$.
            \item If $\deadend{P}{X}$ then, since $P \nsteptau$, there exists a path $R \pathtau R_0 \nsteptau$. Moreover, using Clause 2.a, $\R_1(P,X,R_0)$. Moreover, there exists a path $R_0 \pathtau R'_0$ such that $\R_1(P,R'_0)$, but, since $R_0 \nsteptau$, $R_0 = R'_0$. By Clause 1.a, $\init{P} =\init{R_0}$, so $\deadend{R_0}{X}$. Since $\R_2(R,X,Q)$ and $R \pathtau R_0$, there exists a path $Q \pathtau Q_0$ such that $\R_2(R_0,X,Q_0)$. Moreover, since $\deadend{R_0}{X}$, there exists a path $Q_0 \pathtau Q_0'$ such that $\R_2(R_0,Q'_0)$. Thus, there exists a path $Q \pathtau Q_0'$ such that, by definition of $\R$, $\R(P,Q'_0)$.
            \item If $\deadend{P}{X}$ and $P \step{\rt} P'$ then, since $\R_1(P,X,R)$, there exists a path $R = R_0 \pathtau R_1 \step{\rt} R_2 \pathtau R_3 \step{\rt} ... \pathtau R_{2r-1} \step{\opt{\rt}} R_{2r}$ with $r>0$, such that $\forall i \in [0,r{-}1], \R_1(P,X,R_{2i}) \wedge \deadend{R_{2i+1}}{X}$ and $\R_1(P',X,R_{2r})$. For all $i \in [0,r{-}1]$, $\R_1(P,X,R_{2i})$, $R_{2i} \pathtau R_{2i+1}$ and $P \nsteptau$, therefore, for all $i \in [0,r{-}1]$, $\R_1(P,X,R_{2i+1})$. Since $\R_2(R,X,Q)$, $\deadend{R_{2r-1}}{X}$ and $R_{2r-1} \step{\opt{\rt}} R_{2r}$, there exists a path $Q = Q_0 \pathtau Q_1 \step{\rt} Q_2 \pathtau Q_3 \step{\rt} ... \pathtau Q_{2k-1} \step{\opt{\rt}} Q_{2k}$ with $k>0$, such that $\forall j \in [0,k{-}1], \exists i \in [0,2r{-}1], \R_2(R_i,X,Q_{2j}) \wedge \deadend{Q_{2j+1}}{X}$ and $\R_2(R_{2r},X,Q_{2k})$.
            By definition of $\R$, $\forall j \in [0,k{-}1], \R(P,X,Q_{2j})$ and $\R(P',X,Q_{2k})$.
            \item If $P \nsteptau$ then, since $\R_1(P,X,R)$, there exists a path $R \pathtau R_0 \nsteptau$. Since $\R_2(R,X,Q)$ and $R \pathtau R_0$, there exists a path $Q \pathtau Q_0$ such that $\R_2(R_0,X,Q_0)$. Since $R_0 \nsteptau$, there exists a path $Q_0 \pathtau Q_0' \nsteptau$. Hence there exists a path $Q \pathtau Q_0' \nsteptau$.
           \popQED
        \end{enumerate}
    \end{enumerate}
\end{proof}

\section{Proof of Modal Characterisation} \label{app:modal}

\begin{proof}[Proof of Theorem \ref{thm:modal characterisation}]
    $(\Rightarrow)$ We are going to prove by structural induction on $\logic_b$ and $\logic_b^r$ that, for all $P,Q \in \closed$, $X \subseteq A$, $\varphi \in \logic_b$ and $\psi \in \logic_b^r$,
    \begin{itemize}
        \item if $P \bisimtbr Q$ and $P \models \varphi$ then $Q \models \varphi$
        \item if $P \bisimtbr[X] Q$ and $P \models_X \varphi$ then $Q \models_X \varphi$
        \item if $P \bisimrtbr Q$ and $P \models \psi$ then $Q \models \psi$
        \item if $P \bisimrtbr[X] Q$ and $P \models_X \psi$ then $Q \models_X \psi$
    \end{itemize}
    Note that, in the four cases, we dispose of the contraposition. Let $P, Q \in \closed$, $X \subseteq A$, $\varphi \in \logic_b$ and $\psi \in \logic_b^r$.
    \begin{itemize}
        \item If $P \bisimtbr Q$ and $P \models \varphi$ then
        \begin{itemize}
            \item if $\varphi = \top$ then $Q \models \top$.
            \item if $\varphi = \bigwedge_{i\in I}\varphi_i$ with $(\varphi_i)_{i \in I} \in (\logic_b)^I$ then, for all $i \in I$, $P \models \varphi_i$. Thus, by induction, for all $i \in I$, $Q \models \varphi_i$. Therefore, $Q \models \bigwedge_{i\in I}\varphi_i$.
            \item if $\varphi = \neg\varphi'$ then $P \not\models \varphi'$. Thus, by induction, $Q \not\models \varphi'$. Therefore, $Q \models \neg\varphi'$.
            \item if $\varphi = \langle\epsilon\rangle (\varphi_1\langle\hat{\alpha}\rangle\varphi_2)$ then there exists a path $P \pathtau P_1 \step{\opt{\alpha}} P_2$ such that $P_1 \models \varphi_1$ and $P_2 \models \varphi_2$. Since $P \bisimtbr Q$, there exists a path $Q \pathtau Q_1 \step{\opt{\alpha}} Q_2$ such that $P_1 \bisimtbr Q_1$ and $P_2 \bisimtbr Q_2$. By induction, $Q_1 \models \varphi_1$ and $Q_2 \models \varphi_2$. Therefore, $Q \models \varphi$.
            \item if $\varphi = \varphi_1\langle\epsilon_X\rangle\varphi_2$ then there is a path $P \pathtau P_1 \step{\rt} P_2 \pathtau P_3 \step{\rt} ... \pathtau P_{2r{-}1} \step{\opt{\rt}} P_{2r}$\linebreak[3] with $r>0$, such that $P \models \varphi_1 \wedge \forall i \in [1,2r{-}1]\; P_{i} \models_X \varphi_1 \wedge P_{2r} \models_X \varphi_2$ and, moreover, $\forall i \in [0,r{-}1]\; \deadend{P_{2i+1}}{X}$. Since $P \bisimtbr Q$, there exists a path $Q \pathtau Q_1 \step{\rt} Q_2 \pathtau Q_3 \step{\rt} ... \pathtau Q_{2k-1} \step{\opt{\rt}} Q_{2k}$ with $k>0$, such that $\forall j \in [0,k{-}1]$ $\deadend{Q_{2j+1}}{X}$,  $P_{2r} \bisimtbr[X] Q_{2k}$ and $\forall j \in [1,2k{-}1]$, $\exists i \in [1,2r{-}1]$, $P_{i} \bisimtbr[X] Q_{j}$. By induction, $Q \models \varphi_1$, $\forall j \in [1,2k{-}1]$ $Q_{j} \models_X \varphi_1$ and $Q_{2k} \models_X \varphi_2$. Therefore, $Q \models \varphi_1\langle\epsilon_X\rangle\varphi_2$.
            \item if $\varphi = \langle\epsilon\rangle \neg\langle\tau\rangle\top$ then there exists a path $P \pathtau P_0 \nsteptau$. Since $P \bisimtbr Q$, there exists a path $Q \pathtau Q_0 \nsteptau$. Therefore, $Q \models \varphi$.
        \end{itemize}
        \item If $P \bisimtbr[X] Q$ and $P \models_X \varphi$ then
        \begin{itemize}
            \item if $\varphi = \top$ then $Q \models_X \top$.
            \item if $\varphi = \bigwedge_{i\in I}\varphi_i$ with $(\varphi_i)_{i \in I} \in (\logic_b)^I$ then, for all $i \in I$, $P \models_X \varphi_i$. Thus, by induction, for all $i \in I$, $Q \models_X \varphi_i$. Therefore, $Q \models_X \bigwedge_{i\in I}\varphi_i$.
            \item if $\varphi = \neg\varphi'$ then $P \not\models_X \varphi'$. Thus, by induction, $Q \not\models_X \varphi'$. Therefore, $Q \models_X \neg\varphi'$.
            \item if $\varphi = \langle\epsilon\rangle (\varphi_1\langle\hat{\alpha}\rangle\varphi_2)$ then
            \begin{itemize}
                \item if $\alpha = \tau$ then there exists a path $P \pathtau P_1 \step{\opt{\tau}} P_2$ such that $P_1 \models_X \varphi_1$ and $P_2 \models_X \varphi_2$. Since $P \bisimtbr[X] Q$, there exists a path $Q \pathtau Q_1 \step{\opt{\tau}} Q_2$ such that $P_1 \bisimtbr[X] Q_1$ and $P_2 \bisimtbr[X] Q_2$. By induction, $Q_1 \models_X \varphi_1$ and $Q_2 \models_X \varphi_2$. Therefore, $Q \models_X \varphi$.
                \item if $\alpha \in A$ then $a \in X$ or $\deadend{P}{X}$ and there exists a path $P \pathtau P_1 \step{a} P_2$ such that $P_1 \models_X \varphi_1$ and $P_2 \models \varphi_2$. Since $P \bisimtbr[X] Q$, there exists a path $Q \pathtau Q_1 \step{a} Q_2$ such that $P_1 \bisimtbr[X] Q_1$ and $P_2 \bisimtbr Q_2$. Moreover, with Lemma~\ref{lem:obvious}.4 we can get that $\deadend{P}{X} \Leftrightarrow \deadend{Q_1}{X}$. By induction, $Q_1 \models_X \varphi_1$ and $Q_2 \models \varphi_2$. Therefore, $Q \models_X \varphi$.
            \end{itemize}
            \item if $\varphi = \varphi_1\langle\epsilon_Y\rangle\varphi_2$ then there is a path $P \pathtau P_1 \step{\rt} P_2 \pathtau P_3 \step{\rt} ... \pathtau P_{2r{-}1} \step{\opt{\rt}} P_{2r}$\linebreak[3] with $r>0$, such that $\deadend{P_{1}}{X}$, $\forall i \in [1,2r{-}1]\; P_{i} \models_Y \varphi_1$, $P_{2r} \models_Y \varphi_2$ and $\forall i \in [0,r{-}1]\; \deadend{P_{2i+1}}{Y}$. Since $P \bisimtbr[X] Q$, there exists a path $Q \pathtau Q_1 \step{\rt} Q_2 \pathtau Q_3 \step{\rt} ... \pathtau Q_{2k-1} \step{\opt{\rt}} Q_{2k}$ with $k>0$, such that $\deadend{Q_{1}}{X\cup Y}$, $\forall j \in [1,k{-}1]$ $\deadend{Q_{2j+1}}{Y}$,  $P_{2r} \bisimtbr[Y] Q_{2k}$ and $\forall j \in [1,2k{-}1]$, $\exists i \in [1,2r{-}1]$, $P_{i} \bisimtbr[Y] Q_{j}$. By induction, $\forall j \in [1,2k{-}1]$ $Q_{2j} \models_Y \varphi_1$ and $Q_{2k} \models_Y \varphi_2$. Therefore, $Q \models_X \varphi_1\langle\epsilon_Y\rangle\varphi_2$.
            \item if $\varphi = \langle\epsilon\rangle \neg\langle\tau\rangle\top$ then there exists a path $P \pathtau P_0 \nsteptau$. Since $P \bisimtbr Q$, there exists a path $Q \pathtau Q_0 \nsteptau$. Therefore, $Q \models \varphi$.
        \end{itemize}
        \item If $P \bisimrtbr Q$ and $P \models \psi$ then
        \begin{itemize}
            \item if $\psi = \top$ then $Q \models \top$.
            \item if $\psi = \bigwedge_{i\in I}\psi_i$ with $(\psi_i)_{i \in I} \in (\logic_b^r)^I$ then, for all $i \in I$, $P \models \psi_i$. Thus, by induction, for all $i \in I$, $Q \models \psi_i$. Therefore, $Q \models \bigwedge_{i\in I}\psi_i$.
            \item if $\psi = \neg\psi'$ then $P \not\models \psi'$. Thus, by induction, $Q \not\models \psi'$. Therefore, $Q \models \neg\psi'$.
            \item if $\psi = \langle\alpha\rangle\varphi$ then there is a transition $P \step{\alpha} P'$ such that $P' \models \varphi$. Since $P \bisimrtbr Q$, there exists a path $Q \step{\alpha} Q'$ such that $P' \bisimtbr Q'$. By induction, $Q' \models \varphi$. Therefore, $Q \models \psi$.
            \item if $\psi = \langle t_X\rangle\varphi$ then $\deadend{P}{X}$ and there exists a transition $P \step{\rt} P'$ such that $P' \models_X \varphi$. Since $P \bisimrtbr Q$, $\deadend{Q}{X}$ and there exists a path $Q \step{\rt} Q'$ such that $P' \bisimtbr[X] Q'$. By induction, $Q' \models_X \varphi$. Therefore, $Q \models \psi$.
        \end{itemize}
        \item If $P \bisimrtbr[X] Q$ and $P \models_X \psi$ then
        \begin{itemize}
            \item if $\psi = \top$ then $Q \models_X \top$.
            \item if $\psi = \bigwedge_{i\in I}\psi_i$ with $(\psi_i)_{i \in I} \in (\logic_b^r)^I$ then, for all $i \in I$, $P \models_X \psi_i$. Thus, by induction, for all $i \in I$, $Q \models_X \psi_i$. Therefore, $Q \models_X \bigwedge_{i\in I}\psi_i$.
            \item if $\psi = \neg\psi'$ then $P \not\models_X \psi'$. Thus, by induction, $Q \not\models_X \psi'$. Therefore, $Q \models_X \neg\psi'$.
            \item if $\psi = \langle\alpha\rangle\varphi$
            \begin{itemize}
                \item if $\alpha = \tau$ then there exists a transition $P \step{\tau} P'$ such that $P' \models_X \varphi$. Since $P \bisimrtbr[X] Q$, there exists a transition $Q \step{\tau} Q'$ such that $P' \bisimtbr[X] Q'$. By induction, $Q' \models_X \varphi$. Therefore, $Q \models_X \psi$.
                \item if $\alpha \in A$ then $a \in X$ or $\deadend{P}{X}$ and there exists a transition $P \step{a} P'$ such that $P' \models \varphi$. Since $P \bisimrtbr[X] Q$, $\deadend{P}{X} \Leftrightarrow \deadend{Q}{X}$ and there exists a transition $Q \step{a} Q'$ such that $P' \bisimtbr Q'$. By induction, $Q' \models \varphi$. Therefore, $Q \models_X \psi$.
            \end{itemize}
            \item if $\psi = \langle t_Y\rangle\varphi$ then $\deadend{P}{(X\cup Y)}$ and there exists a transition $P \step{\rt} P'$ such that $P_1 \models_Y \varphi$. Since $P \bisimrtbr[X] Q$, $\deadend{Q}{(X\cup Y)}$ and there exists a transition $Q \step{\rt} Q'$ such that $P' \bisimtbr[Y] Q'$. By induction, $Q' \models_Y \varphi$. Therefore, $Q \models_X \psi$.
        \end{itemize}
    \end{itemize}

    $(\Leftarrow)$ Let $\equiv \; := \{(P,Q) \mid \forall \varphi \in \logic_{tb}, P \models \varphi \Leftrightarrow Q \models \varphi\} \cup \{(P,X,Q) \mid \forall \varphi \in \logic_{tb}, P \models_X \varphi \Leftrightarrow Q \models_X \varphi\}$, and $\equiv^r \; := \{(P,Q) \mid \forall \psi \in \logic_{tb}^r, P \models \psi \Leftrightarrow Q \models \psi\} \cup \{(P,X,Q) \mid \forall \psi \in \logic_{tb}^r,\linebreak[3] P \models_X \psi \Leftrightarrow Q \models_X \psi\}$. $(P,X,Q) \in {\equiv}$ will be denoted $P \equiv_X Q$ for clarity. Note that ${\equiv^r} \subseteq {\equiv}$. We are going to check that $\equiv$ is a generalised \tb bisimulation and $\equiv^r$ a generalised rooted \tb reactive bisimulation. Let $P,Q \in \closed$ and $X \subseteq A$.
    \begin{enumerate}
        \item If $P \equiv Q$
        \begin{enumerate}
            \item if $P \step{\alpha} P'$ then define $\mathcal{Q}^\dag := \{Q^\dag \mid Q \pathtau Q^\dag \wedge P \not\equiv Q^\dag\}$ and $\mathcal{Q}^\ddag := \{Q^\ddag \mid Q \pathtau Q^\dag \step{\opt{\alpha}} Q^\ddag \wedge P' \not\equiv Q^\ddag\}$. Since $\logic_b$ is closed under negation and conjunction, there exist two formulas $\varphi^\dag, \varphi^\ddag \in \logic_b$ such that $P \models \varphi^\dag$, $P' \models \varphi^\ddag$, for all $Q^\dag \in \mathcal{Q}^\dag$, $Q^\dag \not\models \varphi^\dag$ and, for all $Q^\ddag \in \mathcal{Q}^\ddag$, $Q^\ddag \not\models \varphi^\ddag$. Note that $P \models \langle\epsilon\rangle(\varphi^\dag\langle\hat{\alpha}\rangle\varphi^\ddag)$. Thus, $Q \models \langle\epsilon\rangle(\varphi^\dag\langle\hat{\alpha}\rangle\varphi^\ddag)$. Therefore, there exists a path $Q \pathtau Q_1 \step{\opt{\alpha}} Q_2$ such that $Q_1 \models \varphi^\dag$ and $Q_2 \models \varphi^\ddag$. By definition of $\mathcal{Q}^\dag$ and $\mathcal{Q}^\ddag$, $P \equiv Q_1$ and $P' \equiv Q_2$.
            \item if $\deadend{P}{X}$ and $P \step{\rt} P'$ then define $\mathcal{Q}^\dag := \{Q^\dag \mid Q \pathtau Q_1^\dag \step{\rt} Q_2^\dag \pathtau Q_3^\dag \step{\rt} ... \pathtau Q^\dag \wedge P \not\equiv_X Q^\dag\}$ and $\mathcal{Q}^\ddag := \{Q^\ddag \mid Q \pathtau Q_1^\dag \step{\rt} Q_2^\dag \pathtau Q_3^\dag \step{\rt} ... \pathtau Q^\dag_{2r-1} \step{\opt{\rt}} Q^\ddag \wedge P' \not\equiv_X Q^\ddag\}$. Since $\logic_b$ is closed under negation and conjunction, there exist two formulas $\varphi^\dag, \varphi^\ddag \mathbin\in \logic_b$ such that $P \models_X \varphi^\dag\!$, $P' \models_X \varphi^\ddag\!$, for all $Q^\dag \mathbin\in \mathcal{Q}^\dag$, $Q^\dag \mathbin{\not\models_X} \varphi^\dag$ and, for all $Q^\ddag \in \mathcal{Q}^\ddag$, $Q^\ddag \not\models_X \varphi^\ddag$. Note that $P \models \varphi^\dag\langle\epsilon_X\rangle\varphi^\ddag$. Thus, $Q \models \varphi^\dag\langle\epsilon_X\rangle\varphi^\ddag$. Therefore, there exists a path $Q \pathtau Q_1 \step{\rt} Q_2 \pathtau Q_3 \step{\rt} ... \pathtau Q_{2r-1} \step{\opt{\rt}} Q_{2r}$ with $r>0$, such that  $\forall i \in [0,r{-}1]$ $\deadend{Q_{2i+1}}{X}$, $\forall i \in [1,2r{-}1]$ $Q_{i} \models_X \varphi^\dag$ and $Q_{2r} \models_X \varphi^\ddag$.
            By definition of $\mathcal{Q}^\dag$ and $\mathcal{Q}^\ddag$, $\forall i \in [1,r{-}1]$, $P \equiv_X Q_{2i}$ and $P' \equiv_X Q_{2r}$.
            \item if $P \nsteptau$ then $P \models \langle\epsilon\rangle \neg\langle\tau\rangle\top$. Thus $Q \models \langle\epsilon\rangle \neg\langle\tau\rangle\top$. Therefore, $Q \pathtau Q_1 \nsteptau$.
       \end{enumerate}
        \item If $P \equiv_X Q$ then
        \begin{enumerate}
            \item if $P \step{\tau} P'$ then define $\mathcal{Q}^\dag := \{Q^\dag \mid Q \pathtau Q^\dag \wedge P \not\equiv_X Q^\dag\}$ and $\mathcal{Q}^\ddag := \{Q^\ddag \mid Q \pathtau Q^\ddag \wedge P' \not\equiv_X Q^\ddag\}$. Since $\logic_b$ is closed under negation and conjunction, there exist two formulas $\varphi^\dag, \varphi^\ddag \in \logic_b$ such that $P \models_X \varphi^\dag$, $P' \models_X \varphi^\ddag$, for all $Q^\dag \in \mathcal{Q}^\dag$, $Q^\dag \not\models_X \varphi^\dag$ and, for all $Q^\ddag \in \mathcal{Q}^\ddag$, $Q^\ddag \not\models_X \varphi^\ddag$. Note that $P \models_X \langle\epsilon\rangle(\varphi^\dag\langle\hat{\tau}\rangle\varphi^\ddag)$. Thus, $Q \models_X \langle\epsilon\rangle(\varphi^\dag\langle\hat{\tau}\rangle\varphi^\ddag)$. Therefore, there exists a path $Q \pathtau Q_1 \step{\opt{\tau}} Q_2$ such that $Q_1 \models_X \varphi^\dag$ and $Q_2 \models_X \varphi^\ddag$. By definition of $\mathcal{Q}^\dag$ and $\mathcal{Q}^\ddag$, $P \equiv_X Q_1$ and $P' \equiv_X Q_2$.
            \item if $P \step{a} P'$ with $a \in X$ or $\deadend{P}{X}$ then define $\mathcal{Q}^\dag := \{Q^\dag \mid Q \pathtau Q^\dag \wedge P \not\equiv_X Q^\dag\}$ and $\mathcal{Q}^\ddag := \{Q^\ddag \mid Q \pathtau Q^\dag \step{a} Q^\ddag \wedge P' \not\equiv Q^\ddag\}$. Since $\logic_b$ is closed under negation and conjunction, there exist two formulas $\varphi^\dag, \varphi^\ddag \in \logic_b$ such that $P \models_X \varphi^\dag$, $P' \models \varphi^\ddag$, for all $Q^\dag \in \mathcal{Q}^\dag$, $Q^\dag \not\models_X \varphi^\dag$ and, for all $Q^\ddag \in \mathcal{Q}^\ddag$, $Q^\ddag \not\models \varphi^\ddag$. Note that $P \models_X \langle\epsilon\rangle(\varphi^\dag\langle\hat{\alpha}\rangle\varphi^\ddag)$. Thus, $Q \models_X \langle\epsilon\rangle(\varphi^\dag\langle\hat{\alpha}\rangle\varphi^\ddag)$. Therefore, there exists a path $Q \pathtau Q_1 \step{\opt{\alpha}} Q_2$ such that $a \in X \vee \deadend{Q_1}{X}$, $Q_1 \models_X \varphi^\dag$ and $Q_2 \models \varphi^\ddag$. By definition of $\mathcal{Q}^\dag$ and $\mathcal{Q}^\ddag$, $P \equiv_X Q_1$ and $P' \equiv Q_2$.
            \item if $\deadend{P}{X\cup Y}$ and $P \step{\rt} P'$ then define $\mathcal{Q}^\dag := \{Q^\dag \mid Q \pathtau Q_1^\dag \step{\rt} Q_2^\dag \pathtau Q_3^\dag \step{\rt} ... \pathtau Q^\dag \wedge P \not\equiv_Y Q^\dag\}$ and $\mathcal{Q}^\ddag := \{Q^\ddag \mid Q \pathtau Q_1^\dag \step{\rt} Q_2^\dag \pathtau Q_3^\dag \step{\rt} ... \pathtau Q^\dag_{2r-1} \step{\opt{\rt}} Q^\ddag \wedge P' \not\equiv_Y Q^\ddag\}$. Since $\logic_b$ is closed under negation and conjunction, there exist two formulas $\varphi^\dag, \varphi^\ddag \mathbin\in \logic_b$ such that $P \models_Y \varphi^\dag\!$, $P' \models_Y \varphi^\ddag\!$, for all $Q^\dag \mathbin\in \mathcal{Q}^\dag$, $Q^\dag \mathbin{\not\models_Y} \varphi^\dag$ and, for all $Q^\ddag \in \mathcal{Q}^\ddag$, $Q^\ddag \not\models_Y \varphi^\ddag$. Note that $P \models_X \varphi^\dag\langle\epsilon_Y\rangle\varphi^\ddag$. Thus, $Q \models_X \varphi^\dag\langle\epsilon_Y\rangle\varphi^\ddag$. Therefore, there exists a path $Q \pathtau Q_1 \step{\rt} Q_2 \pathtau Q_3 \step{\rt} ... \pathtau Q_{2r-1} \step{\opt{\rt}} Q_{2r}$ with $r>0$, such that $\deadend{Q_1}{X}$, $\forall i \in [0,r{-}1]$ $\deadend{Q_{2i+1}}{Y}$, $\forall i \in [1,2r{-}1]$ $Q_{i} \models_Y \varphi^\dag$ and $Q_{2r} \models_Y \varphi^\ddag$. By definition of $\mathcal{Q}^\dag$ and $\mathcal{Q}^\ddag$, $\forall i \in [1,r{-}1]$, $P \equiv_Y Q_{2i}$ and $P' \equiv_Y Q_{2r}$.
            \item if $P \nsteptau$ then $P \models_X \langle\epsilon\rangle \neg\langle\tau\rangle\top$. Thus $Q \models_X \langle\epsilon\rangle \neg\langle\tau\rangle\top$. Therefore, $Q \pathtau Q_1 \nsteptau$.
        \end{enumerate}
    \end{enumerate}
    \begin{enumerate}
        \item If $P \equiv^r Q$ then
        \begin{enumerate}
            \item if $P \step{\alpha} P'$ with $\alpha \in A_\tau$ then define $\mathcal{Q}^\ddag := \{Q^\ddag \mid Q \step{\alpha} Q^\ddag \wedge P' \not\equiv Q^\ddag\}$. Since $\logic_b^r$ is closed under negation and conjunction, there exist a formula $\varphi^\ddag \in \logic_b^r$ such that $P' \models \varphi^\ddag$ and, for all $Q^\ddag \in \mathcal{Q}^\ddag$, $Q^\ddag \not\models \varphi^\ddag$. Note that $P \models \langle\alpha\rangle\varphi^\ddag$. Thus, $Q \models \langle\alpha\rangle\varphi^\ddag$. Therefore, there\linebreak[3] exists a transition $Q \step{\alpha} Q'$ such that $Q' \models \varphi^\ddag$. By definition of $\mathcal{Q}^\ddag$, $P' \equiv Q'$.
            \item if $\deadend{P}{X}$ and $P \step{\rt} P'$ then define $\mathcal{Q}^\ddag := \{Q^\ddag \mid Q \step{\rt} Q^\ddag \wedge P' \not\equiv_X Q^\ddag\}$. Since $\logic_b^r$ is closed under negation and conjunction, there exist a formula $\varphi^\ddag \in \logic_b^r$ such that $P' \models_X \varphi^\ddag$ and, for all $Q^\ddag \in \mathcal{Q}^\ddag$, $Q^\ddag \not\models_X \varphi^\ddag$. Note that $P \models \langle\rt_X\rangle\varphi^\ddag$. Thus, $Q \models \langle\rt_X\rangle\varphi^\ddag$. Therefore, there exists a transition $Q \step{\rt} Q'$ such that $Q' \models_X \varphi^\ddag$. By definition of $\mathcal{Q}^\ddag$, $P' \equiv_X Q'$.
        \end{enumerate}
        \item If $P \equiv^r_X Q$ then 
        \begin{enumerate}
            \item if $P \step{\tau} P'$ then define $\mathcal{Q}^\ddag := \{Q^\ddag \mid Q \step{\tau} Q^\ddag \wedge P' \not\equiv_X Q^\ddag\}$. Since $\logic_b^r$ is closed under negation and conjunction, there exist a formula $\varphi^\ddag \in \logic_b^r$ such that $P' \models_X \varphi^\ddag$ and, for all $Q^\ddag \in \mathcal{Q}^\ddag$, $Q^\ddag \not\models_X \varphi^\ddag$. Note that $P \models_X \langle\tau\rangle\varphi^\ddag$. Thus, $Q \models_X \langle\tau\rangle\varphi^\ddag$. Therefore, there exists a transition $Q \step{\tau} Q'$ such that $Q' \models_X \varphi^\ddag$. By definition of $\mathcal{Q}^\ddag$, $P' \equiv_X Q'$.
            \item if $P \step{a} P'$ with $a \in X \vee \deadend{P}{X}$ then define $\mathcal{Q}^\ddag := \{Q^\ddag \mid Q \step{a} Q^\ddag \wedge P' \not\equiv Q^\ddag\}$. Since $\logic_b^r$ is closed under negation and conjunction, there exist a formula $\varphi^\ddag \in \logic_b^r$ such that $P' \models \varphi^\ddag$ and, for all $Q^\ddag \in \mathcal{Q}^\ddag$, $Q^\ddag \not\models \varphi^\ddag$. Note that $P \models_X \langle a\rangle\varphi^\ddag$. Thus, $Q \models_X \langle a\rangle\varphi^\ddag$. Therefore, there exists a path $Q \step{\alpha} Q'$ such that $Q' \models \varphi^\ddag$. By definition of $\mathcal{Q}^\ddag$, $P' \equiv Q'$.
            \item if $\deadend{P}{(X\cup Y)}$ and $P \step{\rt} P'$ then define $\mathcal{Q}^\ddag := \{Q^\ddag \mid Q \step{\rt} Q^\ddag \wedge P' \not\equiv_Y Q^\ddag\}$. Since $\logic_b^r$ is closed under negation and conjunction, there exist a formula $\varphi^\ddag \in \logic_b^r$ such that $P' \models_Y \varphi^\ddag$ and, for all $Q^\ddag \in \mathcal{Q}^\ddag$, $Q^\ddag \not\models_Y \varphi^\ddag$. Note that $P \models_X \langle\rt_Y\rangle\varphi^\ddag$. Thus, $Q \models_X \langle\rt_Y\rangle\varphi^\ddag$. Therefore, there exists a path $Q \step{\rt} Q'$ such that $Q' \models_Y \varphi^\ddag$. By definition of $\mathcal{Q}^\ddag$, $P' \equiv_Y Q'$.
        \popQED
        \end{enumerate}
    \end{enumerate}
\end{proof}

\section{Correctness of Time-out Bisimulation} \label{app:time-out}

\begin{proof}[Proof of Proposition \ref{prop:time-out bisim}]
    Let $\R$ be a \tb reactive bisimulation, let's define
    \begin{align*}
        \tbisim := \{(P,Q) \mid \R(P,Q)\} \cup \{(\theta_X(P),\theta_X(Q)) \mid \R(P,X,Q)\}
    \end{align*}
    We are going to show that $\tbisim$ is a \tb time-out bisimulation. Let $P,Q \in \closed$ such that $P \tbisim Q$. By definition of $\tbisim$, $\R(P,Q)$ or $P = \theta_X(P^\dag)$, $Q = \theta_X(Q^\dag)$ and $\R(^\dag,X,Q^\dag)$.
    \begin{enumerate}
        \item If $P \step{\alpha} P'$ with $\alpha \in A_\tau$ then 
        \begin{itemize}
            \item if $\R(P,Q)$ then there exists a path $Q \pathtau Q_1 \step{\opt{\alpha}} Q_2$ such that $\R(P,Q_1)$ and $\R(P',Q_2)$. Thus, by definition of $\tbisim$, $P \tbisim Q_1$ and $P \tbisim Q_2$.
            \item if $P = \theta_X(P^\dag)$, $Q = \theta_X(Q^\dag)$ and $\R(P^\dag,X,Q^\dag)$ then
            \begin{itemize}
                \item if $\alpha = \tau$ then, by the semantics, $P' = \theta_X(P^\ddag)$ and $P^\dag \steptau P^\ddag$. Since $\R(P^\dag,X,Q^\dag)$, there exists a path $Q^\dag \pathtau Q^\dag_1 \step{\opt{\tau}} Q^\dag_2$ such that $\R(P^\dag,X,Q^\dag_1)$ and $\R(P^\ddag,X,Q^\dag_2)$. By the semantics, there exists a path $Q \pathtau \theta_X(Q^\dag_1) \step{\opt{\tau}} \theta_X(Q^\dag_2)$ such that, by the definition of $\tbisim$, $P \tbisim \theta_X(Q^\dag_1)$ and $P' \tbisim \theta_X(Q^\dag_2)$.
                \item if $\alpha = a \in A$ then, by the semantics, $P^\dag \step{a} P'$ and $a \in X \vee \deadend{P^\dag}{X}$.
                \begin{itemize}
                   \item if $a \in X$ then, since $\R(P^\dag,X,Q^\dag)$, there exists a path $Q^\dag \pathtau Q^\dag_1 \step{a} Q_2$ such that $\R(P^\dag,X,Q^\dag_1)$ and $\R(P',Q_2)$.  By the semantics, there exists a path $Q \pathtau \theta_X(Q^\dag_1) \step{a} Q_2$ such that, by the definition of $\tbisim$, $P \tbisim \theta_X(Q^\dag_1)$ and $P' \tbisim Q_2$.
                   \item if $\deadend{P^\dag}{X}$, then there is a path $Q^\dag \pathtau Q^\dag_0\nsteptau$ with $\R(P^\dag,Q^\dag_0)$. Now there exists a path $Q^\dag_0 \pathtau Q^\dag_1 \step{a} Q_2$ such that $\R(P^\dag,Q^\dag_{1})$ and $\R(P',Q_2)$. Moreover, we find that $\deadend{P^\dag}{X} \Leftrightarrow \deadend{Q^\dag_1}{X}$. By the semantics, there exists a path $Q \pathtau \theta_X(Q^\dag_1) \step{a} Q_2$ such that, by the definition of $\tbisim$, $P \tbisim \theta_X(Q_1)$ and $P' \tbisim Q_2$.
                \end{itemize}
            \end{itemize}
        \end{itemize}
        \item If $\deadend{P}{X}$ and $P \step{\rt} P'$ then 
        \begin{itemize}
            \item if $\R(P,Q)$ then $\R(P,X,Q)$, so there exists a path $Q = Q_0 \pathtau Q_1 \step{\rt} Q_2 \pathtau Q_3 \step{\rt} ... \pathtau Q_{2r-1} \step{\opt{\rt}} Q_{2r}$ with $r>0$, such that $\forall i \in [0,r{-}1],\; \R(P,X,Q_{2i}) \wedge \deadend{Q_{2i+1}}{X}$ and $\R(P',X,Q_{2r})$. So $Q_1\nsteptau$. By definition of $\tbisim$, $\forall i \in [1,r{-}1],\; \theta_X(P) \tbisim \theta_X(Q_{2i})$ and $\theta_X(P') \tbisim \theta_X(Q_{2r})$.
            \item if $P = \theta_Y(P^\dag)$, $Q = \theta_Y(Q^\dag)$ and $\R(P^\dag,Y,Q^\dag)$ then, by the semantics of $\theta_Y$, $\deadend{P^\dag}{Y}$ and $P^\dag \step{\rt} P'$. By Clause 2.c of Definition~\ref{def:intuitive}, there is a path $Q^\dag \pathtau Q'_0$ with $\R(P^\dag,Q'_0)$, and thus also $\R(P^\dag,X,Q'_0)$.
            Therefore, since $P^\dag \step{\rt} P'$, there exists a path $Q'_0 \pathtau Q_1' \step{\rt} Q_2' \pathtau Q_3' \step{\rt} ... \pathtau Q'_{2r-1} \step{\opt{\rt}} Q'_{2r}$ with $r>0$, such that $\forall i \in [0,r{-}1], \R(P^\dag,X,Q'_{2i}) \wedge \deadend{Q'_{2i+1}}{X}$ and $\R(P',X,Q'_{2r})$.
            Write $Q_0:=Q=\theta_Y(Q^\dag)$, $Q_1 := \theta_Y(Q'_1)$ and $Q_j := Q'_j$ for $j \in [2,2r]$. By the semantics, there exists a path $Q_0 \pathtau Q_1 \step{\rt} Q_2 \pathtau Q_3 \step{\rt} ... \pathtau Q_{2r+1} \step{\opt{\rt}} Q_{2r}$ with $r>0$, such that $Q_1\nsteptau$ and, by the definition of $\tbisim$, $\forall i \in [1,r{-}1],\; \theta_X(P) \tbisim \theta_X(Q_{2i}) \wedge \deadend{Q_{2i+1}}{X}$ and $\theta_X(P') \tbisim \theta_X(Q_{2r})$.
        \end{itemize}
        \item If $P \nsteptau$ then
        \begin{itemize}
            \item if $\R(P,Q)$ then $\R(P,\emptyset,Q)$, so there exists a path $Q \pathtau Q_0 \nsteptau$.
            \item if $P = \theta_X(P^\dag)$, $Q = \theta_X(Q^\dag)$ and $\R(P^\dag,X,Q^\dag)$ then, by the semantics, $P^\dag \nsteptau$. Since $\R(P^\dag,X,Q^\dag)$, there exists a path $Q^\dag \pathtau Q_0 \nsteptau$. By the semantics, there exists a path $Q \pathtau \theta_X(Q_0) \nsteptau$.
        \end{itemize}
    \end{enumerate}
    Let $\tbisim$ be a \tb time-out bisimulation, let's define
    \begin{align*}
        \R = \{(P,Q) \mid P \tbisim Q\} \cup \{(P,X,Q) \mid \theta_X(P) \tbisim \theta_X(Q)\}
    \end{align*}
    We are going to show that $\R$ is a generalised \tb reactive bisimulation. Let $P,Q \in \closed$ and $X \subseteq A$.
    \begin{enumerate}
        \item If $\R(P,Q)$ then $P \tbisim Q$.
        \begin{enumerate}
            \item If $P \step{\alpha} P'$ then there exists a path $Q \pathtau Q_1 \step{\opt{\alpha}} Q_2$ such that $P \tbisim Q_1$ and $P' \tbisim Q_2$, thus, by definition of $\R$, $\R(P,Q_1)$ and $\R(P',Q_2)$.
            \item If $\deadend{P}{X}$ and $P \step{\rt} P'$ then there exists a path $Q \pathtau Q_1 \step{\rt} Q_2 \pathtau Q_3 \step{\rt} ... \pathtau Q_{2r-1} \step{\opt{\rt}} Q_{2r}$ with $r>0$, such that $Q_1\nsteptau$, $\forall i \in [1,r{-}1],\; \theta_X(P) \tbisim \theta_X(Q_{2i}) \wedge \deadend{Q_{2i+1}}{X}$ and $\theta_X(P') \tbisim \theta_X(Q_{2r})$. Thus, by definition of $\R$, $\forall i \in [1,r{-}1],\; \R(P,X,Q_{2i})$ and $\R(P',X,Q_{2r})$.
            \item If $P \nsteptau$ then there exists a path $Q \pathtau Q_0 \nsteptau$ such that $P \tbisim Q_0$, thus, by definition of $\R$, $\R(P,Q_0)$.
        \end{enumerate}
        \item If $\R(P,X,Q)$ then $\theta_X(P) \tbisim \theta_X(Q)$.
        \begin{enumerate}
            \item If $P \steptau P'$ then, by the semantics, $\theta_X(P) \steptau \theta_X(P')$. Therefore, there exists a path $\theta_X(Q) \pathtau Q^\dag \step{\opt{\tau}} Q^\ddag$ such that $\theta_X(P) \tbisim Q^\dag$ and $\theta_X(P') \tbisim Q^\ddag$. By the semantics, $Q^\dag = \theta_X(Q_1)$, $Q^\ddag = \theta_X(Q_2)$ and $Q \pathtau Q_1 \step{\opt{\tau}} Q_2$. Moreover, by definition of $\R$, $\R(P,X,Q_1)$ and $\R(P',X,Q_2)$.
            \item If $P \step{a} P'$ with $a \in X \vee \deadend{P}{X}$ then, by the semantics, $\theta_X(P) \step{a} P'$. Therefore, there exists a path $\theta_X(Q) \pathtau Q^\dag \step{a} Q_2$ such that $\theta_X(P) \tbisim Q^\dag$ and $P' \tbisim Q_2$. By the semantics, $Q^\dag = \theta_X(Q_1)$ and $Q \pathtau Q_1 \step{a} Q_2$. Moreover, by definition of $\R$, $\R(P,X,Q_1)$ and $\R(P',Q_2)$.
            \item If $\deadend{P}{(X\cup Y)}$ and $P \step{\rt} P'$ then, by the semantics, $\deadend{\theta_X(P)}{Y}$ and $\theta_X(P) \step{\rt} P'$. Therefore, $\theta_X(Q) \pathtau Q^\dag_1 \step{\rt} Q_2 \pathtau Q_3 \step{\rt} ... \pathtau Q_{2r-1} \step{\opt{\rt}} Q_{2r}$ with $r>0$, such that $Q_1^\dag\nsteptau$, $\forall i \in [1,r{-}1]\; \theta_Y(P) \tbisim \theta_Y(Q_{2i}) \wedge \deadend{Q_{2i+1}}{Y}$ and $\theta_Y(P') \tbisim \theta_Y(Q_{2r})$. By the semantics, $Q^\dag_1 = \theta_X(Q_1)$ with $Q_1\nsteptau$ and we have $Q \pathtau Q_1 \step{\rt} Q_2 \pathtau Q_3 \step{\rt} ... \pathtau Q_{2r-1} \step{\opt{\rt}} Q_{2r}$. Moreover, by definition of $\R$, $\forall i \in [1,r{-}1],\; \R(P,Y,Q_{2i}) \wedge \deadend{Q_{2r+1}}{X}$ and $\R(P',Y,Q_{2r})$.
            \item If $P \nsteptau$ then, by the semantics, $\theta_X(P) \nsteptau$. Therefore, there exists a path $\theta_X(Q) \pathtau Q^\dag \nsteptau$ such that $\theta_X(P) \tbisim Q^\dag$. By the semantics, $Q^\dag = \theta_X(Q_0)$ and $Q \pathtau Q_0 \nsteptau$. Moreover, by definition of $\R$, $\R(P,X,Q_0)$.
        \end{enumerate}
    \end{enumerate}
    This ends the proof of Proposition~\ref{prop:time-out bisim}.1, and thereby its corollary~\ref{prop:time-out bisim}.2.\\
    Let $\R$ be a generalised rooted \tb reactive bisimulation, let's define
    \begin{align*}
        \tbisim := \{(P,Q) \mid \R(P,Q)\} \cup \{(\theta_X(P),\theta_X(Q)) \mid \R(P,X,Q)\}
    \end{align*}
    We are going to show that $\tbisim$ is a rooted \tb time-out bisimulation. Let $P,Q \in \closed$ such that $P \tbisim Q$, by definition of $\tbisim$, $\R(P,Q)$ or $P = \theta_X(P^\dag)$, $Q = \theta_X(Q^\dag)$ and $\R(P,X,Q)$.
    \begin{enumerate}
        \item If $P \step{\alpha} P'$ with $\alpha \in A_\tau$ then
        \begin{itemize}
            \item if $\R(P,Q)$ then there exists a transition $Q \step{\alpha} Q'$ such that $P'  \bisimtbr Q'$.
            \item if $P = \theta_X(P^\dag)$, $Q = \theta_X(Q^\dag)$ and $\R(P^\dag,X,Q^\dag)$ then 
            \begin{itemize}
                \item if $\alpha = \tau$ then, by the semantics, $P' = \theta_X(P^\ddag)$ and $P^\dag \steptau P^\ddag$. Since $\R(P^\dag,X,Q^\dag)$, there exists a transition $Q^\dag \steptau Q^\ddag$ such that $P^\ddag  \bisimtbr[X] Q^\ddag$. By the semantics, there exists a transition $Q \steptau \theta_X(Q^\ddag)$. Moreover, by Proposition~\ref{prop:time-out bisim}.\ref{corr}, $P'  \bisimtbr \theta_X(Q^\ddag)$.
                \item if $\alpha = a \in A$ then, by the semantics, $P^\dag \step{a} P'$ and $a \in X \vee \deadend{P^\dag}{X}$. Since $\R(P^\dag,X,Q^\dag)$, there exists a transition $Q^\dag \step{a} Q'$ such that $P'  \bisimtbr Q'$. Moreover, $\deadend{P^\dag}{X} \Leftrightarrow \deadend{Q^\dag}{X}$. By the semantics, there exists a transition $Q \step{a} Q'$ such that $P'  \bisimtbr Q'$.
            \end{itemize}
        \end{itemize}
        \item If $\deadend{P}{X}$ and $P \step{\rt} P'$ then
        \begin{itemize}
            \item if $\R(P,Q)$ then there exists a transition $Q \step{\rt} Q'$ such that $P' \bisimtbr[X] Q'$. Thus, $\theta_X(P')  \bisimtbr \theta_X(Q')$ by Proposition~\ref{prop:time-out bisim}.\ref{corr}.
            \item if $P = \theta_Y(P^\dag)$, $Q = \theta_Y(Q^\dag)$ and $\R(P^\dag,Y,Q^\dag)$ then, by the semantics, $P^\dag \step{\rt} P'$ and $\deadend{P^\dag}{Y}$. Since $\R(P^\dag,Y,Q^\dag)$, $\deadend{P}{(X\cup Y)}$ and $P^\dag \step{\rt} P'$, there exists a transition $Q^\dag \step{\rt} Q'$ such that $P' \bisimtbr[X] Q'$. Thus, $\theta_X(P')  \bisimtbr \theta_X(Q')$. Moreover, $\deadend{Q^\dag}{X}$. By the semantics, there exists a transition $Q \step{\rt} Q'$ such that $\theta_X(P')  \bisimtbr \theta_X(Q')$.
        \end{itemize}
    \end{enumerate}
    Let $\tbisim$ be a rooted \tb time-out bisimulation, let's define
    \begin{align*}
        B := \{(P,Q) \mid P \tbisim Q\} \cup \{(P,X,Q) \mid \theta_X(P) \tbisim \theta_X(Q)\}
    \end{align*}
    We are going to show that $\R$ is a generalised rooted \tb reactive bisimulation. Let $P,Q \in \closed$ and $X \subseteq A$.
    \begin{enumerate}
        \item If $\R(P,Q)$ then $P \tbisim Q$.
        \begin{enumerate}
            \item If $P \step{\alpha} P'$ with $\alpha \in A_\tau$ then there exists a transition $Q \step{\alpha} Q'$ such that $P'  \bisimtbr Q'$.
            \item If $\deadend{P}{X}$ and $P \step{\rt} P'$ then there exists a transition $Q \step{\rt} Q'$ such that $\theta_X(P')  \bisimtbr \theta_X(Q')$. Thus, $P' \bisimtbr[X] Q'$, by Proposition~\ref{prop:time-out bisim}.\ref{corr}.
        \end{enumerate}
        \item If $\R(P,X,Q)$ then $\theta_X(P) \tbisim \theta_X(Q)$.
        \begin{enumerate}
            \item If $P \steptau P'$ then, by the semantics, $\theta_X(P) \steptau \theta_X(P')$. Since $\theta_X(P) \tbisim \theta_X(Q)$, there exists a transition $\theta_X(Q) \steptau Q^\ddag$ such that $\theta_X(P')  \bisimtbr Q^\ddag$. By the semantics, $Q^\ddag = \theta_X(Q')$ and there exists a transition $Q \steptau Q'$. By Proposition~\ref{prop:time-out bisim}.\ref{corr}, $P' \bisimtbr[X] Q'$.
            \item If $P \step{a} P'$ with $a \in X \vee \deadend{P}{X}$ then, by the semantics, $\theta_X(P) \step{a} P'$. Since $\theta_X(P) \tbisim \theta_X(Q)$, there exists a transition $\theta_X(Q) \step{a} Q'$ such that $P'  \bisimtbr Q'$. By the semantics, there exists a transition $Q \step{a} Q'$ such that $P'  \bisimtbr Q'$.
            \item If $\deadend{P}{(X\cup Y)}$ and $P \step{\rt} P'$ then, by the semantics, $\theta_X(P) \step{\rt} P'$ and $\deadend{P}{Y}$. Since $\theta_X(P) \tbisim \theta_X(Q)$, there exists a transition $\theta_X(Q) \step{\rt} Q'$ such that $\theta_Y(P')  \bisimtbr \theta_Y(Q')$. By the semantics, there exists a transition $Q \step{\rt} Q'$. By Proposition~\ref{prop:time-out bisim}.\ref{corr}, $P' \bisimtbr[Y] Q'$.
        \popQED
        \end{enumerate}
    \end{enumerate}
\end{proof}

\section{Congruence Proofs for \texorpdfstring{$\bisimtbr$ and $\bisimtb$}{Branching Reactive Bisimilarity}}  \label{app:stability}

To prove congruence properties, the notion of bisimulation \emph{up to}, introduced by Milner in \cite{Mi90ccs}, is going to be helpful. Let $\bisim\,$ denote the classical notion of strong bisimilarity~\cite{Mi90ccs}:\linebreak[3] A \emph{(strong) bisimulation} is a symmetric relation ${\R} \subseteq \closed\times\closed$ such that, for all $P,Q \in \closed$ with $P \mathrel\R Q$, if $P \step{\alpha} P'$ with $\alpha \in Act$ then there is a transition $Q \step{\alpha} Q'$ such that $P' \mathrel\R Q'$; write $P \bisim Q$ if $P \mathrel\R Q$ for some strong bisimulation $\R$.
\begin{definition}\rm \label{def:up to}
    A \emph{\tb time-out bisimulation up to $\bisim$\,} is a symmetric relation ${\tbisim} \subseteq \closed\times\closed$ such that, for all $P,Q \in \closed$ with $P \tbisim Q$,
    \begin{enumerate}
        \item if $P \step{\alpha} P'$ with $\alpha \in A_\tau$ then there exists a path $Q \pathtau Q_1 \step{\opt{\alpha}} Q_2$ such that $P \upto[\bisim] Q_1$ and $P' \upto[\bisim] Q_2$
        \item if $\deadend{P}{X}$ and $P \step{\rt} P'$ then there exists a path $Q = Q_0 \pathtau Q_1 \step{\rt} Q_2 \pathtau Q_3 \step{\rt} ... \pathtau Q_{2r-1} \step{\opt{\rt}} Q_{2r}$ with $r\mathbin>0$, such that $Q_1\nsteptau$, $\forall i \mathbin\in [1,r{-}1],\; \theta_X(P) \upto[\bisim] \theta_X(Q_{2i})\linebreak[3] \wedge \deadend{Q_{2i+1}}{X}$ and $\theta_X(P') \upto[\bisim] \theta_X(Q_{2r})$
        \item if $P \nsteptau$ then there exists a path $Q \pathtau Q_0 \nsteptau$,
    \end{enumerate}
    where $\upto[\bisim]$ stands for the relational composition $\bisim \circ \tbisim \circ \bisim$\,.
\end{definition}

\begin{proposition} \label{prop:up to}
    Let $P,Q \in \closed$.  Then $P \bisimtbr Q$ iff there exists a \tb time-out bisimulation $\mathcal{B}$ up to $\bisim$ such that $P \tbisim Q$.
\end{proposition}

\begin{proof}
    First of all, a \tb time-out bisimulation is a \tb time-out bisimulation up to $\bisimtbr$ by reflexivity of $\bisim$\,. Conversely, let $\tbisim$ be a \tb bisimulation up to $\bisim$. We are going to show that $\upto[\bisim]$ is a \tb time-out bisimulation. By the reflexivity of $\bisimtbr$ this will suffice. Let $P,Q \in \closed$ such that $P \upto[\bisim] Q$. Then there exists $P^\dag, Q^\dag \in \closed$ such that $P \bisim P^\dag \tbisim Q^\dag \bisim Q$.
    \begin{enumerate}
        \item If $P \step{\alpha} P'$ with $\alpha \in A_\tau$ then, since $P \bisim P^\dag$, there exists a transition $P^\dag \step{\alpha} P^\ddag$ such that $P' \bisim P^\ddag$. Since $P^\dag \tbisim Q^\dag$, there exists a path $Q^\dag \pathtau Q^\star \step{\opt{\alpha}} Q^\ddag$ such that $P^\dag \upto[\bisim] Q^\star$ and $P^\ddag \upto[\bisim] Q^\ddag$. Since $Q^\dag \bisim Q$, there exists a path $Q \pathtau Q_1 \step{\opt{\alpha}} Q_2$ such that $Q^\star \bisim Q_1$ and $Q^\ddag \bisim Q_2$. Since $\bisim$ is transitive, $P \upto[\bisim] Q_1$ and $P' \upto[\bisim] Q_2$.
        \item If $\deadend{P}{X}$ and $P \step{\rt} P'$ then, since $P \bisim P^\dag$, $\deadend{P^\dag}{X}$ and there exists a transition $P^\dag \step{\rt} P^\ddag$ such that $P' \bisim P^\ddag$. Since $P^\dag \tbisim Q^\dag$, there exists a path $Q^\dag \pathtau Q^\dag_1 \step{\rt} Q^\dag_2 \pathtau Q^\dag_3 \step{\rt} ... \pathtau Q^\dag_{2r-1} \step{\opt{\rt}} Q^\dag_{2r}$ with $r>0$, such that $Q^\dag_1\nsteptau$, $\forall i \in [1,r{-}1],\; \theta_X(P^\dag) \upto[\bisim] \theta_X(Q^\dag_{2i}) \wedge \deadend{Q^\dag_{2i+1}}{X}$ and $\theta_X(P^\ddag) \upto[\bisim] \theta_X(Q^\dag_{2r})$. Since $Q^\dag \bisim Q$, there exists a path $Q \pathtau Q_1 \step{\rt} Q_2 \pathtau Q_3 \step{\rt} ... \pathtau Q_{2r-1} \step{\opt{\rt}} Q_{2r}$ such that $Q_1\nsteptau$, $\forall i \in [1,r{-}1],\; Q^\dag_{2i} \bisim Q_{2i} \wedge \deadend{Q_{2i+1}}{X}$ and $Q^\dag_{2r} \bisim Q_{2r}$. Since $\bisim$ is transitive and a congruence for $\theta_X$ \cite{strongreactivebisimilarity}, $\forall i \in [1,r{-}1],\; \theta_X(P) \tbisim \theta_X(Q_{2i}) \wedge \deadend{Q_{2i+1}}{X}$ and $\theta_X(P') \upto[\bisim] \theta_X(Q_{2r})$.
        \item If $P \nsteptau$ then, since $P \bisim P^\dag$, $P^\dag \nsteptau$. Since $P^\dag \tbisim Q^\dag$, there exists a path $Q^\dag \pathtau Q^\star \nsteptau$. Since $Q^\dag \bisim Q$, there exists a path $Q \pathtau Q_0 \nsteptau$ such that $Q^\star \bisim Q_0$.
    \popQED
    \end{enumerate}
\end{proof}

\noindent
The following lemma was proven in \cite[Appendix B]{strongreactivebisimilarity}. It will be useful in the proof of Proposition \ref{prop:stability}.

\begin{lemma}\label{lem:strong identities}
    Let $P,Q \in \closed$, $X,S,I \subseteq A$, $\rename \subseteq A\times A$.
    \begin{itemize}
        \item If $P \nsteptau$ and $\init{P}\cap X \subseteq S$ then $\theta_X(P \parallel_S Q) \bisim \theta_X(P \parallel_S \theta_{X \setminus (S \setminus \init{P})}(Q))$.
        \item $\theta_X(\tau_I(P)) \bisim \theta_X(\tau_I(\theta_{X \cup I}(P)))$.
        \item $\theta_X(\rename(P)) \bisim \theta_X(\rename(\theta_{\rename^{-1}(X)}(P)))$.
    \end{itemize}
\end{lemma}

\begin{proof}[Proof of Proposition \ref{prop:stability}]
    Let $\tbisim$ be the smallest relation satisfying, for all $P,Q \in \closed$,
    \begin{itemize}
        \item if $P \bisimtbr Q$ then $P \tbisim Q$
        \item if $P \tbisim Q$ and $\alpha \in Act$ then $\alpha.P \tbisim \alpha.Q$
        \item if $P_1 \tbisim Q_1$, $P_2 \tbisim Q_2$ and $S \subseteq A$ then $P_1 \parallel_S P_2 \tbisim Q_1 \parallel_S Q_2$
        \item if $P \tbisim Q$ and $I \subseteq A$ then $\tau_I(P) \tbisim \tau_I(Q)$
        \item if $P \tbisim Q$ and $\rename \subseteq A\times A$ then $\rename(P) \tbisim \rename(Q)$
        \item if $P \tbisim Q$ and $L \subseteq U \subseteq A$ then $\theta_L^U(P) \tbisim \theta_L^U(Q)$.
    \end{itemize}
    We are going to show that $\tbisim$ is a \tb time-out bisimulation up to $\bisim$\,. This implies that ${\B}={\bisimtbr}$\,, using Proposition~\ref{prop:up to}, and as $\B$ is a congruence for the operators of  Proposition \ref{prop:stability}, so is $\bisimtbr$\,. Before we do so, we show, by induction on the construction of $\B$, that
    \begin{equation}\label{stability}
    \mbox{if $P \B Q$ and $P\nsteptau$ then $Q\pathtau Q'$ for some $Q'$ with $P \B Q'$ and $\init{Q'}=\init{P}$.}
    \end{equation}
    Let $P \B Q$ and $P\nsteptau$.
    \begin{itemize}
      \item If $P \bisimtbr Q$ then, by Clause 3 of Definition~\ref{def:time-out bisim}, $Q\pathtau Q'$ for some $Q'$ with $Q'\nsteptau$. By (the symmetric counterpart of) Clause 1, one obtains $P \B Q'$. Clause 1 gives $\init{Q'}=\init{P}$.
      \item If $P = \alpha.P^\dag$ and $Q = \alpha.Q^\dag$ with $\alpha \in Act$ then note that $\alpha \ne \tau$ and take $Q':=Q$. One has $\init{Q'}=\init{P}$.
      \item If $P = P_1 \parallel_S P_1$ and $Q = Q_1 \parallel_S Q_2$ with $S \subseteq A$ and $P_i \tbisim Q_i$ for $i=1,2$, then, for $i=1,2$, $P_i\nsteptau$, so by induction $Q_i\pathtau Q'_i$ for some $Q'_i$ with $P_i \B Q'_i$ and $\init{Q'_i}=\init{P_i}$.
      Now $Q \pathtau Q'_1\|_S Q'_2$, $P \B  Q'_1\|_S Q'_2$ and $\init{Q'_1\|_S Q'_2}=\init{P}$.
      \item If $P = \tau_I(P_1)$ and $Q = \tau_I(Q_1)$ with $I \subseteq A$ and $P_1 \tbisim Q_1$, then $P_1\nsteptau$, so by induction $Q_1\pathtau Q'_1$ for some $Q'_1$ with $P_1 \B Q'_1$ and $\init{Q'_1}=\init{P_1}$. Now $Q \pathtau \tau_I(Q'_1)$, $P \B \tau_I(Q'_1)$ and $\init{\tau_I(Q'_1)}=\init{P}$.
      \item If $P = \rename(P_1)$ and $Q = \rename(Q_1)$ with $\rename \subseteq A\times A$ and $P_1 \tbisim Q_1$, then $P_1\nsteptau$, so by induction $Q_1\pathtau Q'_1$ for some $Q'_1$ with $P_1 \B Q'_1$ and $\init{Q'_1}=\init{P_1}$. Now $Q \pathtau \rename(Q'_1)$, $P \B \rename(Q'_1)$ and $\init{\rename(Q'_1)}=\init{P}$.
      \item  If $P = \theta_L^U(P_1)$ and $Q = \theta_L^U(Q_1)$ with  $L \subseteq U \subseteq A$ and $P_1 \tbisim Q_1$, then $P_1\nsteptau$, so by induction $Q_1\pathtau Q'_1$ for some $Q'_1$ with $P_1 \B Q'_1$ and $\init{Q'_1}=\init{P_1}$. Now $Q \pathtau \theta_L^U(Q'_1)$, $P \B \theta_L^U(Q'_1)$ and $\init{\theta_L^U(Q'_1)}=\init{P}$.
    \end{itemize}
We now check that $\tbisim$ is a \tb time-out bisimulation up to $\bisim$\,.
Note that $\tbisim$ is symmetric because $\bisimtbr$ is. $\bisim$ was proven to be a congruence for $\ccsp$ in \cite{strongreactivebisimilarity}. Let $P,Q \in \closed$ such that $P \tbisim Q$.
    \begin{enumerate}
        \item If $P \step{\alpha} P'$ with $\alpha \in A_\tau$ then we have to find a path $Q \pathtau Q_1 \step{\opt{\alpha}} Q_2$ such that $P \upto[\bisim] Q_1$ and $P' \upto[\bisim] Q_2$. Remember that ${\tbisim} \subseteq {\upto[\bisim]}$. We are going to proceed by structural induction on $P$ and by case distinction on the derivation of $P \tbisim Q$.
        \begin{itemize}
            \item If $P \bisimtbr Q$ then, by definition of $\bisimtbr$\,, there exists a path $Q \pathtau Q_1 \step{\opt{\alpha}} Q_2$ such that $P \bisimtbr Q_1$ and $P' \bisimtbr Q_2$, thus, by definition of $\tbisim$, $P \tbisim Q_1$ and $P' \tbisim Q_2$.
            \item If $P = \beta.P^\dag$ and $Q = \beta.Q^\dag$ with $\beta \in Act$ and $P^\dag \tbisim Q^\dag$ then, by the semantics, $P' = P^\dag$, $\beta = \alpha$, and thus there exists a path $Q \step{\alpha} Q^\dag$ such that $P \tbisim Q$ and $P' \tbisim Q^\dag$.
            \item If $P = P^\dag \parallel_S P^\ddag$ and $Q = Q^\dag \parallel_S Q^\ddag$ with $S \subseteq A$, $P^\dag \tbisim Q^\dag$ and $P^\ddag \tbisim Q^\ddag$ then
            \begin{itemize}
                \item if $\alpha \in S$ then, by the semantics, $P' = P'^\dag \parallel_S P'^\ddag$, $P^\dag \step{\alpha} P'^\dag$ and $P^\ddag \step{\alpha} P'^\ddag$. Note that $\alpha \ne \tau$ because $\alpha \in A$. Since $P^\dag \tbisim Q^\dag$ and $P^\ddag \tbisim Q^\ddag$, by induction, there exist two paths $Q^\dag \pathtau Q_1^\dag \step{\alpha} Q_2^\dag$ and $Q^\ddag \pathtau Q_1^\ddag \step{\alpha} Q_2^\ddag$ such that $P^\dag \upto[\bisim] Q^\dag_1$,\linebreak[4] $P'^\dag \upto[\bisim] Q^\dag_2$, $P^\ddag \upto[\bisim] Q_1^\ddag$ and $P'^\ddag \upto[\bisim] Q^\ddag_2$. By the semantics, $Q \pathtau Q^\dag_1 \parallel_S Q^\ddag_1\linebreak[3] \step{\alpha} Q^\dag_2 \parallel_S Q^\ddag_2$. Moreover, by definition of $\tbisim$ and the congruence property of $\bisim$, $P \upto[\bisim] Q^\dag_1 \parallel_S Q^\ddag_1$ and $P' \upto[\bisim] Q^\dag_2 \parallel_S Q^\ddag_2$.
                \item if $\alpha \not\in S$ then, by the semantics, two cases are possible. Suppose that $P' = P'^\dag \parallel_S P^\ddag$ and $P^\dag \step{\alpha} P'^\dag$; the other case is symmetrical. Since $P^\dag \tbisim Q^\dag$, by induction, there exists a path $Q^\dag \pathtau Q^\dag_1 \step{\opt{\alpha}} Q^\dag_2$ such that $P^\dag \upto[\bisim] Q^\dag_1$ and $P'^\dag \upto[\bisim] Q^\dag_2$. By the semantics, there exists a path $Q \pathtau Q^\dag_1 \parallel_S Q^\ddag \step{\opt{\alpha}} Q^\dag_2 \parallel_S Q^\ddag$. Moreover, by definition of $\tbisim$ and the congruence property of $\bisim$, $P \upto[\bisim] Q^\dag_1 \parallel_S Q^\ddag$ and $P' \upto[\bisim] Q^\dag_2 \parallel_S Q^\ddag$.
            \end{itemize}
            \item If $P = \tau_I(P^\dag)$ and $Q = \tau_I(Q^\dag)$ with $I \subseteq A$ and $P^\dag \tbisim Q^\dag$ then, by the semantics, $P' = \tau_I(P'^\dag)$, $P^\dag \step{\beta} P'^\dag$ and $(\beta \in I \wedge \alpha = \tau) \vee \beta = \alpha$. Since $P^\dag \tbisim Q^\dag$, by induction, there exists a path $Q^\dag \pathtau Q^\dag_1 \step{\opt{\beta}} Q^\dag_2$ such that $P^\dag \upto[\bisim] Q^\dag_1$ and $P'^\dag \upto[\bisim] Q^\dag_2$. By the semantics, $Q \pathtau \tau_I(Q_1^\dag) \step{\opt{\alpha}} \tau_I(Q^\dag_2)$ such that, by definition of $\tbisim$ and the congruence property of $\bisim$, $P \upto[\bisim] \tau_I(Q_1^\dag)$ and $P' \upto[\bisim] \tau_I(Q_2^\dag)$.
            \item If $P = \rename(P^\dag)$ and $Q = \rename(Q^\dag)$ with $\rename \subseteq A\times A$ and $P^\dag \tbisim Q^\dag$ then, by the semantics, $P' = \rename(P'^\dag)$, $P^\dag \step{\beta} P'^\dag$ and $(\beta,\alpha) \in \rename \vee \alpha = \beta = \tau$. Since $P^\dag \tbisim Q^\dag$, by induction, there exists a path $Q^\dag \pathtau Q^\dag_1 \step{\opt{\beta}} Q^\dag_2$ such that $P^\dag \upto[\bisim] Q^\dag_1$ and $P'^\dag \upto[\bisim] Q^\dag_2$. By the semantics, $Q \pathtau \rename(Q_1^\dag) \step{\opt{\alpha}} \rename(Q^\dag_2)$ such that, by definition of $\tbisim$ and the congruence property of $\bisim$, $P \upto[\bisim] \rename(Q_1^\dag)$ and $P' \upto[\bisim] \rename(Q_2^\dag)$.
            \item If $P = \theta_L^U(P^\dag)$ and $Q = \theta_L^U(Q^\dag)$ with $L \subseteq U \subseteq A$ and $P^\dag \tbisim Q^\dag$ then
            \begin{itemize}
                \item if $\alpha = \tau$ then, by the semantics, $P' = \theta_X(P'^\dag)$ and $P^\dag \steptau P'^\dag$. Since $P^\dag \tbisim Q^\dag$, by induction, there exists a path $Q^\dag \pathtau Q^\dag_1 \step{\opt{\tau}} Q^\dag_2$ such that $P^\dag \upto[\bisim] Q^\dag_1$ and $P'^\dag \upto[\bisim] Q^\dag_2$. By the semantics, there exists a path $Q \pathtau \theta_L^U(Q^\dag_1) \step{\opt{\tau}} \theta_L^U(Q^\dag_2)$ such that, by definition of $\tbisim$ and the congruence property of $\bisim$, $P \upto[\bisim] \theta_L^U(Q^\dag_1)$ and $P' \upto[\bisim] \theta_L^U(Q^\dag_2)$.
                \item if $\alpha = a \in A$ then, by the semantics, $a \in U \vee \deadend{P^\dag}{L}$ and $P^\dag \step{a} P'$. Since $P^\dag \tbisim Q^\dag$, by induction there exists a path $Q^\dag \pathtau Q_1^\dag \step{a} Q_2^\dag$ such that $P^\dag \upto[\bisim] Q^\dag_1$ and $P' \upto[\bisim] Q^\dag_2$. Moreover, in case $a \notin U$ we have $P^\dag\nsteptau$ so (\ref{stability}) ensures that $Q^\dag\pathtau Q'$ for some $Q'$ with $P^\dag \B Q'$ and $\init{Q'}=\init{P}$. This implies that we may choose $Q_1^\dag$ such that $Q' \pathtau Q_1^\dag$, and thus $Q'=Q_1^\dag$. This gives us $\deadend{P^\dag}{L} \Leftrightarrow \deadend{Q^\dag_1}{L}$. By the semantics, there exists a path $Q \pathtau \theta_L^U(Q^\dag_1) \step{a} Q^\dag_2$ such that, by definition of $\tbisim$ and the congruence property of $\bisim$, $P \upto[\bisim] \theta_L^U(Q^\dag_1)$ and $P' \upto[\bisim] Q^\dag_2$.
            \end{itemize}
        \end{itemize}
        \item If $\deadend{P}{X}$ and $P \step{\rt} P'$ then we have to find a path $Q \pathtau
        Q_1 \step{\rt} Q_2 \pathtau Q_3 \step{\rt} ... \pathtau Q_{2r-1} \step{\opt{\rt}} Q_{2r}$
        with $r>0$, such that $Q_1\nsteptau$, $\forall i \in
        [1,r{-}1],\; \theta_X(P) \upto[\bisim] \theta_X(Q_{2i}) \wedge \deadend{Q_{2i+1}}{X}$ and
        $\theta_X(P') \upto[\bisim] \theta_X(Q_{2r})$. Remember that ${\tbisim} \subseteq
        {\upto[\bisim]}$. We are going to proceed by by structural induction on $P$ and by case distinction on the derivation of $P \tbisim Q$.
        \begin{itemize}
            \item If $P \bisimtbr Q$ then, by Definition~\ref{def:time-out bisim}, there exists a path $Q \pathtau Q_1 \step{\rt} Q_2 \pathtau Q_3 \step{\rt} ... \pathtau Q_{2r{-}1} \step{\opt{\rt}} Q_{2r}$ with $r>0$, such that $Q_1\nsteptau$, $\forall i \in [1,r{-}1],\; \theta_X(P) \tbisim \theta_X(Q_{2i}) \wedge \deadend{Q_{2i+1}}{X}$ and $\theta_X(P') \tbisim \theta_X(Q_{2r})$.
            \item If $P = \beta.P^\dag$ and $Q = \beta.Q^\dag$ with $\beta \in Act$ and
            $P^\dag \tbisim Q^\dag$ then, by the semantics, $P' = P^\dag$ and $\beta = t$. Thus, by the semantics, there exists a path $Q \step{\rt} Q^\dag$ such that $Q\nsteptau$ and, by definition of $\tbisim$, $\theta_X(P') \tbisim \theta_X(Q^\dag)$.
            \item If $P = P^\dag \parallel_S P^\ddag$ and $Q = Q^\dag \parallel_S Q^\dag$ with $S \subseteq A$, $P^\dag \tbisim Q^\dag$ and $Q^\dag \tbisim Q^\dag$ then, since $\rt \not\in S$, by the semantics, two cases are possible. Suppose that $P' = P^\dag \parallel_S P'^\ddag$ and $P^\ddag \step{\rt} P'^\ddag$; the other case is symmetrical. Since $\deadend{P}{X}$, $P^\dag \nsteptau$ and $\init{P^\dag} \cap X \subseteq S$. Moreover, $\deadend{P^\ddag}{(X\setminus S) \cup (X \cap S \cap \init{P^\dag})}$. Note that $(X\setminus S) \cup (X \cap S \cap \init{P^\dag}) = X \setminus (S \setminus \init{P^\dag})$. Since $P^\ddag \tbisim Q^\ddag$, $P^\ddag \step{\rt} P'^\ddag$ and $\deadend{P^\ddag}{X \setminus (S\setminus\init{P^\dag})}$, by induction, there exists a path $Q^\ddag \pathtau Q^\ddag_1 \step{\rt} Q^\ddag_2 \pathtau Q^\ddag_3 \step{\rt} ... \pathtau Q^\ddag_{2r-1} \step{\opt{\rt}} Q^\ddag_{2r}$ with $r>0$, such that $Q^\ddag_1\nsteptau$, $\forall i \in [1,r{-}1],\linebreak[4] \theta_{X \setminus (S\setminus\init{P^\dag})}(P^\ddag) \upto[\bisim] \theta_{X \setminus (S\setminus\init{P^\dag})}(Q^\ddag_{2i}) \wedge \deadend{Q^\ddag_{2i+1}}{X \setminus (S\setminus\init{P^\dag})}$ and $\theta_{X \setminus (S\setminus\init{P^\dag})}(P'^\ddag) \upto[\bisim] \theta_{X \setminus (S\setminus\init{P^\dag})}(Q^\ddag_{2r})$. Moreover, by (\ref{stability}), since $P^\dag \nsteptau$, there exists a path $Q^\dag \pathtau Q^\dag_0 \nsteptau$ such that $P^\dag \tbisim Q^\dag_0$ and $\init{Q^\dag_0} = \init{P^\dag}$. By the semantics, there exists a path $Q \pathtau Q^\dag_0 \parallel_S Q^\ddag_1 \step{\rt} Q^\dag_0 \parallel_S Q^\ddag_2 \pathtau Q^\dag_0 \parallel_S Q^\ddag_3 \step{\rt} ... \pathtau Q^\dag_0 \parallel_S Q^\ddag_{2r-1} \step{\opt{\rt}} Q^\dag_0 \parallel_S Q^\ddag_{2r}$. Since $Q^\dag_0 \nsteptau$ and $Q^\ddag_1\nsteptau$ we have $Q^\dag_0 \parallel_S Q^\ddag_1\nsteptau$. By Lemma~\ref{lem:strong identities}, the definition of $\tbisim$ and the congruence property of $\bisim$, $\forall i \in [1,r{-}1],\; \theta_X(P) \bisim \mbox{}$ $$\theta_X(P^\dag \parallel_S \theta_{X \setminus (S\setminus\init{P^\dag})}(P^\ddag)) \upto[\bisim] \theta_X(Q^\dag_0 \parallel_S \theta_{X \setminus (S\setminus\init{Q_0^\dag})}(Q^\ddag_{2i})) \bisim \theta_X(Q^\dag_0 \parallel_S Q^\ddag_{2i}).$$ Moreover, $\deadend{Q^\dag_0 \parallel_S Q^\ddag_{2i+1}}{X}$ and $\theta_X(P') \bisim \mbox{}$ $$\theta_X(P^\dag \parallel_S \theta_{X \setminus (S\setminus\init{P^\dag})}(P'^\ddag)) \upto[\bisim] \theta_X(Q^\dag_0 \parallel_S \theta_{X \setminus (S\setminus\init{Q_0^\dag})}(Q^\ddag_{2r})) \bisim \theta_X(Q^\dag_0 \parallel_S Q^\ddag_{2r}).$$
            \item If $P = \tau_I(P^\dag)$ and $Q = \tau_I(Q^\dag)$ with $I \subseteq A$ and $P^\dag \tbisim Q^\dag$ then, by the semantics, $P' = \tau_I(P'^\dag)$, $P^\dag \step{\rt} P'^\dag$ and $\deadend{P^\dag}{(X\cup I)}$. Since $P^\dag \tbisim Q^\dag$ and $P\nsteptau$, also $P^\dag\nsteptau$, so (\ref{stability}) ensures that $Q^\dag \pathtau Q^\dag_0$ for some $Q^\dag_0$ with $P^\dag \B Q^\dag_0$ and $\init{Q^\dag_0}=\init{P^\dag}$. Since $P^\dag \B Q^\dag_0$, by induction, there exists a path $Q^\dag_0 \pathtau Q^\dag_1 \step{\rt} Q^\dag_2 \pathtau Q^\dag_3 \step{\rt} ... Q^\dag_{2r-1} \step{\opt{\rt}} Q^\dag_{2r}$ with $r>0$, such that $\forall i \in [1,r{-}1],\; \theta_{X\cup I}(P^\dag) \upto[\bisim] \theta_{X\cup I}(Q^\dag_{2i}) \wedge \deadend{Q^\dag_{2i+1}}{X\cup I}$ and $\theta_{X\cup I}(P'^\dag) \upto[\bisim] \theta_{X\cup I}(Q^\dag_{2r})$. As $Q^\dag_0\nsteptau$ we have $Q^\dag_0=Q^\dag_1$. By the semantics, $Q \pathtau \tau_I(Q_1^\dag) \step{\rt} \tau_I(Q^\dag_2) \pathtau \tau_I(Q^\dag_3) \step{\rt} ... \pathtau \tau_I(Q^\dag_{2r-1}) \step{\opt{\rt}} \tau_I(Q^\dag_{2r})$. Since $\init{Q^\dag_1}=\init{P^\dag}$ and $\tau_I(P^\dag)\nsteptau$, also $\tau_I(Q^\dag_1)\nsteptau$. Lemma~\ref{lem:strong identities}, the definition of $\tbisim$ and the congruence property of $\bisim$, $\forall i \in [1,r{-}1]$, $$\theta_X(P) \bisim \theta_X(\tau_I(\theta_{X\cup I}(P^\dag))) \upto[\bisim] \theta_X(\tau_I(\theta_{X\cup I}(Q^\dag_{2i}))) \bisim \theta_X(\tau_I(Q_{2i}^\dag))$$ and $\deadend{\tau_I(Q^\dag_{2i+1})}{X}$. Moreover, $$\theta_X(P') \bisim \theta_X(\tau_I(\theta_{X\cup I}(P'^\dag))) \upto[\bisim] \theta_X(\tau_I(\theta_{X\cup I}(Q^\dag_{2r}))) \bisim \theta_X(\tau_I(Q_{2r}^\dag)).$$
            \item If $P = \rename(P^\dag)$ and $Q = \rename(Q^\dag)$ with $\rename \subseteq A\times A$ and $P^\dag \tbisim Q^\dag$ then, by the semantics, $P' = \rename(P'^\dag)$, $P^\dag \step{\rt} P'^\dag$ and $\deadend{P^\dag}{\rename^{-1}(X)}$. Since $P^\dag \tbisim Q^\dag$, by induction, there exists a path $Q^\dag \pathtau Q^\dag_1 \step{\rt} Q^\dag_2 \pathtau Q^\dag_3 \step{\rt} ... Q^\dag_{2r-1} \step{\opt{\rt}} Q^\dag_{2r}$ with $r>0$, such that $Q^\dag_1\nsteptau$,\; $\forall i \in [1,r{-}1],\; \theta_{\rename^{-1}(X)}(P^\dag) \upto[\bisim] \theta_{\rename^{-1}(X)}(Q^\dag_{2i}) \wedge \deadend{Q^\dag_{2i+1}}{\rename^{-1}(X)}$ and $\theta_{\rename^{-1}(X)}(P'^\dag) \upto[\bisim] \theta_{\rename^{-1}(X)}(Q^\dag_2)$. By the semantics, $Q \pathtau \rename(Q_1^\dag) \step{\rt} \rename(Q^\dag_2) \pathtau \rename(Q^\dag_3) \step{\rt} ... \pathtau \rename(Q^\dag_{2r-1}) \step{\opt{\rt}} \rename(Q^\dag_{2r})$ and $\rename(Q_1^\dag)\nsteptau$. By Lemma~\ref{lem:strong identities}, the definition of $\tbisim$ and the congruence property of $\bisim$, $\forall i \in [1,r{-}1]$, $$\theta_X(P) \bisim \theta_X(\rename(\theta_{\rename^{-1}(X)}(P^\dag))) \upto[\bisim] \theta_X(\rename(\theta_{\rename^{-1}}(Q^\dag_{2i}))) \bisim \theta_X(\tau_I(Q_{2i}^\dag))$$ and $\deadend{\rename(Q^\dag_{2i+1})}{X}$. Moreover, $$\theta_X(P') \bisim \theta_X(\rename(\theta_{\rename^{-1}(X)}(P'^\dag))) \upto[\bisim] \theta_X(\rename(\theta_{\rename^{-1}(X)}(Q^\dag_{2r}))) \bisim \theta_X(\rename(Q_{2r}^\dag)).$$
            \item If $P = \theta_L^U(P^\dag)$ and $Q = \theta_L^U(Q^\dag)$ with $L \subseteq U \subseteq A$ and $P^\dag \tbisim Q^\dag$ then, by the semantics, $\deadend{P^\dag}{L\cup X}$ and $P^\dag \step{\rt} P'$. Since $P^\dag \tbisim Q^\dag$ and $P\nsteptau$, also $P^\dag\nsteptau$, so (\ref{stability}) ensures that $Q^\dag \pathtau Q^\dag_0$ for some $Q^\dag_0$ with $P^\dag \B Q^\dag_0$ and $\init{Q^\dag_0}=\init{P^\dag}$. Since $P^\dag \B Q^\dag_0$, by induction, there exists a path $Q^\dag_0 \pathtau Q_1^\dag \step{\rt} Q_2 \pathtau Q_3 \step{\rt} ... \pathtau Q_{2r-1} \step{\opt{\rt}} Q_{2r}$ with $r>0$, such that $Q_1^\dag\nsteptau$, $\forall i \in [1,r{-}1],\; \theta_X(P^\dag) \upto[\bisim] \theta_X(Q_{2i}) \wedge \deadend{Q_{2i+1}}{X}$ and $\theta_X(P') \upto[\bisim] \theta_X(Q_{2r})$. As $Q^\dag_0\nsteptau$ we have $Q^\dag_0=Q^\dag_1$. As $\init{Q^\dag_1}=\init{P^\dag}$ one has  $\deadend{Q^\dag_1}{L}$. By the semantics, there exists a path $Q \pathtau \theta_L^U(Q^\dag_1) =: Q_1 \step{\rt} Q_2 \pathtau Q_3 \step{\rt} ... \pathtau Q_{2r-1} \step{\opt{\rt}} Q_{2r}$. Since $Q_1^\dag\nsteptau$ we have $Q_1\nsteptau$.
            Moreover, note that $P \bisim P^\dag$ because $\deadend{P^\dag}{L}$. By the congruence property of $\bisim$, $\forall i \in [1,r{-}1]$, $$\theta_X(P) \bisim \theta_X(P^\dag) \upto[\bisim] \theta_X(Q_{2i})$$ and $\deadend{Q_{2i+1}}{X}$. Moreover, $\theta_X(P') \upto[\bisim] \theta_X(Q_{2r})$.
        \end{itemize}
        \item The last condition of Definition~\ref{def:up to} is implied by (\ref{stability}).
    \popQED
    \end{enumerate}
\end{proof}

\section{Full Congruence Proofs for \texorpdfstring{$\bisimrtbr$ and $\bisimrtb$}{Rooted Branching Reactive Bisimilarity}} \label{app:congruence}

\begin{definition}\rm \label{def:rooted time-out bisim up to}
    Here, a \emph{rooted \tb time-out bisimulation up to $\bisimtbr$} is a symmetric relation ${\tbisim} \subseteq \closed\times\closed$ such that, for all $P,Q \in \closed$ with $P \tbisim Q$,
    \begin{enumerate}
        \item if $P \step{\alpha} P'$ with $\alpha \in A_\tau$ then there is a transition $Q \step{\alpha} Q'$ such that $P' \bisim\,\B\,\bisimtbr Q'$
        \item if $\deadend{P}{X}$ and $P \step{\rt} P'$ then there is a transition $Q \step{\rt} Q'$ such that $\theta_X(P') \bisim\,\B\,\bisimtbr \theta_X(Q')$.
    \end{enumerate}
\end{definition}

\begin{proposition}
    Let $P,Q \in \closed$. Then $P \bisimrtbr Q$ iff there exists a rooted \tb time-out bisimulation $\B$ up to $\bisimtbr$ such that $P \B Q$.
\end{proposition}

\begin{proof}
    First of all, a rooted \tb time-out bisimulation is a rooted \tb time-out bisimulation up to $\bisimtbr$ by reflexivity of $\bisim$ and $\bisimtbr$\,. Conversely, we are going to show that $\bisim\,\B\,\bisimtbr$ is a \tb time-out bisimulation. This implies that ${\bisim\,\B\,\bisimtbr} \subseteq {\bisimtbr}$\,, so that each rooted \tb time-out bisimulation up to $\bisimtbr$ is in fact a rooted \tb time-out bisimulation. Let $P,Q \in \closed$ such that $P \bisim\,\B\,\bisimtbr Q$. There exists $P^\dag,Q^\dag \in \closed$ such that $P \bisim P^\dag \B Q^\dag \bisimtbr Q$.
    \begin{enumerate}
        \item If $P \step{\alpha} P'$ then, since $P \bisim P^\dag$, there is a transition $P^\dag \step{\alpha} P^\ddag$ such that $P' \bisim P^\ddag$. Thus, by Clause 1 of Definition~\ref{def:rooted time-out bisim up to}, there is a transition $Q^\dag \step{\alpha} Q^\ddag$ such that $P^\ddag \bisim\,\B\,\bisimtbr Q^\ddag$. By Clause~1 of Definition~\ref{def:time-out bisim} there is a path $Q \pathtau Q_1 \step{\opt{\alpha}} Q_2$ with $Q^\dag \bisimtbr Q_1$ and $Q^\ddag \bisimtbr Q_2$. By the transitivity of $\bisim$ and $\bisimtbr$ we obtain $P \bisim\,\B\,\bisimtbr Q_1$ and $P' \bisim\,\B\,\bisimtbr Q_2$.
        \item If $\deadend{P}{X}$ and $P \step{t} P'$ then, since $P \bisim P^\dag$, $\deadend{P^\dag}{X}$ and $P^\dag \step{t} P^\ddag$ for some $P^\ddag$ with $P' \bisim P^\ddag$. Thus, by Clauses 1 and 2 of Definition~\ref{def:rooted time-out bisim up to}\linebreak[4] $\deadend{Q^\dag}{X}$ and there is a transition $Q^\dag \step{\rt} Q^\ddag$ with $\theta_X(P^\ddag) \bisim\,\B\,\bisimtbr \theta_X(Q^\ddag)$.  By Clause~2 of Definition~\ref{def:time-out bisim} there is a path $Q = Q_0 \pathtau Q_1 \step{\rt} Q_2 \pathtau Q_3 \step{\rt} ... \pathtau Q_{2r{-}1} \step{\opt{\rt}} Q_{2r}$ with $r>0$, such that $Q_1\nsteptau$, $\forall i \in [1,r{-}1],\; \theta_X(Q^\dag) \bisimtbr \theta_X(Q_{2i}) \wedge \deadend{Q_{2i+1}}{X}$ and $\theta_X(Q^\ddag) \bisimtbr \theta_X(Q_{2r})$. Since $\deadend{Q^\dag}{X}$ we have $Q^\dag \bisim \theta_X(Q^\dag)$ and hence  $Q^\dag \bisimtbr \theta_X(Q^\dag)$. Thus $\theta_X(P) \bisim P \bisim P^\dag \tbisim Q^\dag \bisimtbr \theta_X(Q^\dag) \bisimtbr \theta_X(Q_{2i})$. Since $\bisim$ is a congruence for $\theta_X$ \cite[Theorem 20]{strongreactivebisimilarity}, we have $\theta_X(P')\bisim\theta_X(P^\ddag)$. By the transitivity of $\bisim$ and $\bisimtbr$ we obtain $\theta_X(P) \bisim\,\B\,\bisimtbr \theta_X(Q_{2i})$ and $\theta_X(P') \bisim\,\B\,\bisimtbr \theta_X(Q_{2r})$.
        \item If $P\nsteptau$ then, since $P \bisim P^\dag$, $P^\dag\nsteptau$, so by Clause 1 of Definition~\ref{def:rooted time-out bisim up to} $Q^\dag\nsteptau$, and by Clause~3 of Definition~\ref{def:time-out bisim} there is a path $Q \pathtau Q_0 \nsteptau$.
        \popQED
    \end{enumerate}
\end{proof}

\begin{proof}[Proof of Theorem \ref{thm:congruence}]
    Let ${\tbisim} \subseteq \closed\times\closed$ be the smallest relation such that 
    \begin{itemize}
        \item if $P \bisimrtbr Q$ then $P \tbisim Q$
        \item if $P \tbisim Q$ and $\alpha \in Act$ then $\alpha.P \tbisim \alpha.Q$
        \item if $P_1 \tbisim Q_1$ and $P_2 \tbisim Q_2$ then $P_1 + P_2 \tbisim Q_1 + Q_2$
        \item if $P_1 \tbisim Q_1$, $P_2 \tbisim Q_2$ and $S \subseteq A$ then $P_1 \parallel_S P_2 \tbisim Q_1 \parallel_S Q_2$
        \item if $P \tbisim Q$ and $I \subseteq A$ then $\tau_I(P) \tbisim \tau_I(Q)$
        \item if $P \tbisim Q$ and $\rename \subseteq A\times A$ then $\rename(P) \tbisim \rename(Q)$
        \item if $P \tbisim Q$ and $L \subseteq U \subseteq A$ then $\theta_L^U(P) \tbisim \theta_L^U(Q)$
        \item if $P \tbisim Q$ and $X \subseteq A$ then $\psi_X(P) \tbisim \psi_X(Q)$
        \item if $\equa$ is a recursive specification with $z \in V_\equa$ and $\rho, \nu \in V\setminus V_\equa \rightarrow \closed$ are substitutions such that $\forall x \in V\setminus V_\equa,\; \rho(x) \tbisim \nu(x)$, then $\langle z|\equa\rangle[\rho] \tbisim \langle z|\equa\rangle[\nu]$.
        \item if $\equa$ and $\equa'$ are recursive specifications and $x \in V_\equa=V_{\equa'}$ with $\langle x|\equa\rangle, \langle x|\equa'\rangle \in \closed$ such that $\forall y \in V_\equa,\; \equa_y \bisimrtbr \equa'_y$, then $\langle x|\equa\rangle \tbisim \langle x|\equa'\rangle$.
    \end{itemize}
    Note that since $\bisim$, $\tbisim$ and $\bisimtb$ are congruences for the operators listed in Proposition~\ref{prop:stability}, so are the composed relations $\tbisim\,\bisimtb$ and $\bisim\,\tbisim\,\bisimtb$\;.\hfill (\textsterling)

   Let ${=_\mathcal{I}} := \{(P,Q) \mid \init{P} = \init{Q}\}$. A trivial induction on the derivation of $P \tbisim Q$, using the fact that $=_\mathcal{I}$ is a full congruence for $\ccsp$ \cite{strongreactivebisimilarity}, shows that 
    \begin{align*}
        P \tbisim Q \Rightarrow \init{P} = \init{Q} \tag{$@$}
    \end{align*}
    (For the second last case, the assumption that $\rho(x) \B \nu(x)$ for all $x\in V\setminus V_\equa$ implies $\rho =_{\mathcal{I}} \nu$ by induction. Since $=_{\mathcal{I}}$ is a lean congruence, this implies $\langle z|\equa\rangle[\rho] =_{\mathcal{I}}\langle{z|\equa}\rangle[\nu]$.)\\
    A trivial induction on $\expr$ shows that
    \begin{align*}
        \forall E \in \expr, \rho,\nu \in V \rightarrow \closed,\; (\forall x \mathbin\in V, \rho(x) \tbisim \nu(x)) \Rightarrow E[\rho] \tbisim E[\nu] \tag{$\star$}
    \end{align*}
    A useful corollary is
    \begin{equation}
        \begin{aligned}
            \forall E\in\expr, \equa \mbox{ a recursive specification}, \rho, \nu \in V\setminus V_\equa \rightarrow \closed,\\ (\forall x \in V\setminus V_\equa,\; \rho(x) \tbisim \nu(x)) \Rightarrow \langle E|\equa\rangle[\rho] \tbisim \langle E|\equa\rangle[\nu]
        \end{aligned} \tag{$\$$}
    \end{equation}
    Applied in the context of the last condition of $\tbisim$, it implies
    \begin{align}
        \forall E \in \expr, \mbox{the variables of \textit{E} are in }V_\equa \Rightarrow \langle E|\equa\rangle \tbisim \langle E|\equa'\rangle \tag{$\#$}
    \end{align}
    Since ${\bisimrtbr} \subseteq {\tbisim}$, it suffices to prove that $\tbisim$ is a rooted \tb time-out bisimulation up to $\bisimtbr$ (so that ${\tbisim}={\bisimrtbr}$). Note that $\tbisim$ is symmetric, since $\bisimrtbr$ is. Let $P, Q \in \closed$ such that $P \tbisim Q$.
    \begin{enumerate}
        \item If $P \step{\alpha} P'$ with $\alpha \in A_\tau$ then we need to find a transition $Q \step{\alpha} Q'$ such that $P' \tbisim\,\bisimtbr Q'$. This is sufficient as ${\tbisim\bisimtbr} \subseteq {\bisim\tbisim\bisimtbr}$. We are going to proceed by induction on the proof of $P \step{\alpha} P'$ and by case distinction on the derivation of $P \tbisim Q$.
        \begin{itemize}
            \item If $P \bisimrtbr Q$ then there exists a transition $Q \step{\alpha} Q'$ such that $P' \bisimtbr Q'$, and so $P' \tbisim\bisimtbr Q'$.
            \item If $P = \beta.P^\dag$ and $Q = \beta.Q^\dag$ such that $\beta \in Act$ and $P^\dag \tbisim Q^\dag$ then $\alpha = \beta$ and $P' = P^\dag$. Thus there exists a transition $Q \step{\alpha} Q^\dag$ such that $P^\dag \tbisim Q^\dag$, and so $P^\dag \tbisim \bisimtbr Q^\dag$.
            \item If $P = P^\dag + P^\ddag$ and $Q = Q^\dag + Q^\ddag$ such that $P^\dag \tbisim Q^\dag$ and $P^\ddag \tbisim Q^\ddag$ then, by the semantics, $P^\dag \step{\alpha} P'$ or $P^\ddag \step{\alpha} P'$. Suppose that $P^\dag \step{\alpha} P'$ (the other case proceeds symmetrically). Since $P^\dag \tbisim Q^\dag$, there exists a transition $Q^\dag \step{\alpha} Q'$ such that $P' \tbisim\,\bisimtbr Q'$. By the semantics, there exists a transition $Q \step{\alpha} Q'$ such that $P' \tbisim\,\bisimtbr Q'$.
            \item If $P = P^\dag \parallel_S P^\ddag$ and $Q = Q^\dag \parallel_S Q^\ddag$ such that $S\subseteq A$, $P^\dag \tbisim Q^\dag$ and $P^\ddag \tbisim Q^\ddag$ then
            \begin{itemize}
                \item if $\alpha \not\in S$ then, by the semantics, $P' = P'^\dag \parallel_S P^\ddag$ and $P^\dag \step{\alpha} P'^\dag$ or $P' = P^\dag \parallel_S P'^\ddag$ and $P^\ddag \step{\alpha} P'^\ddag$. Suppose that $P^\dag \step{\alpha} P'^\dag$ (the other case proceeds symmetrically). Since $P^\dag \tbisim Q^\dag$, there exists a transition $Q^\dag \step{\alpha} Q'^\dag$ such that $P'^\dag \tbisim\,\bisimtbr Q'^\dag$. By the semantics, there exists a transition $Q \step{\alpha} Q'^\dag \parallel_S Q^\ddag$ such that, by (\textsterling), $P' \tbisim\,\bisimtbr Q'^\dag \parallel_S Q^\ddag$.
                \item if $\alpha \in S$ then, by the semantics, $P' = P'^\dag \parallel_S P'^\ddag$, $P^\dag \step{\alpha} P'^\dag$ and $P^\ddag \step{\alpha} P'^\ddag$. Since $P^\dag \tbisim Q^\dag$ and $P^\ddag \tbisim Q^\ddag$, there exists two transitions $Q^\dag \step{\alpha} Q'^\dag$ and $Q^\ddag \step{\alpha} Q'^\ddag$ such that $P'^\dag \tbisim\,\bisimtbr Q'^\dag$ and $P'^\ddag \tbisim\,\bisimtbr Q'^\ddag$. By the semantics, there exists a transition $Q \step{\alpha} Q'^\dag \parallel_S Q'^\ddag$ such that, by (\textsterling), $P' \tbisim\,\bisimtbr Q'^\dag \parallel_S Q'^\ddag$.
            \end{itemize}
            \item If $P = \tau_I(P^\dag)$ and $Q = \tau_I(Q^\dag)$ with $I \subseteq A$ and $P^\dag \tbisim Q^\dag$ then, by the semantics, $P' = \tau_I(P'^\dag)$, $P^\dag \step{\beta} P'^\dag$ and $(\beta \in I\cup\{\tau\} \wedge \alpha = \tau) \vee \beta = \alpha \not\in I$. Since $P^\dag \tbisim Q^\dag$, there exists a transition $Q^\dag \step{\beta} Q'^\dag$ such that $P'^\dag \tbisim\,\bisimtbr Q'^\dag$. By the semantics, there exists a transition $Q \step{\alpha} \tau_I(Q'^\dag)$ such that, by (\textsterling), $P'\tbisim\,\bisimtbr \tau_I(Q'^\dag)$.
            \item If $P = \rename(P^\dag)$ and $Q = \rename(Q^\dag)$ with $\rename \subseteq A\times A$ and $P^\dag \tbisim Q^\dag$ then, by the semantics, $P' = \rename(P'^\dag)$, $P^\dag \step{\beta} P'^\dag$ and $(\beta,\alpha) \in \rename \vee \beta = \alpha = \tau$. Since $P^\dag \tbisim Q^\dag$, there exists a transition $Q^\dag \step{\beta} Q'^\dag$ such that $P'^\dag \tbisim\,\bisimtbr Q'^\dag$. By the semantics, there exists a transition $Q \step{\alpha} \rename(Q'^\dag)$ such that, by (\textsterling), $P'\tbisim\,\bisimtbr \rename(Q'^\dag)$.
            \item If $P = \theta_L^U(P^\dag)$ and $Q = \theta_L^U(Q^\dag)$ with $L \subseteq U \subseteq A$ and $P^\dag \tbisim Q^\dag$ then
            \begin{itemize}
                \item if $\alpha = \tau$ then, by the semantics, $P' = \theta_X(P'^\dag)$ and $P^\dag \step{\tau} P'^\dag$. Since $P^\dag \tbisim Q^\dag$, there exists a transition $Q^\dag \step{\tau} Q'^\dag$ such that $P'^\dag \tbisim\,\bisimtbr Q'^\dag$. By the semantics, there exists a transition $Q \step{\tau} \theta_L^U(Q'^\dag)$ such that, by (\textsterling), $P'\tbisim\,\bisimtbr \theta_L^U(Q'^\dag)$.
                \item if $\alpha = a \in A$ then, by the semantics, $P^\dag \step{a} P'$ and $a \in U \vee \deadend{P^\dag}{L}$. Since $P^\dag \tbisim Q^\dag$, there exists a transition $Q^\dag \step{a} Q'$ such that $P' \tbisim\,\bisimtbr Q'$. According to $(@)$, $\deadend{P^\dag}{L} \Leftrightarrow \deadend{Q^\dag}{L} $, thus, by the semantics, there exists a transition $Q \step{a} Q'$ such that $P'\tbisim\,\bisimtbr Q'$.
            \end{itemize}
            \item If $P = \psi_X(P^\dag)$ and $Q = \psi_X(Q^\dag)$ with $X \subseteq A$ and $P^\dag \tbisim Q^\dag$ then, by the semantics, $P^\dag \step{\alpha} P'$. Since $P^\dag \tbisim Q^\dag$, there exists a transition $Q^\dag \step{\alpha} Q'$ such that $P' \tbisim\,\bisimtbr Q'$. By the semantics, there exists a transition $Q \step{\alpha} Q'$ such that $P'\tbisim\,\bisimtbr Q'$.
            \item Let $P = \langle z|\equa\rangle[\rho]$ and $Q = \langle z|\equa\rangle[\nu]$ with $\equa$ a recursive specification, $z \in V_\equa$ and $\rho, \nu \in V\setminus V_\equa \rightarrow\closed$ such that $\forall x \in V\setminus V_\equa,\; \rho(x) \tbisim \nu(x)$. By the semantics, $\langle\equa_z|\equa\rangle[\rho] \step{\alpha} P'$ is provable by a strict sub-proof of $P \step{\alpha} P'$. Moreover, according to $(\$)$, $\langle\equa_z|\equa\rangle[\rho] \tbisim \langle\equa_z|\equa\rangle[\nu]$. By induction, there exists a transition $\langle\equa_z|\equa\rangle[\nu] \step{\alpha} Q'$ such that $P'\tbisim\,\bisimtbr Q'\!$. By the semantics, there exists a transition $\langle z|\equa\rangle[\nu] \step{\alpha} Q'$ such that $P'\tbisim\,\bisimtbr Q'$.
            \item Let $P = \langle x|\equa\rangle$ and $Q = \langle x|\equa'\rangle$ with $\equa$ and $\equa'$ two recursive specifications such that $\forall y \in V_\equa = V_{\equa'},\; \equa_y \bisimrtbr \equa'_y$ and $x \in V_\equa$. By the semantics, $\langle\equa_x|\equa\rangle \step{\alpha} P'$ is provable by a strict sub-proof of $P \step{\alpha} P'$. Moreover, according $(\#)$, $\langle \equa_x|\equa\rangle \tbisim \langle\equa_x|\equa'\rangle$. By induction, there exists a transition $\langle\equa_x|\equa'\rangle \step{\alpha} R'$ such that $P' \tbisim\,\bisimtbr R'$. Since $\langle \_ |\equa'\rangle \in V_{\equa'} \rightarrow \closed$ and $\equa_x \bisimrtbr \equa'_x$, $\langle\equa_x|\equa'\rangle \bisimrtbr \langle\equa_x'|\equa'\rangle$. Therefore, there exists a transition $\langle\equa'_x|\equa'\rangle \step{\alpha} Q'$ such that $R' \bisimtbr Q'$. By the semantics, there exists a transition $Q \step{\alpha} Q'$ such that, by transitivity of $\bisimtbr$\,, $P' \tbisim\,\bisimtbr Q'$.
        \end{itemize}
        \item If $\deadend{P}{X}$ and $P \step{\rt} P'$ then we need to find a transition $Q \step{\rt} Q'$ such that $\theta_X(P') \bisim\,\tbisim\,\bisimtbr \theta_X(Q')$. We are going to proceed by induction on the proof of $P \step{\rt} P'$ and by case distinction on the derivation of $P \tbisim Q$.
        \begin{itemize}
            \item If $P \bisimrtbr Q$ then there exists a transition $Q \step{\rt} Q'$ such that $\theta_X(P') \bisimtbr \theta_X(Q')$ and so $\theta_X(P') \bisim\,\tbisim\,\bisimtbr \theta_X(Q')$.
            \item If $P = \beta.P^\dag$ and $Q = \beta.Q^\dag$ such that $\beta \in Act$ and $P^\dag \tbisim Q^\dag$ then $\beta = \rt$ and $P' = P^\dag$. Thus there is a transition $Q \step{t} Q^\dag$ such that, by definition of $\tbisim$, $\theta_X(P^\dag) \bisim\,\tbisim\,\bisimtbr \theta_X(Q^\dag)$.
            \item If $P = P^\dag + P^\ddag$ and $Q = Q^\dag + Q^\ddag$ such that $P^\dag \tbisim Q^\dag$ and $P^\ddag \tbisim Q^\ddag$ then, by the semantics, $\deadend{P^\dag}{X}$, $\deadend{P^\ddag}{X}$ and $P^\dag \step{\rt} P'$ or $P^\ddag \step{\rt} P'$. Suppose that $P^\dag \step{\rt} P'$ (the other case is symmetrical). Since $P^\dag \tbisim Q^\dag$, there exists a transition $Q^\dag \step{\rt} Q'$ such that $\theta_X(P') \bisim\,\tbisim\,\bisimtbr \theta_X(Q')$. By the semantics, there exists a transition $Q \step{\rt} Q'$ such that $\theta_X(P') \bisim\,\tbisim\, \bisimtbr \theta_X(Q')$.
            \item If $P = P^\dag \parallel_S P^\ddag$ and $Q = Q^\dag \parallel_S Q^\ddag$ such that $S\subseteq A$, $P^\dag \tbisim Q^\dag$ and $P^\ddag \tbisim Q^\ddag$ then, by the semantics, $P' = P'^\dag \parallel_S P^\ddag$ and $P^\dag \step{\rt} P'^\dag$ or $P' = P^\dag \parallel_S P'^\ddag$ and $P^\ddag \step{\rt} P'^\ddag$. Suppose that $P^\ddag \step{\rt} P'^\ddag$ (the other case is symmetrical). Since $\deadend{P}{X}$, $P^\dag \nsteptau$ and $\init{P^\dag} \cap X \subseteq S$. Moreover, $\deadend{P^\ddag}{X \setminus S \cup (X\cap S\cap\init{P^\dag})}$. Note that $X \setminus S \cup (X\cap S\cap\init{P^\dag}) = X \setminus(S\setminus\init{P^\dag})$. Since $P^\ddag \tbisim Q^\ddag$, there exists a transition $Q^\ddag \step{\rt} Q'^\ddag$ such that $\theta_{X \setminus(S\setminus\init{P^\dag})}(P'^\ddag) \bisim\,\tbisim\,\bisimtbr \theta_{X \setminus(S\setminus\init{P^\dag})}(Q'^\ddag)$. Since $P^\dag \tbisim Q^\dag$ and $P^\ddag \tbisim Q^\ddag$, $\init{P^\dag} = \init{Q^\dag}$ and $\init{P^\ddag} = \init{Q^\ddag}$. By the semantics, there exists a transition $Q \step{\rt} Q^\dag \parallel_S Q'^\ddag$. By (\textsterling) and Lemma~\ref{lem:strong identities}, \(\theta_X(P') = \theta_X(P^\dag \parallel_S P'^\ddag) \bisim \theta_X(P^\dag \parallel_S\theta_{X \setminus(S\setminus\init{P^\dag})}(P'^\ddag)) \bisim\,\tbisim\,\bisimtbr \theta_X(Q^\dag \parallel_S \theta_{X \setminus(S\setminus\init{Q^\dag})}(Q'^\ddag)) \bisim \theta_X(Q^\dag \parallel_S Q'^\ddag)\). In the last step we use that $Q^\dag \nsteptau$, since $P^\dag\nsteptau$ and $\init{P}=\init{Q}$, using ($@$). Now apply that ${\bisim}\subseteq{\bisimtbr}$ and the transitivity of $\bisim$ and $\bisimtbr$\,.
            \item If $P = \tau_I(P^\dag)$ and $Q = \tau_I(Q^\dag)$ with $I \subseteq A$ and $P^\dag \tbisim Q^\dag$ then, by the semantics, $P' = \tau_I(P'^\dag)$ and $P^\dag \step{\rt} P'^\dag$. Since $\deadend{P}{X}$, $\deadend{P^\dag}{X\cup I}$. Since $P^\dag \tbisim Q^\dag$, there exists a transition $Q^\dag \step{\rt} Q'^\dag$ such that $\theta_{X\cup I}(P'^\dag) \bisim\,\tbisim\,\bisimtbr \theta_{X\cup I}(Q'^\dag)$. By the semantics, there exists a transition $Q \step{\rt} \tau_I(Q'^\dag)$ such that, by (\textsterling) and Lemma~\ref{lem:strong identities}, $\theta_X(P') \bisim \theta_X(\tau_I(\theta_{X\cup I}(P'^\dag))) \bisim\,\tbisim\,\bisimtbr \theta_X(\tau_I(\theta_{X\cup I}(Q'^\dag))) \bisim \tau_I(Q'^\dag)$.
            \item If $P = \rename(P^\dag)$ and $Q = \rename(Q^\dag)$ with $\rename \subseteq A\times A$ and $P^\dag \tbisim Q^\dag$ then, by the semantics, $P' = \rename(P'^\dag)$ and $P^\dag \step{\rt} P'^\dag$. Since $\deadend{P}{X}$, $\deadend{P^\dag}{\rename^{-1}(X)}$. Since $P^\dag \tbisim Q^\dag$, there exists a transition $Q^\dag \step{\rt} Q'^\dag$ such that $\theta_{\rename^{-1}(X)}(P'^\dag) \bisim\,\tbisim\,\bisimtbr \theta_{\rename^{-1}(X)}(Q'^\dag)$. By the semantics, there exists a transition $Q \step{\rt} \rename(Q'^\dag)$ such that, by (\textsterling) and Lemma~\ref{lem:strong identities}, $\theta_X(P') \bisim \theta_X(\rename(\theta_{\rename^{-1}(X)}(P'^\dag))) \bisim\,\tbisim\,\bisimtbr \theta_X(\rename(\theta_{\rename^{-1}(X)}(Q'^\dag))) \bisim \rename(Q'^\dag)$.
            \item If $P = \theta_L^U(P^\dag)$ and $Q = \theta_L^U(Q^\dag)$ with $L \subseteq U \subseteq A$ and $P^\dag \tbisim Q^\dag$ then, by the semantics, $P^\dag \step{\rt} P'$ and $\deadend{P^\dag}{L}$. Since $P^\dag \tbisim Q^\dag$, there exists a transition $Q^\dag \step{\rt} Q'$ such that $\theta_X(P') \bisim\,\tbisim\,\bisimtbr \theta_X(Q')$. According to $(@)$, $\deadend{P^\dag}{L} \Leftrightarrow \deadend{Q^\dag}{L} $, thus, by the semantics, there exists a transition $Q \step{\rt} Q'$ such that $\theta_X(P')\bisim\,\tbisim\,\bisimtbr \theta_X(Q')$.
            \item If $P = \psi_Y(P^\dag)$ and $Q = \psi_Y(Q^\dag)$ with $Y \subseteq A$ and $P^\dag \tbisim Q^\dag$ then, by the semantics, $P' = \theta_Y(P'^\dag)$, $P^\dag \step{\rt} P'^\dag$ and $\deadend{P^\dag}{Y}$. Since $P^\dag \tbisim Q^\dag$, there exists a transition $Q^\dag \step{\rt} Q'^\dag$ such that $\theta_Y(P'^\dag) \bisim\,\tbisim\,\bisimtbr \theta_Y(Q'^\dag)$. Using ($@$), $\deadend{Q^\dag}{Y}$, so by the semantics, there exists a transition $Q \step{\rt} \theta_Y(Q'^\dag)$. By (\textsterling), $\theta_X(P') \bisim\,\tbisim\,\bisimtbr \theta_X(\theta_Y(Q'^\dag))$.
            \item Let $P = \langle z|\equa\rangle[\rho]$ and $Q = \langle z|\equa\rangle[\nu]$ with $\equa$ a recursive specification, $z \in V_\equa$ and $\rho, \nu \in V\setminus V_\equa \rightarrow\closed$ such that $\forall x \in V\setminus V_\equa, \rho(x) \tbisim \nu(x)$. By the semantics, $\langle\equa_z|\equa\rangle[\rho] \step{\rt} P'$ is provable by a strict sub-proof of $P \step{\rt} P'$ and $\init{P} = \init{\langle\equa_z|\equa\rangle[\rho]}$. Moreover, according to $(\$)$, $\langle\equa_z|\equa\rangle[\rho] \tbisim \langle\equa_z|\equa\rangle[\nu]$. By induction, there exists a transition $\langle\equa_z|\equa\rangle[\nu] \step{\rt} Q'$ such that $\theta_X(P') \bisim\,\tbisim\,\bisimtbr \theta_X(Q')$. By the semantics, there exists a transition $\langle z|\equa\rangle[\nu] \step{\rt} Q'$ such that $\theta_X(P') \bisim\,\tbisim\,\bisimtbr \theta_X(Q')$.
            \item Let $P = \langle x|\equa\rangle$ and $Q = \langle x|\equa'\rangle$ with $\equa$ and $\equa'$ two recursive specifications such that $\forall y \in V_\equa = V_{\equa'}, \equa_y \bisimrtbr \equa'_y$ and $x \in V_\equa$. By the semantics, $\langle\equa_x|\equa\rangle \step{\rt} P'$ is provable be a strict sub-proof of $P \step{\alpha} P'$ and $\init{P} = \init{\langle\equa_x|\equa\rangle}$. Moreover, according $(\#)$, $\langle \equa_x|\equa\rangle \tbisim \langle\equa_x|\equa'\rangle$. By induction, there exists a transition $\langle\equa_x|\equa'\rangle \step{\rt} R'$ such that $\theta_X(P') \bisim\,\tbisim\,\bisimtbr \theta_X(R')$. Since $\langle \_ |\equa'\rangle \in V_{\equa'} \rightarrow \closed$ and $\equa_x \bisimrtbr \equa'_x$, $\langle\equa_x|\equa'\rangle \bisimrtbr \langle\equa_x'|\equa'\rangle$. Moreover, according to $(@)$, $\deadend{\langle\equa_x|\equa'\rangle}{X}$. Therefore, there exists a transition $\langle\equa'_x|\equa'\rangle \step{\alpha} Q'$ such that $\theta_X(R') \bisimtbr \theta_X(Q')$. By the semantics, there exists a transition $Q \step{\alpha} Q'$ such that, by transitivity of $\bisimtbr$, $\theta_X(P') \bisim\,\tbisim\,\bisimtbr \theta_X(Q')$.
        \end{itemize}
    \end{enumerate}
    As a result, $\mathcal{B}$ is a rooted \tb time-out bisimulation up to $\bisimtbr$\,, and $(\star)$ gives us that $\bisimrtbr$ is a lean congruence and the last condition of $\tbisim$ adds that it is a full congruence.
\end{proof}

\section{Proof of RSP}\label{app:RSP}

To prove RSP, another version of $\bisimtbr$ is needed, this time up to itself.

\begin{definition}\rm \label{def:up to b}
    A \textit{\tb time-out bisimulation up to $\bisimtbr$} is a symmetric relation ${\tbisim} \subseteq \closed\times\closed$ such that, for all $P,Q \in \closed$ such that $P \tbisim Q$, and for all $X\subseteq A$,
    \begin{enumerate}
        \item if $P \pathtau P' \step{\alpha} P''$ with $\alpha \in A_\tau$ and $P \bisimtbr P'$ then there exists a path $Q \pathtau Q_1 \step{\opt{\alpha}} Q_2$ such that $P' \upto[\bisimtbr] Q_1$ and $P'' \upto[\bisimtbr] Q_2$
        \item if $P=P_0 \pathtau P_1 \step{\rt} P_2 \pathtau P_3 \step{\rt} ... \pathtau P_{2r-1} \step{\opt{\rt}} P_{2r}$ with $r > 0$, such that $\forall i \in [0,r{-}1]$, $\theta_X(P) \bisimtbr \theta_X(P_{2i}) \wedge P \bisimtbr P_{2i+1} \wedge \deadend{P_{2i+1}}{X}$, then there exists a path $Q = Q_0 \pathtau Q_1 \step{\rt} Q_2 \pathtau Q_3 \step{\rt} ... \pathtau Q_{2n-1} \step{\opt{\rt}} Q_{2n}$ with $n>0$, such that $\forall i \in [1,2n{-}1]\;\theta_X(P) \upto[\bisimtbr] \theta_X(Q_{i})$,\;$\forall j \in [0,n{-}1]\;\deadend{Q_{2j+1}}{X}$ and $\theta_X(P_{2r}) \upto[\bisimtbr] \theta_X(Q_{2n})$
        \item if $P \pathtau P_0 \nsteptau$ with $P \bisimtbr P_0$ then there exists a path $Q \pathtau Q_0 \nsteptau$.
    \end{enumerate}
\end{definition}

\begin{proposition} \label{prop:up to b}
    Let $P,Q \in \closed$. Then $P \bisimtbr Q$ iff there exists a \tb time-out bisimulation $\mathcal{B}$ up to $\bisimtbr$ such that $P \tbisim Q$.
\end{proposition}

\begin{proof}
    Let $\tbisim$ be a \tb time-out bisimulation up to $\bisimtbr$\,. We are going to show that $\upto[\bisimtbr]$ is a \tb time-out bisimulation. Let $P,Q \in \closed$ such that $P \upto[\bisimtbr] Q$. Then there exists $P^\dag,Q^\dag \in \closed$ such that $P \bisimtbr P^\dag \tbisim Q^\dag \bisimtbr Q$.
    \begin{enumerate}
        \item If $P \step{\alpha} P'$ with $\alpha \in A_\tau$ then, since $P \bisimtbr P^\dag$, there exists a path $P^\dag \pathtau P^\star \step{\opt{\alpha}} P^\ddag$ such that $P \bisimtbr P^\star$ and $P' \bisimtbr P^\ddag$. Since $P^\dag \pathtau P^\star \step{\opt{\alpha}} P^\ddag$ and $P^\dag \tbisim Q^\dag$, there exists a path $Q^\dag \pathtau Q^\star \step{\opt{\alpha}} Q^\ddag$ such that $P^\star \upto[\bisimtbr] Q^\star$ and $P^\ddag \upto[\bisimtbr] Q^\ddag$. Since $Q^\dag \pathtau Q^\star$ and $Q^\dag \bisimtbr Q$, there exists a path $Q \pathtau Q_0$ such that $Q^\star \bisimtbr Q_0$; moreover, since $Q^\star \step{\opt{\alpha}} Q^\ddag$, there exists a path $Q_0 \pathtau Q_1 \step{\opt{\alpha}} Q_2$ such that $Q^\star \bisimtbr Q_1$ and $Q^\ddag \bisimtbr Q_2$. As a result, there exists a path $Q \pathtau Q_1 \step{\opt{\alpha}} Q_2$ such that, by transitivity of $\bisimtbr$\,, $P \upto[\bisimtbr] Q_1$ and $P' \upto[\bisimtbr] Q_2$.
        \item If $\deadend{P}{X}$ and $P \step{\rt} P'$ then, since $P \bisimtbr P^\dag$, there exists a path $P^\dag=P^\dag_0 \pathtau P^\dag_1 \step{\rt} P^\dag_2 \pathtau P^\dag_3 \step{\rt} ... \pathtau P^\dag_{2r-1} \step{\opt{\rt}} P^{\dag}_{2r}$ with $r>0$, such that $P^\dag_1\nsteptau$,\linebreak[3] $\forall i \in [1,r{-}1],\; \theta_X(P) \bisimtbr \theta_X(P^\dag_{2i}) \wedge \deadend{P^\dag_{2i+1}}{X}$ and $\theta_X(P') \bisimtbr \theta_X(P^\dag_{2r})$. As remarked in Section~\ref{sec:brb}, we even have $P \bisimtbr P^\dag_{2i+1}$ for all $i \in [1,r{-}1]$. As $\bisimtbr$ is a congruence, $\theta_X(P^\dag) \bisimtbr \theta_X(P)$, so $\forall i \in [0,r{-}1],\; \theta_X(P^\dag) \bisimtbr \theta_X(P^\dag_{2i}) \wedge P^\dag \bisimtbr P^\dag_{2i+1}$. Since $P^\dag \tbisim Q^\dag$, there exists a path $Q^\dag = Q^\dag_0 \pathtau Q^\dag_1 \step{\rt} Q^\dag_2 \pathtau Q^\dag_3 \step{\rt} ... \pathtau Q^\dag_{2n-1} \step{\opt{\rt}} Q^\dag_{2n}$ with $n>0$, such that 
        $\forall i \in [1,2n{-}1]\;\theta_X(P^\dag) \upto[\bisimtbr] \theta_X(Q^\dag_{i})$,\;$\forall j \in [0,n{-}1]\;\deadend{Q_{2j+1}}{X}$ and $\theta_X(P^\dag_{2r}) \upto[\bisimtbr] \theta_X(Q^\dag_{2n})$.
        Since $Q^\dag \bisimtbr Q$, there exists a path $Q \pathtau Q_1 \step{\rt} Q_2 \pathtau Q_3 \step{\rt} ... \pathtau Q_{2m-1} \step{\opt{\rt}} Q_{2m}$ with $m>0$, such that $Q_1\nsteptau$, $\forall k \in [1,m{-}1], \exists j \in [0,2n{-}1], \theta_X(Q^\dag_j) \bisimtbr \theta_X(Q_{2k}) \wedge \deadend{Q_{2k+1}}{X}$ and $\theta_X(Q^\dag_{2n}) \bisimtbr \theta_X(Q_{2m})$. As a result, there exists a path $Q \pathtau Q_1 \step{\rt} Q_2 \pathtau Q_3 \step{\rt} \dots \pathtau Q_{2m-1} \step{\opt{\rt}} Q_{2m}$ with $m>0$, such that $Q_1\nsteptau$, and, by transitivity of $\bisimtbr$\,, $\forall i \in [1,m{-}1],\; \theta_X(P) \upto[\bisimtbr] \theta_X(Q_{2i}) \wedge \deadend{Q_{2i+1}}{X}$ and $\theta_X(P') \upto[\bisimtbr] \theta_X(Q_{2m})$.
        \item If $P \nsteptau$ then, since $P \bisimtbr P^\dag$, there exists a path $P^\dag \pathtau P^\dag_0 \nsteptau$, and $P^\dag \bisimtbr P^\dag_0$. Since $P^\dag \tbisim Q^\dag$, there exists a path $Q^\dag \pathtau Q^\dag_0 \nsteptau$. Since $Q^\dag \bisimtbr Q$, there exists a path $Q \pathtau Q_0 \nsteptau$.
\popQED
    \end{enumerate}
\end{proof}

\noindent
The following lemma will be useful to deal with the matching of paths.

\begin{lemma} \label{lem:guarded}
    Let $H \mathbin\in \expr$ be well-guarded and have free variables from $W \mathbin\subseteq V$ only, and let $\rho,\nu \mathbin\in \closed^W\!\!$. 
    \begin{enumerate}
        \item $\init{H[\rho]} = \init{H[\nu]}$. \label{init}
        \item If $H[\rho] \step{\alpha} R$ with $\alpha \in Act$ then there exists $H' \in \expr$ with free variables in $W$ only such that $R = H'[\rho]$ and $H[\nu] \step{\alpha} H'[\nu]$.
        Moreover, in case $\alpha\in\{\tau,\rt\}$, also $H'$ is well-guarded.\label{2}
    \end{enumerate}
\end{lemma}

\begin{proof}
\newcommand{\spar}[1]{\mathbin{\|^{}_{#1}}}        
\ref{init}.\ has been proven in \cite{strongreactivebisimilarity}.
We obtain \ref{2}.\ by induction on the derivation of $H[\rho] \step\alpha R$, making a case distinction on the shape of $H$.

Let $H=\alpha.G$, so that $H[\rho] = \alpha.G[\rho]$.
Then $R = G[\rho]$ and $H[\nu] \step\alpha G[\nu]$.
In case $\alpha\in\{\tau,\rt\}$, also $G$ is well-guarded.

The case $H=0$ cannot occur. Nor can the case $H=x\in V$, as $H$ is well-guarded.

Let $H = H_1 \spar{S} H_2$, so that $H[\rho] = H_1[\rho]\spar{S} H_2[\rho]$. Note that $H_1$ and $H_2$ are well-guarded and have free variables in $W$ only. One possibility is that $a\notin S$, $H_1[\rho]\step\alpha R_1$ and $R= R_1 \spar{S} H_2[\rho]$. By induction, $R_1$ has the form $H'_1[\rho]$ for some term $H'_1\in\expr$ with free variables in $W$ only, and in case $\alpha\in\{\tau,\rt\}$,
also $H'_1$ is well-guarded. Moreover, $H_1[\nu] \step\alpha H'_1[\nu]$.
Thus $R = (H'_1 \spar{S} H_2)[\rho]$, and $H':= H'_1 \spar{S} H_2$ has free variables in $W$ only. In case $\alpha\in\{\tau,\rt\}$, $H$ is well-guarded. Moreover, $H[\nu] =  H_1[\nu]\spar{S} H_2[\nu] \step\alpha  H'_1[\nu]\spar{S} H_2[\nu] = H'[\nu]$.

The other two cases for $\spar{S}$, and the cases for the operators $+$ and $\rename$, are equally trivial.

Let $H= \theta_L^U(H^\dagger)$, so that $H[\rho] = \theta_L^U(H^\dagger[\rho])$. Note that $H^\dagger$ is well-guarded and has free variables in $W$ only. The case $\alpha = \tau$ is again trivial, so assume $\alpha\neq\tau$. Then {$H^\dagger[\rho] \step\alpha R$} and either $\alpha\in X$ or $\init{H^\dagger[\rho]} \cap (L\cup\{\tau\}) = \emptyset$. By induction, $R$ has the form $H'[\rho]$ for some term $H'\in\expr$ with free variables in $W$ only, and in case $\alpha=\rt$ this term is well-guarded. Moreover, $H^\dagger[\nu] \step\alpha H'[\nu]$. Since $\init{H^\dagger[\rho]} = \init{(H^\dagger[\nu]}$ by Lemma~\ref{lem:guarded}.\ref{init}, either $\alpha\in X$ or $\init{H^\dagger[\nu]} \cap (L\cup\{\tau\}) = \emptyset$. Consequently, $H[\nu] = \theta_L^U(H^\dagger[\nu])\step\alpha H'[\nu]$.

Let $H= \psi_X(H^\dagger)$, so that $H[\rho] = \psi_X(H^\dagger[\rho])$. Note that $H^\dagger$ is well-guarded and has free variables in $W$ only. The case $\alpha \in A\cup\{\tau\}$ is trivial, so assume $\alpha=\rt$. Then {$H^\dagger[\rho] \step\rt R^\dagger$} for some $R^\dagger$ such that $R=\theta_X(R^\dagger)$. Moreover, $H^\dagger[\rho] \cap (X\cup\{\tau\}) = \emptyset$. By induction, $R^\dagger$ has the form $H'[\rho]$ for some well-guarded term $H'\in\expr$ with free variables in $W$ only. Moreover, $H^\dagger[\nu] \step\rt H'[\nu]$. Thus $R=(\theta_X(H'))[\rho]$ and $\theta_X(H')$ is well-guarded and has free variables in $W$ only. Since $\init{H^\dagger[\rho]} = \init{(H^\dagger[\nu]}$ by Lemma~\ref{lem:guarded}.\ref{init}, $H^\dagger[\nu] \cap (X\cup\{\tau\}) = \emptyset$. Consequently, $H[\nu] = \psi_X(H^\dagger[\nu])\step\rt \theta_X(H'[\nu]) = (\theta_X(H'))[\nu]$.

Finally, let $H = \langle x|\equa \rangle$, so that $H[\rho] = \langle x|\equa[\rho^\dagger]\rangle$, where $\rho^\dagger \in \closed^{W {\setminus} V_\equa}$ is the restriction of $\rho$ to $W {\setminus} V_\equa$.
The transition $\langle\equa_x[\rho^\dagger]|\equa[\rho^\dagger]\rangle \step\alpha R$ is derivable through a subderivation of the one for $\langle x|\equa[\rho^\dagger]\rangle \step\alpha R$.
Moreover, $\langle\equa_x[\rho^\dagger]|\equa[\rho^\dagger]\rangle = \langle\equa_x|\equa\rangle[\rho]$.
So by induction, $R$ has the form $H'[\rho]$ for some term $H'\mathbin\in\expr$ with free variables in $W$
only, and $\langle\equa_x|\equa\rangle[\nu] \step\alpha H'[\nu]$.
Moreover, in case $\alpha\in\{\tau,\rt\}$, also $H'$ is well-guarded.
Since $\langle\equa_x|\equa\rangle[\nu] = \langle\equa_x[\nu^\dagger]|\equa[\nu^\dagger]\rangle$, it follows that 
$H[\nu] = \langle x|\equa\rangle[\nu]= \langle x|\equa[\nu^\dagger]\rangle\step\alpha H'[\nu]$. 
\end{proof}

\begin{corollary} \label{cor:guarded path}
    Let $H \in \expr$ be well-guarded and have free variables from $W \subseteq V$ only, and let $\rho,\nu \in \closed^W$. 
    \begin{itemize}
        \item If $H[\rho] \pathtau R \step{\alpha} S$ with $\alpha \in Act$ then there exists $H',H'' \in \expr$ with free variables in $W$ only such that $R = H'[\rho]$, $S = H'[\rho]$ and $H[\nu] \pathtau H'[\nu] \step{\alpha} H'[\nu]$.
        \item if $H[\rho] \pathtau R_1 \step{\rt} R_2 \pathtau R_3 \step{\rt} \dots \pathtau R_{2r-1} \step{\opt{\rt}} R_{2r}$ with $r>0$ then there exists $(H_i)_{i\in [1,2r]} \in \expr^{2r}$ with free variables in $W$ only such that $\forall i \in [1,2r],\; R_i = H_i[\rho]$ and $H[\rho] \pathtau H_1[\nu] \step{\rt} H_2[\nu] \pathtau H_3[\nu] \step{\rt} \dots \pathtau H_{2r-1}[\nu] \step{\opt{\rt}} H_{2r}[\nu]$.
\qed
    \end{itemize}
\end{corollary}

\begin{proof}[Proof of Proposition \ref{prop:rsp}]
    It suffices to prove the proposition when $\rho, \nu \in \closed^{V_\equa}$ and only variables of $V_\equa$ can occur in the expressions $\equa_x$ for $x \in V_\equa$. Indeed, the general case requires to prove that, for all $\sigma: V \rightarrow \closed$, $\rho[\sigma] \bisimrtbr \nu[\sigma]$. Let $\hat{\sigma}: V\setminus V_\equa \rightarrow \closed$ be defined as $\forall x \in V \setminus V_\equa, \hat{\sigma}(x) = \sigma(x)$. Since $\rho \bisimrtbr \equa[\rho]$, $\rho[\sigma] \bisimrtbr \equa[\rho][\sigma] = \equa[\hat{\sigma}][\rho[\sigma]]$, therefore, proving the proposition with $\rho[\sigma]$, $\nu[\sigma]$ and $\equa[\hat{\sigma}]$ is sufficient.

    It also suffices to prove the proposition for the case that $\equa$ is manifestly well-guarded. Indeed, if $\equa$ is well-guarded, let $\equa'$ be the manifestly well-guarded specification into which $\equa$ can be converted. Since $\bisimrtbr$ is a lean congruence, a solution to $\equa$ up to $\bisimrtbr$ is a solution to $\equa'$ up to $\bisimrtbr$.
    
    Let $\equa$ be a manifestly well-guarded recursive specification with free variables from $V_\equa$ only, and $\rho, \nu \in \closed^{V_\equa}$ two of its solutions up to $\bisimrtbr$. We are going to show that the symmetric closure of 
    \begin{align*}
        \tbisim := \{(H[\equa[\rho]],H[\equa[\nu]]) \mid H \in \expr \mbox{ is without $\tau_I$ and with free variables from } V_\equa \mbox{ only}\}
    \end{align*}
    is a branching time-out bisimulation up to $\bisimrtbr$. Here $\equa[\rho] := \{x = \equa_x[\rho] \mid x \in V_\equa\}\in\closed^{V_\equa}$ is employed as a substitution. Let $P,Q \in \closed$ such that $P \tbisim Q$. Then there exists $H \in \expr$ with free variables from $V_\equa$ only such that $P = H[\equa[\rho]]$ and $Q = H[\equa[\nu]]$, the other case being symmetrical. Note that $H[\equa[\rho]] = H[\equa][\rho]$. Since $H$ and $\equa$ have free variables from $V_\equa$ only, so does $H[\equa]$. Moreover, since $\equa$ is manifestly well-guarded, $H[\equa]$ is well-guarded.
    \begin{enumerate}
        \item Let $H[\equa[\rho]] \pathtau P_1 \step{\alpha} P_2$. By Corollary \ref{cor:guarded path}, there exists $H_1, H_2 \in \expr$ with free variables from $V_\equa$ only such that $P_1 = H_1[\rho]$, $P_2 = H_2[\rho]$ and $H[\equa][\nu] \pathtau H_1[\nu] \step{\alpha} H_2[\nu]$. Furthermore, since $\bisimrtbr$ is a congruence and $\rho$ and $\nu$ are solutions of $\equa$ up to $\bisimrtbr$\,, $H_1[\rho] \bisimrtbr H_1[\equa[\rho]]$, $H_1[\nu] \bisimrtbr H_1[\equa[\nu]]$, $H_2[\rho] \bisimrtbr H_2[\equa[\rho]]$ and $H_2[\nu] \bisimrtbr H_2[\equa[\nu]]$, therefore, by definition of $\tbisim$, $H_1[\rho] \upto[\bisimrtbr] H_1[\nu]$ and $H_2[\rho] \upto[\bisimrtbr] H_2[\nu]$.
        \item Let $H[\equa[\rho]] \pathtau P_1 \step{\rt} P_2 \pathtau P_3 \step{\rt} ... \pathtau P_{2r-1} \step{\opt{\rt}} P_{2r}$ with $r>0$, such that $\forall i \in [0,r{-}1],\; \theta_X(P) \bisimtbr \theta_X(P_{2i}) \wedge P \bisimtbr P_{2i+1} \wedge \deadend{P_{2i+1}}{X}$. By Corollary \ref{cor:guarded path}, there exists $(H_i)_{i \in [1,2r]} \in \expr^{2r}$ with free variables from $V_\equa$ only such that $\forall i \in [1,2r],\; P_i = H_i[\rho]$ and $H[\equa][\nu] \pathtau H_1[\nu] \step{\rt} H_2[\nu] \pathtau H_3[\nu] \step{\rt} ... \pathtau H_{2r-1} \step{\opt{\rt}} H_{2r}[\nu]$. Since all $H_{2i+1}$ are well-guarded, by Lemma \ref{lem:guarded}, $\forall i \in [0,r{-}1],\; \deadend{H_{2i+1}[\nu]}{X}$. Furthermore, since $\bisimrtbr$ is a congruence and $\rho$ and $\nu$ are solutions of $\equa$ up to $\bisimrtbr$\,, for all $i \in [1,2r]$, $H_i[\rho] \bisimrtbr H_i[\equa[\rho]]$ and $H_i[\nu] \bisimrtbr H_i[\equa[\nu]]$; therefore, by definition of $\tbisim$, for $i \in [1,2r]$, $\theta_X(H_i[\rho]) \upto[\bisimrtbr] \theta_X(H_i[\nu])$ (notice that $\theta_X(H_i[\equa][\rho]) = \theta_X(H_i)[\equa][\rho]$). It follows that $\forall i \in [1,2r{-}1]\;\theta_X(P) \upto[\bisimtbr] \theta_X(H_i[\nu])$.
        \item Let $H[\equa[\rho]] \pathtau P_0 \nsteptau$. By Lemma \ref{lem:guarded}, there exists a well-guarded $H_1\in \expr$ with free variables from $V_\equa$ only such that $P_0 = H_0[\rho]$ and $H[\equa][\nu] \pathtau H_0[\nu]$. Since $H_0$ is well-guarded and $P_0 \nsteptau$, according to Lemma \ref{lem:guarded}.1, $H_0[\nu] \nsteptau$.
    \end{enumerate}

\noindent
    Next, we will prove that $\tbisim$ is a rooted \tb time-out bisimulation. Let $P,Q \in \closed$ such that $P \tbisim Q$. Then there exists $H \in \expr$ with free variables from $V_\equa$ only such that $P = H[\equa[\rho]]$ and $Q = H[\equa[\nu]]$, the other case being symmetrical. Note that $H[\equa[\rho]] = H[\equa][\rho]$. Since $H$ and $\equa$ have free variables from $V_\equa$ only, so does $H[\equa]$. Moreover, since $\equa$ is manifestly well-guarded, $H[\equa]$ is well-guarded.
    \begin{enumerate}
        \item Let $P \step{\alpha} P'$. By Lemma \ref{lem:guarded}, there exists $H' \in \expr$ with free variables from $V_\equa$ only such that $P' = H'[\rho]$ and $Q = H[\equa][\nu] \step{\alpha} H'[\nu]$. Furthermore, since $\bisimrtbr$ is a congruence and $\rho$ and $\nu$ are solutions of $\equa$ up to $\bisimrtbr\,$, $H'[\rho] \bisimrtbr H'[\equa[\rho]]$ and $H'[\nu] \bisimrtbr H'[\equa[\nu]]$. Therefore, by definition of $\tbisim$, $H'[\rho] \upto[\bisimrtbr] H'[\nu]$. But, $\tbisim$ is a \tb time-out bisimulation up to $\bisimrtbr\,$, thus, by Proposition~\ref{prop:up to b}, $H'[\rho] \bisimrtbr H'[\nu]$.
        \item Let $\deadend{P}{X}$ and $P \step{\rt} P'$. By Lemma \ref{lem:guarded}, there exists $H' \in \expr$ with free variables from $V_\equa$ only such that $P' = H'[\rho]$ and $Q=H[\equa][\nu] \step{\rt} H'[\nu]$. Exactly as above, not even using $\deadend{P}{X}$, this implies $H'[\rho] \bisimrtbr H'[\nu]$. Thus, since $\bisimrtbr$ is a congruence, $\theta_X(H'[\rho]) \bisimrtbr \theta_X(H'[\nu])$.
    \end{enumerate}

\noindent
    By considering $H = x$ with $x \in V_\equa$, this yields $\equa_x[\rho] \bisimrtbr \equa_x[\nu]$ and so $\rho(x) \bisimrtbr \equa_x[\rho] \bisimrtbr \equa_x[\nu] \bisimrtbr \nu(x)$. Consequently, $\rho \bisimrtbr \nu$.
\end{proof}

\section{Proof of Lemma \ref{lem:independent}}\label{app:contraintuitive}

\begin{proof}[Proof of Lemma \ref{lem:independent}]
    Suppose that there exists $X \subseteq A$ such that $\deadend{P}{X}$ and there exists
    $P \step{\rt} P'$ such that $\theta_X(P) \bisimtbr \theta_X(P')$, yet there exists $Y \subseteq A$ such that $\deadend{P}{Y}$ and $\neg(\theta_Y(P) \bisimtbr \theta_Y(P'))$.

    Since $P \nsteptau$ and $\theta_X(P) \bisimtbr \theta_X(P')$, By Definition~\ref{def:time-out bisim}.3 there exists a path $\theta_X(P') \pathtau P^\dag \nsteptau$, such that $\theta_X(P') \bisimtbr P^\dag$ by Definition~\ref{def:time-out bisim}.1. By the semantics, there exists a path $P' \pathtau P'' \nsteptau$ such that $P^\dag = \theta_X(P'')$. Proposition~\ref{prop:time-out bisim}.2 yields $P\bisimtbr[X] P''$, so $P\bisimtbr P''$ by Lemma~\ref{lem:obvious}.3 and $\init{P} = \init{P''}$ by Lemma~\ref{lem:obvious}.2.

    We are going to show that if $P \bisimtbr Q$ and $\init{P} = \init{Q}$ then we can find a nonempty path $Q (\pathtau\step{\rt})^*\pathtau Q'$ such that $P \bisimtbr Q'$ and $\init{P} = \init{Q'}$.

    Since $P \bisimtbr Q$, $\deadend{P}{Y}$ and $P \step{\rt} P'$, there exists a path $Q \pathtau Q_1 \step{\rt} Q_2 \pathtau Q_3 \step{\rt} ... \pathtau Q_{2r-1} \step{\opt{\rt}} Q_{2r}$ with $r>0$, such that $\forall i \in [1,r{-}1], \theta_Y(P) \bisimtbr \theta_Y(Q_{2i})\linebreak[2] \wedge\linebreak[2]\deadend{Q_{2i+1}}{Y}$ and $\theta_Y(P') \bisimtbr \theta_Y(Q_{2r})$. Since $\neg(\theta_Y(P) \bisimtbr \theta_Y(P'))$, $Q \ne Q_{2r}$. 

    Since $\theta_Y(P') \pathtau \theta_Y(P'') \nsteptau$ and $\theta_Y(P') \bisimtbr \theta_Y(Q_{2r})$, there exists a path $\theta_Y(Q_{2r}) \pathtau Q^\dag \nsteptau$ such that $Q^\dag \bisimtbr \theta_Y(P'')$. According to the semantics, there is a path $Q_{2r} \pathtau Q' \nsteptau$ such that $Q^\dag = \theta_Y(Q')$. Since $\deadend{P''}{Y}$, $Q' \nsteptau$ and $\theta_Y(P'') \bisimtbr \theta_Y(Q')$, Proposition~\ref{prop:time-out bisim}.2 and Lemma~\ref{lem:obvious} yields $Q' \bisimtbr P'' \bisimtbr P$ and $\init{Q'} = \init{P''} = \init{P}$.

    As a result, there exists an infinite $\tau/t$-path starting in $P''$, but that contradicts the strong guardedness of $P$.
\end{proof}

\section{Soundness of the Reactive Approximation Axiom} \label{app:RA}

\begin{lemma}\label{lem:theta twice}
    $\forall P \in \closed, \theta_X(P) \bisim \theta_X(\theta_X(P))$
\end{lemma}

\begin{proof}
    Trivial when considering the semantics of $\theta_X$.
\end{proof}

\begin{proof}[Proof of Proposition \ref{prop:soundness}]
    We show that $\tbisim := \{(P,Q),(Q,P) \mid \forall X \subseteq A, \; \psi_X(P) \bisimrtbr \psi_X(Q)\}$ is a rooted branching time-out bisimulation. Let $P,Q \in \closed$ such that $P \tbisim Q$. Thus, $\forall X \subseteq A, \psi_X(P) \bisimrtbr \psi_X(Q)$.
    \begin{enumerate}
        \item If $P \step{\alpha} P'$ with $\alpha \in A_\tau$ then, by the semantics, $\psi_A(P) \step{\alpha} P'$. Since $\psi_A(P) \bisimrtbr \psi_A(Q)$, there exists a transition $\psi_A(Q) \step{\alpha} Q'$ such that $P' \bisimtbr Q'$. By the semantics, $Q \step{\alpha} Q'$.
        \item If $\deadend{P}{X}$ and $P \step{\rt} P'$ then, by the semantics, $\psi_X(P) \step{\rt} \theta_X(P')$. Since $\psi_X(P) \bisimrtbr \psi_X(Q)$, there exists a transition $\psi_X(Q) \step{\rt} Q^\ddag$ with $\theta_X(\theta_X(P')) \bisimtbr \theta_X(Q^\ddag)$. By the semantics, $Q^\ddag = \theta_X(Q')$ and $Q \step{\rt} Q'$. By Lemma~\ref{lem:theta twice}, $\theta_X(P') \bisim \theta_X(\theta_X(P')) \bisimtbr \theta_X(\theta_X(Q')) \bisim \theta_X(Q')$.
    \popQED
    \end{enumerate}
\end{proof}

\section{Proofs of Completeness for Finite Processes} \label{app:completeness finite}

\begin{definition}\rm
Call a process $P$ \emph{brb-stable} if (a) there is no process $P^\dagger$ with $P\steptau P^\dagger$ and $P \bisimtbr P^\dagger$, and (b) there is no set $X\subseteq A$ and process $P^\ddagger$ with
$\deadend{P}{X}$, $P \step\rt P^\ddagger$ and $\theta_X(P) \bisimtbr \theta_X(P^\ddagger)$.
\end{definition}

\begin{lemma}\label{lem:brb-stable}
If $P$ and $Q$ are brb-stable and $P \bisimtbr Q$ then $P \bisimrtbr Q$.
\end{lemma}
\begin{proof}
Assume that $P$ and $Q$ are brb-stable and $P \bisimtbr Q$.
If $P \step{\alpha} P'$ with $\alpha \in A_\tau$, then there is a path $Q \pathtau Q_1 \step{\opt{\alpha}} Q_2$ with $P \bisimtbr Q_1$ and $P' \bisimtbr Q_2$. By symmetry and transitivity of $\bisimtbr$ we have $Q_1 \bisimtbr Q$, so by the brb-stability of $Q$ it follows that $Q_1=Q$. Moreover, if $\alpha=\tau$ and $Q_2=Q_1$ then $P'\bisimtbr P$, contradicting the brb-stability of $P$. Thus $Q \step\alpha Q_2$. This argument also yields that $\init{P} = \init{Q}$.

Furthermore, if $\init{P}\cap(X\cup\{\tau\}) = \emptyset$ and $P \step{\rt} P'$ then there is a path $Q \pathtau Q_1 \step{\rt} Q_2 \pathtau Q_3 \step{\rt} ... \pathtau Q_{2r{-}1} \step{\opt{\rt}} Q_{2r}$ with $r>0$, such that $Q_1\nsteptau$, $\forall i \in [1,r{-}1],\; \theta_X(P) \bisimtbr \theta_X(Q_{2i})\linebreak[2] \wedge\linebreak[2] \deadend{Q_{2i+1}}{X}$ and $\theta_X(P') \bisimtbr \theta_X(Q_{2r})$. By Lemma~\ref{lem:obvious}.1, Lemma~\ref{lem:obvious}.3 and the transitivity of $\bisimtbr$\,, we have $P \bisimtbr[X] Q_1$, $P \bisimtbr Q_1$ and $Q \bisimtbr Q_1$, respectively, so the brb-stability of $Q$ yields $Q_1=Q$. If $r>1$ we would have $\theta_X(Q_2) \bisimtbr \theta_X(P) \bisimtbr \theta_X(Q)$, contradicting the brb-stability of $Q$. Thus $r=1$ and \raisebox{0pt}[0pt]{$Q \step{\opt{\rt}} Q_2$}. If $Q_2=Q$ we would obtain $\theta_X(P') \bisimtbr \theta_X(Q_{2})\bisimtbr \theta_X(P)$, contradicting the  brb-stability of $P$. Hence $Q \step{\rt} Q_2$. So indeed $P \bisimrtbr Q$.
\end{proof}

\begin{proof}[Proof of Proposition \ref{prop:collapse}]
    We define the \textit{length} of a path $P_0 \step{\alpha_1} P_1 \step{\alpha_2} ... \step{\alpha_n} P_n$ to be $n$ and the \textit{depth} of a process $P$, denoted $d(P)$, to be the length of the longest path starting from $P$. It is well defined because $P$ is a recursion-free $\ccsp$ process. Note that $d(\theta_X(P))\leq d(P)$.

    $(\bisimtbr)$ We will proceed by induction on $\max(d(P), d(Q))$. Let $n \mathbin\in\nat$ and suppose that the property holds for any recursion-free $\ccsp$ processes $P, Q$ such that $\max(d(P), d(Q)) < n$. Let $P,Q$ be two recursion-free $\ccsp$ processes such that $\max(d(P), d(Q)) = n$ and $P \bisimtbr Q$.
    
    Since $P$ is recursion-free, using Lemma \ref{lem:obvious}.4, there exists a path $P \pathtau P_1 \step{\rt} P_2 \pathtau P_3 \step{\rt} ... \pathtau P_{2r-1} \step{\rt} P_{2r} \pathtau P_0$ with $r \in \nat$, such that $P \bisimtbr P_1$ and $\forall i \in [1,r], \exists X_i \subseteq A, \deadend{P_{2i-1}}{X_i} \wedge \theta_{X_i}(P) \bisimtbr \theta_X(P_{2r})$, $P \bisimtbr P_0$ and $P_0$ is brb-stable. We are going to show that, for all $\alpha \in Act$, $Ax_r \vdash \alpha.\hat{P} = \alpha.\hat{P}_0$. If $P$ is brb-stable then $P = P_1$ and $r=0$ so $P = P_0$ and this is trivial. Thus, suppose that $P$ is not brb-stable. Then, as $P_0$ is brb-stable, $P \ne P_0$ and so $d(P_0) < d(P)$.

    Let $I := \{(\alpha,P') \mid P \step{\alpha} P' \wedge \alpha \in A_\tau \wedge (\alpha \ne \tau \vee P \not\bisimtbr P')\}$, listing the outgoing transitions of $P$ not labelled by $\rt$ and not elidable w.r.t.\ $\bisimtbr\,$. Let $(\alpha,P') \in I$. Since $P \bisimtbr P_0$, there exists a path $P_0 \pathtau P_1 \step{\opt{\alpha}} P_2$ such that $P \bisimtbr P_1$ and $P' \bisimtbr P_2$. Since $P_0$ is brb-stable and $P_1 \bisimtbr P \bisimtbr P_0$, $P_0 = P_1$ and $P_0 \step{\opt{\alpha}} P_2$. If $\alpha = \tau$ then $P_0 \bisimtbr P \not\bisimtbr P' \bisimtbr P_2$ so $P_0 \ne P_2$ and $P_0 \step{\alpha} P_2$. Since $\max(d(P_2),d(P')) < d(P)$, by induction, $Ax_r \vdash \alpha.\hat{P}' = \alpha.\hat{P}_2$, so by Lemma~\ref{lem:head-normal form} $Ax_r \vdash \alpha.P' = \alpha.P_2$. As a result, $Ax_r \vdash \hat{P}_0 = \sum_{(\alpha,P')\in I}\alpha.{P'} + \hat{P}_0$.

    Let $J := \{(\tau,P') \mid P \steptau P' \wedge P \bisimtbr P'\}$, listing the outgoing $\tau$-transitions of $P$ elidable w.r.t.\ $\bisimtbr\,$. Let $(\tau,P') \in J$. Since $P' \bisimtbr P \bisimtbr P_0$ and $\max(d(P'),d(P_0)) < d(P)$, by induction, $Ax_r \vdash \tau.\hat{P'} = \tau.\hat{P_0}$, so by Lemma~\ref{lem:head-normal form} $Ax_r \vdash \tau.P' = \tau.\hat{P_0}$.

    Suppose that $P \steptau$. Since $P$ is not brb-stable, there exists a transition $P \steptau P'$ such that $P \bisimtbr P'$ (i.e.\ $J \ne \emptyset$). Now the following equality can be derived from $Ax_r$, for all $\beta \in Act$. Here the first step applies $\hyperlink{Lt}{\mbox{\bf L}\tau}$.
    \begin{align*}
        Ax_r \vdash \beta.\hat{P} = &~ \beta.(\sum_{\{(\alpha,P') \mid P\step{\alpha} P' \wedge \alpha \in A_\tau\}}\alpha.{P'}) = \beta.(\sum_{(\tau,P') \in J}\tau.{P}' + \sum_{(\alpha,P') \in I}\alpha.{P}') \\
        = &~ \beta.(\tau.\hat{P}_0 + \sum_{(\alpha,P') \in I}\alpha.{P'}) = \beta.(\tau.(\hat{P}_0 + \sum_{(\alpha,P') \in I}\alpha.{P'}) + \sum_{(\alpha,P') \in I}\alpha.{P'}) \\
        = &~ \beta.(\hat{P}_0 + \sum_{(\alpha,P') \in I}\alpha.{P'}) = \beta.\hat{P}_0
    \end{align*}

    Now, suppose that $P \nsteptau$. Then $J = \emptyset$ and $I = \{(\alpha,P') \mid P \step{\alpha} P' \wedge \alpha \in A\}$. Since $P$ is not brb-stable and $P \nsteptau$, there exists $X \subseteq A$ and $P \step{\rt} P'$ such that $\deadend{P}{X}$ and $\theta_X(P) \bisimtbr \theta_X(P')$ $(\star)$. Moreover, since $P_0$ is brb-stable, $P \nsteptau$ and $P \bisimtbr P_0$, $P_0 \nsteptau$. Let $I_0 := \{(\alpha,P'_0) \mid P_0 \step{\alpha} P'_0 \wedge \alpha \in A_\tau\}$. Since $P\nsteptau$, for all $(\alpha, P'_0) \in I_0$, there exists $P \step{\alpha} P'$ with $P' \bisimtbr P'_0$, thus, $\init{P} = \init{P_0}$ and, by induction, $Ax_r \vdash \sum_{(\alpha,P') \in I}\alpha.\hat{P}' = \sum_{(\alpha,P'_0) \in I_0}\alpha.\hat{P}'_0$.

    Let $H := \{(\rt,P') \mid \forall X \subseteq A, (\deadend{P}{X} \Rightarrow \theta_X(P) \bisimtbr \theta_X(P'))\}$, listing the outgoing  time-outs of $P$ that can be elided. Note that $(\star)$ implies $H \ne \emptyset$. Let $(\rt,P') \in H$. Since $P \bisimtbr P_0$, for all $X \subseteq A$, if $\deadend{P}{X}$ then $\theta_X(P') \bisimtbr \theta_X(P_0)$ and $\max(d(\theta_X(P')),d(\theta_X(P_0))) \le \max(d(P'),d(P_0)) < d(P)$. Therefore, by induction, for all $X \subseteq A$, if $\deadend{P}{X}$ then $Ax_r \vdash \rt.\widehat{\theta_X(P')} = \rt.\widehat{\theta_X(P_0)}$. As a result, using the reactive approximation axiom (RAA), $Ax_r \vdash \sum_{(\alpha,P') \in I}\alpha.\hat{P}' + \sum_{(\rt,P') \in H}\rt.\hat{P}' = \sum_{(\alpha,P'_0) \in I_0}\alpha.\hat{P}'_0 + \rt.\hat{P}_0$, so with Lemma~\ref{lem:head-normal form}
    \begin{equation}\label{C}
      Ax_r \vdash \sum_{(\alpha,P') \in I}\alpha.{P}' + \sum_{(\rt,P') \in H}\rt.{P}' = \sum_{(\alpha,P'_0) \in I_0}\alpha.{P}' + \rt.\hat{P}_0.
    \end{equation}
    
    Let $K := \{(\rt,P') \mid \forall X \subseteq A, (\deadend{P}{X} \Rightarrow \theta_X(P) \not\bisimtbr \theta_X(P'))\}$, listing the outgoing  time-outs of $P$ that cannot be elided. Let $(\rt,P') \in K$. Since $P \bisimtbr P_0$, for all $X \subseteq A$, if $\deadend{P}{X}$ then there exists a path $P_0 \pathtau P^X_1 \step{\rt} P^X_2 \pathtau P_3^X \step{\rt} ... \pathtau P_{2r-1}^X \step{\opt{\rt}} P^X_{2r}$ with $r>0$, such that $Q_1\nsteptau$, $\forall i \in [0,r{-}1],\, \theta_X(P_{2i}^X) \bisimtbr \theta_X(P) \linebreak[2]\wedge\linebreak[2] \deadend{P_{2i+1}}{X}$ and $\theta_X(P_{2r}^X) \bisimtbr \theta_X(P')$. Since $P_0$ is brb-stable, $P_0 \step{\opt{\rt}} P^X_{2r}$ and, since $\theta_X(P) \not\bisimtbr \theta_X(P')$, $P_0 \step{\rt} P^X_{2r}$ and $\max(d(\theta_X(P')),d(\theta_X(P_{2r}^X)))<n$. As a result, by induction, for all $X \subseteq A$, if $\deadend{P}{X}$ then there exists a transition $P_0 \step{\rt} P^X$ such that $Ax_r \vdash \rt.\widehat{\theta_X(P')} = \rt.\widehat{\theta_X(P^X)}$. Therefore, using RAA and Lemma~\ref{lem:head-normal form},
    \begin{equation}\label{D}
      Ax_r \vdash \hat{P}_0 = \hat{P}_0 + \sum_{(\rt,P') \in K}\rt.{P}'.
    \end{equation}

    According to Lemma \ref{lem:independent}, since $P$ is recursion-free and thus strongly guarded, for all $(\rt,P')$ such that $P \step{\rt} P'$, either, for all $X \subseteq A$ with $\deadend{P}{X}$ we have $\theta_X(P) \bisimtbr \theta_X(P')$; or, for all $X \subseteq A$ with $\deadend{P}{X}$ we have $\theta_X(P) \,\not\!\bisimtbr \theta_X(P')$. As a result,
    \begin{equation}\label{A}
      \hat{P} = \sum_{(\alpha,P') \in I}\alpha.P' + \sum_{(\rt,P') \in H}\rt.P' + \sum_{(t,P') \in K}\rt.P'.
    \end{equation}
    Let $\alpha \in Act$.
    Then, using (\ref{A}), (\ref{C}) and (\ref{D}), respectively, writing $R$ for $\sum_{(\alpha,P'_0) \in I_0}\alpha.{P}'_0 + \sum_{(\rt,P') \in K}\rt.P'$ and applying the $\rt$-branching axiom,
    \begin{align*}
        Ax_r \vdash \alpha.\hat{P} = &~ \alpha.(\sum_{(\alpha,P') \in I}\alpha.P' + \sum_{(\rt,P') \in H}\rt.P' + \sum_{(\rt,P') \in K}\rt.P') \\
        = &~ \alpha.(\sum_{(\alpha,P'_0) \in I_0}\alpha.{P}'_0 + \rt.\hat{P}_0 + \sum_{(\rt,P') \in K}\rt.P') \\
        = &~ \alpha.(\sum_{(\alpha,P'_0) \in I_0}\alpha.{P}'_0 + \rt.(\hat{P}_0 + \sum_{(\rt,P') \in K}\rt.P') + \sum_{(\rt,P') \in K}\rt.P') \\
        = &~ \alpha.(R+ \rt.(\sum_{(\alpha,P'_0) \in I_0}\alpha.\hat{P}'_0 + \sum_{\{P''_0\mid P_0 \step\rt P''_0\}}\rt.P''_0 + \sum_{(\rt,P') \in K}\rt.P')) \\
        = &~ \alpha.(R+ \rt.(R + \sum_{\{P''_0\mid P_0 \step\rt P''_0\}}\rt.P''_0)) = \alpha.(R + \sum_{\{P''_0\mid P_0 \step\rt P''_0\}}\rt.P''_0) \\
        = &~ \alpha.(\hat{P}_0 + \sum_{(\rt,P') \in K}\rt.P') = \alpha.\hat{P}_0
    \end{align*}

    As a result, in any case, for all $\alpha \in Act$, $Ax_r \vdash \alpha.\hat{P} = \alpha.\hat{P}_0$.
    
    Likewise, since $Q$ is recursion-free, a similar brb-stable $Q_0$ can be defined. By the same reasoning, it can be proved that, for all $\alpha \in Act$, $Ax_r \vdash \alpha.\hat{Q} = \alpha.\hat{Q}_0$. Since $P_0$ and $Q_0$ are brb-stable and $P_0 \bisimtbr P \bisimtbr Q \bisimtbr Q_0$, $P_0 \bisimrtbr Q_0$ according to Lemma \ref{lem:brb-stable}. 
    
    To end the proof, it suffices to show that, for all $\alpha \in Act$, $Ax_r \vdash \alpha.\hat{P}_0 = \alpha.\hat{Q}_0$, but \hypertarget{here}{we are going to prove the stronger statement $Ax_r \vdash \hat{P}_0 = \hat{Q}_0$}. Using RAA, it suffices to prove that, for all $X \subseteq A$, $Ax_r \vdash \psi_X(P_0) = \psi_X(Q_0)$.

    Let $(\alpha,P_0') \in I_0$. Since $P_0 \bisimrtbr Q_0$, there exists a transition $Q_0 \step{\alpha} Q_0'$ such that $P_0' \bisimtbr Q_0'$. By induction, $Ax_r \vdash \alpha.\hat{P}_0' = \alpha.\hat{Q}_0'$.
    
    Let $X \subseteq A$ and $(\rt,P_0')$ such that $\deadend{P_0}{X}$ and $P_0 \step{\rt} P_0'$. Since $P_0 \bisimrtbr Q_0$, there exists a transition $Q_0 \step{\rt} Q_0'$ such that $\theta_X(P_0') \bisimtbr \theta_X(Q_0')$. By induction, $Ax_r \vdash \rt.\widehat{\theta_X(P_0')} = \rt.\widehat{\theta_X(Q_0')}$.

    Let $X \subseteq A$. If $\init{P_0}\cap(X\cup\{\tau\}) \ne \emptyset$ then $\init{Q_0}\cap(X\cup\{\tau\}) \ne \emptyset$ and, using Lemma~\ref{lem:head-normal form},
    \begin{align*}
        Ax_r \vdash \psi_X(Q_0) = & ~\sum_{\{(\alpha,Q_0') \mid Q_0 \step{\alpha} Q_0' \wedge \alpha \ne \rt\}}\alpha.Q_0' \\
        = & ~\sum_{\{(\alpha,Q_0') \mid Q_0 \step{\alpha} Q_0' \wedge \alpha \ne\rt\}}\alpha.Q_0' + \sum_{(\alpha,P_0') \in I_0}\alpha.P_0' = ~\psi_X(P_0 + Q_0)
    \end{align*}
    If $\deadend{P_0}{X}$ then $\init{Q_0}\cap(X\cup\{\tau\}) = \emptyset$ and
    \begin{align*}
        Ax_r \vdash \psi_X(Q_0) = &~ \sum_{\{(\alpha,Q_0') \mid Q_0 \step{\alpha} Q_0' \wedge \alpha \ne \rt\}}\alpha.Q_0' + \sum_{\{(t,Q_0') \mid Q_0 \step{\rt} Q_0'\}}\rt.\theta_X(Q_0') \\
        = &~ \sum_{\{(\alpha,Q_0') \mid Q_0 \step{\alpha} Q_0' \wedge \alpha \ne\rt\}}\alpha.Q_0' + \sum_{(\alpha,P_0') \in I_0}\alpha.P_0' + \sum_{\{(\rt,Q_0') \mid Q_0 \step{\rt} Q_0'\}}\rt.\theta_X(Q_0') \\
        & + \sum_{\{(\rt,P_0') \mid P_0 \step{\rt} P_0'\}}\rt.\theta_X(P_0') = \psi_X(P_0 + Q_0)
    \end{align*}
    As a result, for all $X \subseteq A$, $Ax_r \vdash \psi_X(Q_0) = \psi_X(P_0+Q_0)$, and so,
    $Ax_r \vdash Q_0 = P_0+Q_0$. Symmetrically, $Ax_r \vdash P_0 = P_0+Q_0$. Therefore, $Ax_r \vdash P_0=Q_0$.

\bigskip

    $(\bisimtb)$ We will proceed by induction on $\max(d(P), d(Q))$. Let $n \mathbin\in\nat$ and suppose that the property holds for any recursion-free $\ccsp$ processes $P, Q$ such that $\max(d(P), d(Q)) < n$. Let $P,Q$ be two recursion-free $\ccsp$ processes such that $\max(d(P), d(Q)) = n$ and $P \bisimtb Q$.

    Since $P$ is recursion-free, there exists a path $P \pathtau P_1 \step{\rt} P_2 \pathtau P_3 \step{\rt} ... \pathtau P_{2r-1} \step{\rt} P_{2r} \pathtau P_0$ with $r \in \nat$ such that $\forall i \in [0,2r],\, P \bisimtb P_i$ and, for all $P_0 \step{\tau/\rt} P^\dag$, $\neg(P_0 \bisimtb P^\dag)$. We are going to show that, for all $\alpha \in Act$, $Ax \vdash \alpha.\hat{P} = \alpha.\hat{P_0}$. If $P = P_0$ then it is trivial. Thus, suppose $P_0 \ne P$. Then $d(P_0) < d(P)$ and there exists $P \step{\tau/\rt} P'$ with $P \bisimtb P'$.

    Let $J := \{(\tau,P') \mid P \steptau P' \wedge P \bisimtb P'\}$, listing the outgoing $\tau$-transitions of $P$ that can be elided w.r.t.\ $\bisimtb\,$. Let $(\tau,P') \in J$. Since $P' \bisimtb P \bisimtb P_0$ and $\max(d(P'),d(P_0)) < d(P)$, by induction, $Ax \vdash \tau.\hat{P'} = \tau.\hat{P}_0$.

    Let $K := \{(\rt,P') \mid P \step{\rt} P' \wedge P \bisimtb P'\}$, listing the outgoing time-outs of $P$ that can be elided w.r.t.\ $\bisimtb\,$. Let $(\rt,P') \in J$. Since $P' \bisimtb P \bisimtb P_0$ and $\max(d(P'),d(P_0)) < d(P)$, by induction, $Ax \vdash \rt.\hat{P'} = \rt.\hat{P}_0$.
    
    Let $I := \{(\alpha,P') \mid P \step{\alpha} P'\} \setminus (J \cup K)$, listing the outgoing transitions of $P$ that cannot be elided. Let $(\alpha,P') \in I$. If $\alpha \in A_\tau$ then, since $P \bisimtb P_0$, there exists a path $P_0 \pathtau P_1 \step{\opt{\alpha}} P_2$ such that $P \bisimtb P_1$ and $P' \bisimtb P_2$. Thus, $P_1 \bisimtb P \bisimtb P_0$, but, for all $P_0 \steptau P^\dag$, $P_0 \,\not\!\bisimtb P^\dag$, so $P_0 = P_1$ and $P \step{\opt{\alpha}} P_2$. Since $(\alpha,P') \not\in J$, $\alpha \in A$ or $P_0 \bisimtb P \not\bisimtb P' \bisimtb P_2$ so $P_0 \step{\alpha} P_2$ and $\max(d(P_2),d(P')) < d(P)$. Thus, by induction, $Ax \vdash \alpha.\hat{P'} = \alpha.\hat{P}_2$. If $\alpha = \rt$ then, since $P \bisimtb P_0$, there exists a path $P_0 \pathtau P_1 \step{\rt} P_2 \pathtau P_3 \step{\rt} ... \pathtau P_{2r-1} \step{\opt{\rt}} P_{2r}$ with $r>0$, such that, for all $i \in [0,2r-1]$, $P \bisimtb P_i$ and $P' \bisimtb P_{2r}$. Thus, $P_{2r-1} \bisimtb P \bisimtb P_0$, but, for all $P_0 \step{\tau/t} P^\dag$, $\neg(P_0 \bisimtb P^\dag)$, therefore, $P_0 = P_{2r-1}$ and $P_0 \step{\opt{\rt}} P_{2r}$. Since $(\rt,P') \not\in K$, $P_0 \bisimtb P \not\bisimtb P' \bisimtb P_2$ so $P_0 \step{\rt} P_2$ and $\max(d(P_2),d(P')) < d(P)$. Thus, by induction, $Ax \vdash \rt.\hat{P'} = \rt.\hat{P}_2$. As a result, $Ax \vdash \hat{P}_0 + \sum_{(\alpha,P') \in I}\alpha.\hat{P'} = \hat{P}_0$.

    Since there exists $P \step{\tau/\rt} P'$ with $P \bisimtb P'$, $J \cup K \ne \emptyset$. We are going to perform a case distinction on the emptiness of $J$ and $K$.
    \begin{itemize}
        \item If $J \ne \emptyset$ and $K \ne \emptyset$ then, employing the $\tau/\rt$-branching axiom,
        \begin{align*}
            Ax \vdash \alpha.\hat{P} = &~ \alpha.(\sum_{(\tau,P') \in J}\tau.\hat{P}' + \sum_{(t,P') \in K}\rt.\hat{P}' + \sum_{(\alpha,P') \in I}\alpha.\hat{P}') \\
            = &~ \alpha.(\tau.\hat{P}_0 + \rt.\hat{P}_0 + \sum_{(\alpha,P') \in I}\alpha.\hat{P}') \\
            = &~ \alpha.(\tau.(\hat{P}_0 + \sum_{(\alpha,P') \in I}\alpha.\hat{P}') + \rt.(\hat{P}_0 + \sum_{(\alpha,P') \in I}\alpha.\hat{P}') + \sum_{(\alpha,P') \in I}\alpha.\hat{P}') = \alpha.\hat{P}_0 
        \end{align*}
        \item If $J \ne \emptyset$ and $K = \emptyset$ then, employing the branching axiom,
        \begin{align*}
            Ax \vdash \alpha.\hat{P} = &~ \alpha.(\sum_{(\tau,P') \in J}\tau.\hat{P}' + \sum_{(t,P') \in K}\rt.\hat{P}' + \sum_{(\alpha,P') \in I}\alpha.\hat{P}') = \alpha.(\tau.\hat{P}_0 + \!\sum_{(\alpha,P') \in I}\!\alpha.\hat{P}') \\
            = &~ \alpha.(\tau.(\hat{P}_0 + \sum_{(\alpha,P') \in I}\alpha.\hat{P}') + \sum_{(\alpha,P') \in I}\alpha.\hat{P}') = \alpha.\hat{P}_0 
        \end{align*}
        \item If $J = \emptyset$ and $K \ne \emptyset$ then $\{(\alpha, P') \mid P \step{\alpha} P' \wedge \alpha \in A_\tau\} \subseteq I$. Let $(\alpha,P_2)$ such that $\alpha \in A_\tau$ and $P_0 \step{\alpha} P_2$. Since $P \bisimtb P_0$, there exists a path $P \pathtau P' \step{\opt{\alpha}} P''$ such that $P' \bisimtb P_0$ and $P'' \bisimtb P_2$. Since $J = \emptyset$, $P = P'$. If $\alpha = \tau$ and $P = P''$ then $P_0 \bisimtb P' \bisimtb P_2$ and $P_0 \steptau P_2$, but that contradicts the definition of $P_0$. Thus, $P\step{\alpha} P_2$. Therefore, by induction, for all $(\alpha,P_2) \in \{(\alpha,P_2) \mid P_0 \step{\alpha} P_2 \wedge \alpha \in A_\tau\}$, there exists $(\alpha,P') \in \{(\alpha, P') \mid P \step{\alpha} P' \wedge \alpha \in A_\tau\}$ such that $Ax \vdash \alpha.\hat{P}_2 = \alpha.\hat{P'}$. Symmetrically, for all $(\alpha,P') \in \{(\alpha, P') \mid P \step{\alpha} P' \wedge \alpha \in A_\tau\}$, there exists $(\alpha,P_2) \in \{(\alpha,P_2) \mid P_0 \step{\alpha} P_2 \wedge \alpha \in A_\tau\}$ such that $Ax \vdash \alpha.\hat{P}_2 = \alpha.\hat{P'}$. As a result, $$Ax \vdash \sum_{\{(\alpha, P') \mid P \step{\alpha} P' \wedge \alpha \in A_\tau\}}\alpha.\hat{P}' = \sum_{\{(\alpha,P_2) \mid P_0 \step{\alpha} P_2 \wedge \alpha \in A_\tau\}}\alpha.\hat{P}_2.$$

        Moreover, the reasoning that yields $Ax \vdash \hat{P}_0 + \sum_{(\alpha,P') \in I}\alpha.\hat{P'} = \hat{P}_0$ can be used to get $Ax \vdash \hat{P}_0 + \sum_{(\rt,P') \in I}\rt.\hat{P}' = \hat{P}_0$. Thus, writing $R$ for $\sum_{\{(\alpha, P_2) \mid P_0 \step{\alpha} P_2 \wedge \alpha \in A_\tau\}}\alpha.\hat{P}_2 + \sum_{(\rt,P') \in I}\rt.\hat{P}'$, and using the $\rt$-branching axiom,

\begin{align*}
            Ax \vdash \alpha.\hat{P} = &~ \alpha.(\sum_{(\tau,P') \in J}\tau.\hat{P}' + \sum_{(t,P') \in K}\rt.\hat{P}' + \sum_{(\alpha,P') \in I}\alpha.\hat{P}') \\
            = &~ \alpha.(\rt.\hat{P}_0 + \sum_{\{(\alpha, P') \mid P \step{\alpha} P' \wedge \alpha \in A_\tau\}}\alpha.\hat{P}' +  \sum_{(\rt,P') \in I}\rt.\hat{P}') \\
            = &~ \alpha.(\rt.(\hat{P}_0 + \sum_{(\rt,P') \in I}\rt.\hat{P'}) + \sum_{\{(\alpha, P_2) \mid P_0 \step{\alpha} P_2 \wedge \alpha \in A_\tau\}}\alpha.\hat{P}_2 +  \sum_{(\rt,P') \in I}\rt.\hat{P}') \\
            = &~ \alpha.(\rt.(R+\sum_{P_0 \step\rt P'_0}\rt.P'_0) + R) = \alpha.(R+\sum_{P_0 \step\rt P'_0}\rt.P'_0)\\
            = &~ \alpha.(\hat{P}_0 + \sum_{(\rt,P') \in I}\rt.\hat{P}') = \alpha.\hat{P}_0
        \end{align*}
    \end{itemize}
    As a result, in any case, for all $\alpha \in Act$, $Ax \vdash \alpha.\hat{P} = \alpha.\hat{P}_0$. Similarly, since $Q$ is recursion-free, there exists a recursion-free $\ccsp$ process $Q_0$ such that $Q (\step{\tau/\rt})^* Q_0$, $Q \bisimtb Q_0$ and, for all $Q_0 \step{\tau/t} Q^\dag$, $\neg(Q_0 \bisimtb Q^\dag)$. Moreover, for all $\alpha \in Act$, $Ax \vdash \alpha.\hat{Q} = \alpha.\hat{Q}_0$. Notice that $P_0 \bisimtb P \bisimtb Q \bisimtb Q_0$ and, since, for all $Q_0 \step{\tau/t} Q^\dag$, $\neg(Q_0 \bisimtb Q^\dag)$ and, for all $P_0 \step{\tau/t} P^\dag$, $\neg(P_0 \bisimtb P^\dag)$, $P_0 \bisimrtb Q_0$.

    Let $(\alpha,P_0')$ such that $P_0 \step{\alpha} P_0'$. Since $P_0 \bisimrtb Q_0$, there exists a path $Q_0 \step{\alpha} Q_2$ such that $P_0' \bisimtb Q_2$. Since $\max(d(P_0'),d(Q_2)) < n$, by induction, $Ax \vdash \alpha.P_0' = \alpha.Q_2$. As a result, $Ax \vdash \hat{P}_0 + \hat{Q}_0 = \hat{Q}_0$. Symmetrically, $Ax \vdash \hat{P}_0 + \hat{Q}_0 = \hat{P}_0$, and so, $Ax \vdash \hat{P}_0 = \hat{Q}_0$. Finally, for all $\alpha \in Act$, $Ax \vdash \alpha.\hat{P} = \alpha.\hat{Q}$.
\end{proof}

\begin{proof}[Proof of Theorem \ref{thm:completeness finite}]
    Let $P,Q \in \closed$ be two recursion-free $\ccsp$ processes. Let $P \bisimrtbr Q$. $P$ and $Q$ can be equated in the same manner as \hyperlink{here}{$P_0$ and $Q_0$} in the proof of Proposition \ref{prop:collapse}. 

    Suppose that $P \bisimrtb Q$. Let $(\alpha,P')$ such that $P \step{\alpha} P'$ with $\alpha \in Act$. Since $P \bisimrtb Q$, there exists a transition $Q \step{\alpha} Q'$ such that $P' \bisimtb Q'$. According to the previous proposition, $Ax \vdash \alpha.P' = \alpha.Q'$, thus, 
    \begin{align*}
        Ax \vdash Q = \sum_{\{(\alpha,Q') \mid Q \step{\alpha} Q'\}}\alpha.Q' = \sum_{\{(\alpha,Q') \mid Q \step{\alpha} Q'\}}\alpha.Q' + \sum_{\{(\alpha,P') \mid P \step{\alpha} P'\}}\alpha.P' = Q + P
    \end{align*}
    As a result, $Ax \vdash Q = P+Q$. Symmetrically, $Ax \vdash P = P+Q$. Therefore, $Ax \vdash P=Q$.
\end{proof}

\section{Proof of Completeness by Equation Merging} \label{app:completeness}

\begin{proof}[Proof of Theorem \ref{thm:completeness}]
    Let $E_0$ and $F_0$ two strongly guarded $\ccsp$ processes such that $E_0 \bisimrtb F_0$. We are going to build a recursive specification $\equa$ such that $E_0$ and $F_0$ will be components of solutions of $\equa$ in the same variable. Let $\mathcal{E}_{E_0}$ (resp.\ $\mathcal{E}_{F_0}$) be the set of reachable expressions from $E_0$ (resp.\ $F_0$). Let $V_\equa$ be a set of new variables $\{x_{EF} \mid (E,F) \in \mathcal{E}_{E_0}\times\mathcal{E}_{F_0} \wedge E \bisimtb F\}$. We denote $x_0 = x_{E_0F_0} \in V_\equa$ and we define the following set of equations $\equa$, for all $x_{EF} \in V_\equa$, with $\alpha\in A \cup\{\tau,\rt\}$.
    \begin{align*}
        \equa_{x_{EF}} := & \sum_{E \step{\alpha} E', F \step{\alpha} F', E' \bisimtb F'}\hspace{-8pt}\alpha.x_{E'F'}  + \hspace{-8pt}\sum_{E \steptau E', E' \bisimtb F, x_{EF} \neq x_0}\hspace{-8pt}\tau.x_{E'F} \\
        & + \hspace{-8pt}\sum_{F \steptau F', E \bisimtb F', x_{EF} \neq x_0}\hspace{-8pt}\tau.x_{EF'} + \hspace{-8pt}\sum_{E \step{\rt} E', E' \bisimtb F, x_{EF} \neq x_0}\hspace{-8pt}\rt.x_{E'F} + \hspace{-8pt}\sum_{F \step{\rt} F', E \bisimtb F', x_{EF} \neq x_0}\hspace{-8.5pt}\rt.x_{EF'}
    \end{align*}
    Note that $\equa$ is well-guarded since $E_0$ and $F_0$ are strongly guarded $\ccsp$ processes. For $x_{EF} \in V_\equa$, we define $H_{EF}, G_{EF} \in \expr$ such that 
    \begin{align*}
        H_{EF} := & \sum_{E \step{\alpha} E', F \step{\alpha} F', E' \bisimtb F'}\hspace{-8pt}\alpha.E' + \hspace{-8pt}\sum_{E \steptau E', E' \bisimtb F, x_{EF} \neq x_0}\hspace{-8pt}\tau.E'
          + \hspace{-8pt}\sum_{E \step{\rt} E', E' \bisimtb F, x_{EF} \neq x_0}\hspace{-8pt}\rt.E' \\
        G_{EF} := &  \begin{cases} 
                        H_{EF} + \tau.E + \rt.E & \mbox{if }x_0 \mathbin{\neq} x_{EF} \mbox{, } \exists F \steptau F', E \bisimtb F' \mbox{ and }\exists F \step{\rt} F', E \bisimtb F'  \\ 
                        H_{EF} + \tau.E & \mbox{if }x_0 \mathbin{\neq} x_{EF} \mbox{, } \exists F \steptau F', E \bisimtb F' \mbox{ and } \forall F \step{\rt} F', E \,\not\!\bisimtb F'  \\ 
                        H_{EF} + \rt.E & \mbox{if }x_0 \mathbin{\neq} x_{EF} \mbox{, } \forall F \steptau F', E \,\not\!\bisimtb F'
                        \mbox{ and }\exists F \step{\rt} F', E \bisimtb F'  \\ 
                        E & \mbox{otherwise}
                    \end{cases}
    \end{align*}
    According to Lemma~\ref{lem:head-normal form}, for all $(E,F) \in \mathcal{E}_{E_0}\times\mathcal{E}_{F_0}$, $Ax^\infty \vdash E + H_{EF} = \widehat{(E + H_{EF})} = \hat{E} = E$. Let $(E,F) \in \mathcal{E}_{E_0}\times\mathcal{E}_{F_0}$.
    \begin{itemize}
        \item If $x_0 \ne x_{EF}$ and $\exists F\steptau F', E \bisimtb F'$ and $\exists F \step{\rt} F', E \bisimtb F'$ then, for all $\alpha \in Act$, $Ax^\infty \vdash \alpha.G_{EF} = \alpha.(H_{EF} + \tau.E + \rt.E) = \alpha.E$ using the $\tau/\rt$-branching axiom.
        \item If $x_0 \ne x_{EF}$ and $\exists F\steptau F', E \bisimtb F'$ and $\forall F \step{\rt} F', E \,\not\!\bisimtb F'$ then, for all $\alpha \in Act$, $Ax^\infty \vdash \alpha.G_{EF} = \alpha.(H_{EF} + \tau.E) = \alpha.E$ using the branching axiom.
        \item If $x_0 \ne x_{EF}$ and $\forall F\steptau F', E \,\not\!\bisimtb F'$ and $\exists F \step{\rt} F', E \bisimtb F'$ then, for all $\alpha \in Act$, $Ax^\infty \vdash \alpha.G_{EF} = \alpha.(H_{EF} + \rt.E)$. Let $(\alpha,E')$ such that $E \step{\alpha} E'$ with $\alpha \in A_\tau$. Since $E \bisimtb F$, there exists a path $F \pathtau F_1 \step{\opt{\alpha}} F_2$ such that $E \bisimtb F_1$ and $E' \bisimtb F_2$. Since $\forall F\steptau F', E \,\not\!\bisimtb F'$, $F = F_1$ and $F \step{\opt{\alpha}} F_2$. If $\alpha = \tau$ and $F = F_2$ then $E \steptau E'$ and $E' \bisimtb F$ so $H_{EF} \steptau E'$; else $F \step{\alpha} F_2$ and $E' \bisimtb F'$ so $H_{EF} \step{\alpha} E'$. Therefore, $Ax^\infty \vdash H_{EF} + \sum_{\{(\alpha,E') \mid E \step{\alpha} E' \wedge \alpha \in A_\tau\}}\alpha.E' = H_{EF}$. Consequently, $Ax^\infty \vdash E = H_{EF} + E = H_{EF} + \sum_{\{E'\mid E \step\rt E'\}}\rt.E'$. As a result, the $\rt$-branching axiom can be applied and so, for all $\alpha \in Act$, $Ax^\infty \vdash \alpha.G_{EF} = \alpha.(H_{EF} + \rt.E) = \alpha.E$.
    \end{itemize}
    In any case, for all $\alpha \in Act$, $Ax^\infty \vdash \alpha.G_{EF} = \alpha.E$. \hfill $(*)$
    
    If we prove that the family $(G_{EF})_{(E,F) \in \mathcal{E}_{E_0}\times\mathcal{E}_{F_0}}$ is a solution of $\equa$ then, by definition of $G_{E_0F_0}$, there would exist a solution whose value for the variable $x_0$ is $E$. According to $(*)$, we need to prove that, for all $x_{EF} \in V_\equa$, 
    \begin{align*}
        Ax^\infty \vdash G_{EF} = & \sum_{E \step{\alpha} E', F \step{\alpha} F', E' \bisimtb F'}\hspace{-8pt}\alpha.G_{E'F'} \\ 
        & + \hspace{-8pt}\sum_{E \steptau E', E' \bisimtb F, x_{EF} \neq x_0}\hspace{-8pt}\tau.G_{E'F} + \hspace{-8pt}\sum_{F \steptau F', E \bisimtb F', x_{EF} \neq x_0}\hspace{-8pt}\tau.G_{EF'} \\
        & + \hspace{-8pt}\sum_{E \step{\rt} E', E' \bisimtb F, x_{EF} \neq x_0}\hspace{-8pt}\rt.G_{E'F} + \hspace{-8pt}\sum_{F \step{\rt} F', E \bisimtb F', x_{EF} \neq x_0}\hspace{-8pt}\rt.G_{EF'}\\
        = & \sum_{E \step{\alpha} E', F \step{\alpha} F', E' \bisimtb F'}\hspace{-8pt}\alpha.E' \\ 
        & + \hspace{-8pt}\sum_{E \steptau E', E' \bisimtb F, x_{EF} \neq x_0}\hspace{-8pt}\tau.E' + \hspace{-8pt}\sum_{F \steptau F', E \bisimtb F', x_{EF} \neq x_0}\hspace{-8pt}\tau.E \\
        & + \hspace{-8pt}\sum_{E \step{\rt} E', E' \bisimtb F, x_{EF} \neq x_0}\hspace{-8pt}\rt.E' + \hspace{-8pt}\sum_{F \step{\rt} F', E \bisimtb F', x_{EF} \neq x_0}\hspace{-8pt}\rt.E \\
        = &~ H_{EF} + \hspace{-8pt}\sum_{x_{EF} \neq x_0, F \steptau F', E \bisimtb F'}\hspace{-8pt}\tau.E + \hspace{-8pt}\sum_{F \step{\rt} F', E \bisimtb F', x_{EF} \neq x_0}\hspace{-8pt}\rt.E 
    \end{align*}
    \begin{itemize}
        \item If $x_{EF} \neq x_0$ and $(\exists F \steptau F', E \bisimtb F') \vee (\exists F \step{\rt} F', E \bisimtb F')$ then this follows from the definition of $G_{EF}$.
        \item If $x_{EF} \neq x_0$ and $\forall F\steptau F', E \,\not\!\bisimtb F'$ and $\forall F\step{\rt} F', E \,\not\!\bisimtb F'$ then, by definition of $G_{EF}$, we have to prove $Ax^\infty \vdash E = H_{EF}$. Let $(\alpha,E')$ such that $E \step{\alpha} E'$ and $\alpha \in A_\tau$. Since $E \bisimtb F$, there exists a path $F \pathtau F_1 \step{\opt{\alpha}} F_2$ such that $E \bisimtb F_1$ and $E' \bisimtb F_2$. Since $\forall F\steptau F', \neg(E \bisimtb F')$, $F = F_1$, so there exists a transition $F \step{\opt{\alpha}} F_2$ such that either $F \step{\alpha} F_2$ and $E' \bisimtb F_2$, or $\alpha = \tau$ and $E' \bisimtb F$. In either case, $H_{EF} \step{\alpha} E'$. \\
        Let $(t,E')$ such that $E \step{\rt} E'$. Since $E \bisimtb F$, there exists a path $F \pathtau F_1 \step{\rt} F_2 \pathtau F_3 \step{\rt} ... \pathtau F_{2r-1} \step{\opt{\rt}} F_{2r}$ with $r>0$, such that, for all $i \in [0,2r{-}1]$, $E \bisimtb F_i$ and $E' \bisimtb F_{2r}$. Since $\forall F\steptau F', E \,\not\!\bisimtb F'$ and $\forall F\step{\rt} F', E \,\not\!\bisimtb F'$, $F = F_{2r-1}$, so $F \step{\opt{\rt}} F_{2r}$. Thus either there exists a transition $F \step{\rt} F_{2r}$ such that $E' \bisimtb F_{2r}$ or $E' \bisimtb F$. In either case, $H_{EF} \step{\rt} E'$. \\
        As a result, $Ax^\infty \vdash E = \hat{E} = \widehat{E + H}_{EF} = \hat{H}_{EF} = H_{EF}$.
        \item If $x_{EF} = x_0$ then $E = E_0$, $F = F_0$ and we have to show that $Ax^\infty \vdash E_0 = H_{E_0F_0} = \sum_{E_0 \step{\alpha} E', F_0 \step{\alpha} F', E' \bisimtb F'}\alpha.E'$. Let $(\alpha,E')$ such that $E_0 \step{\alpha} E'$. Since $E_0 \bisimrtb F_0$, there exists a transition $F_0 \step{\alpha} F'$ such that $E' \bisimtb F'$. $Ax^\infty \vdash E = \hat{E} = \hat{H}_{EF} = H_{EF}$.
    \popQED
    \end{itemize}
Note that we could define $H'_{EF}$ and $G'_{EF}$ by reverting the role of $E$ and $F$ and also get a solution whose value for the variable $x_0$ is $F_0$. Consequently, RSP yields $Ax^\infty \vdash E_0 = F_0$.
\end{proof}

\section{The Canonical Representative} \label{app:canonical rep}

We are going to start by proving some lemmas facilitating the handling of classes. 

\newpage

\begin{lemma} \label{lem:transition class}
    Let $P \in \closed^g$.
    \begin{enumerate}
        \item $\forall \alpha \in A_\tau, ([P] \step{\alpha} R' \Leftrightarrow \exists P \pathtau P_1 \step{\alpha} P_2, (P_1,P_2) \in [P]\times R' \wedge (\alpha \in A \vee [P] \ne R'))$.
        \item ${[P] \nsteptau} \Leftrightarrow \exists P_0 \in [P], P \pathtau P_0 \nsteptau$.
        \item Let $X\subseteq A$. Then $\deadend{[P]}{X} \Leftrightarrow \exists P \pathtau P_0, P_0 \in [P] \wedge \deadend{P_0}{X}$.
        \item If $[P] \step{\rt} R'$ and $\deadend{[P]}{X}$ then $\exists r>0,\, \exists P \pathtau P_1 \step{\rt} P_2 \pathtau P_3 \step{\rt} ... \pathtau P_{2r-1} \step{\rt} P_{2r}$, $P_1 \in [P] \wedge \forall i \in [0,2r{-}1],\; \theta_X(P_{i}) \in [\theta_X(P)] \wedge \forall j \in [0,r{-}1],\; \deadend{P_{2j+1}}{X}\linebreak[2] \wedge \theta_X(P_{2r}) \in [\theta_X(\chi(R'))] \wedge [\theta_X(P)] \ne [\theta_X(\chi(R'))]$.
        \item If $\exists X \mathbin\subseteq A, r\mathbin>0,\, \exists P \pathtau P_1 \step{\rt} P_2 \pathtau P_3 \step{\rt} ... \pathtau P_{2r-1} \step{\rt} P_{2r}, P_1 \in [P] \wedge \forall i \in [0,r{-}1]$, $\theta_X(P_{2i}) \in [\theta_X(P)] \wedge \deadend{P_{2i+1}}{X} \wedge \theta_X(P_{2r}) \not\in [\theta_X(P)]$ then there exists an $R'$ with $[P] \step{\rt} R' \wedge \theta_X(P_{2r}) \in [\theta_X(\chi(R'))]$.
    \end{enumerate}
\end{lemma}

\begin{proof}
    Let $P \in \closed^g$.
    \begin{enumerate}
        \item Let $\alpha \in A_\tau$.
        \begin{itemize}
            \item If $[P] \step{\alpha} R'$ then, by definition of $\rightarrow$, there exists a path $\chi([P]) \pathtau P_1 \step{\alpha} P_2$ such that $P_1 \in [P]$, $P_2 \in R'$ and $\alpha \in A \vee [P] \ne R'$. Since $\chi([P]) \bisimtbr P$, there exists a path $P \pathtau P_1' \step{\opt{\alpha}} P_2'$ such that $P_1 \bisimtbr P_1'$ and $P_2 \bisimtbr P_2'$, thus, $P_1' \in [P]$ and $P_2' \in R'$. If $\alpha = \tau$ then $[P] \ne R'$, so $P_1' \,\not\!\bisimtbr P_2'$ and so $P_1' \step{\alpha} P_2$, otherwise, $P_1' \step{\alpha} P_2'$. 
            \item If there exists a path $P \pathtau P_1 \step{\alpha} P_2$ such that $P_1 \in [P]$, $P_2 \in R'$ and $\alpha \in A \vee [P] \ne R'$ then, since $P \bisimtbr \chi([P])$, there exists a path $\chi([P]) \pathtau P_1' \step{\opt{\alpha}} P_2'$ such that $P_1 \bisimtbr P_1'$ and $P_2 \bisimtbr P_2'$, thus, $P'_1 \in [P]$ and $P'_2 \in R'$. If $\alpha = \tau$ then $[P] \ne R'$, therefore, $P_1' \,\not\!\bisimtbr P_2'$ and so $P_1' \step{\alpha} P_2'$, otherwise, $P_1' \step{\alpha} P_2'$. By definition of $\rightarrow$, $[P] \step{\alpha} R'$.
        \end{itemize}
        \item 
        \begin{itemize}
            \item If $[P] \steptau$ then, according to the previous point, there exists a path $P \pathtau P_1 \steptau P_2$ such that $P_1 \in [P]$ and $P_2 \not\in [P]$. Suppose that there is a path $P \pathtau P_0 \nsteptau$ with $P_0 \in [P]$. Then $P_0 \bisimtbr P$, so there exists a path $P_0 \pathtau P^\dag \step{\opt{\tau}} P^\ddag$ such that $P_1 \bisimtbr P^\dag$ and $P_2 \bisimtbr P^\ddag$. Since $P_0 \nsteptau$, $P_0 = P^\ddag \bisimtbr P_2 \not\in [P]$, but that's impossible.
            \item Suppose that, for all paths $P \pathtau P_0 \nsteptau$, $P_0 \not\in [P]$. Since $P$ is strongly guarded, there exists a path $P \pathtau P_1$ such that $P_1 \in [P]$ and, for all $P_1 \steptau P'$, $P_1 \,\not\!\bisimtbr P'$. Since $P \pathtau P_1$ and $P_1 \in [P]$, there exists a transition $P_1 \steptau P_2$, and $P_1 \not\!\bisimtbr P_2$. Thus, there exists a path $P \pathtau P_1 \steptau P_2$ such that $P_1 \in [P]$ and $P_2 \not\in [P]$. According to the previous point, $[P] \steptau [P_2]$.
        \end{itemize}
        \item This a corollary of the two previous points.
        \item Suppose $[P] \step{\rt} R'$ and $\deadend{[P]}{X}$. By definition of $\rightarrow$, there exists $Z \subseteq A$, $r>0$ and a path $\chi([P])=P_0 \pathtau P_1 \step{\rt} P_2 \pathtau P_3 \step{\rt} ... \pathtau P_{2r-1} \step{\rt} P_{2r}$ such that $P_1 \in [P]$, $P_{2r} \in R'$, $[\theta_Z(\chi([P]))] \ne [\theta_Z(\chi(R'))]$, and for all $i \in [0,r{-}1]$, $\theta_Z(P_{2i}) \in [\theta_Z(\chi([P]))]$ and $\deadend{P_{2i+1}}{Z}$. For all $i \in [0,r{-}1]$, considering that $P_1 \bisimtbr \chi([P]) \bisimtbr[Z] P_{2i}$, Lemma~\ref{lem:obvious}.1 gives $P_1 \bisimtbr[Z] P_{2i+1}$, and by Lemma~\ref{lem:obvious}.3  $P_1 \bisimtbr P_{2i+1}$, i.e., $P_{2i+1} \in [P]$. 

        By Statement 3 of this lemma, $\exists P \pathtau P', P' \in [P] \wedge \deadend{P'}{X}$. According to Lemma~\ref{lem:obvious}.2, for $i \in [0,r{-}1]$, since $P' \bisimtbr P_{2i+1}$, one obtains  $\deadend{[P_{2i+1}]}{X}$.

        For all $i \in [1,r{-}1]$, since $P_{2i-1}\step\rt P_{2i}$, $\deadend{P_{2i-1}}{Z}$, $\deadend{[P]}{X}$ and $\theta_Z(P_{2i-1}) \bisimtbr \theta_Z(P_{2i})$, Lemma~\ref{lem:independent} yields $\theta_X(P_{2i-1}) \bisimtbr \theta_X(P_{2i})$, i.e., $\theta_X(P_{2i}) \in\linebreak[4] [\theta_X(\chi([P]))]$. Moreover,  $[\theta_X(P)] \ne [\theta_X(\chi(R'))]$, i.e., $\theta_X(P_{2r-1}) \,\not\!\bisimtbr \theta_X(P_{2r})$, for if we had $\theta_X(P_{2r-1}) \bisimtbr \theta_X(P_{2r})$ then the same reasoning would yield $\theta_Z(P_{2r-1}) \bisimtbr \theta_Z(P_{2r})$, contradicting that $[\theta_Z(\chi([P]))] \ne [\theta_Z(\chi(R'))]$.

        Since $\chi([P]) \bisimtbr P$, $[\theta_X(\chi([P]))] = [\theta_X(P)]$ and there exists a path $P=P'_0 \pathtau P_1' \step{\rt} P_2' \pathtau P_3' \step{\rt} ... \pathtau P_{2k-1}' \step{\opt{\rt}} P_{2k}'$ with $k>0$, such that $P_1 \bisimtbr P_1'$, $\forall j \in [0,k{-}1], \exists i \in [0,2r{-}1]$, $\theta_X(P_i) \bisimtbr \theta_X(P_{2j}') \wedge \deadend{P'_{2j+1}}{X}$ and $\theta_X(P_{2r}) \bisimtbr \theta_X(P_{2k}')$. With Lemma~\ref{lem:obvious}.1 we even have $\theta_X(P) \bisimtbr\theta_X(\chi([P])) \bisimtbr \theta_X(P_{i}) \bisimtbr \theta_X(P_{2j}') \bisimtbr \theta_X(P_{2j+1}')$. As a result, $P'_1 \in [P] \wedge \forall i \in [0,2k{-}1],\; \theta_X(P'_{i}) \in [\theta_X(P)] \wedge\linebreak[3] \forall j \in [0,k{-}1],\; \deadend{P'_{2j+1}}{X} \wedge \theta_X(P'_{2r}) \in [\theta_X(\chi(R'))] \wedge [\theta_X(P)] \ne [\theta_X(\chi(R'))]$. Since $[\theta_X(P)] \ne [\theta_X(\chi(R'))]$, $P_{2k-1}' \step{\rt} P_{2k}'$.
         \item Suppose there exists $X \subseteq A$ and a path $P \pathtau P_1 \step{\rt} P_2 \pathtau P_3 \step{\rt} ... \pathtau P_{2r-1} \step{\rt} P_{2r}$ such that $P_1 \in [P]$, $\forall i \in [0,r{-}1]$, $\theta_X(P_{2i}) \in [\theta_X(P)] \wedge \deadend{P_{2i+1}}{X}$ and $\theta_X(P_{2r}) \not\in [\theta_X(P)]$. Since $P \bisimtbr \chi([P])$, $[\theta_X(P)] = [\theta_X(\chi([P]))]$ and there exists a path $\chi([P]) \pathtau P_1' \step{\rt} P_2' \pathtau P_3' \step{\rt} ... \pathtau P_{2k-1}' \step{\opt{\rt}} P'_{2k}$ such that $P_1 \bisimtbr P_1'$, $\forall j \in [0,k{-}1],\linebreak[3] \exists i \in [0,2r{-}1]$, $\theta_X(P_i) \bisimtbr \theta_X(P'_{2j}) \wedge \deadend{P'_{2j+1}}{X}$ and $\theta_X(P_{2r}) \bisimtbr \theta_X(P'_{2k})$.\linebreak[4] With Lemma~\ref{lem:obvious}.1 we even have $\theta_X(\chi([P])) \bisimtbr \theta_X(P)) \bisimtbr \theta_X(P_{i}) \bisimtbr \theta_X(P_{2j}') \bisimtbr \theta_X(P_{2j+1}')$. As a a result, $P_1' \in [P]$, $\forall j \in [0,k{-}1], \theta_X(P'_{2j}) \in [\theta_X(\chi([P]))]$. Notice that $\theta_X(P_{2k}') \in [\theta_X(\chi([P'_{2k}]))]$ and so $[\theta_X(\chi([P]))] \ne [\theta_X(\chi([P'_{2k}]))]$ since $\theta_X(P_{2r}) \mathbin{\not\in} [\theta_X(\chi([P]))]$ but $\theta_X(P_{2r}) \in [\theta_X(\chi([P'_{2k}]))]$. Since $\theta_X(P_{2k}') \not\in [\theta_X(\chi([P]))]$, $P_{2k-1}' \step{\rt} P'_{2k}$. Therefore, by definition of $\rightarrow$, $[P] \step{\rt} [P_{2k}']$ and $\theta_X(P_{2r}) \in [\theta_X(\chi([P'_{2k}]))]$.
\popQED
    \end{enumerate}
\end{proof}

\begin{corollary}  \label{cor:transition class}
    Let $P \in \closed^g$ and $X \subseteq A$.
    \begin{enumerate}
        \item If $[P] \pathtau R'$ and $\init{R'}\cap (X \cup \{\tau\})=\emptyset$ then $\exists P \pathtau P'\mathbin\in R'$ with $\init{P'}\cap (X \cup \{\tau\})=\emptyset$.
        \item If $\theta_X(P) \in [\theta_X(\chi(R))]$, $R \pathtau R'$ and $\init{R'}\cap (X \cup \{\tau\})=\emptyset$ then $\exists P \pathtau P' \in R'$ with $\init{P'}\cap (X \cup \{\tau\})=\emptyset$.
    \end{enumerate}
\end{corollary}

\begin{proof}
    The first statement follows directly from Lemma~\ref{lem:transition class}.1--3. For the second, suppose $\theta_X(P) \in [\theta_X(\chi(R))]$, $R \pathtau R'$ and $\init{R'}\cap (X \cup \{\tau\})=\emptyset$. By the first statement, $\chi(R)\pathtau Q'\in R'$ for some $Q'$ with $\init{Q'}\cap (X \cup \{\tau\})=\emptyset$.  By the semantics of $\theta_X$, there is a path $\theta_X(\chi(R)) \pathtau \theta_X(Q')\nsteptau$. Since $\theta_X(P) \bisimtbr \theta_X(\chi(R))$, there is a path $\theta_X(P) \pathtau P^\dag \nsteptau$ with $P^\dag \bisimtbr \theta_X(Q')$. By the semantics, $P^\dag = \theta_X(P')$ for some $P'$ with $P\pathtau P'\nsteptau$. So $P' \bisimtbr[X] Q'$ by Proposition~\ref{prop:time-out bisim}.2, and Lemma~\ref{lem:obvious}.3 yields $P' \bisimtbr Q'$. Thus $P'\in R'$ and Lemma~\ref{lem:obvious}.2 gives $\init{P'}\cap (X \cup \{\tau\})=\emptyset$.
\end{proof}

\begin{remark}\label{tortau}
    Let $R\in[\closed^g]$. If $R \step\rt$ then $R\nsteptau$.
\end{remark}

\begin{proof}
    Suppose $R \step\rt R'$. By the definition in Section~\ref{subsec:canonical}, there is a path $\chi(R)\pathtau P_1 \nsteptau$ with $P_1\in R$. So by Lemma~\ref{lem:transition class}.2 $R\nsteptau$.
\end{proof}

\begin{definition}\rm\label{def:uptoRT}
    A \emph{concrete branching time-out bisimulation up to reflexivity and transitivity} is a symmetric relation ${\B}$ on $\closed^g \uplus [\closed^g] \uplus \{\theta_X([P]) \mathbin| X \mathbin\subseteq A \wedge P\mathbin\in \closed^g\}$, such that, for all $P^\dag\!\B Q$,
    \begin{itemize}\itemsep 0pt \parsep 0pt
        \item if $P^\dag \step{\alpha} P^\ddagger$ with $\alpha\mathbin\in A_\tau$, then $\exists$ path $Q\pathtau Q^\dag \step{\opt{\alpha}} Q^\ddagger$ with $P^\dag \B^* Q^\dag$ and $P^\ddagger \B^* Q^\dagger$,
        \item if $P^\dag \step\rt P^\ddagger$ with $\init{P^\dag}\cap (X\cup\{\tau\})=\emptyset$, then there is a path $Q\pathtau Q^\dag \step{\rt} Q^\ddagger$ with $\init{Q^\dag}\cap (X\cup\{\tau\})=\emptyset$ and $\theta_{X}(P^\ddagger) \B^* \theta_X(Q^\ddagger)$,
        \item if $P^\dag \nsteptau$ then there is a path $Q\pathtau Q^\dag \nsteptau$.
    \end{itemize}
    Here $\B^*  := \{(P^\dag,Q^\dag) \mid \exists n\geq 0. ~ \exists P_0,\dots, P_n.~ P^\dag = P_0 \B P_1 \B \dots \B P_n = Q^\dag\}$.
\end{definition}

\begin{proposition}\label{prop:uptoRT}
    If $P \B Q$ for a concrete branching time-out bisimulation $\B$ up to reflexivity and transitivity, then $P \bisimtbr Q$.
\end{proposition}

\begin{proof}
    It suffices to show that $\B^*$ is a branching time-out bisimulation. Clearly this relation is symmetric.
    \begin{itemize}
        \item Suppose $P_0 \B P_1 \B \dots \B P_n$ for some $n\geq 0$ and $P \pathtau P^\dag_0 \step{\opt{\alpha}} P^\ddagger_0$ with $\alpha\in A_\tau$. It suffices to find $P^\dag_n,P^\ddagger_n$ such that $P_n \pathtau P^\dag_n \step{\opt{\alpha}} P^\ddag_n$, $P^\dag_0 \B^* P^\dag_n$ and $P^\ddag_0 \B^* P^\ddag_n$. (In fact, we need this only in the special case where $P_0 = P^\dag_0 \neq P^\ddag_0$, but establish the more general claim.) We proceed with induction on $n$. The case $n=0$ is trivial.

        Fixing an $n>0$, by Definition~\ref{def:uptoRT} there are $P^\dag_1,P^\ddag_1$ such that $P_1 \pathtau P^\dag_1 \step{\opt{\alpha}} P^\ddag_1$, $P^\dag_0 \B^* P^\dag_1$ and $P^\ddag_0 \B^* P^\ddag_1$. Now by induction there are $P^\dag_n,P^\ddagger_n$ such that $P_n \pathtau P^\dag_n \step{\opt{\alpha}} P^\ddag_n$, $P^\dag_1 \B^* P^\dag_n$ and $P^\ddag_1 \B^* P^\ddag_n$. Hence $P^\dag_0 \B^* P^\dag_n$ and $P^\ddag_0 \B^* P^\ddag_n$.
        
        \item Suppose $P_0 \B P_1 \B \dots \B P_n$ for some $n\geq 0$ and there is a path $P_0 \pathtau P^\dag_0 \step\rt P^\ddagger_0$ with $\init{P^\dag_0}\cap(X\cup\{\tau\})=\emptyset$. It suffices to find $P^\dag_n,P^\ddagger_n$ such that $P_n \pathtau P^\dag_n \step\rt P^\ddagger_n$, $\init{P^\dag_n}\cap(X\cup\{\tau\})=\emptyset$ and $\theta_X(P^\ddagger_0) \B^* \theta_X(P^\ddagger_n)$. (In fact, we need this only in the special case where $P^\dag_0 = P_0$, but establish the more general claim.) We proceed with induction on $n$. The case $n=0$ is trivial.

        Fixing an $n>0$, by Definition~\ref{def:uptoRT} there exist $P^\dag_1,P^\ddagger_1$ such that $P_1 \pathtau P^\dag_1 \step\rt P^\ddagger_1$, $\init{P^\dag_1}\cap(X\cup\{\tau\})=\emptyset$ and $\theta_X(P^\ddagger_0) \B^* \theta_X(P^\ddagger_1)$. By induction there are $P^\dag_n,P^\ddagger_n$ with $P_n \mathbin{\pathtau} P^\dag_n \mathbin{\step\rt} P^\ddagger_n$, $\init{P^\dag_n}\!\cap\!(X\!\cup\!\{\tau\})\mathbin=\emptyset$ and $\theta_X(P^\ddagger_1) \B^* \theta_X(P^\ddagger_n)$. Hence $\theta_X(P^\ddagger_0) \B^* \theta_X(P^\ddagger_n)$.
    
        \item Suppose $P_0 \B P_1 \B \dots \B P_n$ for some $n\geq 0$ and there is a path $P_0 \pathtau P^\dag_0\nsteptau$. It suffices to find a path $P_n \pathtau P^\dag_n\nsteptau$. (In fact, we need this only in the special case where $P^\dag_0 = P_0$, but establish the more general claim.) We proceed with induction on $n$. The case $n=0$ is trivial.

        Fixing an $n>0$, by Definition~\ref{def:uptoRT} there exists a path $P_1 \pathtau P^\dag_1\nsteptau$. By induction, there exists a path $P_n \pathtau P^\dag_n\nsteptau$.
        \popQED          
    \end{itemize}
\end{proof}

\begin{lemma} \label{lem:class deadlock theta}
    Let $P \in \closed^g$ and $X \subseteq A$. If $\deadend{[P]}{X}$ then $[P] \bisim \theta_X([P])$.
\end{lemma}

\begin{proof}
    It suffices to see that $\tbisim := \{([P],\theta_X([P])),(\theta_X([P]),[P])\} \cup Id$ is a strong bisimulation thanks to the semantics of $\theta_X$.
\end{proof}

\begin{lemma} \label{lem:theta class}
    Let $P \in \closed^g$ and $X \subseteq A$. Then $\theta_X([P]) \bisimtbr[] [\theta_X(P)]$.
\end{lemma}

\begin{proof}
    We will show that ${\tbisim} := {{\bisim} \cup \!\{(\theta_X([P]),[\theta_X(\!P)]), ([\theta_X(\!P)],\theta_X([P])) \!\mid\! P \mathbin\in \closed^g\! \wedge X \mathbin\subseteq A\}}$ is a concrete branching time-out bisimulation up to reflexivity and transitivity.
     \begin{itemize}
        \item If $\theta_X([P]) \steptau R'$ then $[P] \steptau R^\dag$ with $R' = \theta_X(R^\dag)$. According to Lemma \ref{lem:transition class}.1, there exists a path $P \pathtau P_1 \steptau P_2$ such that $P_1 \in [P]$ and $P_2 \in R^\dag = [P_2]$ and $[P] \ne R^\dag$. Thus, there exists a path $\theta_X(P) \pathtau \theta_X(P_1) \steptau \theta_X(P_2)$ such that $\theta_X(P_1) \in [\theta_X(P)]$. If $\theta_X(P_1) \bisimtbr \theta_X(P_2)$ then $[\theta_X(P)] = [\theta_X(P_2)]$. Otherwise, $[\theta_X(P)] \ne [\theta_X(P_2)]$, thus, according to Lemma \ref{lem:transition class}.1, $[\theta_X(P)] \steptau [\theta_X(P_2)]$. In either case, there exists a transition $[\theta_X(P)] \step{\opt{\tau}} [\theta_X(P_2)]$ such that, by definition of $\tbisim$, $[\theta_X(P_2)] \tbisim \theta_X([P_2]) = R'$.
        \item If $[\theta_X(P)] \steptau R'$ then, according to Lemma \ref{lem:transition class}.1, there exists a path $\theta_X(P) \pathtau P_1 \steptau P_2$ such that $P_1 \in [\theta_X(P)]$, $P_2 \in R'$ and $[\theta_X(P)] \ne R'$. Thus, there exists a path $P \pathtau P^\dag \steptau P^\ddag$ such that $P_1 = \theta_X(P^\dag)$ and $P_2 = \theta_X(P^\ddag)$. Notice that, since $[P_1] \ne [P_2]$, $[P^\dag] \ne [P^\ddag]$. According to Lemma \ref{lem:transition class}.1, there exists a path $[P] \pathtau [P^\dag] \steptau [P^\ddag]$. Thus, $\theta_X([P]) \pathtau \theta_X([P^\dag]) \steptau \theta_X([P^\ddag])$. Moreover, by definition of $\tbisim$, $\theta_X([P^\dag]) \tbisim [\theta_X(P^\dag)] = [\theta_X(P)]$ and $\theta_X([P^\ddag]) \tbisim [\theta_X(P^\ddag)] = R'$.
        \item If $\theta_X([P]) \step{a} R'$ with $a \in A$ then $[P] \step{a} R'$ and $a \in X \vee \deadend{[P]}{X}$. If $\deadend{[P]}{X}$, according to Lemma \ref{lem:transition class}.3, there exists a path $P \pathtau P_0$ such that $P_0 \in [P]$ and $\deadend{P_0}{X}$. Otherwise, set $P_0 := P$. Since $[P_0] \step{a} R'$, according to Lemma \ref{lem:transition class}.1, there exists a path $P_0 \pathtau P_1 \step{a} P_2$ such that $P_1 \in [P_0]$ and $P_2 \in R'$. Notice that $\deadend{[P]}{X} \Rightarrow \deadend{P_0}{X} \wedge P_0 = P_1$. Thus, there exists a path $\theta_X(P_0) \pathtau \theta_X(P_1) \step{a} P_2$ such that $\theta_X(P_1) \in [\theta_X(P)]$. According to Lemma \ref{lem:transition class}.1, there exists a transition $[\theta_X(P)] \step{a} R'$.
        \item If $[\theta_X(P)] \step{a} R'$ with $a \in A$ then, according to Lemma \ref{lem:transition class}.1, there exists a path $\theta_X(P) \pathtau P_1 \step{a} P_2$ such that $P_1 \in [\theta_X(P)]$ and $P_2 \in R'$. Thus, there exists a path $P \pathtau P^\dag \step{a} P_2$ such that $P_1 = \theta_X(P^\dag)$ and $a \in X \vee \deadend{P^\dag}{X}$. According to Lemma \ref{lem:transition class}.1, there exists a path $[P] \pathtau [P^\dag] \step{a} [P_2]$. Thus, $\theta_X([P]) \pathtau \theta_X([P^\dag]) \step{a} [P_2]$ since $\deadend{P^\dag}{X} \Rightarrow \deadend{[P^\dag]}{X}$ by Lemma \ref{lem:transition class}.3. Moreover, by definition of $\tbisim$, $\theta_X([P^\dag]) \tbisim [\theta_X(P^\dag)] = [\theta_X(P)]$ and $[P_2] = R' \tbisim^* R'$.
        \item If $\deadend{\theta_X([P])}{Y}$ and $\theta_X([P]) \step{\rt} R'$ then $\deadend{[P]}{X}$; thus, according to Lemma \ref{lem:transition class}.3, there exists a path $P \pathtau P_0$ such that $P_0 \in [P]$ and $\deadend{P_0}{X}$. Since $\deadend{P_0}{X}$, $P \bisimtbr P_0 \bisim \theta_X(P_0) \bisimtbr \theta_X(P)$ and so $[P] = [\theta_X(P)]$. Since $\deadend{[P]}{X}$, according to Lemma \ref{lem:class deadlock theta}, $\theta_X([P]) \bisim [P] = [\theta_X(P)]$.
         \item If $\deadend{[\theta_X(P)]}{Y}$ and $[\theta_X(P)] \step{\rt} R'$ then, according to Lemma \ref{lem:transition class}.4, there exists $r>0$ and a path $\theta_X(P) \pathtau \theta_X(P_1) \step{\rt} P_2 \pathtau P_3 \step{\rt} ... \pathtau P_{2r-1} \step{\rt} P_{2r}$ such that $\theta_X(P_1) \in [\theta_X(P)]$, $\deadend{\theta_X(P_1)}{Y}$ and $\forall i \in [1,r{-}1]$, $\theta_Y(P_{2i}) \in [\theta_Y(\theta_X(P))] \wedge \deadend{P_{2i+1}}{Y}$, $\theta_Y(P_{2r}) \in [\theta_Y(\chi(R'))]$ and $[\theta_Y(\theta_X(P))] \ne [\theta_Y(\chi(R'))]$. Since $\theta_X(P_1) \step{\rt} P_2$, $P_1 \step{\rt} P_2$ and $\deadend{P_1}{X}$, thus, $\theta_X(P) \bisimtbr \theta_X(P_1) \bisim P_1$. Therefore, $P \pathtau P_1$ and there exists a path $P_1 \step{\rt} P_2 \pathtau P_3 \step{\rt} ... \pathtau P_{2r-1} \step{\rt} P_{2r}$ such that $P_1 \in [P_1]$, $\forall i \in [1,r{-}1]$, $\theta_Y(P_{2i}) \in [\theta_Y(P_1)] \wedge \deadend{P_{2i+1}}{Y}$ and $[\theta_Y(P_1)] \ne [\theta_Y(\chi(R'))] = [\theta_Y(P_{2r})]$. According to Lemma \ref{lem:transition class}.1 and~\ref{lem:transition class}.5, there exists a path $[P] \pathtau [P_1] \step{\rt} R''$ for some $R''\in[\closed^g]$ with $\theta_Y(P_{2r}) \in [\theta_Y(\chi(R''))]$. Thus, $\theta_X([P]) \pathtau \theta_X([P_1]) \step{\rt} R'$ since $\deadend{[P_1]}{X}$. Since $P_1 \bisim \theta_X(P_1)$ and $\deadend{\theta_X(P_1)}{Y}$, Lemma~\ref{lem:transition class}.3 yields $\deadend{[P_1]}{Y}$, and hence $\deadend{\theta_X([P_1])}{Y}$. Moreover, by definition of $\tbisim$, $\theta_Y(R') \tbisim [\theta_Y(\chi(R'))] = [\theta_Y(P_{2r})] =  [\theta_Y(\chi(R''))] \tbisim \theta_Y(R'')$.
        \item If $\theta_X([P]) \nsteptau$ then $[P] \nsteptau$. According to Lemma \ref{lem:transition class}.2, there exists a path $P \pathtau P_0 \nsteptau$ such that $P_0 \in [P]$. Thus, there exists a path $\theta_X(P) \pathtau \theta_X(P_0) \nsteptau$ such that $\theta_X(P_0) \in [\theta_X(P)]$. According to Lemma \ref{lem:transition class}.2, $[\theta_X(P)] \nsteptau$.
        \item If $[\theta_X(P)] \nsteptau$ then, according to Lemma \ref{lem:transition class}.2, there exists a path $\theta_X(P) \pathtau P_0 \nsteptau$ such that $P_0 \in [\theta_X(P)]$. Thus, there exists a path $P \pathtau P^\dag \nsteptau$ such that $P_0 = \theta_X(P^\dag)$. According to Lemma \ref{lem:transition class}.1--2, there exists a path $[P] \pathtau [P^\dag]\nsteptau$. Thus, $\theta_X([P]) \pathtau \theta_X([P^\dag]) \nsteptau$. 
\popQED          
    \end{itemize}
\end{proof}

\begin{proof}[Proof of Proposition \ref{prop:class}]
    We are going to show that $\tbisim := \{(P,[P]),([P],P) \mid P \in \closed^g\}$ is a branching time-out bisimulation up to $\bisimtbr$ (see Definition~\ref{def:up to b}).
    \begin{enumerate}
        \item 
        \begin{itemize}
            \item Let $P \pathtau P' \step{\alpha} P''$ with $\alpha \in A_\tau$ and $P \bisimtbr P'$. If $\alpha \in A \vee P \,\not\!\bisimtbr P''$ then, according to Lemma \ref{lem:transition class}.1, $[P] \step{\alpha} [P'']$ and, by definition of $\tbisim$, $P' \tbisim [P'] = [P]$ and $P'' \tbisim [P'']$. Otherwise, $\alpha = \tau \wedge P \bisimtbr P''$ thus, by definition of $\tbisim$, $P'' \tbisim [P''] = [P]$. In either case, there exists a path $[P] \step{\opt{\alpha}} [P'']$ such that $P' \tbisim [P]$ and $[P''] \tbisim P''$.
            \item If $[P] \pathtau R' \step{\alpha} R''$ with $\alpha \in A_\tau$ and $[P] \bisimtbr R'$ then, according to Lemma \ref{lem:transition class}.1, $P \pathtau P_1 \step{\alpha} P_2$ such that $P_1 \in R'$ and $P_2 \in R''$. Thus, by definition of $\tbisim$, $P_1 \tbisim [P_1] = R'$ and $P_2 \tbisim [P_2] = R''$. 
        \end{itemize}
        \item
        \begin{itemize}
            \item Let $P \pathtau P_1 \step{\rt} P_2 \pathtau P_3 \step{\rt} ... \pathtau P_{2r-1} \step{\opt{\rt}} P_{2r}$ with $r>0$, and $\forall i \mathbin\in [0,r{-}1]$, $\theta_X(P) \bisimtbr \theta_X(P_{2i}) \wedge P \bisimtbr P_{2i+1} \wedge \deadend{P_{2i+1}}{X}$. If $\theta_X(P) \bisimtbr \theta_X(P_{2r})$ then $\forall i \mathbin\in [0,r],\; \theta_X(P_{2i}) \tbisim [\theta_X(P_{2i})] = [\theta_X(P)] \bisimtbr \theta_X([P])$, in the last step applying Lemma~\ref{lem:theta class}. Otherwise, $P_{2r-1}\mathbin{\neq} P_{2r}$ and by Lemma \ref{lem:transition class}.5 there exists a transition $[P] \mathbin{\step{\rt}} R'$ such that $\theta_X(P_{2r}) \bisimtbr \theta_X(\chi(R')) \tbisim [\theta_X(\chi(R'))] \bisimtbr \theta_X(R')$ and $\forall i \mathbin\in [0,r{-}1]$,\linebreak $\theta_X(P_{2i}) \tbisim [\theta_X(P_{2i})] = [\theta_X(P)] \bisimtbr \theta_X([P])$. In either case, there exists a transition $[P] \step{\opt{\rt}} R'$ such that $\theta_X(P) \bisimtbr \tbisim \bisimtbr \theta_X([P])$ and $\theta_X(P_{2r}) \upto[\bisimtbr] \theta_X(R')$. Moreover, by Lemma~\ref{lem:transition class}.3, $\deadend{[P]}{X}$ since $\deadend{P_1}{X}$ and $P_1 \in [P]$.
            \item Let $[P] = R_0 \pathtau R_1 \step{\rt} R_2 \pathtau R_3 \step{\rt} ... \pathtau R_{2r-1} \step{\opt{\rt}} R_{2r}$ with $r > 0$, such that $\forall i \in [0,r{-}1]$, $\theta_X([P]) \bisimtbr \theta_X(R_{2i}) \wedge [P] \bisimtbr R_{2i+1} \wedge \deadend{R_{2i+1}}{X}$.\linebreak[4] Then, according to Lemma \ref{lem:transition class}.4 and Corollary~\ref{cor:transition class}, there exists a path $P = P_0 \pathtau P_1 \step{\rt} P_2 \pathtau P_3 \step{\rt} ... \pathtau P_{2k-1} \step{\opt{\rt}} P_{2k}$ with $k>0$, such that $\forall j \in [0,k{-}1]$, $\deadend{P_{2j+1}}{X}$ and $\forall j \in [0,2k{-}1]$, $\exists i \in [0,2r{-}1]$, $\theta_X(P_{j}) \in [\theta_X(\chi(R_{i}))]$ and $\theta_X(P_{2k}) \in [\theta_X(\chi(R_{2r}))]$. Thus, applying Lemma~\ref{lem:theta class}, $\forall j \in [0,2k{-}1], \exists i \in [0,2r{-}1]$, $\theta_X(P_{j}) \tbisim [\theta_X(P_{j})] = [\theta_X(\chi(R_i))] \bisimtbr \theta_X(R_i) \bisimtbr \theta_X([P])$ and $\theta_X(P_{2k}) \tbisim [\theta_X(P_{2k})] = [\theta_X(\chi(R_{2r}))] \bisimtbr \theta_X(R_{2r})$.
        \end{itemize}
        \item
        \begin{itemize}
            \item If $P \pathtau P_0 \nsteptau$ with $P \bisimtbr P_0$ then, according to Lemma \ref{lem:transition class}.2, $[P] \nsteptau$.
            \item If $[P] \pathtau R' \nsteptau$ with $[P] \bisimtbr R'$ then, according to Lemma \ref{lem:transition class}.1--2, there exists a path $P \pathtau P' \pathtau P_0 \nsteptau$ such that $P',P_0 \in R'$.
            \popQED
        \end{itemize}
    \end{enumerate}
\end{proof}

\section{Completeness Proof by Canonical Representatives} \label{app:canonical}

\begin{lemma} \label{lem:simplication}
    Let $P,Q \in \closed^g$.
    \begin{itemize}
        \item $[P] \bisimtbr[] [Q] \Rightarrow [P] = [Q]$.
        \item $\theta_X([P]) \bisimtbr \theta_X([Q]) \Rightarrow \theta_X([P]) \bisimtb \theta_X([Q])$.
    \end{itemize}
\end{lemma}

\begin{proof}
    ~
    \begin{itemize}
        \item If $[P] \bisimtbr[] [Q]$ then, by Proposition~\ref{prop:class}, $P \bisimtbr[] [P] \bisimtbr[] [Q] \bisimtbr Q$. Thus, $[P] = [Q]$.
        \item We are going to show that ${\tbisim} := {\it Id} \cup \{(\theta_X([P]),\theta_X([Q])) \mid \theta_X([P]) \bisimtbr \theta_X([Q])\}$ is a $\rt$-branching bisimulation. Suppose $\theta_X([P]) \bisimtbr \theta_X([Q])$. The first and third clause of Definition~\ref{def:non-reactive} are trivially satisfied, because Definition~\ref{def:time-out bisim} features the same clauses. Towards the second clause, suppose $\theta_X([P]) \step\rt R'$. Then $\deadend{[P]}{X}$ and $[P]\step\rt R'$. As $[P] \bisimtbr[X] [Q]$, by Clause 2.c of Definition~\ref{def:intuitive} there is a path $[Q] \pathtau R$ for some $R\in[\closed^g]$ with $[P] \bisimtbr R$. By the previous statement of this lemma, $R=[P]$. Thus $\theta_X([Q])\pathtau \theta_X([P]) \step\rt R'$, which suffices to satisfy the second clause of  Definition~\ref{def:non-reactive}.
\popQED
    \end{itemize}
\end{proof}

\begin{proof}[Proof of Proposition \ref{prop:canonical}]
    Let $\equa'$ be a recursive specification such that $V_{\equa'} := \{y_{P'} \mid P' \in \reach{P}\}$ and, for all $P' \in \reach{P}$, $\equa_{y_{P'}} := \sum_{\{(\alpha,P'') \mid P' \step{\alpha} P''\}}\alpha.y_{P''}$. Note that $\equa'$ is strongly guarded since $P$ is. We are going to show that $P$ and $\langle x_P|\equa\rangle$ are both $y_P$-components of solutions of $\equa'$, so that the proposition follows by RSP.

    First of all, consider $\rho: V_{\equa'} \rightarrow \closed$ such that $\forall P' \in \reach{P},\; \rho(y_{P'}) := P'$. For all $P' \in \reach{P}$, $Ax^\infty_r \vdash \rho(y_{P'}) = P' = \sum_{\{(\alpha,P'') \mid P' \step{\alpha} P''\}}\alpha.P'' = \sum_{\{(\alpha,P'') \mid P' \step{\alpha} P''\}}\alpha.\rho(y_{P''})$ is a direct application of Lemma \ref{lem:head-normal form}. Thus, for all $P' \in \reach{P}$, $Ax^\infty_r \vdash \rho(y_{P'}) = \equa'_{y_{P'}}[\rho]$, i.e., $\rho$ is a solution of $\equa'$ up to $\bisimtbr$\,, and $\rho(y_P) = P$.

    Next, consider $\nu: V_{\equa'} \rightarrow \closed$ such that, for all $P' \in \reach{P}$, 
    \begin{align*}
        \nu(y_{P'}) := \sum_{\{(\alpha,P'') \mid P' \step{\alpha} P''\}}\alpha.\langle x_{[P'']}\mid\equa\rangle
    \end{align*}
    We are going to show that, for all $\alpha \in Act$ and all $P' \in \reach{P}$, $Ax^\infty_r \vdash \alpha.\nu(y_{P'}) = \alpha.\langle x_{[P']}\mid\equa\rangle$. Let $P' \in \reach{P}$.
    \begin{itemize}
        \item If $\exists P' \steptau P'',\; P' \bisimtbr P''$ then $\{(\alpha,[P'']) \mid P' \step{\alpha} P'' \wedge (\alpha \in A \vee (\alpha=\tau \wedge P' \,\not\!\bisimtbr P''))\} \subseteq \{(\alpha,R) \mid [P'] \step{\alpha} R\}$. Thus,
        \begin{align*}
            Ax^\infty_r \vdash \alpha.\nu(y_{P'}) & = \alpha.(\sum_{\{(\alpha,P'') \mid P' \step{\alpha} P'' \wedge \alpha \ne \rt\}}\alpha.\langle x_{[P'']}\mid\equa\rangle) \qquad\qquad(\hyperlink{Lt}{\mbox{\bf L}\tau})\\
            & = \alpha.(\sum_{\{(\alpha,P'') \mid P' \step{\alpha} P'' \wedge (\alpha \in A \vee (\alpha=\tau \wedge P \,\,\not\!\scriptrbis{}{\!br}\, P'))\}}\hspace{-50pt}\alpha.\langle x_{[P'']}\mid\equa\rangle + \tau.\langle x_{[P']}\mid\equa\rangle) \\
            & = \alpha.\langle x_{[P']}\mid\equa\rangle \qquad\qquad\qquad\qquad\mbox{(branching axiom and RDP)}
        \end{align*}
        \item If there exists $X \subseteq A$ and a transition $P' \step{\rt} P''$ such that $\deadend{P'}{X}$ and $\theta_X(P') \bisimtbr \theta_X(P'')$ then, since $P' \nsteptau$, for all $\alpha \in A_\tau$, $P' \step{\alpha} P'' \wedge P'' \in R \iff [P'] \step{\alpha} R$, using Lemma~\ref{lem:transition class}.1. Moreover, according to Lemma \ref{lem:independent}, $\{(\rt,P'') \mid P' \step{\rt} P''\} = \mbox{}$
         \[\begin{array}{c@{~}l}&\{(\rt,P'') \mid P' \step{\rt} P'' \wedge \forall Y \subseteq A,\, \deadend{P'}{Y} \Rightarrow \theta_Y(P') \bisimtbr \theta_Y(P'')\} \\ \uplus  & \{(\rt,P'') \mid P' \step{\rt} P'' \wedge \forall Y \subseteq A,\, \deadend{P'}{Y} \Rightarrow \theta_Y(P') \,\not\!\bisimtbr \theta_Y(P'')\}.\end{array}\] 
         For all $(\rt,P'')$ in the first class, and for all $Z \subseteq A$ such that $\deadend{P'}{Z}$, \hyperlink{recall}{recalling that $\langle x_R |\equa\rangle \bisim R$ for all $R$}, and applying Lemma~\ref{lem:theta class} twice, we obtain
         $$\theta_Z(\langle x_{[P'']}|\equa\rangle) \bisim \theta_Z([P'']) \bisimtbr[] [\theta_Z(P'')] = [\theta_Z(P')] \bisimtbr \theta_Z([P']) \bisim \theta_Z(\langle x_{[P']}|\equa\rangle).$$
         Moreover, $\init{[P']}=\init{P'}$.
         Therefore the reactive approximation axiom yields
        \[ 
        \sum_{\{(\alpha,R) \mid [P']\step{\alpha}R \wedge \alpha \in A_\tau\}}\hspace{-10pt}\alpha.\langle x_R\mid\equa\rangle + \rt.\langle x_{[P'']}\mid\equa\rangle = \sum_{\{(\alpha,R) \mid [P']\step{\alpha}R \wedge \alpha \in A_\tau\}}\hspace{-10pt}\alpha.\langle x_R\mid\equa\rangle + \rt.\langle x_{[P']}\mid\equa\rangle
        \]
        According to Lemma \ref{lem:transition class}.5, for all $(\rt,P'')$ in the second class, and for all $Z \subseteq A$ such that $\deadend{P'}{Z}$, there exists a transition $[P'] \step{\rt} R'$ such that $\theta_Z(P'') \in [\theta_Z(\chi(R'))]$. Thus, applying Lemma~\ref{lem:theta class} twice, $\theta_Z([P'']) \bisimtbr[] [\theta_Z(P'')] = [\theta_Z(\chi(R'))] \bisimtbr \theta_Z(R')$ and so, according to Lemma \ref{lem:simplication}, $\theta_Z([P'']) \bisimtb \theta_Z(R')$. \hyperlink{recall}{Recalling that $\langle x_R |\equa\rangle \bisim R$ for all $R$}, we obtain $\theta_Z(\langle x_{[P'']}|\equa\rangle) \bisimtb \theta_Z(\langle x_{R'}|\equa\rangle)$, and thus $\rt.\theta_Z(\langle x_{[P'']}|\equa\rangle) \bisimrtb \rt.\theta_Z(\langle x_{R'}|\equa\rangle)$. Therefore, thanks to Theorem \ref{thm:completeness}, using that $Ax^\infty$ can be derived from $Ax^\infty_r$, $Ax^\infty_r \vdash \rt.\theta_Z(\langle x_{[P'']}\mid\equa\rangle) = \rt.\theta_Z(\langle x_{R'}\mid\equa\rangle)$. As a result, using RDP and the reactive approximation axiom,
        \begin{align*}
            Ax^\infty_r \vdash \!\!\sum_{\{(\rt,P'') \mid P'\step{\rt}P'' \wedge \forall X \subseteq A,\; \deadend{P'}{X} ~\Rightarrow~ \theta_X(P') \,\not\!\scriptrbis{}{br}\theta_X(P'')\}}\hspace{-80pt}\rt.\langle x_{[P'']}\mid\equa\rangle) + \langle x_{[P']}\mid\equa\rangle = \langle x_{[P']}\mid\equa\rangle
        \end{align*}
        Therefore, for all $\alpha \in Act$,
        \begin{align*}
            Ax^\infty_r \vdash \alpha.\nu(y_{P'})& = \alpha.(\sum_{\{(\alpha,P'') \mid P'\step{\alpha}P''\}}\alpha.\langle x_{[P'']}\mid\equa\rangle) \\
            & = \alpha.(\sum_{\{(\alpha,R) \mid [P']\step{\alpha}R \wedge \alpha \in A_\tau\}}\alpha.\langle x_R\mid\equa\rangle + \rt.\langle x_{[P']}\mid\equa\rangle \\
            & + \sum_{\{(\rt,P'') \mid P'\step{\rt}P'' \wedge \forall X \subseteq A,\; \deadend{P'}{X} ~\Rightarrow~ \theta_X(P') \,\not\!\scriptrbis{}{br} \theta_X(P'')\}}\hspace{-60pt}\rt.\langle x_{[P'']}\mid\equa\rangle) \\
            & = \alpha.\langle x_{[P']}\mid\equa\rangle \qquad\qquad\qquad\qquad\mbox{(t-branching axiom and RDP)}
        \end{align*}
        \item If $\forall P \steptau P',\, P \,\not\!\bisimtbr P'$ and $\forall X \subseteq A,\; \deadend{P}{X} \Rightarrow \forall P \step{\rt} P',\; \theta_X(P) \mathrel{\,\not\!\bisimtbr} \theta_X(P')$ then, for all $\alpha \in A_\tau$, $P \step{\alpha} P' \wedge P' \in R' \iff [P] \step{\alpha} R'$ and $\init{P} = \init{[P]}$. Moreover, if $\deadend{P}{X}$ and $[P] \step{\rt} R'$ then there exists a transition $P \step{\rt} P'$ with $\theta_X(P') \in [\theta_X(\chi(R'))]$. Thus $\theta_X([P']) \bisimtbr \theta_X(R')$ so $\theta_X([P']) \bisimtb \theta_X(R')$ by Lemma~\ref{lem:simplication}, and therefore $Ax^\infty_r \vdash \rt.\theta_X([P']) = \rt.\theta_X(R')$. Conversely, if $\deadend{P}{X}$ and $P \step{\rt} P'$ then there exists a transition $[P] \step{\rt} R'$ such that $\theta_X(P') \in [\theta_X(\chi(R'))]$ and thus $Ax^\infty_r \vdash  \rt.\theta_X([P']) =  \rt.\theta_X(R')$. Using the reactive approximation axiom, $Ax^\infty_r \vdash \nu(y_{P'}) = \langle x_{[P']} \mid\equa\rangle$ and so, for all $\alpha \in Act$, $Ax^\infty_r \vdash \alpha.\nu(y_{P'}) = \alpha.\langle x_{[P']} \mid\equa\rangle$.
    \end{itemize}
    As a result, for all $P' \in \reach{P}$, $ Ax^\infty_r \vdash \nu(y_{P'}) = \sum_{\{(\alpha,P'') \mid P' \step{\alpha} P''\}}\alpha.\langle x_{[P'']}|\equa\rangle = \sum_{\{(\alpha,P'') \mid P' \step{\alpha} P''\}}\alpha.\nu(y_{P''}) = \equa_{y_{P'}}[\nu]$, so $\nu$ is a solution of $\equa'$ up to $\bisimtbr$\,. Moreover, $\nu(y_P) = \sum_{\{(\alpha,P') \mid P \step{\alpha} P'\}}\alpha.\langle x_{[P']}|\equa\rangle$ which can be equated to $\langle x_P | \equa\rangle$ by a single application of RDP.
\end{proof}

\begin{proof}[Proof of Theorem \ref{thm:canonical}]
    According to Proposition \ref{prop:canonical}, it suffices to establish that $Ax^\infty_r \vdash \langle x_P|\equa\rangle = \langle x_Q|\equa\rangle$. By applying RDP, this amounts to proving that
    \begin{align*}
        Ax^\infty_r \vdash \sum_{\{(\alpha,P') \mid P\step{\alpha}P'\}}\alpha.\langle x_{[P']}\mid\equa\rangle = \sum_{\{(\alpha,Q') \mid Q\step{\alpha}Q'\}}\alpha.\langle x_{[Q']}\mid\equa\rangle
    \end{align*}
    Let $(\alpha,P')$ such that $P \step{\alpha} P'$ and $\alpha \in A_\tau$. Since $P \bisimrtbr Q$, there exists a transition $Q \step{\alpha} Q'$ such that $P' \bisimtbr Q'$. Thus, $[P'] = [Q']$ and so $\langle x_{[P']}\mid\equa\rangle = \langle x_{[Q']}\mid\equa\rangle$. The same observation can be made for all $(\alpha,Q')$ such that $Q \step{\alpha} Q'$ and $\alpha \in A_\tau$. As a result, $\init{P} = \init{Q}$ and
    \begin{align*}
        Ax^\infty_r \vdash \sum_{\{(\alpha,P') \mid P\step{\alpha}P' \wedge \alpha \in A_\tau\}}\alpha.\langle x_{[P']}\mid\equa\rangle = \sum_{\{(\alpha,Q') \mid Q\step{\alpha}Q' \wedge \alpha \in A_\tau\}}\alpha.\langle x_{[Q']}\mid\equa\rangle
    \end{align*}
    Let $(\rt,P')$ be such that $P \step{\rt} P'$. Since $P \bisimrtbr Q$, for all $X \subseteq A$ such that $\deadend{P}{X}$, there exists a transition $Q \step{\rt} Q'$ such that $\theta_X(P') \bisimtbr \theta_X(Q')$. Thus, $\theta_X([P']) \bisimtbr \theta_X([Q'])$ by Proposition~\ref{prop:class}, and $\theta_X([P']) \bisimtb \theta_X([Q'])$ by Lemma \ref{lem:simplication}. \hyperlink{recall}{Recalling that $\langle x_R |\equa\rangle \bisim R$ for all $R$}, $\theta_X(\langle x_{[P']}\mid\equa\rangle)\linebreak[2] \bisimtb \theta_X(\langle x_{[Q']}\mid\equa\rangle)$ and hence $\rt.\theta_X(\langle x_{[P']}\mid\equa\rangle) \bisimrtb \rt.\theta_X(\langle x_{[Q']}\mid\equa\rangle)$. Since $Ax^\infty$ can be derived from $Ax^\infty_r$, according to Theorem \ref{thm:completeness}, $Ax^\infty_r \vdash \rt.\theta_X(\langle x_{[P']}\mid\equa\rangle) = \rt.\theta_X(\langle x_{[Q']}\mid\equa\rangle)$. The same observation can be made for all $(\rt,Q')$ such that $Q \step{\rt} Q'$. Let $X \subseteq A$. If $P \step{\alpha}$ with $\alpha \in X \cup\{\tau\}$ then 
    \begin{align*}
        Ax^\infty_r \vdash \psi_X(\langle x_{P}|\equa\rangle) & = \sum_{\{(\alpha,P') \mid P\step{\alpha}P' \wedge \alpha \in A_\tau\}}\alpha.\langle x_{[P']}\mid\equa\rangle \\
        & = \sum_{\{(\alpha,Q') \mid Q\step{\alpha}Q' \wedge \alpha \in A_\tau\}}\alpha.\langle x_{[Q']}\mid\equa\rangle \\
        & = \psi_X(\langle x_{Q}\mid\equa\rangle)
    \end{align*}
    Otherwise, $\deadend{P}{X}$ so $\deadend{Q}{X}$, thus,
    \begin{align*}
        Ax^\infty_r \vdash \psi_X(\langle x_{P}|\equa\rangle) & = \sum_{\{(\alpha,P') \mid P\step{\alpha}P' \wedge \alpha \in A_\tau\}}\!\!\alpha.\langle x_{[P']}\mid\equa\rangle + \sum_{\{(t,P') \mid P\step{\rt}P'\}}\rt.\theta_X(\langle x_{[P']}\mid\equa\rangle) \\
        & = \sum_{\{(\alpha,Q') \mid Q\step{\alpha}Q' \wedge \alpha \in A_\tau\}}\!\!\alpha.\langle x_{[Q']}\mid\equa\rangle + \sum_{\{(t,Q') \mid Q\step{\rt}Q'\}}\rt.\theta_X(\langle x_{[Q']}\mid\equa\rangle) \\
        & = \psi_X(\langle x_{Q}\mid\equa\rangle)
    \end{align*}
    Using the reactive approximation axiom, $Ax^\infty_r \vdash \langle x_{P}|\equa\rangle = \langle x_{Q}|\equa\rangle$.
\end{proof}